\documentclass{amsart}
\usepackage[latin1]{inputenc}
\usepackage{latexsym}
\usepackage{amsmath}
\usepackage{amssymb}
\usepackage{amsfonts}
\usepackage{mathrsfs}               
\usepackage{bbm}                    
\usepackage{bbold}                  
\usepackage{textcomp}               
\usepackage{amstext}
\usepackage{enumerate}
\usepackage{units}
\usepackage{stmaryrd}
\usepackage{multicol}
 \newtheorem{thm}{Theorem}[section]
 \newtheorem{cor}[thm]{Corollary}
 \newtheorem{lem}[thm]{Lemma}
 \newtheorem{prop}[thm]{Proposition}
 \theoremstyle{definition}
 \newtheorem{defn}[thm]{Definition}
 \theoremstyle{remark}
 \newtheorem{rem}[thm]{Remark}
 \newtheorem{ex}[thm]{Example}
 \newtheorem{hyp}[thm]{Hypothesis}
 \newtheorem{notation}[thm]{Notation}
 \numberwithin{equation}{section}
\newcommand{\CC}{\mathbb{C}}
\newcommand{\EE}{\mathbb{E}}

\newcommand{\NN}{\mathbb{N}}
\newcommand{\PP}{\mathbb{P}}

\newcommand{\RR}{\mathbb{R}}

\newcommand{\BB}{\mathbb{B}}
\newcommand{\supp}{\mathrm{supp}}
\newcommand{\dist}{\mathrm{dist}}
\newcommand{\Ran}{\mathrm{Ran}}

\newcommand{\loc}{\mathrm{loc}}
\newcommand{\fin}{\mathrm{fin}}

\newcommand{\Div}{{\mathrm{div}}}

\newcommand{\Id}{\mathrm{d}}
\newcommand{\Ad}{\mathrm{ad}}
\newcommand{\SPn}[2]{\langle #1|#2\rangle} 
\newcommand{\SPb}[2]{\big\langle #1\big|#2\big\rangle} 
\newcommand{\SPB}[2]{\Big\langle \,#1\,\Big|\,#2\, \Big\rangle}

\newcommand{\ol}[1]{\overline{#1}} 
\newcommand{\ul}[1]{\underline{#1}} 
\newcommand{\mr}[1]{\mathring{#1}}
\newcommand{\wh}[1]{\widehat{#1}}
\newcommand{\wt}[1]{\widetilde{#1}}
\newcommand{\nf}[2]{\nicefrac{#1}{#2}}
\newcommand{\eh}{{\nf{1}{2}}}
\newcommand{\mh}{{-\nf{1}{2}}}
\newcommand{\cN}{\mathcal{N}}
\newcommand{\cO}{\mathcal{O}} 
\newcommand{\cC}{\mathcal{C}}
\newcommand{\cD}{\mathcal{D}} 
\newcommand{\cE}{\mathcal{E}}\newcommand{\cQ}{\mathcal{Q}}

\newcommand{\cT}{\mathcal{T}}

\newcommand{\cK}{\mathcal{K}}
 
\newcommand{\cM}{\mathcal{M}}       
 
\newcommand{\sA}{\mathscr{A}}
\newcommand{\sB}{\mathscr{B}} 
\newcommand{\sC}{\mathscr{C}}
\newcommand{\sD}{\mathscr{D}}\newcommand{\sP}{\mathscr{P}} 

\newcommand{\sF}{\mathscr{F}}\newcommand{\sR}{\mathscr{R}}
\newcommand{\sS}{\mathscr{S}}

\newcommand{\sI}{\mathscr{I}}\newcommand{\sU}{\mathscr{U}}
\newcommand{\sJ}{\mathscr{J}}
\newcommand{\sK}{\mathscr{K}}
\newcommand{\sL}{\mathscr{L}}\newcommand{\sX}{\mathscr{X}}         
\newcommand{\sM}{\mathscr{M}}\newcommand{\sY}{\mathscr{Y}}       
 
\newcommand{\fA}{\mathfrak{A}}
\newcommand{\fB}{\mathfrak{B}} 
\newcommand{\fC}{\mathfrak{C}}

\newcommand{\fF}{\mathfrak{F}}

\newcommand{\fH}{\mathfrak{H}}

\newcommand{\fa}{\mathfrak{a}}
\newcommand{\fb}{\mathfrak{b}}
\newcommand{\fc}{\mathfrak{c}}
\newcommand{\fd}{\mathfrak{d}}
\newcommand{\V}[1]{\boldsymbol{#1}}
\newcommand{\vsigma}{\boldsymbol{\sigma}}

\newcommand{\veps}{\boldsymbol{\varepsilon}}

\newcommand{\vxi}{\boldsymbol{\xi}}
\newcommand{\vgamma}{\boldsymbol{\gamma}}
\newcommand{\ve}{\varepsilon}
\newcommand{\vp}{\varphi}
\newcommand{\vo}{\varpi}
\newcommand{\vk}{\varkappa}
\newcommand{\vr}{\varrho}
\newcommand{\vt}{\vartheta}

\newcommand{\id}{\mathbbm{1}}                   
\newcommand{\dom}{\cD}                          
\newcommand{\fdom}{\cQ}                         
\newcommand{\HR}{\mathscr{H}}                   
\newcommand{\FHR}{\hat{\mathscr{H}}}        
\newcommand{\HP}{\mathfrak{h}}                  
\newcommand{\LO}{\mathscr{B}}                   
\newcommand{\ad}{a^\dagger}                     
\newcommand{\const}{\mathfrak{c}}     

\newcommand{\WW}[2]{\mathbb{W}_{#1}^{#2}}

\renewcommand{\Re}{\mathrm{Re}}
\renewcommand{\Im}{\mathrm{Im}}
\renewcommand{\le}{\leqslant}        
\renewcommand{\ge}{\geqslant}  
\begin{document}

\title[The semi-group in NRQED] {Continuity properties of the semi-group and its
integral kernel in non-relativistic QED}

\author[O.~Matte]{O.~Matte}

\address{Oliver Matte, Institut for Matematik, {\AA}rhus Universitet,
Ny Munkegade 118, DK-8000 Aarhus C, Denmark}
\email{matte@math.au.dk}
%

\begin{abstract}
Employing recent results on stochastic differential equations associated
with the standard model of non-relativistic quantum electrodynamics by B.~G\"{u}neysu,
J.S.~M{\o}ller, and the present author, 
we study the continuity of the corresponding semi-group between weighted 
vector-valued $L^p$-spaces, continuity properties of elements in the range of the
semi-group, and the pointwise continuity of an operator-valued semi-group kernel. 
We further discuss the continuous dependence of the semi-group and its integral kernel
on model parameters. All these results are obtained for Kato decomposable electrostatic 
potentials and the actual assumptions on the model are general enough to cover the Nelson 
model as well. As a corollary we obtain some new pointwise 
exponential decay and continuity results on elements of low-energetic spectral subspaces of atoms 
or molecules that also take spin into account. In a simpler situation where spin is 
neglected we explain how to verify the joint continuity of positive ground state eigenvectors with 
respect to spatial coordinates and model parameters. There are no smallness assumptions
imposed on any model parameter.
\end{abstract}
\maketitle
\setcounter{tocdepth}{1}
\tableofcontents
%

\section{Introduction}\label{sec-intro}

\noindent
In this article we extend some well-known results on Schr\"odinger semi-groups to the case where 
the quantum mechanical matter particles modeled by the non-relativistic Schr\"odinger operator are 
coupled to relativistic quantized radiation fields. The prime example for such a situation is the 
standard model of non-relativistic quantum electrodynamics (QED), describing the interaction of a 
fixed number of non-relativistic electrons with the second  quantized photon field; here a quantized 
vector potential is introduced in the Hamiltonian via minimal coupling. Another example is Nelson's 
model for the interaction of nucleons with a quantized, linearly coupled meson field.

In these cases the corresponding  Feynman-Kac (FK) formula for the semi-group is 
given by a vector-valued expectation, where the target Hilbert space, call it $\FHR$, is given by the 
tensor product of a finite-dimensional space accounting for spin degrees of freedom (if any) with the 
infinite-dimensional state space of the quantized radiation field (bosonic Fock space). In comparison
to the usual Schr\"odinger operator, the FK
integrand involves an additional process, attaining values in the set of bounded operators on 
$\FHR$, which is defined by means of certain $\FHR$-valued semi-martingales solving a stochastic 
differential equation (SDE) associated with the model. While FK representations
have been available in non-relativistic QED for quite some time 
\cite{Hiroshima1997,HiroshimaLorinczi2008,LHB2011},
the latter SDE has been derived and investigated only recently in our 
earlier work together with B.~G\"uneysu and J.S.~M{\o}ller \cite{GMM2014}. 
(Our article \cite{GMM2014} also provides a new version of the FK formula {\em with spin}, which 
will be employed in our examples; compare Rem.~\ref{rem-FK} below.)

The additional operator-valued process in
the FK integrand is -- roughly speaking -- a perturbation of the semi-group
associated with the radiation field energy operator, and it therefore has a
regularizing effect on the position coordinates of the bosons that constitute the quantized radiation 
field. Taking these new regularizing effects into account, in addition to the familiar regularity
properties of  Schr\"odinger semi-groups, poses a new mathematical problem. To discuss these 
effects we employ our SDE and derive some Burkholder-Davis-Gundy (BDG) type estimates 
involving unbounded weight functions of the radiation field energy. As a result we obtain, for 
instance, the following canonical regularity result: 
To start with it turns out that -- as usual -- 
elements, $\Psi$, in the range of the semi-group at strictly positive times are represented by 
continuous functions of the position of the matter particles (electrons, nucleons), in our case 
attaining values in $\FHR$. For a fixed position, $\V{x}\in\RR^\nu$, of the matter particles
we therefore obtain a well-defined element of $\Psi(\V{x})\in\FHR$. By means of the 
BDK type estimates we can show that the $n$-boson functions constituting 
$\Psi(\V{x})$ belong to Sobolev spaces of some order $\alpha\ge1$, provided that the coefficient 
functions in the SDE that couple the matter and radiation degrees of freedom belong to a Sobolev 
space (with respect to the bosonic position variable) of the same
order. If the coefficient functions depend continuously on $\V{x}$ as elements of that Sobolev 
space, then so does $\Psi(\V{x})$. If the coefficient functions are continuous functions of $\V{x}$ as 
elements of Sobolev spaces of arbitrary high orders, then the Sobolev embedding theorem 
implies that elements in the range of the semi-group are given by sequences (indexed by the boson 
number) of complex-valued functions that are {\em jointly continuous} in the position
variables of the matter particles {\em and} the bosons, together with all derivatives with respect to
any boson positions. All this holds true for Kato decomposable electrostatic potentials. 

To show that the regularity of $\Psi(\V{x})$ is not worse than the one of the
coefficients, sufficiently good bounds on multiple commutators of creation and annihilation operators
with functions of the radiation field energy are required in the derivation of the BDG type 
estimates. Obtaining such commutator bounds {\em with sufficiently good right hand sides}
turns out to be non-trivial. Hence we study them systematically in the appendix, which might also be
useful elsewhere. 

Of course, the analysis of the usual Schr\"odinger semi-groups  by means of FK formulas 
is by now a well-known subject, which has been extensively studied by many authors in the past 
decades. One standard reference, most of the time treating Kato decomposable electrostatic 
potentials and neglecting magnetic fields, is \cite{Simon1982}; see, e.g., also Ch.~I.1 of 
\cite{Sznitman1998} for a short introduction. In the presence of singular, classical magnetic 
fields, various continuity properties of the semi-group and its integral kernel are studied in 
\cite{BHL2000}. Parts of the analysis in \cite{BHL2000} are pushed forward to a matrix-valued 
case in \cite{Gueneysu2011}. The present article is mainly motivated by the work
quoted so far; for more information on the literature on Schr\"odinger semi-groups we refer to the 
remarks and reference lists in \cite{BHL2000,Simon1982,Sznitman1998}.

In the following we shall describe some more results obtained in this article and explain its 
organization.

In Sect.~\ref{sec-definitions} we fix our notation and standing assumptions and survey some earlier 
results. In Sect.~\ref{sec-SG-CK} we verify that our Feynman-Kac operators define a self-adjoint
semi-group between $\FHR$-valued $L^p$-spaces; we also verify the Chapman-Kolmogoroff 
equations for an operator-valued integral kernel of the semi-group.
The BDG type estimates mentioned above are established in Sect.~\ref{sec-weights}
employing the commutator bounds of App.~\ref{app-comm}. The results of Sect.~\ref{sec-weights}
will also play an important role in a forthcoming analysis of a Bismut-Elworthy-Li type formula in 
non-relativistic QED \cite{Matte2015}.

After that we start analyzing continuity properties of the semi-group:
Sect.~\ref{sec-cont-V} is devoted to weighted operator norm bounds on the semi-group between 
$\FHR$-valued $L^p$-spaces. At the same time, we shall study the continuous dependence
of the semi-group on the electrostatic potential, which is always assumed to be Kato decomposable.
The continuous dependence in weighted $L^p$-to-$L^q$-norms of the semi-group on the choice of
the coupling functions, that determine the interaction between the matter particles and the radiation
field, is discussed in Sect.~\ref{sec-coup}.
After that we study the strong continuity with respect to the time parameter of the semi-group
in Sect.~\ref{sec-SG}. In Sect.~\ref{sec-eq-cont} we show that, at strictly positive times, 
the semi-group maps bounded sets in $L^p$ into equicontinuous sets of functions from 
$\RR^\nu$ to $\FHR$.  The results of Sects.~\ref{sec-cont-V}--\ref{sec-eq-cont} will imply the
regularity results discussed above; cf. Ex.~\ref{ex-QED-exp}. 
With our crucial BDG type estimates at hand, 
we may proceed along the lines of the usual Schr\"odinger semi-group theory
in the proofs in these sections.

As an application we mention a simple argument proving the pointwise exponential decay of the
partial Fock-space norms of elements of low-lying spectral subspaces of atomic or molecular
Hamiltonians in non-relativistic QED; see Ex.~\ref{ex-QED-exp}. 
This complements earlier pointwise decay results \cite{HiHi2010,LHB2011} in several 
aspects: it applies to several electrons, it takes spin into account, the bound on the 
exponential decay rate is the natural one given in terms of the ionization energy, and it is not
restricted to eigenvectors. 
In fact, corresponding $L^2$-exponential localization results already exist \cite{Griesemer2004}
and all we have to do is to apply our weighted operator norm bounds to go from $L^2$ to $L^\infty$.

The results of Sects.~\ref{sec-SG-CK}--\ref{sec-eq-cont} will all be necessary
in Sect.~\ref{sec-cont-ker} to prove the (joint) continuity on $(0,\infty)\times\RR^\nu\times\RR^\nu$ 
of the operator-valued integral kernel of the semi-group with respect to some weighted operator 
norm. The fact that we study the {\em operator norm} continuity of the kernel also causes some new
technical problems that do not show up in the usual Schr\"odinger semi-group theory.
We shall also discuss the dependence on model parameters of the semi-group kernel.

The continuity results on the semi-group kernel will be complemented in Sect.~\ref{sec-pos}.
There it is shown that, if no spin degrees of freedom are present, the integral kernel 
(at arbitrary fixed $(t,\V{x},\V{y})$, $t>0$)
is positivity improving with respect to a suitable positive cone in the Fock space. 
The well-known positivity improvement by the corresponding semi-group \cite{Hiroshima2000}
is an immediate corollary of this result. The relatively simple proof given in this section is based 
on a novel factorization of the FK integrand found in \cite{GMM2014} and traditional ideas 
associated with Perron-Frobenius type arguments in mathematical quantum field theory;
see, e.g., \cite{Faris1972,Simon1974}.

Sect.~\ref{sec-pos} also provides a link to the final Sect.~\ref{sec-cont-GS} where, again in the
absence of any spin degrees of freedom, the joint continuity of positive ground state eigenvectors
with respect to spatial coordinates and model parameters is discussed. In fact, what remains to
prove in this section is only the continuous dependence on the model parameters in the Hilbert
space norm. Since the ground state eigenvalues are typically embedded in the continuous spectrum
this is, however, still a non-trivial task. We shall show that a certain compactness argument used
to prove the existence of ground states in \cite{GLL2001} can be adapted in such a way that is
reveals their continuous dependence on model parameters. In fact, we shall present a simplified
version of the compactness argument that works for a larger class of coupling functions and does
not require the photon derivative bound used in \cite{GLL2001}. This last section is rather sketchy
at some points; we shall only work out the new observations that we would like to communicate.

 \medskip

{\em General notation.} $\dom(T)$ (resp. $\fdom(T)$) denotes the domain (resp. form domain)
of a suitable linear operator $T$. If $\sX$ and $\sY$ are normed vector spaces, then
$\LO(\sX,\sY)$ denotes the set of bounded operators from $\sX$ to $\sY$ and
$\LO(\sX):=\LO(\sX,\sX)$.
If $\V{v}=(v_1,\ldots,v_\ell)$ is a vector of elements $v_j$ of a fixed normed vector space, then
we abbreviate $\|\V{v}\|^2:=\|v_1\|^2+\dots+\|v_\ell\|^2$.
If $\eta$ is a measure on the $\sigma$-algebra $\fC$, then $\eta^{\otimes_n}$ is the corresponding
$n$-fold product measure on the $n$-fold product $\sigma$-algebra $\fC^{\otimes_n}$.
We set $a\wedge b:=\min\{a,b\}$ and $a\vee b:=\max\{a,b\}$,
for $a,b\in\RR$. If not specified otherwise, then the symbols $c_{a,b,\dots}\:,c_{a,b,\dots}'\:,\ldots\;$
denote positive constants that depend only on the quantities $a,b,\ldots$ (if any) and whose
values might change from one equation array to another.


\section{Definitions, assumptions, and earlier results}\label{sec-definitions}

\noindent
This preliminary section is split into six subsections. In the first one, we recall the definition of the 
bosonic Fock space, which is the state space of the quantized radiation field, as well as the 
definitions of certain operators acting in it. In Subsect.~\ref{ssec-models} we introduce the full 
Hilbert space and vector-valued $L^p$-spaces. The Hamiltonian determining the models covered by 
our results is introduced, for vanishing electrostatic potentials to start with, 
in Subsect.~\ref{ssec-free-Ham}.
In particular, a standing hypothesis on the terms in the Hamiltonian that couple the matter particles 
to the radiation field is introduced in that subsection.  In Subsect.~\ref{ssec-SDE} we fix our notation 
for probabilistic objects and recall some results of \cite{GMM2014} on SDE's associated with the 
models under consideration. In this article we shall only treat Kato decomposable electrostatic 
potentials. Their definition and some of their well-known properties are recalled in 
Subsect.~\ref{ssec-Kato}. 
Finally, in Subsect.~\ref{ssec-FK}, the main objects of our present study, certain FK operators, are 
defined and a FK formula of \cite{GMM2014} is recalled in the special case of Kato 
decomposable potentials. In this article, the FK formula of Thm.~\ref{thm-FK} below will actually 
only be used in Ex.~\ref{ex-QED-exp} and in Sect.~\ref{sec-cont-GS}, where applications to 
non-relativistic QED are discussed. 
(The definition of the standard model of non-relativistic QED is recalled in 
Ex.~\ref{ex-NRQED} and Ex.~\ref{ex-QED-V}.)
All other results are statements on the FK operators introduced in Def.~\ref{def-FKO} that rely on 
the analysis of our SDE in \cite{GMM2014}, but not on the FK formula.


\subsection{Bosonic Fock space}\label{ssec-Fock}

\noindent
Let $(\cM,\fA,\mu)$ be a $\sigma$-finite measure space. To ensure separability of the 
corresponding $L^2$-space,
\begin{align}\label{def-HP}
\HP&:=L^2(\cM,\fA,\mu),
\end{align}
we assume that $\fA$ is generated by some countable semi-ring $\fH$ such that 
$\mu\!\!\upharpoonright_{\fH}$ is $\sigma$-finite. The corresponding bosonic Fock space 
is denoted by
\begin{align*}
{\sF}:=\CC\oplus\bigoplus_{n=1}^\infty{\sF}^{(n)}\ni\psi
=(\psi^{(0)},\psi^{(1)},\ldots,\psi^{(n)},\ldots\:).
\end{align*}
Hence, $\sF^{(1)}:=\HP$ and ${\sF}^{(n)}\subset L^2(\cM^n,\fA^{\otimes_n},\mu^{\otimes_n})$ 
is the closed subspace of all functions $\psi^{(n)}$ which are 
symmetric under permutations of their arguments,
$$
\psi^{(n)}(k_1,\ldots,k_n)=\psi^{(n)}(k_{\pi(1)},\ldots,k_{\pi(n)}),\quad\text{$\mu^{\otimes_n}$-a.e.},
$$
for every permutation $\pi$ of $\{1,\ldots,n\}$. The symbols
$\ad(f)$ and $a(f)$ denote the usual creation and annihilation
operators of a boson $f\in\HP$. We recall the convenient notation
\begin{align}\label{def-a(k)}
(a(k)\psi)^{(n)}(k_1,\ldots,k_n)&=(n+1)^\eh\psi^{(n+1)}(k_1,\ldots,k_n,k),
\quad\text{$\mu^{\otimes_{n+1}}$-a.e.},
\end{align}
for every $n\in\NN$ and $\psi\in\sF$, and $(a(k)\psi)^{(0)}:=\psi^{(1)}(k)$. Then
\begin{align}\label{def-a(f)}
(a(f)\psi)^{(n)}:=\int\ol{f(k)}\,(a(k)\psi)^{(n)}\Id\mu(k),\quad n\in\NN_0,\;
\psi\in\dom(a(f)),
\end{align}
and $\ad(f):=a(f)^*$,
where $\dom(a(f))$ is the set of all $\psi\in\sF$ for which the right hand side of \eqref{def-a(f)}
defines an element of $\sF$.
Let $\fF_{\fin}$ denote the subspace of all elements $(\psi^{(n)})_{n=0}^\infty\in\sF$
such that $\psi^{(n)}\not=0$ holds for only finitely many $n$.
Then we have the following canonical commutation relations on (e.g.) $\sF_{\fin}$,
\begin{align}\label{CCR}
[a(f),a(g)]&=[\ad(f),\ad(g)]=0\,,\quad
[a(f),\ad(g)]=\SPn{f}{g}\,\id\,,\quad f,g\in{{\HP}}.
\end{align}
The field operators given by
\begin{equation}\label{def-vp-vo}
\vp(f):=\ad(f)+a(f),\qquad f\in\HP,
\end{equation}
are essentially self-adjoint, and we denote
their unique self-adjoint extensions again by $\vp(f)$.
Given a vector of boson wave functions, $\V{f}=(f_1,\ldots,f_\nu)$, we set
$\vp(\V{f}):=(\vp(f_1),\ldots,\vp(f_\nu))$
and we shall employ an analogous convention for the creation and annihilation operators.

If $\vk$ is a real-valued measurable function on $\cM$ and $n\in\NN$,
then $\Id\Gamma^{(n)}(\vk)$ denotes the maximal operator of multiplication with
the function $(k_1,\ldots,k_n)\mapsto\sum_{\ell=1}^n\vk(k_\ell)$ in $\sF^{(n)}$.
We also set $\Id\Gamma^{(0)}(\vk):=0$ and recall that the differential
second quantization of $\vk$ is the self-adjoint operator in $\sF$ given by the direct sum
$\Id\Gamma(\vk):=\oplus_{n=0}^\infty\Id\Gamma(\vk)$.
If $\vk>0$ a.e. and $f\in\dom(\vk^\mh)$, then the Cauchy-Schwarz inequality 
and \eqref{CCR} imply the basic relative bounds
\begin{align}\label{rb-a}
\|a(f)\psi\|&\le\|\vk^\mh f\|\,\|\Id\Gamma(\vk)^{\nf{1}{2}}\psi\|,
\\\label{rb-ad}
\|\ad(f)\psi\|^2&\le\|\vk^\mh f\|^2\|\Id\Gamma(\vk)^{\nf{1}{2}}\psi\|^2+\|f\|^2\|\psi\|^2,
\end{align}
for all $\psi\in\fdom(\Id\Gamma(\vk))$. As a direct consequence we obtain
\begin{align}\label{qfb-vp}
\Id\Gamma(\vk)+\vp(f)\ge-\|\vk^\mh f\|^2\quad\text{on}\;\,\fdom(\Id\Gamma(\vk)).
\end{align}
If $\psi\in\dom(\Id\Gamma(\vk))$, then we also have the bound
\begin{align}\label{rb-vp2}
\|\vp(f)^2\psi\|&\le6\|(1+\vk^{-1})^\eh f\|^2\|(1+\Id\Gamma(\omega))\psi\|.
\end{align}


\subsection{Full Hilbert space and weighted, vector-valued $\boldsymbol{L^p}$-spaces}
\label{ssec-models}

\noindent
Next, we add (generalized) spin degrees of freedom to our
model by tensoring the Fock space with ${\CC^L}$, for some $L\in\NN$. We call 
\begin{align}\label{def-fHR}
\FHR:=\CC^L\otimes\sF
\end{align}
the fiber Hilbert space. 
The full Hilbert space for non-relativistic quantum mechanical particles with
a total number of $L$ ``spin'' degrees of freedom and interacting with a
quantized bosonic field is then given by
\begin{align}\label{def-HR}
\HR:=L^2(\RR^\nu,\FHR)=\int_{\RR^\nu}^\oplus\FHR\,\Id\V{x},
\quad\text{for some}\;\nu\in\NN\setminus\{1\}.
\end{align}
In our mathematical analysis of semi-groups we shall, however, most of the time work in the more
general spaces $L^p(\RR^\nu,\FHR;e^{pF}\Id\V{x})$, $p\in[1,\infty]$, where the density $e^{pF}$ 
with respect to the Borel-Lebesgue measure $\Id\V{x}$ is given by some globally Lipschitz 
continuous function $F:\RR^\nu\to\RR$. The norm on 
$L^p(\RR^\nu,\FHR)=L^p(\RR^\nu,\FHR;\Id\V{x})$
is denoted by $\|\cdot\|_p$, for $p\in[1,\infty]$, so that $\|\cdot\|_2=\|\cdot\|_{\HR}$. 
The operator norm on $\LO\big(L^p(\RR^\nu,{\FHR}),L^q(\RR^\nu,{\FHR})\big)$
is denoted by $\|\cdot\|_{p,q}$. We shall use the same symbols for norms on scalar $L^p$-spaces 
as well, which should not cause any confusion. 

Of course, a linear operator 
$T:L^p(\RR^\nu,{\FHR};e^{pF}\Id\V{x})\to L^q(\RR^\nu,{\FHR};e^{qF}\Id\V{x})$
is bounded, if and only if $e^FTe^{-F}\in\LO\big(L^p(\RR^\nu,{\FHR}),L^q(\RR^\nu,{\FHR})\big)$, 
and in the affirmative case
\begin{align*}
\|T\|_{\LO(L^p(\RR^\nu,{\FHR};e^{pF}\Id\V{x}),L^q(\RR^\nu,{\FHR};e^{qF}\Id\V{x}))}
=\|e^FTe^{-F}\|_{p,q}.
\end{align*}


\subsection{Interaction terms and free Hamiltonian}\label{ssec-free-Ham}

\noindent
The intercation of the non-relativistic quantum mechanical particles with the radiation field
is determined by $S\in\NN$ Hermitian matrices $\sigma_1,\ldots,\sigma_S\in\LO({\CC^L})$
and two vectors of boson wave functions,
$\V{G}_{\V{x}}=(G_{1,\V{x}},\ldots,G_{\nu,\V{x}})\in\HP^\nu$ and
$\V{F}_{\V{x}}=(F_{1,\V{x}},\ldots,F_{S,\V{x}})\in\HP^S$, parametrized by $\V{x}\in\RR^\nu$. 
Most of the time we regard the matrices as operators on $\FHR$
by identifying $\sigma_j\equiv\sigma_j\otimes\id$. We shall write
$\vsigma\cdot\V{v}=\sigma_1v_1+\dots+\sigma_Sv_S$, if $\V{v}=(v_1,\ldots,v_S)$ is any
vector of numbers, functions, or suitable operators.
To state our precise standing assumptions on $\V{G}$ and $\V{F}$, we first
recall that a conjugation $C$ on a Hilbert space is an anti-linear isometry with $C^2=\id$.

\begin{hyp}\label{hyp-G}
$\omega:\cM\to[0,\infty)$ is $\mu$-a.e. strictly positive.
The map $\RR^\nu\times\cM\ni(\V{x},k)\mapsto (\V{G}_{\V{x}},\V{F}_{\V{x}})(k)\in\CC^{\nu+S}$ is
measurable, $\V{x}\mapsto\V{G}_{\V{x}}$ belongs to $C^2(\RR^\nu,\HP^{\nu})$,
and $\V{x}\mapsto \V{F}_{\V{x}}\in\HP^S$ is globally Lipschitz continuous on $\RR^\nu$.
All components of $\V{G}_{\V{x}}$, $\V{F}_{\V{x}}$, and  $\partial_{x_\ell}\V{G}_{\V{x}}$ belong to 
$$
\mathfrak{k}:=L^2\big(\cM,\fA,(\omega^{-1}+\omega^{2})\mu\big).
$$
and the map
$$
\RR^\nu\ni\V{x}\longmapsto
(\V{G}_{\V{x}},\partial_{x_1}\V{G}_{\V{x}},
\ldots,\partial_{x_\nu}\V{G}_{\V{x}},\V{F}_{\V{x}})\in\mathfrak{k}^{\nu(\nu+1)+S}
$$
is bounded and continuous. Furthermore, there is a conjugation $C\colon\HP\to\HP$,
such that, for all $t>0$, $\V{x}\in\RR^\nu$, $\ell=\{1,\ldots,\nu\}$, and $j\in\{1,\ldots,S\}$, 
\begin{align*}
[C,e^{-t\omega}]=0,\quad
CG_{\ell,\V{x}}=G_{\ell,\V{x}},\quad CF_{j,\V{x}}=F_{j,\V{x}}.
\end{align*}
\end{hyp}
{\em The previous hypothesis will be tacitly assumed throughout the whole paper.}
Explicitly, we shall only mention certain additional assumptions on $\V{G}$ and $\V{F}$,
which are occasionally used to derive some of our results. To this end it is convenient to introduce
a {\em coefficient vector} $\V{c}$ by setting
\begin{equation}\label{def-coeff}
q_{\V{x}}:=\Div_{\V{x}}\V{G}_{\V{x}},\quad
\V{c}_{\V{x}}:=(\V{G}_{\V{x}},q_{\V{x}},\vsigma\cdot\V{F}_{\V{x}})\in\HP^{\nu+1+L^2},
\quad\V{x}\in\RR^\nu.
\end{equation}
Furthermore, we abbreviate
\begin{align}\label{def-lambda}
\mho(\V{x}):=\|M(\V{x})\|^2_{\sB(\CC^L)},\;\;\text{where 
$M(\V{x}):=\big(\|\omega^\mh(\vsigma\cdot\V{F}_{\V{x}})_{ij}\|\big)_{i,j=1}^L$.}
\end{align}
The interaction terms appearing in the free Hamiltonian defined below are now given by
\begin{align}\label{def-A}
\vp(\V{G})&:=\int_{\RR^\nu}^\oplus\id_{{\CC^L}}\otimes\vp(\V{G}_{\V{x}})\,\Id\V{x},
\quad\;
\vsigma\cdot\vp(\V{F}):=\sum_{j=1}^S
\int_{\RR^\nu}^\oplus\sigma_j\otimes\vp({F}_{j,\V{x}})\,\Id\V{x},
\\\label{def-q}
\vp(q)&:=\int_{\RR^\nu}^\oplus\id_{{\CC^L}}\otimes\vp(q_{\V{x}})\,\Id\V{x}.
\end{align}
These direct integrals of self-adjoint operators are well-defined, since 
$\RR^\nu\ni\V{x}\mapsto e^{i\vp(c_{j,\V{x}})}\psi\in\FHR$ is measurable,
for all $\psi\in\FHR$, if $c_j$ denotes any component of $\V{c}$.

\begin{ex}\label{ex-NRQED}
To cover the standard model of non-relativistic QED
for $N\in\NN$ electrons we choose $\nu=3N$ 
and write $\ul{\V{x}}=(\V{x}_1,\ldots,\V{x}_N)\in(\RR^3)^N$ instead of $\V{x}$.
Moreover, we choose $\cM=\RR^3\times\{0,1\}$,
equipped with the product of the Lebesgue and counting measures, and set
$\omega(\V{k},\lambda):=|\V{k}|$, $(\V{k},\lambda)\in\RR^3\times\{0,1\}$.
The coupling function $\V{G}$ is given by 
$$
\V{G}_{\ul{\V{x}}}^{\chi,N}(\V{k},\lambda):=
\big(\V{G}_{\V{x}_1}^{\chi}(\V{k},\lambda),\ldots,\V{G}_{\V{x}_N}^\chi(\V{k},\lambda)\big)
\in(\CC^3)^N,
$$
and
\begin{align*}
\V{G}_{\V{x}}^\chi(\V{k},\lambda)
:=(2\pi)^{-\nf{3}{2}}(2\omega(\V{k},\lambda))^\mh\chi(\V{k})
e^{-i\V{k}\cdot\V{x}}\veps(\V{k},\lambda),\quad\V{x}\in\RR^3,
\end{align*}
where the {\em ultra-violet cut-off function} $\chi:\RR^3\to[0,\infty)$ is always assumed to be even,
$\chi(\V{k})=\chi(-\V{k})$, $\V{k}\in\RR^3$, with $\omega^2\chi\in L^2(\RR^3\times\{0,1\})$.
In the above formulas we chose the Coulomb gauge, i.e., $q_{\ul{\V{x}}}=0$,
and by applying a suitable unitary transformation, if necessary, we may assume that the
{\em polarization vectors} are given as
$$
\veps(\V{k},0)=|\V{e}\times\V{k}|^{-1}\V{e}\times\V{k}, \quad 
\veps(\V{k},1)=|\V{k}|^{-1}\V{k}\times\veps(\V{k},0),\quad\text{a.e.}\;\V{k},
$$
for some $\V{e}\in\RR^3$ with $|\V{e}|=1$.
Several more explicit choices for $\chi$ are common in the literature.
Often a sharp ultra-violet cutoff is chosen, in which case $\chi$ is proportional to the characteristic 
function of some ball about the origin in $\RR^3$. 
Sometimes it is, however, favorable for technical reasons 
if $\chi$ is some Schwartz function, e.g., a Gaussian.

In many papers devoted to the mathematical analysis of non-relativistic QED, spin is neglected.
In this situation we set $L=1$ and $\V{F}=\V{0}$. To include spin, we choose $L=2^N$, so that
${\CC^L}=(\CC^2)^{\otimes_N}$, $S=3N$, and
$$
\sigma_{3\ell+j}:=\id_{\CC^2}^{\otimes_\ell}\otimes\sigma_j\otimes\id_{\CC^2}^{\otimes_{N-\ell-1}},
\qquad\ell=0,\ldots,N-1,\;j=1,2,3,
$$
with the $2$\texttimes$2$ Pauli spin-matrices $\sigma_1$, $\sigma_2$, and $\sigma_3$, as well as
$$
\V{F}_{\ul{\V{x}}}^{\chi,N}(\V{k},\lambda):=
\big(\V{F}^{\chi}_{\V{x}_1}(\V{k},\lambda),\ldots,\V{F}_{\V{x}_N}^\chi(\V{k},\lambda)\big)
\in(\CC^3)^N,\;\;\;\V{F}^\chi_{\V{x}}(\V{k},\lambda):=
-\tfrac{i}{2}\V{k}\times\V{G}_{\V{x}}^\chi(\V{k},\lambda).
$$
A suitable conjugation is given by $(Cf)(\V{k},\lambda):=(-1)^{1+\lambda}\ol{f(-\V{k},\lambda)}$,
a.e. $\V{k}\in\RR^3$ and $\lambda\in\{0,1\}$. Obviously, if $\chi$ is rapidly decaying (resp. a 
Gaussian), then the corresponding coefficient vector
$\V{c}^{\chi,N}_{\ul{\V{x}}}=(\V{G}_{\V{x}}^\chi,0,\V{F}_{\ul{\V{x}}}^{\chi,N})$ satisfies
\begin{equation}\label{hyp-aw-QED}
\sup_{\ul{\V{x}}}\|\omega^{\alpha}\V{c}^{\chi,N}_{\ul{\V{x}}}\|_{\HP}<\infty
\quad
\big(\text{resp.}\quad\sup_{\ul{\V{x}}}\|e^{\delta\omega}
\V{c}^{\chi,N}_{\ul{\V{x}}}\|_{\HP}<\infty\big)
\end{equation}
for all $\alpha>0$ (resp. $\delta>0$). When relevant, additional conditions of this type will  be 
imposed on the coefficient vector $\V{c}$ in the statements of our results.
\end{ex}

We are now in a position to introduce the Hamiltonian determining the {\em free} particle-radiation 
field system, i.e., we first consider the case of vanishing electrostatic potentials. For short, we
will denote operators of the type $\id_{\CC^L}\otimes A$ and
constant direct integrals of the type $\int_{\RR^\nu}\id_{\CC^L}\otimes A\Id\V{x}$
simply by $A$ in what follows. This should not cause any confusion if the underlying Hilbert
space is specified.

\begin{defn}\label{free-Ham}
We define the {\em free Hamiltonian} acting in the Hilbert space $\HR$ by
\begin{align}\nonumber
H^{0}&:=\tfrac{1}{2}(-i\nabla_{\V{x}}-\vp(\V{G}))^2-\vsigma\cdot\vp(\V{F})+\Id\Gamma(\omega),
\quad\dom(H^0):=\dom(\Delta)\cap\dom(\Id\Gamma(\omega)).
\end{align}
For every fixed $\V{x}\in\RR^\nu$, we further define an operator acting in $\FHR$ by
\begin{align*}
\wh{H}(\V{x})&:=\tfrac{1}{2}\vp(\V{G}_{\V{x}})^2
-\tfrac{i}{2}\vp(q_{\V{x}})-\vsigma\cdot\vp(\V{F}_{\V{x}})+\Id\Gamma(\omega),
\quad \dom(\wh{H}(\V{x})):=\dom(\Id\Gamma(\omega)).
\end{align*}
\end{defn}

The $\FHR$-valued second order Sobolev space $\dom(\Delta)$ and the nabla operators acting
on $\FHR$-valued Sobolev functions appearing in Def.~\ref{free-Ham} are defined by means of
a $\FHR$-valued Fourier transform. These objects are introduced in complete
analogy to the scalar case upon replacing the Lebesgue integral by a Bochner-Lebesgue integral
in the formula for the Fourier transform of an element of $L^1\cap L^2$.

It is known \cite{HaslerHerbst2008,Hiroshima2000esa,Hiroshima2002} 
that $H^0$ is well-defined and self-adjoint. In view of \eqref{qfb-vp}, $H^0$ is bounded from below.
It is essentially well-known and not difficult to prove (see, e.g., \cite[App.~A]{GMM2014})
that $\wh{H}(\V{x})$ is closed, and self-adjoint if $q_{\V{x}}=0$.
On account of \eqref{rb-a}, \eqref{rb-ad}, \eqref{rb-vp2}, and Hyp.~\ref{hyp-G},
there is a universal constant $c>0$ such that
\begin{align}\label{rb-H(x)}
\sup_{\V{x}\in\RR^\nu}\|\wh{H}(\V{x})(1+\Id\Gamma(\omega))^{-1}\|
\le c\big(1+\sup_{\V{x}\in\RR^\nu}\|\V{c}_{\V{x}}\|_{\HP}^2\big)<\infty.
\end{align}


\subsection{Associated SDE's}\label{ssec-SDE}

Before we continue with our discussion of the Hamiltonians in the presence of exterior potentials
we shall introduce some probabilistic objects in the present subsection. In particular, we shall recall
some results of \cite{GMM2014} on a certain SDE involving $\wh{H}(\V{x})$.
For basic definitions and results from stochastic analysis we refer to 
\cite{daPrZa1992,HackenbrochThalmaier1994,Me1982,MePe1980}.

Let $I$ be a closed interval with $\inf I=0$
and $\BB=(\Omega,\fF,(\fF_t)_{t\in I},\PP)$ some stochastic basis (i.e. a filtered probability space)
satisfying the {\em usual assumptions} (i.e. it is complete and right continuous). We shall also
consider the time-shifted basis,
$$
I^\tau:=\{t\ge0:\,\tau+t\in I\},\quad\BB_\tau:=(\Omega,\fF,(\fF_{\tau+t})_{t\in I^\tau}),\PP),
\quad\tau\in[0,\sup I).
$$
If $\sK$ is a separable Hilbert space, then $\mathsf{S}_{I^\tau}(\sK)$ denotes the space of
{\em continuous} $\sK$-valued semi-martingales with respect to $\BB_\tau$ and
we set $\mathsf{S}_{I}(\sK):=\mathsf{S}_{I^0}(\sK)$.
The bold letter $\V{B}\in\mathsf{S}_I(\RR^\nu)$ denotes an arbitrary $\nu$-dimensional 
$\BB$-Brownian motion with covariance matrix $\id$. If $\tau\in[0,\sup I)$,
then ${}^\tau\!\V{X}^{\V{q}}\in\mathsf{S}_{I^\tau}(\RR^\nu)$ is always a solution of the It\={o} equation
\begin{equation}\label{eq-Xq}
\V{X}_t=\V{q}+\V{B}_{\tau+t}+\int_0^t\V{\beta}(\tau+s,\V{X}_s)\Id s,\quad t\in[0,\sup I^\tau),
\end{equation}
for some $\fF_\tau$-measurable $\V{q}:\Omega\to\RR^\nu$. 
We set $\V{X}^{\V{q}}:={}^0\!\V{X}^{\V{q}}$.
In fact, we shall only need to consider two
different choices of the time-dependent vector field $\V{\beta}:[0,\sup I)\times\RR^\nu\to\RR^\nu$ 
in the present article. Once and for all we thus agree upon the following convention:
\begin{enumerate}
\item[$\bullet$] either $I=[0,t]$ with $t>0$ and $\V{\beta}(s,\V{x})=\frac{\V{y}-\V{x}}{t-s}$
with $\V{y}\in\RR^\nu$, so that ${}^\tau\!\V{X}^{\V{q}}$ with $\tau\in[0,t)$ is a Brownian bridge from 
some $\fF_\tau$-measurable $\V{q}$ to $\V{y}$ in time $t-\tau$. In this case we denote
${}^\tau\!\V{X}^{\V{q}}$ also by ${}^\tau\V{b}^{t;\V{q},\V{y}}$ and 
$\V{X}^{\V{q}}$ by $\V{b}^{t;\V{q},\V{y}}$.
\item[$\bullet$] or $I=[0,\infty)$ and $\V{\beta}=\V{0}$, so that
${}^\tau\!\V{X}_\bullet^{\V{q}}={}^\tau\!\V{B}_\bullet^{\V{q}}:=\V{q}+\V{B}_{\tau+\bullet}$. 
\end{enumerate}
For every $s\in I$, we denote the process obtained by reversing $\V{X}^{\V{q}}$ at time $s$
by $(\sR_s\V{X}^{\V{q}})_{\tau}:=\V{X}^{\V{q}}_{s-\tau}$, $\tau\in[0,s]$. The associated
filtration $(\fF[\sR_s{\V{X}^{\V{q}}}]_\tau)_{\tau\in[0,s]}$ is the standard extension of the
filtration $(\fH_\tau)_{\tau\in[0,s]}$ given by
$\fH_\tau:=\sigma(\V{X}_{s-\tau}^{\V{q}};\V{B}_{r}-\V{B}_u:s-\tau\le r\le u\le s)$;
see \cite{HaussmannPardoux1986,Pardoux-LNM1204}.
Replacing $(\fF_t)_{t\in I}$ in $\BB$ by $(\fF[\sR_s{\V{X}^{\V{q}}}]_\tau)_{\tau\in[0,s]}$
we obtain a new stochastic basis, denoted by $\BB[\sR_s{\V{X}^{\V{q}}}]$,
again satisfying the usual assumptions.

\begin{ex}\label{ex-stoch}
(1) Let $\Omega_{\mathrm{W}}:=C([0,\infty),\RR^\nu)$ denote the Wiener space and 
$\fF^{\mathrm{W}}$ the completion of the Borel $\sigma$-algebra associated with the 
standard Polish topology on $\Omega_{\mathrm{W}}$ with respect to the
Wiener measure $\PP_{\mathrm{W}}$ on 
$\Omega_{\mathrm{W}}$. Let $(\fF_t^{\mathrm{W}})_{t\ge0}$ be the standard extension of the 
filtration $(\sigma(\mathrm{pr}_s:s\in[0,t]))_{t\ge0}$ generated by the evaluation maps
$\mathrm{pr}_t(\vgamma):=\vgamma(t)$, $\vgamma\in\Omega_{\mathrm{W}}$, $t\ge0$. Then
$(\Omega,\fF,(\fF_t)_{t\in I},\PP,\V{B}):=(\Omega_{\mathrm{W}},\fF^{\mathrm{W}},
(\fF^{\mathrm{W}}_t)_{t\in I},\PP_{\mathrm{W}},\mathrm{pr})$
is the standard example of a Brownian motion.
 
 \smallskip
 
 \noindent(2) Let $t>0$ and $\V{x}\in\RR^\nu$. Then there exists a Brownian motion 
 $\hat{\V{B}}$ with respect to the stochastic basis
${\BB}[\sR_t\V{B}^{\V{x}}]$ such that, up to indistinguishability,
$\sR_t\V{B}^{\V{x}}$ is the unique solution with initial condition 
$\V{q}=\V{B}_{t}^{\V{x}}$ of the stochastic differential equation
\begin{align}\label{bridge-rev}
\V{X}_\tau=\V{q}+\int_0^\tau\frac{\V{x}-\V{X}_u}{t-u}\,\Id u+\hat{\V{B}}_\tau,\quad\tau\in[0,t),
\quad\V{X}_t=\V{x}.
\end{align}
These assertions follow from \cite{HaussmannPardoux1986,Pardoux-LNM1204}. 
We see that $\sR_t\V{B}^{\V{x}}$ is a Brownian bridge
back to $\V{x}$, defined by means of a new stochastic basis and Brownian motion.
We will denote the solution of \eqref{bridge-rev} corresponding an arbitrary
$\fF[\sR_t\V{B}^{\V{x}}]_0$-measurable initial condition $\V{q}\in\RR^\nu$ 
by the symbol $\smash{\hat{\V{b}}}^{t;\V{q},\V{x}}$.
 
 \smallskip
 
 \noindent(3) Let $t>s>0$ and $\V{x},\V{y}\in\RR^\nu$. Then there exists a Brownian motion 
 $\bar{\V{B}}$ with respect to the stochastic basis
${\BB}[\sR_s\V{b}^{t;\V{x},\V{y}}]$ such that, up to indistinguishability,
$\sR_s\V{b}^{t;\V{x},\V{y}}$ is the unique solution with initial condition 
$\V{q}=\V{b}_{s}^{t;\V{x},\V{y}}$ of the stochastic differential equation
\begin{align}\label{bridge-rev2}
\V{b}_\tau=\V{q}+\int_0^\tau\frac{\V{x}-\V{b}_u}{s-u}\,\Id u+\bar{\V{B}}_\tau,\quad\tau\in[0,s),
\quad\V{b}_s=\V{x}.
\end{align}
These assertions follow again from \cite{HaussmannPardoux1986,Pardoux-LNM1204}.
We denote the solution of \eqref{bridge-rev2} with an arbitrary 
$\fF[\sR_s\V{b}^{t;\V{x},\V{y}}]_0$-measurable initial condition $\V{q}:\Omega\to\RR^\nu$ by 
$\bar{\V{b}}^{s;\V{q},\V{x}}$. Then $({\BB}[\sR_s\V{b}^{t;\V{x},\V{y}}],\bar{\V{b}}^{s;\V{q},\V{x}})$
is another valid choice for a basis and a driving process to which Thm.~\ref{thm-WW} below applies.
\end{ex}

In the following three theorems we quote some of the main results of \cite{GMM2014}.
In their statements we shall always consider $\dom(\Id\Gamma(\omega))$ as a Hilbert
space equipped with the graph norm of $\Id\Gamma(\omega)$, whence it might make sense
to recall the following:

\begin{rem}\label{rem-Gronwall}
Let $(\Sigma,\fC)$ be a measurable space and 
$T$ be a closed operator acting in the separable Hilbert space $\sK$. Consider its domain
$\dom(T)$ as a Hilbert space equipped with the graph norm. Then a map
map $\eta:\Sigma\to\dom(T)$ is $\fC$-$\fB(\sK)$-measurable, if and only if it is
$\fC$-$\fB(\dom(T))$-measurable. Therefore, no ambiguities appear when we say, for instance,
that $\eta:\Omega\to\dom(\Id\Gamma(\omega))$ be $\fF_0$-measurable.
\end{rem}

\begin{thm}[\cite{GMM2014}]\label{thm-WW}
Let $V:\RR^\nu\to\RR$ be bounded and continuous. Then,
for all $\tau\in[0,\sup I)$ and $\fF_\tau$-measurable $\V{q}:\Omega\to\RR^\nu$,
there exist operators $\WW{t}{V}[^\tau\!\V{X}^{\V{q}}](\vgamma)\in\LO(\FHR)$,
$t\in I^\tau$, $\vgamma\in\Omega$, such that the following holds:

\smallskip

\noindent
{\rm(1)} For every $t\in I^\tau$, the operator-valued map 
$\WW{t}{V}[^\tau\!\V{X}^{\V{q}}]:\Omega\to\LO(\FHR)$ is 
$\fF_{\tau+t}$-$\fB(\LO(\FHR))$-measurable and $\PP$-almost separably valued.

\smallskip

\noindent
{\rm(2)} Using the notation \eqref{def-lambda}, we have the following bound for all $t\in I^\tau$,
\begin{align}\label{bd-WW}
\ln\|\WW{t}{V}[^\tau\!\V{X}^{\V{q}}]\|
\le\int_0^t\big(\mho(^\tau\!\V{X}_s^{\V{q}})-V(^\tau\!\V{X}_s^{\V{q}})\big)\Id s
\quad\text{on $\Omega$.}
\end{align}
{\rm(3)} If $\eta:\Omega\to{\dom(\Id\Gamma(\omega))}$ is $\fF_\tau$-measurable, then
$\WW{}{V}[^\tau\!\V{X}^{\V{q}}]\eta\in {\sf S}_{{I}^\tau}(\FHR)$
and, up to indistinguishability, $\WW{}{V}[^\tau\!\V{X}^{\V{q}}]\eta$ is the unique
element of ${\sf S}_{{I}^\tau}(\FHR)$ whose paths belong $\PP$-a.s. to  
$C(I^\tau,\dom(\Id\Gamma(\omega)))$ and which $\PP$-a.s. solves
\begin{align}\label{SDE-spin}
Y_\bullet&=\eta+i\int_0^\bullet\vp(\V{G}_{{}^\tau\!\V{X}^{\V{q}}_s})\,Y_s\,\Id^\tau\!
\V{X}^{\V{q}}_s-\int_0^\bullet
\big(\wh{H}(^\tau\!\V{X}^{\V{q}}_s)+V(^\tau\!\V{X}_s^{\V{q}})\big)Y_s\,\Id s
\end{align}
on $[0,\sup I^\tau)$.

\smallskip

\noindent
{\rm(4)} For all $0\le\tau\le t\in I$, $\V{x}\in\RR^\nu$, $\psi\in\FHR$, and $\vgamma\in\Omega$, set
\begin{align}\label{def-flow}
\Lambda_{\tau,t}(\V{x},\psi,\vgamma)
:=(^\tau\!\V{X}_{t-\tau}^{\V{x}},\WW{t-\tau}{V}[^\tau\!\V{X}^{\V{x}}]\psi)(\vgamma)
\in\RR^\nu\times\FHR.
\end{align}
Then $(\Lambda_{\tau,t})_{0\le\tau\le t\in I}$ is a stochastic flow for the system of SDE
comprised of \eqref{eq-Xq} and \eqref{SDE-spin} and in particular we have,
for all $\tau\in I$ and $\fF_\tau$-measurable $(\V{q},\eta):\Omega\to\RR^\nu\times\FHR$,
\begin{align}\label{flow1}
\Lambda_{\tau,\tau+\bullet}(\V{q}(\vgamma),\eta(\vgamma),\vgamma)=
(^\tau\!\V{X}_{\bullet}^{\V{q}},\WW{\bullet}{V}[^\tau\!\V{X}^{\V{q}}]\psi)(\vgamma)
\>\>\text{on $I^\tau$, for $\PP$-a.e. $\vgamma$.}
\end{align}
{\rm(5)} Let $\sK$ be a separable Hilbert space. Set
\begin{align*}
(P_{\tau,t}f)(\V{x},\psi):=\int_\Omega f\big(\Lambda_{\tau,t}(\V{x},\psi,\vgamma)\big)\Id\PP(\vgamma),
\quad\V{x}\in\RR^\nu,\,\psi\in\FHR,
\end{align*}
for all $0\le\tau\le t\in I$ and every bounded Borel measurable function $f:\RR^\nu\times\FHR\to\sK$.
Then the transition operators $(P_{\tau,t})_{0\le\tau\le t\in I}$ satisfy the following Markov
property: If $0\le\sigma\le\tau\le t\in I$, if $f$ is a bounded real-valued Borel function
on $\RR^\nu\times\FHR$ or if $f:\RR^\nu\times\FHR\to\sK$ is bounded and continuous, 
and if $(\V{q},\eta):\Omega\to\RR^\nu\times\FHR$ is $\fF_\sigma$-measurable, then
it follows that, for $\PP$-a.e. $\vgamma$,
\begin{align}\label{Markov}
\big(\EE^{\fF_\tau}\big[f(\Lambda_{\sigma,t}[\V{q},\eta])\big]\big)(\vgamma)
=(P_{\tau,t}f)\big(\Lambda_{\sigma,\tau}(\V{q}(\vgamma),\eta(\vgamma),\vgamma)\big).
\end{align}
Here $\EE^{\fF_\tau}$ denotes conditional expectation with respect to $\PP$ given $\fF_\tau$
and $f(\Lambda_{\sigma,t}[\V{q},\eta])$ is the random variable given by
$\Omega\ni\vgamma\mapsto f(\Lambda_{\sigma,t}(\V{q}(\vgamma),\eta(\vgamma),\vgamma))$.
\end{thm}

\begin{thm}[\cite{GMM2014}]\label{thm-str-sol}
Assume that $V$ is continuous and bounded.
Let $(\Lambda_{\tau,t}^{\mathrm{W}})_{\tau\le t\in I}$ denote the stochastic
flow defined in Thm.~\ref{thm-WW}(3) for the special choices
$\BB=(\Omega_{\mathrm{W}},\fF^{\mathrm{W}},(\fF_t^{\mathrm{W}})_{t\in I},\PP_{\mathrm{W}})$ and
$\V{B}=\mathrm{pr}$; see Ex.~\ref{ex-stoch}(1). Then $(\Lambda_{0,t}^{\mathrm{W}})_{t\in I}$ is a
strong solution of \eqref{eq-Xq}\&\eqref{SDE-spin} in the sense that,
for any stochastic basis $(\Omega,\fF,(\fF_t)_{t\in I},\PP)$ and Brownian
motion $\V{B}$ as explained in the beginning of this subsection, and for any $\fF_0$-measurable
$(\V{q},\eta):\Omega\to\RR^\nu\times\FHR$, the (up to indistinguishability) unique solution of
\eqref{eq-Xq}\&\eqref{SDE-spin} with $\tau=0$ is given, for $\PP$-a.e. $\vgamma\in\Omega$, by \begin{equation}\label{eq-str-sol}
(\V{X}^{\V{q}}_t,\WW{t}{V}[\V{X}^{\V{q}}]\eta)(\vgamma)
=\Lambda_{0,t}^{\mathrm{W}}(\V{q}(\vgamma),\eta(\vgamma),\V{B}_\bullet(\vgamma))
\,,\quad t\in I.
\end{equation}
\end{thm}

\begin{ex}\label{ex-str-sol}
(1) For all $t>s>0$, $\V{y},\V{z}\in\RR^\nu$, and $\psi\in\FHR$, the random variables
$\WW{t-s}{V}[\V{b}^{t-s;\V{z},\V{y}}]\psi$ and $\WW{t-s}{V}[^s\V{b}^{t;\V{z},\V{y}}]\psi$ 
have the same distribution. In fact, both processes 
$(\V{b}^{t-s;\V{z},\V{y}},\WW{}{V}[\V{b}^{t-s;\V{z},\V{y}}]\psi)$
and $(^s\V{b}^{t;\V{z},\V{y}},\WW{}{V}[^s\V{b}^{t;\V{z},\V{y}}]\psi)$ 
solve the same system of equations \eqref{eq-Xq}\&\eqref{SDE-spin}. The only difference
is that the second one is constructed by means of the time shifted basis with
filtration $(\fF_{s+\tau})_{\tau\in[0,t-s]}$ and the time shifted Brownian motion 
$(\V{B}_{s+\tau})_{\tau\in[0,t-s]}$.
Since $(\V{B}_\tau)_{\tau\in[0,t-s]}$ and $(\V{B}_{s+\tau})_{\tau\in[0,t-s]}$ have the same
distribution and since the initial conditions $\V{q}=\V{z}$ and $\eta=\psi$ are constant, 
the claim follows from Thm.~\ref{thm-str-sol}.

\smallskip

\noindent(2) For all $s>0$, $\V{x},\V{z}\in\RR^\nu$, and $\psi\in\FHR$, the random variables
$\WW{s}{V}[\V{b}^{s;\V{z},\V{x}}]$,
$\WW{s}{V}[\bar{\V{b}}{}^{s;\V{z},\V{x}}]\psi$, and $\WW{s}{V}[\hat{\V{b}}{}^{s;\V{z},\V{x}}]\psi$ 
have the same distribution, where we use the notation introduced in Ex.~\ref{ex-stoch}(2)\&(3).
In fact, $\V{b}^{s;\V{z},\V{x}}$, $\bar{\V{b}}{}^{s;\V{z},\V{x}}$ and $\hat{\V{b}}{}^{s;\V{z},\V{x}}$ 
satisfy formally the same SDE; only the underlying bases and the driving Brownian motions are 
different. Hence, the claim follows again from Thm.~\ref{thm-str-sol}.
\end{ex}

\begin{thm}[\cite{GMM2014}]\label{thm-WW-rev}
Let $V:\RR^\nu\to\RR$ be bounded and continuous, $\V{x}\in\RR^\nu$, $s\in I$, 
and let $\WW{s}{V}[\sR_s{\V{X}^{\V{x}}}]$ be given by Thm.~\ref{thm-WW} applied with 
$\sR_s{\BB}$ as underlying stochastic basis. Then the following identity holds $\PP$-a.s.,
\begin{align*}
\WW{s}{V}[\V{X}^{\V{x}}]^*&=\WW{s}{V}[\sR_s{\V{X}^{\V{x}}}].
\end{align*}
\end{thm}

If the Brownian bridge is chosen as driving process, then we also know the following:

\begin{prop}[\cite{GMM2014}]\label{prop-ida}
Let $V:\RR^\nu\to\RR$ be bounded and continuous and $t>0$. 
Then, for all $\V{x},\V{y}\in\RR^\nu$, we can choose versions of
$\WW{}{V}[\V{b}^{t;\V{x},\V{y}}]$ such that 
$[0,t]\times\RR^{2\nu}\times\Omega\ni(s,\V{x},\V{y},\vgamma)\mapsto
\WW{s}{V}[\V{b}^{t;\V{x},\V{y}}](\vgamma)\in\LO(\FHR)$ is measurable with a separable image and
$(s,\V{x},\V{y})\mapsto\WW{s}{V}[\V{b}^{t;\V{x},\V{y}}](\vgamma)$
is continuous from $(0,t]\times\RR^{2\nu}$ into $\LO(\FHR)$ for all $\vgamma\in\Omega$.
\end{prop}


\subsection{Kato-decomposable potentials}\label{ssec-Kato}

For $\V{x}\not=\V{0}$ and $r>0$, we set
\begin{align*}
g_r(\V{x})&:=1_{|\V{x}|<r}\cdot
\left\{\begin{array}{ll}
-\ln|\V{x}|,&\textrm{if}\;\nu=2,\\
|\V{x}|^{2-\nu},&\text{if}\;\nu>2,
\end{array}
\right.
\end{align*}
and we recall that a {\em real-valued} measurable function $V:\RR^\nu\to\RR$ is in the Kato-class,
in symbols $V\in\cK(\RR^\nu)$, iff
\begin{equation*}
\lim_{r\downarrow0}\sup_{\V{x}\in\RR^\nu}\int_{\RR^\nu}g_r(\V{x}-\V{y})\,|V(\V{y})|\,\Id\V{y}=0.
\end{equation*}
It is said to be in the local Kato-class, in symbols $V\in\cK_\loc(\RR^\nu)$,
iff $1_KV\in\cK(\RR^\nu)$, for every compact $K\subset\RR^\nu$.
It is called Kato-decomposable, in symbols $V\in\cK_\pm(\RR^\nu)$,
if there exist $V_-\in\cK(\RR^\nu)$ and $V_+\in\cK_\loc(\RR^\nu)$ with
$V_\pm\ge0$ and $V=V_+-V_-$.

We refer to \cite{AizenmanSimon1982,BHL2000,Simon1982,Sznitman1998} for detailed
treatments of Kato decomposable potentials and examples; for instance,
$L^p_{\RR}(\RR^\nu)\subset\cK(\RR^\nu)$ and $L^p_{\RR,\loc}(\RR^\nu)\subset\cK_\loc(\RR^\nu)$, 
if $p>\nu/2$. In view of the subsequent discussion it makes sense to recall the following:

\begin{rem}\label{rem-Kato}
For all $V\in\cK_\loc(\RR^\nu)$ and $\V{x}\in\RR^\nu$, we $\PP$-a.s. have
$|V|(\V{B}_{\bullet}^{\V{x}})\in L^1_\loc([0,\infty))$; see, e.g., \cite[\textsection I.1.2]{Sznitman1998}.
Likewise, if $V\in\cK_\loc(\RR^\nu)$, $\V{x},\V{y}\in\RR^\nu$, and $t>0$, then we $\PP$-a.s. have
$|V|(\V{b}_{\bullet}^{t;\V{x},\V{y}})\in L^1([0,t])$; see, e.g., the proof of Lem.~C.8 in \cite{BHL2000}.
\end{rem}

Of particular importance for us are the standard facts collected in the following lemma, where
the Euclidean heat kernel is denoted by
$$
p_t(\V{x},\V{y}):=(2\pi t)^{-\nf{\nu}{2}}e^{-|\V{x}-\V{y}|^2/2t},\quad t>0,\;\V{x},\V{y}\in\RR^\nu.
$$

\begin{lem}\label{lem-dirk}
Let $V\in\cK_\pm(\RR^\nu)$ and $p>0$. Then the following holds:

\smallskip

\noindent
{\rm(1)} There exists $c>0$, depending only on $pV_-$, such that
for all $t>0$ and $\V{x},\V{y}\in\RR^\nu$, 
\begin{align}\label{dirk0}
\sup_{\V{z}\in\RR^\nu}\EE\big[e^{p\int_0^tV(\V{B}_s^{\V{z}})\Id s}\big]&\le2e^{ct},
\\\label{dirk1}
p_t(\V{x},\V{y})\,\EE\big[e^{p\int_0^tV(\V{b}_s^{t;\V{x},\V{y}})\Id s}\big]&\le ct^{-\nf{\nu}{2}}
e^{ct-(\V{x}-\V{y})^2/4t}.
\end{align}
{\rm(2)} For every compact $K\subset\RR^\nu$,
\begin{align}\label{dirk2a}
\lim_{s\downarrow0}\sup_{\V{y}\in K}\EE\big[|e^{-\int_0^sV(\V{B}_r^{\V{y}})\Id r}-1|^p\big]&=0.
\end{align}
In the case $V\in\cK(\RR^\nu)$, the compact set $K$ in \eqref{dirk2a} 
can be replaced by $\RR^\nu$.
\end{lem}

\begin{proof}
(1): See, e.g., \cite[\textsection I.1.2]{Sznitman1998}.
(2) is proved, e.g., in Lem.~C.3 and~C.5 of \cite{BHL2000}.
\end{proof}

Let $V\in\cK_\pm(\RR^\nu)$, $F:\RR^\nu\to\RR$ be globally Lipschitz continuous,
and $f\in L^p(\RR^\nu;e^F\Id\V{x})$ with $p\in[1,\infty]$. Then the expectations in 
\begin{equation}\label{def-StV}
(S_t^Vf)(\V{x}):=\EE\big[e^{-\int_0^tV(\V{B}_s^{\V{x}})}f(\V{B}_t^{\V{x}})\big],\quad
\V{x}\in\RR^\nu,\,t\ge0,
\end{equation}
are well-defined and, for all $\V{x}\in\RR^\nu$ and $t>0$, we further have the relation
\begin{align}\label{def-StV(x,y)1}
(S_t^Vf)(\V{x})=\int_{\RR^\nu}S_t^V(\V{x},\V{y})f(\V{y})\Id\V{y},
\end{align}
with
\begin{align}\label{def-StV(x,y)2}
S_t^V(\V{x},\V{y}):=p_t(\V{x},\V{y})\EE\big[e^{-\int_0^tV(\V{b}_s^{t;\V{y},\V{x}})\Id s}\big].
\end{align}
Here the map $(0,\infty)\times\RR^{2\nu}\ni(t,\V{x},\V{y})\mapsto S_t^V(\V{x},\V{y})$
is continuous and can be dominated by means of \eqref{dirk1}; 
see, e.g., \cite[\textsection I.1.3]{Sznitman1998}.
Using \eqref{dirk0}, one can in fact verify that $S_t^V$ with $t>0$ maps $L^p(\RR^\nu;e^F\Id\V{x})$ 
continuously into every $L^q(\RR^\nu;e^F\Id\V{x})$ with $q\in[p,\infty]$. Moreover,
\begin{equation}\label{bd-StV}
\sup_{\mathrm{Lip}(F)\le a}\sup_{\tau_1\le t\le\tau_2}\|e^FS_t^Ve^{-F}\|_{p,q}<\infty,
\quad0<\tau_1\le\tau_2<\infty,\,1\le p\le q\le\infty,
\end{equation}
where $\|\cdot\|_{p,q}$ is the norm on $\LO(L^p(\RR^\nu),L^q(\RR^\nu))$ and the first supremum
is taken over all Lipschitz continuous $F:\RR^\nu\to\RR$ with Lipschitz constant $\le a\in[0,\infty)$.
In the case $p=q$ the choice $\tau_1=0$ is allowed for in \eqref{bd-StV} as well;
compare, e.g., Lem.~C.1 in \cite{BHL2000} or the proof of Thm.~\ref{thm-LNCV}(1) below.

We shall also make use of the following approximation results:

\begin{lem}\label{lem-conny}
{\rm(1)} For every $V\in\cK_\pm(\RR^\nu)$,
there exist $V_n\in C_0^\infty(\RR^\nu,\RR)$, $n\in\NN$, such that, for every compact
$K\subset\RR^\nu$,
\begin{align}\label{approx1}
\lim_{n\to\infty}\sup_{\V{x}\in\RR^\nu}\int_Kg_1(\V{x}-\V{y})|V-V_n|(\V{y})\Id\V{y}&=0,
\end{align}
and, for some $A\ge1$ and every $r\in(0,1]$,
\begin{align}\label{approx2}
\sup_n\sup_{\V{x}\in\RR^\nu}\int_{\RR^\nu}g_r(\V{x}-\V{y})V_{n-}(\V{y})\Id\V{y}
&\le A\sup_{\V{x}\in\RR^\nu}\int_{\RR^\nu}g_r(\V{x}-\V{y})V_{-}(\V{y})\Id\V{y}.
\end{align}
{\rm(2)} Let $V_n\in\cK_\pm(\RR^\nu)$, $n\in\NN$. Then the bound
\begin{align}\label{approx2b}
\sup_n\sup_{\V{x}\in\RR^\nu}\int_{\RR^\nu}g_r(\V{x}-\V{y})V_{n-}(\V{y})\Id\V{y}<\infty,
\end{align}
implies
\begin{align}\label{approx3a}
\forall\,\tau,p>0\,:\quad
\sup_{n\in\NN}\sup_{t\in[0,\tau]}\sup_{\V{x}\in\RR^\nu}
\EE\big[e^{-p\int_0^tV_n(\V{B}^{\V{x}}_s)\Id s}\big]&<\infty.
\end{align}
{\rm (3)} Let $V_n\in\cK(\RR^\nu)$, $n\in\NN$. Then the the analytic condition
\begin{align}\label{approx3c}
\lim_{r\downarrow0}\sup_{n\in\NN}\sup_{\V{x}\in\RR^\nu}\int_{\RR^\nu}
g_r(\V{x}-\V{y})|V_n(\V{y})|\Id\V{y}&=0
\end{align}
is equivalent to the probabilistic condition
\begin{align}\label{approx3b}
\lim_{t\downarrow0}\sup_{n\in\NN}\sup_{\V{x}\in\RR^\nu}
\int_0^t\EE\big[|V_n(\V{B}_s^{\V{x}})|\big]\Id s&=0.
\end{align}
{\rm(4)} Let $V\in\cK_\pm(\RR^\nu)$. Then,
for every sequence $V_n\in\cK_\pm(\RR^\nu)$, $n\in\NN$, satisfying \eqref{approx1}
and \eqref{approx2} and every compact $K\subset\RR^\nu$,
\begin{align}\label{approx3}
\forall\;\tau,p>0:\quad\lim_{n\to\infty}\sup_{t\in[0,\tau]}\sup_{\V{x}\in K}
\EE\big[|e^{-\int_0^tV(\V{B}_s^{\V{x}})\Id s}-e^{-\int_0^tV_n(\V{B}_s^{\V{x}})\Id s}|^p\big]&=0.
\end{align}
If in addition $V-V_n\in\cK(\RR^\nu)$, $n\in\NN$, and \eqref{approx1} holds with $K$ replaced by 
$\RR^\nu$, then the compact set $K$ in \eqref{approx3} can be replaced by $\RR^\nu$.
\end{lem}

\begin{proof}
Detailed proofs of (1) (resp. (4)) can be found, e.g., in App.~A (resp. Lem.~C.4 and Lem.~C.6) of 
\cite{BHL2000}. In fact, Lem.C.4 of \cite{BHL2000} proves a slightly weaker version of the last 
statement of (4), where the condition $V-V_n\in\cK(\RR^\nu)$, $n\in\NN$, is replaced by 
$V\in\cK(\RR^\nu)$ and $V_n\in\cK(\RR^\nu)$, $n\in\NN$.
Combining this weaker statement with \eqref{dirk0} we see, however, that the full assertion
follows from H\"older's inequality. Part~(2) follows from a standard application of H\"older's
inequality and Khasminskii's lemma; compare the proof of Lem.~C.1 in \cite{BHL2000}.
\end{proof}

\begin{lem}\label{lem-Kato-qfb}
Let $V\in\cK(\RR^\nu)$. Then $V$ is infinitesimally form-bounded with respect to the Laplacian.
More precisely,
\begin{align}\label{qfb-Kato}
\int_{\RR^\nu}|V(\V{x})||f(\V{x})|^2\Id\V{x}
&\le\frac{c_\gamma(V)}{2}\int_{\RR^\nu}|\nabla f(\V{x})|^2\Id\V{x}+\gamma c_\gamma(V)\|f\|^2,
\end{align}
for all $\gamma>0$ and $f\in H^1(\RR^\nu)$, where
\begin{align*}
c_\gamma(V)&:=\sup_{\V{x}\in\RR^\nu}\int_0^\infty e^{-t\gamma}\EE\big[|V(\V{B}_t^{\V{x}})|\big]\Id t,
\end{align*}
which, for all $\tau>0$ and $\gamma\ge1$, can be estimated as follows,
\begin{align}\label{bd-cgammaV}
c_\gamma(V)&\le\sup_{\V{x}\in\RR^\nu}\int_0^\tau\EE\big[|V(\V{B}_t^{\V{x}})|\big]\Id t
+Ae^{-\tau\gamma/2}\ln\Big(\sup_{\V{x}\in\RR^\nu}
\EE\big[e^{\int_0^1|V(\V{B}_s^{\V{x}})|\Id s}\big]\Big),
\end{align}
with some universal constant $A>0$. For all $\ve>0$ and every sequence $V_n\in\cK_\pm(\RR^\nu)$, 
$n\in\NN$, satisfying \eqref{approx2}, we find some $\gamma\ge1$ such that
$c_\gamma(V_n)\le\ve$, $n\in\NN$.
\end{lem}

\begin{proof}
The bound \eqref{qfb-Kato} is, e.g., identical to \cite[Eqn.~(58) on p.~92]{ChungZhao1995}.
To derive the bound \eqref{bd-cgammaV}, where $A:=\sum_{\ell=1}^\infty \ell e^{-(\ell-1)/2}$,
follows easily from
\begin{align*}
e^{{\int_0^\ell\EE[|V(\V{B}_s^{\V{x}})|]\Id s}}&\le
\EE\big[e^{\int_0^\ell|V(\V{B}_s^{\V{x}})|\Id s}\big]
\\
&=\|S_{\ell}^{|V|}\|_{\infty,\infty}\le\|S_1^{|V|}\|_{\infty,\infty}^\ell
=\Big(\EE\big[e^{\int_0^1|V(\V{B}_s^{\V{x}})|\Id s}\big]\Big)^\ell,
\end{align*}
where we applied Jensen's inequality in the first step. In view of \eqref{approx3a} and 
\eqref{approx3b}, we may first choose $\tau>0$ small enough and then
$\gamma\ge1$ large enough (depending on $\tau$) to make $\sup_{n}c_\gamma(V_n)$ as
small as we please.
\end{proof}

\begin{ex}\label{ex-QED-V}
(1) Let us consider the standard model of non-relativistic QED (Ex.~\ref{ex-NRQED}) 
for $N\in\NN$ electrons in the electrostatic potential of
$K\in\NN$ nuclei with atomic numbers $\V{Z}=(Z_1,\ldots,Z_K)\in[0,\infty)^K$ 
whose positions are given by the components of
$\ul{\V{R}}=(\V{R}_1,\ldots,\V{R}_K)\in\RR^{2K}$. Then the total Hamiltonian reads
\begin{align}\label{HQED}
H_{\V{Z},\ul{\V{R}}}^{\chi,N,\mathfrak{e}}
&:=\sum_{\ell=1}^N\big\{\tfrac{1}{2}(-i\nabla_{\V{x}_\ell}-\vp(\V{G}_{\V{x}_\ell}^\chi))^2
-\vsigma^{(\ell)}\cdot\vp(\V{F}_{\V{x}_\ell}^\chi)\big\}+\Id\Gamma(\omega)
+V_{\V{Z},\ul{\V{R}}}^{N,\mathfrak{e}},
\end{align}
where $\vsigma^{(\ell)}:=(\sigma_{3\ell-2},\sigma_{3\ell-1},\sigma_{3\ell})$
and where we dropped the cumbersome direct integral signs of \eqref{def-A} in the notation.
Here the Coulomb interaction potential is
\begin{align}\label{def-VN}
V_{\V{Z},\ul{\V{R}}}^{N,\mathfrak{e}}(\V{x}_1,\ldots,\V{x}_N)
&:=-\sum_{i=1}^N\sum_{\vk=1}^K\frac{\mathfrak{e}^2Z_\vk}{|\V{x}_i-\V{R}_\vk|}
+\sum_{1\le i<j\le N}\frac{\mathfrak{e}^2}{|\V{x}_i-\V{x}_j|},
\end{align}
where $\mathfrak{e}>0$ is the elementary charge. It is well-known that $V_{\V{Z}}^N$ blongs to 
$\cK(\RR^{3N})$ \cite[p.~216/7]{AizenmanSimon1982}. 
It is also well-known that the Hamiltonian in \eqref{HQED}
is well-defined and self-adjoint on $\dom(\Delta)\cap\dom(\Id\Gamma(\omega))$ 
\cite{HaslerHerbst2008,Hiroshima2000esa,Hiroshima2002}; it is a special
case of the Hamiltonains constructed for general Kato-decomposable potentials in 
Thm.~\ref{thm-FK} below. According to the Pauli principle, the physically relevant Hamiltonian is actually not $H_{\V{Z}}^{\Lambda,N}$, but rather its restriction to the reducing subspace of 
those functions in $\HR$ which are anti-symmetric under simultaneous permutations of their 
$N$ electronic position-spin variables, i.e., to
$$
\HR^N_{\mathrm{phys}}:=\Big(\bigwedge_{j=1}^NL^2(\RR^3,\CC^2)\Big)\otimes\sF,
$$
if it is considered as a subspace of $\HR$ in the canonical way.
Since all results of this paper proven for $H_{\V{Z}}^{\Lambda,N}$ contain analogous results for
$H_{\V{Z}}^{\Lambda,N}\!\!\upharpoonright_{\HR^N_{\mathrm{phys}}}$ as an obvious special case,
we shall most of the time not comment on the Pauli principle anymore.

\smallskip

\noindent(2)
Let $\mathfrak{e}_n>0$, $n\in\NN$, such that $\mathfrak{e}_n\to\mathfrak{e}$, $n\to\infty$,
for some $\mathfrak{e}>0$. Furthermore, let $\V{Z}_n\in[0,\infty)^K$, $n\in\NN$, be a
converging sequence of $K$ atomic numbers with limit $\V{Z}\in[0,\infty)^K$.
Abbreviate $V:=V_{\V{Z},\ul{\V{R}}}^{N,\mathfrak{e}}$ and 
$V_n:=V_{\V{Z}_n,\ul{\V{R}}}^{N,\mathfrak{e}_n}$, $n\in\NN$. Then, quite obviously,
$V_n$ converges to $V$ in the sense that \eqref{approx1} and \eqref{approx2} are satisfied.

Next, we consider $\ul{\V{R}}_n\in\RR^{2K}$, $n\in\NN$, be a converging sequence of nuclear 
positions with limit $\ul{\V{R}}_n\in\RR^{2K}$ and set 
$V_n':=V_{\V{Z},\ul{\V{R}}_n}^{N,\mathfrak{e}}$, $n\in\NN$. Let us explain why 
\eqref{approx1} and \eqref{approx2} are satisfied by $V$ and $V_n'$, $n\in\NN$.
The validity of \eqref{approx2} is again obvious, because neither its right nor its left hand side
change when the potentials is replaced by their translates. To argue that \eqref{approx1} holds
true as well, it suffices to treat a single term of the form  $|\V{y}_i-\V{R}_{\vk,n}|^{-1}$.
Without loss of generality we may further assume that $\V{R}_{\vk,n}\to\V{0}$. Writing
$$
\frac{1}{|\V{y}_i-\V{R}_{\vk,n}|}-\frac{1}{|\V{y}_i|}
=\Big(\frac{1}{|\V{y}_i-\V{R}_{\vk,n}|^\eh}-\frac{1}{|\V{y}_i|^\eh}\Big)
\Big(\frac{1}{|\V{y}_i-\V{R}_{\vk,n}|^\eh}+\frac{1}{|\V{y}_i|^\eh}\Big),
$$
using that
\begin{align*}
\Big|\frac{1}{|\V{y}_i-\V{R}_{\vk,n}|^\eh}-\frac{1}{|\V{y}_i|^\eh}\Big|&\le
\frac{|\V{R}_{\vk,n}|^\eh}{{|\V{y}_i-\V{R}_{\vk,n}|^\eh}{|\V{y}_i|^\eh}}
\\
&\le\frac{|\V{R}_{\vk,n}|^\eh}{2}\Big(\frac{1}{|\V{y}_i-\V{R}_{\vk,n}|}+\frac{1}{|\V{y}_i|}\Big),
\end{align*}
and applying Young's inequality $(\frac{1}{3}+\frac{2}{3}=1)$ twice, we see that
\begin{align*}
\Big|\frac{1}{|\V{y}_i-\V{R}_{\vk,n}|}-\frac{1}{|\V{y}_i|}\Big|
\le{|\V{R}_{\vk,n}|^\eh}\Big(\frac{1}{|\V{y}_i-\V{R}_{\vk,n}|^{\nf{3}{2}}}
+\frac{1}{|\V{y}_i|^{\nf{3}{2}}}\Big).
\end{align*}
Now, the validity of \eqref{approx1} for $V$ and $V_n'$ follows from the fact that
$\RR^{3N}\ni\ul{\V{x}}\mapsto|\V{x}_i|^{-\alpha}$ is in $\cK_\loc(\RR^{3N})$, for every $\alpha<2$
\cite[p.~216/7]{AizenmanSimon1982}.
\end{ex}


\subsection{FK operators}\label{ssec-FK}

\noindent
Let $V\in\cK_\pm(\RR^\nu)$ and let $\V{X}^{\V{x}}=(\V{X}_t^{\V{x}})_{t\in I}$ be a Brownian motion 
or a Brownian bridge starting at $\V{x}\in\RR^\nu$ as explained in Subsect.~\ref{ssec-SDE}.  
Then, according to Rem.~\ref{rem-Kato}, it $\PP$-a.s. makes sense to define
\begin{equation}\label{def-WWV}
\WW{t}{V}[\V{X}^{\V{x}}]:=e^{-\int_0^tV(\V{X}_s^{\V{x}})\Id s}\,\WW{t}{0}[\V{X}^{\V{x}}],\quad t\in I,
\end{equation}
where $\WW{}{0}[\V{X}^{\V{x}}]$ is given by Thm.~\ref{thm-WW}.

\begin{defn}[{\bf Feynman-Kac operators}]\label{def-FKO} 
Assume that $V\in\cK_\pm(\RR^\nu)$. If $\Psi\in L^p(\RR^\nu,{\FHR};e^{-a|\V{x}|}\Id\V{x})$,
for some $p\in[1,\infty]$ and $a\ge0$, then we define
\begin{align*}
(T_t^V\Psi)(\V{x})&:=\EE\big[\WW{t}{V}[\V{B}^{\V{x}}]^*\Psi(\V{B}_t^{\V{x}})\big],
\quad\V{x}\in\RR^\nu,\,t\ge0.
\end{align*}
Furthermore, we set
\begin{align*}
T_t^V(\V{x},\V{y})&:=p_t(\V{x},\V{y})\,\EE\big[\WW{t}{V}[\V{b}^{t;\V{y},\V{x}}]\big],\quad
\V{x},\V{y}\in\RR^\nu,\, t>0.
\end{align*}
\end{defn}

Invoking some of the results collected in 
Subsect.~\ref{ssec-Kato}, we may in fact ensure the existence of the expectations
appearing in Def.~\ref{def-FKO}; see Sect.~\ref{sec-SG-CK}. For instance, the vector
$(T_t^V\Psi)(\V{x})$ does not depend on the representative of 
$\Psi\in L^p(\RR^\nu,\FHR;e^{-a|\V{x}|}\Id\V{x})$
because $\PP\{\V{B}_t^{\V{x}}\in N\}=0$, for every Borel set $N\subset\RR^\nu$ of
Lebesgue measure zero. The operator-valued expectation defining $T_t^V(\V{x},\V{y})$
exists due to Thm.~\ref{thm-WW}(1). The set of equivalence classes of functions
$\Psi(\cdot):\RR^\nu\to\FHR$, for which we find $a\ge0$ and $p\in[1,\infty]$ such that
$\Psi\in L^p(\RR^\nu,{\FHR};e^{-a|\V{x}|}\Id\V{x})$, is a vector space and we 
consider $T_t^V$ as a linear map on that vector space with values in the measurable
functions from $\RR^\nu$ to $\FHR$. Many of our results deal with suitable
restrictions of $T_t^V$, which are often again denoted by the same symbol; it should always
be clear from the context what is meant.

We close this survey section by stating the FK formula and clarifying the definition of the
full Hamiltonian including a Kato decomposable potential $V$ in the next theorem. It follows
immediately from the results of \cite[\textsection11]{GMM2014},
because any $V_-\in\cK(\RR^\nu)$ is infinitesimally $-\Delta$-form bounded
and $V_+\in\cK_\loc(\RR^\nu)$ entails $V_+\in L^1_\loc(\RR^\nu)$; see \cite{AizenmanSimon1982}.
The symbol $H^{V_+}$ will denote the semi-bounded self-adjoint operator representing
the closure of the semi-bounded and closable quadratic form
$$
\dom(H^0)\cap\fdom(V_+)\ni\Psi\longmapsto
\SPn{\Psi}{H^0\Psi}+\int_{\RR^\nu}V_+(\V{x})\|\Psi(\V{x})\|^2\Id\V{x}.
$$
Of course, if $A$ is a self-adjoint operator in some Hilbert space which is semi-bounded below, then
$(e^{-tA})_{t\ge0}$ denotes its semi-group defined by means of the spectral calculus;
we also recall the notation \eqref{def-lambda}.

\begin{thm}[\cite{GMM2014}]\label{thm-FK}
 Let $V\in\cK_\pm(\RR^\nu)$. Then ${\Psi}\in\fdom(H^{V_+})$ implies
$\|{\Psi}(\cdot)\|\in\fdom(-\Delta)\cap\fdom(V_+)$ with
\begin{align*}
\SPb{\|{\Psi}(\cdot)\|}{(-\tfrac{1}{2}\Delta+V_+)\|{\Psi}(\cdot)\|}
\le\SPb{{\Psi}}{(H^{V_+}+\|\mho\|_\infty){\Psi}}.
\end{align*}
In particular, $V_-$ is infinitesimally $H^{V_+}$-form bounded. Hence, by the KLMN theorem,
there exists a unique semi-bounded self-adjoint operator $H^V$ with 
$\fdom(H^V)=\fdom(H^0)\cap\fdom(V_+)$ representing the closure
of the semi-bounded closable quadratic form
$$
\dom(H^0)\cap\fdom(V_+)\ni\Psi\longmapsto
\SPn{\Psi}{H^0\Psi}+\int_{\RR^\nu}V(\V{x})\|\Psi(\V{x})\|^2\Id\V{x}.
$$
The semi-group of $H^V$ is represented by the following Feynman-Kac formula,
\begin{align}
(e^{-tH^V}\!\Psi)(\V{x})&=(T_t^V\Psi)(\V{x}),\label{feyn}
\end{align}
for all $t>0$, $\Psi\in\HR=L^2(\RR^\nu,\FHR)$, and a.e. $\V{x}\in\RR^{\nu}$.
\end{thm}

\begin{rem}\label{rem-FK}
In the {scalar case} ($\V{F}=\V{0}$) \eqref{feyn} has been proved
(under slightly more restrictive conditions on $\V{G}$) first in \cite{Hiroshima1997}
by intricate applications of Trotter's product formula. That the FK integrand found in
\cite{Hiroshima1997} actually solves a SDE is a result of \cite{GMM2014}.
For one spin-$1/2$ electron, a different version of \eqref{feyn}
has been obtained in \cite{HiroshimaLorinczi2008}, where the spin degrees of freedom are
accounted for by a Poisson jump process. The FK formula found in \cite{HiroshimaLorinczi2008}
involves an additional regularization procedure which is avoided in \eqref{feyn} by putting the spin
degrees of freedom in the target space.
\end{rem}


\section{Chapman-Kolmogoroff equation and semi-group properties}\label{sec-SG-CK}

\noindent
In this section we establish the usual basic identities relating the operators $T_t^V$ and the kernels
$T^V_t(\V{x},\V{y})$, for Kato decomposable potentials. In particular, we shall eventually verify the
Chapman-Kolmogoroff equations for the kernels from which we deduce the semi-group property of
the operators.

Recall the definition of $\mho$ in \eqref{def-lambda} as well as \eqref{def-StV} and \eqref{bd-StV}.

\begin{lem}\label{lem-T-bd}
Let $V\in\cK_\pm(\RR^\nu)$, $t>0$, $1\le p\le q\le\infty$, and $F:\RR^\nu\to\RR$ be globally
Lipschitz continuous. Then the formula for $T_t^V$ in Def.~\ref{def-FKO} yields 
a well-defined element of 
$\LO(L^p(\RR^\nu,\FHR;e^{pF}\Id\V{x}),L^q(\RR^\nu,\FHR;e^{qF}\Id\V{x}))$ 
satisfying
\begin{align}\label{norm-T}
\|e^FT_t^Ve^{-F}\|_{p,q}\le e^{\|\mho\|_\infty t}\,\|e^FS_t^Ve^{-F}\|_{p,q}.
\end{align}
\end{lem}

\begin{proof}
$\Psi\in L^p(\RR^\nu,{\FHR};e^{pF}\Id\V{x})$ says that 
$\|e^F\Psi(\cdot)\|\in L^p(\RR^\nu;\Id\V{x})$ and Thm.~\ref{thm-WW}(1) implies
\begin{align*}
\big\|\EE\big[\WW{t}{V}[\V{B}^{\V{x}}]^*\Psi(\V{B}_t^{\V{x}})\big]\big\|
&\le e^{\|\mho\|_\infty t}\,\EE\big[e^{-\int_0^tV(\V{B}_s^{\V{x}})\,\Id s}
\|\Psi(\V{B}_t^{\V{x}})\|\big],\quad\V{x}\in\RR^\nu.
\end{align*}
Hence, the assertion follows from \eqref{def-StV} and the remarks below it.
\end{proof}

The previous lemma will be complemented later on in Thm.~\ref{thm-LNCV} by
considering additional unbounded weight operators acting in $\FHR$.

Before we turn to the discussion of the integral kernel we first verify that a familiar transformation 
rule for the Brownian bridge applies to the processes given by Thm.~\ref{thm-WW} as well.
To prepare ourselves for the approximation arguments in the proof of the next lemma it might 
make sense to recall that, for all $t>0$, $\V{x},\V{y}\in\RR^\nu$, 
and every Borel set $N\subset\RR^\nu$ of Lebesgue measure zero,
\begin{align}\label{null}
\int_0^t\PP\{\V{B}_s^{\V{x}}\in N\}\Id s=0=
\int_0^t\PP\{\V{b}_s^{t;\V{y},\V{x}}\in N\}\Id s.
\end{align}

\begin{lem}\label{lem-bridget1}
Let $t>s>0$, $\V{x},\V{y}\in\RR^\nu$, $V\in\cK_\pm(\RR^\nu)$, and let
$A:\RR^\nu\to\LO(\FHR)$ be measurable and bounded with a separable image. Then
\begin{align}\label{bridget71}
\int_{\RR^\nu}p_t(\V{x},\V{z})A(\V{z})\EE\big[\WW{t}{V}[\hat{\V{b}}{}^{t;\V{z},\V{x}}]^*\big]\Id\V{z}
&=\EE\big[A(\V{B}_t^{\V{x}})\WW{t}{V}[\V{B}^{\V{x}}]\big],
\\\label{bridget72}
p_t(\V{x},\V{y})\EE\big[A(\V{b}_s^{t;\V{x},\V{y}})\WW{s}{V}[\V{b}^{t;\V{x},\V{y}}]\big]
&=\EE\big[p_{t-s}(\V{B}_{s}^{\V{x}},\V{y})A(\V{B}_s^{\V{x}})\WW{s}{V}[\V{B}^{\V{x}}]\big].
\end{align}
\end{lem}

\begin{proof} 
{\em Step 1.} First, we prove both \eqref{bridget71} and \eqref{bridget72} under the additional
assumption that $V:\RR^\nu\to\RR$ be bounded and continuous.

{\em As to Eqn.~\eqref{bridget71}:}
First, we observe that the left hand side of \eqref{bridget71} is well-defined thanks to
Prop.~\ref{prop-ida} and since the $\LO(\FHR)$-valued Bochner-Lebesgue integral commutes
with the adjoint.

We have seen in Ex.~\ref{ex-stoch}(2) that the reversed process $(\V{B}_{t-\tau}^{\V{x}})_{\tau\in[0,t]}$
is a Brownian bridge from $\V{B}_t^{\V{x}}$ to $\V{x}$ in time $t$. Hence
$\WW{t}{V}[\V{B}^{\V{x}}]^*=\WW{t}{V}[{\hat{\V{b}}{}^{t;\V{B}_t^{\V{x}},\V{x}}}]$, $\PP$-a.s.,
by Thm.~\ref{thm-WW-rev}, where we use the notation introduced in Ex.~\ref{ex-stoch}(2).
By the Markov property stated in Thm.~\ref{thm-WW}(4) we further have, for every
$\fF[\sR_t\V{B}^{\V{x}}]_0$-measurable bounded $\eta:\Omega\to\FHR$,
\begin{align*}
\EE^{\fF[\sR_t\V{B}^{\V{x}}]_0}\big[\WW{t}{V}[\hat{\V{b}}{}^{t;\V{B}_t^{\V{x}},\V{x}}]\eta\big]
=\hat{F}(\V{B}_t^{\V{x}})\eta,\quad\text{$\PP$-a.s.},
\end{align*}
with $\hat{F}(\V{z}):=\EE[\WW{t}{V}[\hat{\V{b}}{}^{s;\V{z},\V{x}}]]$, $\V{z}\in\RR^\nu$. 
(Here we apply the Markov property with $\BB[\sR_t\V{B}^{\V{x}}]$ as underlying stochastic basis, 
$\V{X}^{\V{q}}=\hat{\V{b}}{}^{t;\V{q},\V{x}}$, and with $f(\V{z},\phi):=\chi(\|\phi\|)\phi$,
$\V{z}\in\RR^\nu$, $\phi\in\FHR$, where $\chi\in C_0(\RR)$ equals $1$ on $[0,R]$ where $R>0$ is 
chosen much larger than $e^{(\lambda_0+\|V\|_\infty)t}\|\eta\|_\infty$.) As we may choose
$\eta:=A(\V{B}_t^{\V{x}})\psi$ with an arbitrary $\psi\in\FHR$, this implies
\begin{align*}
\EE\big[\WW{t}{V}[\V{B}^{\V{x}}]^*A(\V{B}_t^{\V{x}})^*\big]
&=\EE\big[\hat{F}(\V{B}_t^{\V{x}})^*A(\V{B}_t^{\V{x}})^*\big].
\end{align*}
This identity is the adjoint of \eqref{bridget71} since $\V{B}_t^{\V{x}}$ is 
$p_t(\V{x},\cdot)$-distributed.

{\em As to Eqn.~\eqref{bridget72}:}
 According to Ex.~\ref{ex-stoch}(3) and Thm.~\ref{thm-WW-rev},
\begin{align*}
\WW{s}{V}[\V{b}^{t;\V{x},\V{y}}]^*=\WW{s}{V}[\bar{\V{b}}{}^{s;\V{b}_s^{t;\V{x},\V{y}},\V{x}}],
\quad\PP\text{-a.s.},
\end{align*}
where we employ the notation introduced in Ex.~\ref{ex-stoch}(3). Furthermore,
the Markov property of Thm.~\ref{thm-WW}(4) yields, 
for all ${\fF}[\sR_s\V{b}^{t;\V{x},\V{y}}]_0$-measurable bounded $\eta:\Omega\to\FHR$,
\begin{align*}
\EE^{{\fF}[\sR_s\V{b}^{t;\V{x},\V{y}}]_0}\big[\WW{s}{V}[\bar{\V{b}}{}^{s;\V{b}_s^{t;\V{x},\V{y}},\V{x}}]
\eta\big]=\bar{F}(\V{b}_s^{t;\V{x},\V{y}})\eta,\quad\PP\text{-a.s.},
\end{align*}
with $\bar{F}(\V{z}):=\EE\big[\WW{s}{V}[\bar{\V{b}}{}^{s;\V{z},\V{x}}]$, $\V{z}\in\RR^\nu$.
Since we may choose $\eta:=A(\V{b}_s^{t;\V{x},\V{y}})\psi$ with $\psi\in\FHR$ and 
the distribution of $\V{b}_s^{t;\V{x},\V{y}}$ is given by
$\RR^\nu\ni\V{z}\mapsto p_s(\V{x},\V{z})p_{t-s}(\V{z},\V{y})/p_t(\V{x},\V{y})$, 
it follows altogether that
\begin{align*}
p_t(\V{x},\V{y})\EE\big[\WW{s}{V}[\V{b}^{t;\V{x},\V{y}}]^*A(\V{b}_s^{t;\V{x},\V{y}})^*\big]
&=\int_{\RR^\nu}p_s(\V{x},\V{z})\EE\big[\WW{s}{V}[\bar{\V{b}}{}^{s;\V{z},\V{x}}]\big]
p_{t-s}(\V{z},\V{y})A(\V{z})^*\Id\V{z}
\\
&=\EE\big[\WW{s}{V}[\V{B}^{\V{x}}]^*p_{t-s}(\V{B}_s^{\V{x}},\V{y})A(\V{B}_s^{\V{x}})^*\big].
\end{align*}
Here we applied \eqref{bridget71} in the second step,
with $t$ replaced by $s$ and $A(\V{z})$ replaced by $p_{t-s}(\V{z},\V{y})A(\V{z})$, and
taking into account that $\WW{s}{V}[\bar{\V{b}}{}^{s;\V{z},\V{x}}]$ and
$\WW{s}{V}[\hat{\V{b}}{}^{s;\V{z},\V{x}}]$ have the same distribution by Ex.~\ref{ex-str-sol}(2).

{\em Step 2.} Next, we extend \eqref{bridget72} to all Kato decomposable
potentials. If $V$ is measurable and bounded, then, by convolution with
standard mollifiers, we find a uniformly bounded sequence of 
potentials $V_n\in C(\RR^\nu)$ such that $V_n\to V$ a.e. and, hence, 
$\int_0^sV_n(\V{b}_r^{t;\V{x},\V{y}})\Id r\to\int_0^sV(\V{b}_r^{t;\V{x},\V{y}})\Id r$
and $\int_0^sV_n(\V{B}_r^{\V{x}})\Id r\to\int_0^sV(\V{B}_r^{\V{x}})\Id r$,
$\PP$-a.s., by \eqref{null}. In view of \eqref{bd-WW}, we may thus extend \eqref{bridget72} to 
bounded $V$ by dominated convergence. 
If $V\in\cK_\pm(\RR^\nu)$ is arbitrary, then
we plug $V_m:=1_{|V|\le m} V$, $m\in\NN$, into \eqref{bridget72} and
employ the dominated convergence theorem together with \eqref{bd-WW}, \eqref{dirk0},
and \eqref{dirk1} to pass to the limit $m\to\infty$.

Since the second identity in \eqref{null} holds true with $\V{b}^{t;\V{y},\V{x}}$ replaced by
$\hat{\V{b}}{}^{t;\V{z},\V{x}}$ as well, it should now be clear how to extend \eqref{bridget71}
to all $V\in\cK_\pm(\RR^\nu)$.
\end{proof}

\begin{prop}\label{lem-T(x,y)-sym}
Let $V\in\cK_\pm(\RR^\nu)$ and $t>0$. Then, for all $\V{x},\V{y}\in\RR^\nu$,
$T_t^V(\V{x},\V{y})$ is a well-defined element of $\LO(\FHR)$ with
\begin{equation}\label{T(x,y)-sym}
T_t^V(\V{x},\V{y})^*=T_t^V(\V{y},\V{x})\quad
\text{and}\quad\|T_t^V(\V{x},\V{y})\|\le
e^{\|\mho\|_\infty t}S_t^V(\V{x},\V{y}).
\end{equation}
Furthermore, the map $\RR^{2\nu}\ni(\V{x},\V{y})\mapsto T_t^V(\V{x},\V{y})\in\LO(\FHR)$ 
is measurable with a separable image and
\begin{align}\label{intK}
(T_t^V\Psi)(\V{x})&=\int_{\RR^{\nu}}T_t^V(\V{x},\V{y})\Psi(\V{y})\,\Id \V{y},\quad\V{x}\in\RR^\nu,
\end{align}
for all $\Psi\in L^p(\RR^\nu,\FHR;e^{-a|\V{x}|}\Id\V{x})$ with $p\in[1,\infty]$ and $a\ge0$.
\end{prop}

\begin{proof}
The existence of the integral defining the kernel 
$T^V_t(\V{x},\V{y})$ and the inequality in \eqref{T(x,y)-sym} 
follow immediately from \eqref{dirk1} and Thm.~\ref{thm-WW}(1)\&(2).
The assertions on its measurability and image are consequences of Prop.~\ref{prop-ida}

For bounded and continuous $V$, the identity in \eqref{T(x,y)-sym} has been shown in 
\cite[Lem.~10.6]{GMM2014}. (It follows from Thms.~\ref{thm-str-sol} and~\ref{thm-WW-rev} 
and the fact that $(\V{b}_{t-\tau}^{t;\V{y},\V{x}})_{\tau\in[0,t]}$ is a Brownian bridge from $\V{x}$
to $\V{y}$ in time $t$ with respect to some new stochastic basis and driving Brownian motion.
The latter fact can be proved by adapting arguments of \cite{Pardoux-LNM1204};
see \cite[Lem.~10.3 \& App.~D]{GMM2014}.) It is, however, straightforward to extend
the identity in \eqref{T(x,y)-sym} to general Kato decomposable potentials by the approximation
scheme applied in Step~2 of the proof of Lem.~\ref{lem-bridget1}.

Likewise, \eqref{intK} follows from Ex.~10.4 in \cite{GMM2014}, provided that $V$ is bounded
and continuous and $\Psi\in L^\infty(\RR^\nu,\FHR)$. 
If $\Psi\in L^p(\RR^\nu,\FHR;e^{-a|\V{x}|}\Id\V{x})$ with
$p\in[1,\infty]$, then we plug the functions 
$\Psi_n\in L^p(\RR^\nu,\FHR)\cap L^\infty(\RR^\nu,\FHR)$
given by $\Psi_n(\V{z}):=1_{\|\Psi(\V{z})\|\le n}\Psi(\V{z})$ into \eqref{intK} and pass to the limit
$n\to\infty$ by means of the dominated convergence theorem, on the right hand side using 
$\|T^V_t(\V{x},\cdot)\|\in L^{p'}(\RR^\nu;e^{a|\V{x}|}\Id\V{x})$ and \eqref{null}, 
on the left hand side employing 
$e^{-\int_0^tV(\V{B}_s^{\V{x}})\Id s}\|\Psi(\V{B}_t^{\V{x}})\|$ as a $\PP$-integrable majorant.
After that we carry through the approximation scheme of Step~2 of the proof of 
Lem.~\ref{lem-bridget1} to generalize the result to Kato decomposable potentials. 
\end{proof}

Before we state the next proposition we 
notice that, in view of \eqref{dirk1} and \eqref{T(x,y)-sym}, for $V\in\cK_\pm(\RR^\nu)$ 
and all fixed $t>0$, $\V{y}\in\RR^\nu$, and $\psi\in\FHR$, the function 
$\RR^\nu\ni\V{x}\mapsto T_t^V(\V{x},\V{y})\psi$ belongs to every $L^p(\RR^\nu,\FHR)$
with $p\in[1,\infty]$.

\begin{prop}\label{prop-CK}
Let $V\in\cK_\pm(\RR^\nu)$, $t>s>0$, and $\V{x},\V{y}\in\RR^\nu$. Then 
\begin{align}\label{CK1}
T_{t}^V(\V{x},\V{y})\psi&=\big(T_{s}^V(T_{t-s}^V(\cdot,\V{y})\psi)\big)(\V{x}),\quad\psi\in\FHR,
\end{align}
and the following Chapman-Kolmogoroff equation is valid,
\begin{align}\label{CK2}
T_{t}^V(\V{x},\V{y})=\int_{\RR^\nu}T_{s}^V(\V{x},\V{z})T_{t-s}^V(\V{z},\V{y})\Id\V{z}.
\end{align}
\end{prop}

\begin{proof}
Let $V$ be bounded and continuous to start with. 
Let $(P_{r,u}^{\V{y}})_{0\le r\le u\le t}$ denote the transition operators associated with
the choices $\V{X}^{\V{q}}=\V{b}^{t;\V{q},\V{y}}$ in Thm.~\ref{thm-WW} and define 
$f:\RR^\nu\times\FHR\to\FHR$ by $f(\V{z},\psi):=\chi(\|\psi\|)\psi$, 
where $\chi\in C_0(\RR)$ is again equal to $1$ on $[0,R]$ with $R>0$ much larger than 
$e^{(\|\mho\|_\infty+\|V\|_\infty)t}$. Then we observe that, 
for all $\V{z}\in\RR^\nu$ and $\psi\in\FHR$,
\begin{align}\label{bridget76}
(P_{s,t}^{\V{y}}f)(\V{z},\psi)&=\EE\big[\WW{t-s}{V}[^s\V{b}^{t;\V{z},\V{y}}]\psi\big]
=\EE\big[\WW{t-s}{V}[\V{b}^{t-s;\V{z},\V{y}}]\psi\big]=A(\V{z})\psi,
\end{align}
with $A(\V{z}):=\EE\big[\WW{t-s}{V}[\V{b}^{t-s;\V{z},\V{y}}]\big]$,
because $\WW{t-s}{V}[^s\V{b}^{t;\V{z},\V{y}}]\psi$ and $\WW{t-s}{V}[\V{b}^{t-s;\V{z},\V{y}}]\psi$
have the same distribution by Ex.~\ref{ex-str-sol}(2). In this notation,
\begin{align}\label{bridget77}
T_{t-s}^V(\V{y},\V{z})&=p_{t-s}(\V{y},\V{z})A(\V{z}).
\end{align}
Therefore, we obtain, for all $\psi\in\FHR$,
\begin{align*}
T_t^V(\V{y},\V{x})\psi&=p_t(\V{x},\V{y})(P_{0,t}^{\V{y}}f)(\V{x},\psi)
=p_t(\V{x},\V{y})(P_{0,s}^{\V{y}}P_{s,t}^{\V{y}}f)(\V{x},\psi)
\\
&=p_t(\V{x},\V{y})
\EE\big[(P_{s,t}^{\V{y}}f)(\V{b}_s^{t;\V{x},\V{y}},\WW{s}{V}[\V{b}^{t;\V{x},\V{y}}]\psi)\big]
\\
&=p_t(\V{x},\V{y})\EE\big[A(\V{b}_s^{t;\V{x},\V{y}})\WW{s}{V}[\V{b}^{t;\V{x},\V{y}}]\psi\big]
\\
&=\EE\big[p_{t-s}(\V{y},\V{B}_s^{\V{x}})A(\V{B}_s^{\V{x}})\WW{s}{V}[\V{B}^{\V{x}}]\big]\psi
\\
&=\EE\big[T_{t-s}^V(\V{y},\V{B}_s^{\V{x}})\WW{s}{V}[\V{B}^{\V{x}}]\big]\psi,
\end{align*}
where we applied \eqref{bridget76} in the fourth step,
Lem.~\ref{lem-bridget1} in the penultimate step, and \eqref{bridget77} in the
last one. In view of \eqref{T(x,y)-sym} this computation implies, for every $\phi\in\FHR$,
\begin{align*}
T_t^V(\V{x},\V{y})\phi&=\EE\big[\WW{s}{V}[\V{B}^{\V{x}}]^*T_{t-s}(\V{B}_s^{\V{x}},\V{y})\phi\big]
=T_s^V(T^V_{t-s}(\cdot,\V{y})\phi)(\V{x}),
\end{align*}
which is \eqref{CK1} with a bounded and continuous $V$. Applying \eqref{intK} we also obtain
\eqref{CK2} for bounded and continuous $V$. In the proof of \eqref{T(x,y)-sym} we have, however,
explained how to approximate $T^V_{\tilde{t}}(\V{x},\V{y})$, $\tilde{t}>0$, 
with a bounded and measurable (resp. Kato  decomposable) $V$ in operator norm by kernels with 
bounded and continuous (resp. bounded and measurable) potentials. Since all kernels can be 
majorized by means of  \eqref{dirk1} and the bound in \eqref{T(x,y)-sym}, 
we may thus employ the dominated convergence theorem to extend \eqref{CK2} to Kato 
decomposable $V$. Together with \eqref{intK} this will also prove \eqref{CK1} for Kato 
decomposable $V$.
\end{proof}

In what follows, we shall sometimes use the symbol $T_t^{V;(p,q),F}$ to denote
the restriction of $T_t^V$ to $L^p(\RR^\nu,\FHR;e^{pF}\Id\V{x})$ considered 
(by means of Lem.~\ref{lem-T-bd}) as an element of
$\LO(L^p(\RR^\nu,\FHR;e^{pF}\Id\V{x}),L^q(\RR^\nu,\FHR;e^{qF}\Id\V{x}))$.
Thanks to \cite{GMM2014} we already know that  $(T_t^{V;(2,2),0})_{t\ge0}$ is a 
strongly continuous semi-group of bounded self-adjoint operators in the Hilbert
space $\HR=L^2(\RR^\nu,\FHR)$. In fact, the following more general statement holds true:

\begin{cor}\label{cor-SASG}
Let $V\in\cK_\pm(\RR^\nu)$ and $F:\RR^\nu\to\RR$ be globally Lipschitz continuous. 
Then $(T_t^V)_{t\ge0}$ defines self-adjoint semi-groups
between the spaces $L^p(\RR^\nu,{\FHR};e^{pF}\Id\V{x})$ in the sense that 
\begin{align}\label{achim}
T_{s+t}^{V;(p,r),F}=T_s^{V;(q,r),F}T_t^{V;(p,q),F},\qquad (T_t^{V;(p,q),F})^*=T_t^{V;(q',p'),-F},
\end{align}
for all $s,t>0$ and $1\le p\le q\le r\le\infty$, where $p'$ is the
exponent conjugate to $p$ and $q'$ the one conjugate to $q$.
\end{cor}

\begin{proof}
The second relation in \eqref{achim} follows immediately from Prop.~\ref{lem-T(x,y)-sym}.
The first asserted relation, expressing the semi-group property, is an easy 
consequence of \eqref{intK}, the Chapman-Kolmogoroff equation, and Fubini's theorem.
The latter is applicable by virtue of \eqref{T(x,y)-sym} which shows that the right hand side of
\begin{align*}
e^{-a|\V{x}|}\int_{\RR^{2\nu}}&\|T_s^V(\V{x},\V{z})T_t^V(\V{z},\V{y})\Psi(\V{y})\|\Id(\V{y},\V{z})
\\
&\le \int_{\RR^\nu}e^{a|\V{x}-\V{z}|}\|T^V_s(\V{x},\V{z})\|\Id\V{z}
\sup_{\tilde{\V{z}}\in\RR^\nu}\big\|e^{a|\tilde{\V{z}}-\cdot|}
\|T_t^V(\tilde{\V{z}},\cdot)\|\big\|_{p'}\|e^{-a|\cdot|}\Psi\|_p
\end{align*}
is finite, for all $\Psi\in L^p(\RR^\nu,\FHR;e^{-a|\V{x}|}\Id\V{x})$ with $p\in[1,\infty]$.
\end{proof}


\section{Weighted estimates on the stochastic flow}\label{sec-weights}

\noindent
In this section we study how the operator-valued process $\WW{}{V}[\V{B}^{\V{q}}]$ 
introduced in Thm.~\ref{thm-WW} behaves when it is multiplied or conjugated with certain weight 
functions, i.e., with certain unbounded multiplication operators acting in the Fock space
which are defined and discussed in Subsect.~\ref{ssec-def-weights}. 
The weighted Burkholder-Davis-Gundy (BDG) type inequalities derived Subsect.~\ref{ssec-wBDG} 
are extensions of similar ones in 
\cite[\textsection7]{GMM2014} in the sense that we consider more general weights and an  
inhomogeneous SDE generalizing \eqref{SDE-spin}. In certain respects the situation here is, 
however, also simpler than in \cite{GMM2014}, 
where we had to study the convergence of a time-ordered integral series defining 
$\WW{}{V}[\V{B}^{\V{q}}]$ in addition. Moreover, we refrain from considering unbounded drift
vectors $\V{\beta}$ in this section to shorten some arguments. A suitable version of the 
stochastic Gronwall lemma \cite{Scheutzow2013} obtained in 
Subsect.~\ref{ssec-stoch-Gronwall} turns out to be convenient
in the derivation of our BDG type bounds. 

In Subsect.~\ref{ssec-flow-cont}, we shall further address continuity properties of 
$\WW{}{V}[\V{B}^{\V{q}}]$ under changes of the initial condition $\V{q}$ and the coefficient 
vector $\V{c}$, employing an elementary algebraic lemma that we prove first in 
Subsect.~\ref{ssec-i1i5}.

The introduction of the inhomogeneity
$\V{R}$ in the SDE \eqref{SDE-inhom} will actually become relevant only in the discussion of
differentiability properties of the stochastic flow in our companion paper \cite{Matte2015}, where the
results of this section will serve as a crucial ingredient. The introduction of the inhomogeneity
$\rho$ in \eqref{SDE-inhom} will already turn out to be useful in Ex.~\ref{ex-antonio}


\subsection{Definition and discussion of the weight functions}\label{ssec-def-weights}

First, we introduce the weight functions we are interested in: Let $t_0>0$ and set
\begin{align*}
\forall\,t\in[0,t_0]:\quad\tau_\alpha(t)&:=\left\{
\begin{array}{ll}
t/2\alpha,&\alpha>0,\\
(t-t_0)/2\alpha,&\alpha<0.
\end{array}
\right.
\end{align*}
Let $\vo,\vk:\cM\to[0,\infty)$ be measurable functions such that $\vo\le\omega$ and define
\begin{align*}
v_{\alpha,t}&:=
\tau_\alpha(t)\,\vo+\vk,\quad v_{\alpha,\ve,t}:=v_{\alpha,t}(1+\ve\,v_{\alpha,t})^{-1},
\\
\Theta_{\ve,t}^{(\alpha)}
&:=({1+\iota\tau_\alpha(t)+\Id\Gamma(v_{\alpha,\ve,t})})^\alpha\big({1+\ve\,(1+\iota\tau_\alpha(t)+\Id\Gamma(v_{\alpha,\ve,t}}))\big)^{-\alpha},
\end{align*}
for all $t\in[0,t_0]$, $\ve\in[0,1]$, and $\alpha\in\RR\setminus\{0\}$, where
$\iota=0$, if $\vo=0$, and $\iota=1$, if $\vo\not=0$.
If one is only interested in the question whether domains of higher
powers of the radiation field energy or the number operator stay {\em invariant}
under the action of our semi-groups, then it suffices to choose
$\vo=0$ and $\vk=\omega$ or $\vk=1$, respectively.
In order to show that $\WW{t}{0}[\V{B}^{\V{q}}]$ {\em improves} the localization in the
photon momentum spaces, we choose $\vo=\omega$ or $\vo=1_{\{\omega>r\}}\omega$,
for some $r>0$; see, e.g., Ex.~\ref{ex-QED-exp}(1). 
Moreover, we choose $\alpha>0$ to control weights to the left of 
$\WW{t_0}{0}[\V{B}^{\V{q}}]$, and $\alpha<0$ 
to control weights to the right of $\WW{t_0}{0}[\V{B}^{\V{q}}]$;
cf. Rem.~\ref{rem-spin4} below.
In technical proofs the regularization parameter $\ve$ will be chosen strictly positive,
thus rendering the operators bounded, and eventually be send to zero again. We shall also be 
interested in the exponential weights given by
\begin{align*}
\forall\,\ve\in[0,1],\,t\in[0,t_0]:\quad\Xi_{\ve,t}^{(\delta)}&:=
\left\{
\begin{array}{ll}
e^{\delta\Theta_{\ve,t}^{(1)}},&\delta\in(0,1],\\
e^{\delta\Theta_{\ve,t_0-t}^{(1)}},&\delta\in[-1,0).
\end{array}
\right.
\end{align*}

\begin{rem}\label{rem-weights}
For $t\in[0,t_0]$, $\ve\in(0,1]$, $|\alpha|>0$, and $0<|\delta|\le1$, the following holds:

\smallskip

\noindent(1) $\Theta_{\ve,t_0-t}^{(\alpha)}={\Theta_{\ve,t}^{(-\alpha)}}^{-1}$,
$\Xi_{\ve,t_0-t}^{(\delta)}={\Xi_{\ve,t}^{(-\delta)}}^{-1}$.

\smallskip

\noindent(2)
If we consider the $\Id\Gamma(\cdot)$'s as multiplication operators, then the following
pointwise bounds hold,
\begin{align}\label{gabi1}
{\Theta_{\ve,t}^{(\alpha)}}^{-1}\frac{\Id}{\Id t}\,\Theta_{\ve,t}^{(\alpha)}
&\le\alpha\,\tau_\alpha'(t)(1+\Id\Gamma(\omega))
\le \tfrac{1}{2}\,\Id\Gamma(\omega)+\tfrac{1}{2},
\\\label{gabi2}
{\Xi_{\ve,t}^{(\delta)}}^{-1}\,\frac{\Id}{\Id t}\,\Xi_{\ve,t}^{(\delta)}
&\le \tfrac{1}{2}\,\Id\Gamma(\omega)+\tfrac{1}{2}.
\end{align}
\end{rem}

\begin{notation}\label{not-Theta}
In what follows, $\Theta_s$ denotes any of the weights $\Theta^{(\alpha)}_{\ve,s}$ or
$\Xi_{\ve,s}^{(\delta)}$ with {\em non-zero} $\ve\in(0,1]$ as described in the previous paragraphs. 
\end{notation}

Using that the operator norm of $\Theta_s$ is 
bounded uniformly in $s$ by some $\ve$-dependent constant and employing the dominated 
convergence theorem togeher with \eqref{gabi1} and \eqref{gabi2}, it is straightforward to verify 
that, for every $\psi\in\dom(\Id\Gamma(\omega))$, the map $s\mapsto\Theta_s\psi$ belongs
to $C^1([0,t_0],\FHR)$ and its derivative is given by $\dot{\Theta}_s\psi$, where
$\dot{\Theta}_s$ denotes the maximal operator of multiplication with the function obtained
by taking pointwise derivatives of the weight functions at $s\in[0,t_0]$.

Since the $\LO(\FHR)$-valued maps $s\mapsto\Theta_s$ and $s\mapsto\dot{\Theta}_s$ are only
{\em strongly} continuous, it might make sense to verify the following:

\begin{lem}\label{lem-Ito-Theta}
Let $Y:[0,t_0]\times\Omega\to\FHR$ be a semi-martingale whose paths belong 
to $C([0,t_0],\dom(\Id\Gamma(\omega)))$ and which can $\PP$-a.s. be written as
\begin{align}
Y_t=\eta+\int_0^t A_{0,s}\Id s+\int_0^t\V{A}_s\Id\V{B}_s,\quad t\in[0,t_0],
\end{align}
for some $\fF_0$-measurable $\eta:\Omega\to\dom(\Id\Gamma(\omega))$ and an adapted
$\LO(\RR^{1+\nu},\FHR)$-valued process $(A_0,\V{A})$ with continuous paths.
Then $(\Theta_tY_t)_{t\in[0,t_0]}$
is a $\FHR$-valued semi-martingale and it $\PP$-a.s. satisfies
\begin{align}\label{laura0}
\Theta_tY_t&=\Theta_0\eta+\int_0^t\dot\Theta_sY_s\Id s+\int_0^t\Theta_sA_{0,s}\Id s
+\int_0^t\Theta_s\V{A}_s\Id\V{B}_s,\quad t\in[0,t_0],
\\\nonumber
\|\Theta_tY_t\|^2&=\|\Theta_0\eta\|^2+\int_0^t2\Re\SPn{\Theta_sY_s}{\Theta_sA_{0,s}
+\dot\Theta_sY_s}\Id s
\\\label{laura0b}
&\quad+\int_0^t2\Re\SPn{\Theta_sY_s}{\Theta_s\V{A}_s}\Id\V{B}_s+\int_0^t\|\Theta_s\V{A}_s\|^2\Id s,
\quad t\in[0,t_0].
\end{align}
\end{lem}

\begin{proof}
Let $0=\tau_0<\tau_1^n<\ldots<\tau_n^n=t_0$, $n\in\NN$, be partitions of $[0,t_0]$ with
$\max_\ell|\tau_{\ell+1}^n-\tau_\ell^n|\to0$, as $n\to\infty$. On $\Omega$ and for
all $t\in[0,t_0]$ and $n\in\NN$, we then have
\begin{align*}
\Theta_tY_t-\Theta_0Y_0
&=\sum_{\ell=0}^{n-1}(\Theta_{\tau_{\ell+1}^n\wedge t}Y_{\tau_{\ell+1}^n\wedge t}
-\Theta_{\tau_\ell^n\wedge t}Y_{\tau_\ell^n\wedge t})
\\
&=\int_0^t\sum_{\ell=0}^{n-1}1_{(\tau_\ell^n,\tau_{\ell+1}^n]}(s)\dot\Theta_sY_{\tau_{\ell+1}^n}\Id s
+\sum_{\ell=0}^{n-1}\Theta_{\tau_\ell^n}(Y_{\tau_{\ell+1}^n\wedge t}-Y_{\tau_\ell^n\wedge t}).
\end{align*}
Employing \eqref{gabi1} and \eqref{gabi2} and using 
$\max_{s\in[0,t_0]}\|(1+\Id\Gamma(\omega))\Theta_sY_s\|$
pathwise as dominating (constant) function, we see that the Bochner-Lebesgue integrals in 
the previous identity converge to $\int_0^t\dot\Theta_sY_s\Id s$, pathwise and for every $t\in[0,t_0]$,
as $n$ goes to infinity. Furthermore, we $\PP$-a.s. have
\begin{align}\nonumber
\sum_{\ell=0}^{n-1}&\Theta_{\tau_\ell^n}(Y_{\tau_{\ell+1}^n\wedge t}-Y_{\tau_\ell^n\wedge t})
=\int_0^t\sum_{\ell=0}^{n-1}1_{(\tau_\ell^n,\tau_{\ell+1}^n]}(s)\Theta_{\tau_\ell^n}\Id Y_s
\\\label{laura1}
&=\int_0^t\sum_{\ell=0}^{n-1}1_{(\tau_\ell^n,\tau_{\ell+1}^n]}(s)\Theta_{\tau_\ell^n}A_{0,s}\Id s+
\int_0^t\sum_{\ell=0}^{n-1}1_{(\tau_\ell^n,\tau_{\ell+1}^n]}(s)\Theta_{\tau_\ell^n}\V{A}_s\Id\V{B}_s,
\end{align}
for all $t\in[0,t_0]$ and $n\in\NN$. Taking the remarks preceding the statement into account, 
we see that $\Theta(A_0,\V{A})$ is again an adapted, continuous 
$\LO(\RR^{1+\nu},\FHR)$-valued process, whose paths are bounded on $[0,t_0]$.
Using $R_t:=\max_{0\le s\le t}\|\Theta_s(A_{0,s},\V{A}_s)\|$, $t\in[0,t_0]$, as a predictable 
dominating process, we may invoke the dominated convergence theorem for stochastic integrals
\cite[Thm.~24.2]{Me1982} to conclude that, along a subsequence, the expression in the last line
of \eqref{laura1} converges $\PP$-a.s. uniformly on $[0,t_0]$ to the sum of the last two
integrals in \eqref{laura0}. 

Recall that a conjugation $C$ was introduced in Hyp.~\ref{hyp-G}.
To infer \eqref{laura0b} from \eqref{laura0}, we observe that $\FHR=\sF_C^L+i\sF_C^L$,
with the completely real subspace $\sF_C:=\{\psi\in\sF:\,\Gamma(-C)\psi=\psi\}$,
and that every $\psi\in\FHR$ has a unique decomposition $\psi=\psi_1+i\psi_2$ with 
$\psi_1,\psi_2\in\sF_C^L$
and $\|\psi\|^2=\|\psi_1\|^2+\|\psi_2\|^2$. Hence, we can interpret the squared norm on $\FHR$
as a function on the real Hilbert space $\sF_C^{2L}$ and read \eqref{laura0} as a $\PP$-equality
between $\sF_C^{2L}$-valued processes. Then \eqref{laura0b} follows from 
\eqref{laura0} and the It\={o} formula of \cite[Thm.~4.32]{daPrZa1992}.
\end{proof}

\begin{notation}\label{not-T}
The operator $\vt$ is equal to $1$, if $\Theta_s=\Theta^{(\alpha)}_{\ve,s}$, and
equal to $1+\Id\Gamma(\omega)$, if $\Theta_s={\Xi_{\ve,s}^{(\delta)}}$. 
Using this convention we set
\begin{align}\label{eva1}
T_{1,s}&:=\tfrac{1}{2}\,\vt^\mh\Theta_s^{-1}\big[[\Theta_s^2,\vp(\V{G}_{\V{B}^{\V{q}}_s})]\,,\,
\vp(\V{G}_{\V{B}^{\V{q}}_s})\big]\,\Theta_s^{-1}\vt^\mh,
\\\label{eva2}
T_{2,s}&:=\tfrac{i}{2}[\Theta_s,\vp(q_{\V{B}^{\V{q}}_s})]\,\Theta_s^{-1}\vt^\mh
+\tfrac{1}{2}\Theta_s^{-1}\big[\Theta_s\,,\,[\Theta_s,\vsigma\cdot\vp(\V{F}_{\V{B}^{\V{q}}_{s}})]\big]
\Theta_s^{-1}\vt^\mh,
\\\label{eva4}
\V{T}_s^\pm&:=2[\Theta_s^{\pm1},\vp(\V{G}_{\V{B}^{\V{q}}_s})]\,\Theta_s^{\mp1}\vt^\mh.
\end{align}
\end{notation}

The operators in \eqref{eva1}--\eqref{eva4} 
are well-defined a priori on $\dom(\Id\Gamma(\omega))$. In fact, under suitable additional
conditions on the coefficient vector $\V{c}$, they
extend to bounded operators on ${\FHR}$ whose norms are bounded on $\Omega$
uniformly in $s$ {\em and $\ve$}. 
For we have the following result, whose proof is deferred to the appendix, where a systematic
treatment of multiple commutator estimates is presented:

\begin{lem}\label{lem-bd-T}
{\rm(1)} If $|\alpha|\ge1/2$ and the operators in \eqref{eva1}--\eqref{eva4} are defined
by means of $\Theta_s=\Theta_{\ve,s}^{(\alpha)}$ and $\vt=1$, 
then we find some $c_\alpha>0$ such that
\begin{align}\nonumber
\|T_{1,s}\|,\|\V{T}^\pm_s\|^2
&\le c_\alpha^2\|(\vo+\vk)^\eh(1+\vo+\vk)^{|\alpha|-\eh}\V{G}_{\V{B}_s^{\V{q}}}\|^2,
\\\label{tina1}
\|T_{2,s}\|&\le c_\alpha\|(\vo+\vk)^\eh(1+\vo+\vk)^{|\alpha|-\eh}(q,\vsigma\cdot\V{F})_{\V{B}_s^{\V{q}}}\|,
\end{align}
for all $s\in[0,t_0]$ and $\ve\in(0,1]$.

\smallskip

\noindent{\rm(2)} If the operators in \eqref{eva1}--\eqref{eva4} are defined
by means of $\Theta_s={\Xi_{\ve,s}^{(\delta)}}$ with 
$\vt=1+\Id\Gamma(\omega)$, then we find some $c>0$ such that
\begin{align}\nonumber\
\|T_{1,s}\|,\|\V{T}^\pm_s\|^2
&\le c^2|\delta|\big\|\big\{(\tfrac{t_0\vo}{2}+\vk)\vee\tfrac{(t_0\vo/2+\vk)^2}{\omega}\big\}^\eh 
e^{|\delta|(t_0\vo/2+\vk)}\V{G}_{\V{B}_s^{\V{q}}}\big\|^2,
\\\label{tina2}
\|T_{2,s}\|&\le c|\delta|^\eh\big\|\big\{(\tfrac{t_0\vo}{2}+\vk)\vee\tfrac{(t_0\vo/2+\vk)^2}{\omega}\big\}^\eh 
e^{|\delta|(t_0\vo/2+\vk)}(q,\vsigma\cdot\V{F})_{\V{B}_s^{\V{q}}}\big\|,
\end{align}
on $\Omega$, for all $s\in[0,t_0]$ and $0<\ve\le|\delta|$.
\end{lem}

\begin{proof}
The proof of (1) can be found in Ex.~\ref{ex-comm-omega}, the one of (2) in Ex.~\ref{ex-comm-exp}.
\end{proof}

\begin{notation}\label{not-normcirc}
For any vector $\V{v}$ with components in $\HP$, we write
\begin{align}\label{def-normcirc}
\|\V{v}\|_{\circ}&:=\left\{\begin{array}{ll}
c_\alpha\|(\vo+\vk)^\eh
(1+\vo+\vk)^{|\alpha|-\eh}\V{v}\|_{\HP},&\text{if $\Theta_s=\Theta_{\ve,s}^{(\alpha)}$,}
\\
c|\delta|^\eh\big\|\big\{(\tfrac{t_0\vo}{2}+\vk)\vee\tfrac{(t_0\vo/2+\vk)^2}{\omega}\big\}^\eh 
e^{|\delta|(t_0\vo/2+\vk)}\V{v}\big\|_{\HP},&
\text{if $\Theta_s={\Xi_{\ve,s}^{(\delta)}}$,}
\end{array}\right.
\end{align}
where $c_\alpha$ and $c$ are the constants in \eqref{tina1} and \eqref{tina2}, respectively.
If $\V{v}$ depends on $\V{x}\in\RR^\nu$, then we abbreviate
$\|\V{v}\|_{\circ,\infty}:=\sup_{\V{x}\in\RR^\nu}\|\V{v}_{\V{x}}\|_\circ$.
\end{notation}

We will use the following simple remark without further notice:

\begin{rem}
Let $\RR^\nu\ni\V{x}\mapsto v_{\V{x}}\in\HP$ be continuous and $\kappa:\cM\to[0,\infty)$
be measurable. Then the numerical function 
$\RR^\nu\ni\V{x}\mapsto\int_{\cM}\kappa|{v}_{\V{x}}|^2\Id\mu$ is measurable as well. In fact, 
$\int_{\cM}\kappa|{v}_{\V{x}}|^2\Id\mu=\sup_{n\in\NN}\int_\cM(n\wedge\kappa)|{v}_{\V{x}}|^2\Id\mu$, 
for all $\V{x}\in\RR^\nu$, by the monotone convergence theorem.
\end{rem}


\subsection{A remark on the stochastic Gronwall lemma}\label{ssec-stoch-Gronwall}

\noindent
In our derivation of the BDK type estimates we shall apply a variant of the stochastic 
Gronwall lemma \cite{Scheutzow2013} that we state first. The proof of the following lemma is a 
simple modification of a proof appearing in \cite{Scheutzow2013}.

\begin{lem}\label{lem-stoch-Gronwall}
Let $t_0>0$. Consider real processes $M$, $Z$, $R$, and $b$ on $[0,t_0]$, where
$M$ is a continuous local martingale starting at $0$, 
$Z$ and $R$ are non-negative and adapted with continuous paths, 
and $b$ is non-negative and progressively measurable.  Assume that, $\PP$-a.s.,
\begin{align*}
Z_t&\le Z_0+\int_0^tb_sZ_s\Id s+M_t+\int_0^tZ_s^\delta R_s\Id s, \quad t\in[0,t_0],
\end{align*}
for some $\delta\in[0,1)$. Pick some $\gamma\in(0,1)$ and set $B_t:=\int_0^tb_s\Id s$. 
Then the a priori bound
\begin{align}\label{inger1}
\EE\big[\sup_{s\le t_0} e^{-\gamma B_s}Z_s^\gamma\big]&<\infty,
\end{align}
implies the following inequality, for all $t\in[0,t_0]$, 
\begin{align*}
\EE\big[\sup_{s\le t} e^{-\gamma B_s}Z_s^\gamma\big]&\le c_{\gamma,\delta}\EE[Z_0^\gamma]
+c_{\gamma,\delta}\EE\Big[\Big(\int_0^te^{-(1-\delta)B_s}R_s\Id s
\Big)^{\frac{\gamma}{1-\delta}}\Big].
\end{align*}
Here the constant is given by
$c_{\gamma,\delta}^{1-\delta}=2^{({\gamma+\delta-1})\vee0}
\{(4\wedge\frac{1}{\gamma})\frac{\pi\gamma}{\sin(\pi\gamma)}+1\}$.
\end{lem}

\begin{proof}
Let $H_\bullet:=\int_0^\bullet Z_s^\delta R_s\Id s$ and
$L_\bullet:=\int_0^\bullet e^{-B_s}\Id M_s$. Then a pathwise application of Gronwall's lemma 
and a partial integration $\PP$-a.s. yields
\begin{align*}
e^{-B_t}Z_t&\le e^{-B_t}(Z_0+M_t+H_t)+\int_0^t(Z_0+M_s+H_s)b_se^{-B_s}\Id s
\\
&=Z_0+L_t+\int_0^te^{-B_s}Z_s^\delta R_s\Id s,\quad t\in[0,t_0].
\end{align*}
In particular, we $\PP$-a.s. have
\begin{align*}
\sup_{s\le t}e^{-B_s}Z_s&\le Z_0+\sup_{s\le t}L_s+\int_0^te^{-B_s}Z_s^\delta R_s\Id s,
\\
-\inf_{s\le t}L_s&=\sup_{s\le t}(-L_s)\le Z_0+\int_0^te^{-B_s}Z_s^\delta R_s\Id s,
\end{align*}
for all $t\in[0,t_0]$. Now the key step consists in applying the inequality
\begin{align*}
\EE\big[(\sup_{s\le t}L_s)^\gamma\big]\le(c_{\gamma,0}-1)\EE\big[(-\inf_{s\le t}L_s)^\gamma\big],
\quad t\in[0,t_0].
\end{align*}
This is a special case of an inequality for continuous local martingales starting at zero
due to Burkholder with an improved constant obtained in \cite[Prop.~1]{Scheutzow2013}.
In combination with $(x+y+z)^\gamma\le x^\gamma+y^\gamma+z^\gamma$, $x,y,z\ge0$,
the above estimates entail
\begin{align*}
\EE\big[\sup_{s\le t} e^{-\gamma B_s}Z_s^\gamma\big]
&\le c_{\gamma,0}\EE\Big[\Big(Z_0+\int_0^te^{-B_s}Z_s^\delta R_s\Id s\Big)^\gamma\Big]
\\
&\le c_{\gamma,0}\EE\Big[\big(\sup_{s\le t}e^{-\gamma B_s}Z_{s}^{\gamma}\big)^\delta
\Big(Z_0^{1-\delta}+\int_0^te^{-(1-\delta)B_s}R_s\Id s\Big)^\gamma\Big]
\\
&\le c_{\gamma,0}\EE\big[\sup_{s\le t}e^{-\gamma B_s}Z_{s}^{\gamma}\big]^{\delta}
\EE\Big[\Big(Z_0^{1-\delta}+\int_0^te^{-(1-\delta)B_s}R_s\Id s\Big)^{\frac{\gamma}{1-\delta}}
\Big]^{1-\delta}.
\end{align*}
Now the a priori bound \eqref{inger1} permits to get
\begin{align*}
\EE\big[\sup_{s\le t} e^{-\gamma B_s}Z_s^\gamma\big]
&\le c_{\gamma,0}^{\frac{1}{1-\delta}}
\EE\Big[\Big(Z_0^{1-\delta}+\int_0^te^{-(1-\delta)B_s}R_s\Id s\Big)^{\frac{\gamma}{1-\delta}}\Big],
\quad t\in[0,t_0].
\end{align*}
Finally, the desired bound follows from $(x+y)^r\le2^{(r-1)\vee0}(x^r+y^r)$, $x,y,r\ge0$.
\end{proof}


\subsection{Weighted BDG type estimates}\label{ssec-wBDG}

\noindent
In the following lemma and its proof we will employ the conventions of
Notation~\ref{not-Theta}, \ref{not-T}, and~\ref{not-normcirc}. The reader should keep in mind that
the weights $\Theta_s$ depend on $\ve$ and notice that the constants in \eqref{diane0} and
\eqref{diane1} are $\ve$-independent.

\begin{lem}\label{cor-spin1}
Let $V\in\cK_\pm(\RR^\nu)$ and assume that $\|\V{c}_{\V{x}}\|_\circ<\infty$, 
for every $\V{x}\in\RR^\nu$. Let 
$(\rho_t,\V{R}_t)_{t\in[0,t_0]}$ be an adapted $\LO(\RR^{1+\nu},\FHR)$-valued process 
with continuous paths. Let $p\in\NN$ and assume that $(\psi_t)_{t\in[0,t_0]}$ is a 
$\FHR$-valued semi-martingale whose paths belong $\PP$-a.s. to 
$C([0,t_0],\dom(\Id\Gamma(\omega)))$ and which $\PP$-a.s. satisfies
\begin{align}\nonumber
\psi_t&=\eta-\int_0^t\big(\wh{H}(\V{B}_s^{\V{q}})+V(\V{B}_s^{\V{q}})\big)\psi_s\Id s
+\int_0^ti\vp(\V{G}_{\V{B}_s^{\V{q}}})\psi_s\Id\V{B}_s
\\\label{SDE-inhom}
&\quad+\int_0^t\rho_s\Id s+\int_0^t\V{R}_s\Id\V{B}_s,\quad t\in[0,t_0],
\end{align}
where $(\V{q},\eta):\Omega\to\RR^\nu\times\dom(\Id\Gamma(\omega))$ is $\fF_0$-measurable
with $\|\eta\|_{\FHR}\in L^{p}(\PP)$. Let $\mu>0$, set
\begin{align}\label{def-fpmu}
f_{p,\mu}(s):=
\left\{\begin{array}{ll}
2p\|\V{G}_{\V{B}_s^{\V{q}}}\|_\circ^2+\mu\|(q,\vsigma\cdot\V{F})_{\V{B}_s^{\V{q}}}\|_\circ,
&\text{if $\vt=1+\Id\Gamma(\omega)$,}
\\
0,&\text{if $\vt=1$,}
\end{array}
\right.
\end{align}
and assume that $f_{p,\mu}(s)\le1/8$, for all $s\in[0,t_0]$. Abbreviate
\begin{align}\nonumber
R_{p}(s)&:=
\big\|(1+\Id\Gamma(\omega))^\mh\Theta_s\rho_s\big\|^2
+\big(p+\|\V{G}_{\V{B}_s^{\V{q}}}\|^2_{\mathfrak{k}^\nu}\big)\|\Theta_s\V{R}_s\|^2,
\\\label{def-bpmu}
b_{p,\mu}(s)&:=p\|\V{G}_{\V{B}_s^{\V{q}}}\|_\circ^2+(\tfrac{\mu}{2}
+\tfrac{1}{2\mu}){\|(q,\vsigma\cdot\V{F})_{\V{B}_s^{\V{q}}}\|_\circ}+4{\mho(\V{B}_s^{\V{q}})}
+{V_-(\V{B}_s^{\V{q}})},
\end{align}
and assume that
\begin{align*}
\EE\Big[\Big(\int_0^{t_0}R_p(s)\Id s\Big)^{\nf{p}{2}}\Big]<\infty.
\end{align*}
Then the following bounds hold for all $t\in[0,t_0]$,
\begin{align}\nonumber
\EE\big[&\sup_{s\le t}\{e^{-p\int_0^sb_{p,\mu}(r)\Id r}\|\Theta_s\psi_s\|^p\}\big]
\\\label{diane0}
&\le(7e^{t})^p\EE\big[\|\Theta_0\psi_0\|^p\big]+(28e^{t}p^\eh)^p\EE\Big[\Big(\int_0^t
e^{-2\int_0^sb_{p,\mu}(r)\Id r}R_{p}(s)\Id s\Big)^{\nf{p}{2}}\Big],
\\\nonumber
\EE\Big[&\Big(\int_0^te^{-2\int_0^sb_{p,\mu}(r)\Id r}
\|\Id\Gamma(\omega)^\eh\Theta_s\psi_s\|^2\Id s\Big)^{\nf{p}{2}}\Big]
\\\nonumber
&\le c_pe^{2pt}\EE\big[\|\Theta_0\psi_0\|^{p}\big]
+c_pe^{pt}\EE\Big[\Big(\int_0^te^{-2\int_0^sb_{p,\mu}(r)\Id r}
\|\V{G}_{\V{B}_s^{\V{q}}}\|_\circ^2\|\Theta_s\psi_s\|^2\Id s\Big)^{\nf{p}{2}}\Big]
\\\label{diane1}
&\quad
+c_pe^{2pt}\EE\Big[\Big(\int_0^te^{-2\int_0^sb_{p,\mu}(r)\Id r}R_p(s)\Id s\Big)^{\nf{p}{2}}\Big].
\end{align}
Here $c_p>0$ depends only on $p$.
\end{lem}

\begin{proof}
{\it Step 1.} If $f$ is one of the components of $\V{c}=(\V{G},q,\V{F})$,
then $\V{x}\mapsto\vp(f_{\V{x}})$, $\V{x}\mapsto\vp(\V{G}_{\V{x}})^2$, and,
hence, $\V{x}\mapsto\wh{H}(\V{x})$ are well-defined continuous maps from
$\RR^\nu$ into $\sB(\dom(\Id\Gamma(\omega)),\FHR)$ as a consequence of 
\eqref{rb-a}, \eqref{rb-ad}, Hyp.~\ref{hyp-G}, and Ex.~\ref{ex-isi}.
This shows that Lem.~\ref{lem-Ito-Theta} applies to the semi-martingale $\psi$.
It is also known that $\vp(f_{\V{x}})$ maps $\dom(\Id\Gamma(\omega))$ continuously
into $\dom(\Id\Gamma(\omega)^\eh)\subset\dom(\vp(f_{\V{x}}))$, which altogether shows that
all algebraic manipulations in the following steps are justified.

{\it Step 2.} 
Combining Lem.~\ref{lem-Ito-Theta} with \eqref{SDE-inhom} we $\PP$-a.s. find, for all $t\in[0,t_0]$,
\begin{align}\nonumber
\|\Theta_t\psi_t\|^2
&=\|\Theta_0\psi_0\|^2-\int_0^t2\Re\SPb{\psi_s}{
\Theta_s^2\,\tfrac{1}{2}\vp(\V{G}_{\V{B}_s^{\V{q}}})^2\,\psi_s}\Id s
\\\nonumber
&\quad-\int_0^t2\Re\SPb{\psi_s}{\Theta_s^2\big(\Id\Gamma(\omega)
-\tfrac{i}{2}\vp(q_{\V{B}^{\V{q}}_s})-\vsigma\cdot\vp(\V{F}_{\V{B}^{\V{q}}_s})
-\tfrac{\dot\Theta_s}{\Theta_s}+V(\V{B}_s^{\V{q}})\big)\psi_s}\Id s
\\\nonumber
&\quad+\int_0^t
2\Re\SPb{\psi_s}{\Theta_s^2\,i\vp(\V{G}_{\V{B}^{\V{q}}_s})\psi_s}\Id \V{B}_s+\int_0^t
\big\|\Theta_s\vp(\V{G}_{\V{B}^{\V{q}}_s})\psi_s\big\|^2\Id s
\\\nonumber
&\quad+\int_0^t2\Re\SPn{\Theta_s\psi_s}{\Theta_s\rho_s}\Id s
+\int_0^t2\Re\SPn{\Theta_s\psi_s}{\Theta_s\V{R}_s}\Id\V{B}_s
\\\nonumber
&\quad
+\int_0^t2\Re\SPn{\Theta_s\V{R}_s}{\Theta_s\,i\vp(\V{G}_{\V{B}^{\V{q}}_s})\psi_s}\Id s
+\int_0^t\|\Theta_s\V{R}_s\|^2\Id s.
\end{align}
Now we commute $\Theta_s^2$ with one of the factors $\vp(\V{G}_{\V{B}^{\V{q}}_s})$
in the integral in the first line, take the cancellation with the second integral
in the third line into account, and observe that 
\begin{align*}
\Re\SPb{\psi_s}{[\Theta_s^2,\vp(\V{G}_{\V{B}^{\V{q}}_s})]\vp(\V{G}_{\V{B}_s^{\V{q}}})\psi}
&=\SPb{\vt^\eh\Theta_s\psi_s}{T_1(s)\vt^\eh\Theta_s\psi_s},
\\
2\Re\SPb{\psi_s}{\Theta_s[\Theta_s,\vsigma\cdot\vp(\V{F}_{\V{B}^{\V{q}}_s})]\psi_s}
&=\SPb{\psi_s}{\big[\Theta_s\,,\,[\Theta_s,\vsigma\cdot\vp(\V{F}_{\V{B}^{\V{q}}_{s}})]\big]\psi_s}.
\end{align*}
Moreover, we commute one factor $\Theta_s$
with $\vp(\V{G}_{\V{B}^{\V{q}}_s})$ in the first $\Id\V{B}$-integral and use that
$\Re\SPn{\psi_s}{\Theta_si\vp(\V{G}_{\V{B}^{\V{q}}_s})\Theta_s\psi_s}=0$. 
The same observation can be used in the second line 
with $\tfrac{i}{2}\vp(q_{\V{B}^{\V{q}}_s})$ in place of $i\vp(\V{G}_{\V{B}^{\V{q}}_s})$.
In this way we $\PP$-a.s. arrive at 
\begin{align}\label{spin80}
\|\Theta_t\psi_{t}\|^2&=\|\Theta_0\psi_{0}\|^2+\int_0^t\sJ(s)\Id s+\int_0^t\V{J}(s)\Id \V{B}_s,
\end{align}
for all $t\in[0,t_0]$, with
\begin{align}\nonumber
\sJ(s)&:=-2\SPb{\Theta_s\psi_s}{
\big(\Id\Gamma(\omega)-\vsigma\cdot\vp(\V{F}_{\V{B}^{\V{q}}_s})-\tfrac{\dot\Theta_s}{\Theta_s}
+V(\V{B}_s^{\V{q}})\big)\Theta_s\psi_s}
\\\nonumber
&\quad-\SPb{\vt^\eh\Theta_s\psi_s}{T_{1,s}\vt^\eh\Theta_s\psi_s}+
2\Re\SPb{\Theta_s\psi_s}{T_{2,s}\vt^\eh\Theta_s\psi_{s}}
\\\nonumber
&\quad+2\Re\SPn{\Theta_s\psi_s}{\Theta_s\rho_s}+\|\Theta_s\V{R}_s\|^2
\\\nonumber
&\quad
+2\Re\SPn{\Theta_s\V{R}_s}{i\vp(\V{G}_{\V{B}^{\V{q}}_s})\Theta_s\psi_s}
+\Re\SPn{\Theta_s\V{R}_s}{i\V{T}_s^+\vt^\eh\Theta_s\psi_s},
\\\label{def-vecJ}
\V{J}(s)&:=\Re\SPb{\Theta_s\psi_s}{i\V{T}_s^+\vt^\eh\Theta_s\psi_s+2\Theta_s\V{R}_s}.
\end{align}
Next, we apply the bounds
\begin{align}\label{flora}
\vsigma\cdot\vp(\V{F}_{\V{B}^{\V{q}}_s})+\tfrac{\dot\Theta_s}{\Theta_s}
&\le(1/4+1/2)\Id\Gamma(\omega)+4\mho(\V{B}_s^{\V{q}})+1/2,
\\\label{flora2}
\|\vp(\V{G}_{\V{B}_s^{\V{q}}})(1+\Id\Gamma(\omega))^\mh\|&
\le2^\eh\|\V{G}_{\V{B}_s^{\V{q}}}\|_{\mathfrak{k}}.
\end{align}
The first one of them
follows from \eqref{qfb-vp} (with $\vk=\omega/4$) and \eqref{gabi1} (resp. \eqref{gabi2}), 
and the second one is implied by \eqref{rb-a} and \eqref{rb-ad}. We thus obtain
\begin{align}\nonumber
\sJ(s)&\le-\big(\tfrac{1}{2}-2\mu'\big)\|\Id\Gamma(\omega)^\eh\Theta_s\psi_s\|^2
\\\nonumber
&\quad+\big(\|T_{1,s}\|+\|\V{T}_s^+\|^2+\mu\|T_{2,s}\|\big)\|\vt^\eh\Theta_s\psi_s\|^2
\\\nonumber
&\quad+\Big(8\mho(\V{B}_s^{\V{q}})+1-2V(\V{B}_s^{\V{q}})
+\frac{\|T_{2,s}\|}{\mu}+2\mu'\Big)\|\Theta_s\psi_s\|^2
\\\label{diane99}
&\quad+\frac{1}{\mu'}\big\|(1+\Id\Gamma(\omega))^\mh\Theta_s\rho_s\big\|^2
+\Big(\frac{5}{4}+\frac{2\|\V{G}_{\V{B}_s^{\V{q}}}\|^2_{\mathfrak{k}}}{\mu'}\Big)
\|\Theta_s\V{R}_s\|^2,
\end{align}
for all $s\in[0,t_0]$ and $\mu,\mu'>0$. We further have
\begin{align}\label{diane99b}
\tfrac{1}{2}\V{J}(s)^2&\le2\|\Theta_s\psi_s\|^2\|\V{T}_s^+\|^2\|\vt^\eh\Theta_s\psi_s\|^2
+4\|\Theta_s\psi_s\|^2\|\Theta_s\V{R}_s\|^2.
\end{align}

{\em Step 3.} Let $p\ge 2$. 
In view of \eqref{spin80} an application of It\={o}'s formula $\PP$-a.s. yields
\begin{align}\nonumber
\|\Theta_t\psi_t\|^{2p}&=\|\Theta_0\psi_0\|^{2p}+{p}\int_0^t\|\Theta_s\psi_s\|^{2p-2}\sJ(s)\Id s
\\\label{robert2000}
&\quad+p\int_0^t\|\Theta_s\psi_s\|^{2p-2}\V{J}(s)\Id \V{B}_s+\frac{p(p-1)}{2}\int_0^t
\|\Theta_s\psi_s\|^{2p-4}\V{J}(s)^2\Id s,
\end{align}
for all $t\in[0,t_0]$. Using \eqref{diane99} with $\mu'=1/8$, \eqref{diane99b}, 
and Lem.~\ref{lem-bd-T} together with \eqref{def-normcirc}, we $\PP$-a.s. obtain
\begin{align}\nonumber
\|\Theta_t\psi_t\|^{2p}
&\le\|\Theta_t\psi_t\|^{2p}+p\int_0^t\big(\tfrac{1}{4}-f_{p,\mu}(s)\big)\|\Theta_s\psi_s\|^{2p-2}
\|\Id\Gamma(\omega)^\eh\Theta_s\psi_s\|^2\Id s
\\\nonumber
&\le\|\Theta_0\psi_0\|^{2p}+
2p\!\int_0^t(b_{p,\mu}(s)+1)\|\Theta_s\psi_s\|^{2p}\Id s
+p\!\int_0^t\|\Theta_s\psi_s\|^{2p-2}\V{J}(s)\Id \V{B}_s
\\\nonumber
&\quad+8p\int_0^t\|\Theta_s\psi_s\|^{2p-2}
\big\|(1+\Id\Gamma(\omega))^\mh\Theta_s\rho_s\big\|^2\Id s
\\\label{spin81n}
&\quad+p\int_0^t\|\Theta_s\psi_s\|^{2p-2}
\Big(\frac{5}{4}+4(p-1)+{16\|\V{G}_{\V{B}_s^{\V{q}}}\|^2_{\mathfrak{k}}}\Big)
\|\Theta_s\V{R}_s\|^2\Id s,
\end{align}
for all $t\in[0,t_0]$ and $\mu>0$. Taking \eqref{spin80} and \eqref{diane99} into account,
we see that \eqref{spin81n} is actually valid for $p=1$ as well.

{\em Step 4.}
Applying \eqref{diane99}--\eqref{robert2000} with $\Theta_s:=\vt:=\id$ and $\mu'=1/8$, 
we $\PP$-a.s. obtain
\begin{align}\nonumber
\|\psi_t\|^{2p}+\frac{p}{4}\int_0^t\|\psi_s\|^{2p-2}\|\Id\Gamma(\omega)^\eh&\psi_s\|^2\Id s
\le\|\eta\|^{2p}+16p\int_0^t\|\psi_s\|^{2p-2}R_p(s)\Id s
\\\label{pernille1}
&+2p\int_0^t(4\mho(\V{B}_s^{\V{q}})+1+V_-(\V{B}_s^{\V{q}}))\|\psi_s\|^{2p}\Id s,
\end{align}
for all $t\in[0,t_0]$ and $p\in\{1\}\cup[2,\infty)$. 
Together with Gronwall's lemma and an integration by parts \eqref{pernille1} $\PP$-a.s. yields
\begin{align*}
\sup_{s\le t}\{e^{-2pB_s}\|\psi_s\|^{2p}\}&\le\|\eta\|^{2p}
+16p\int_0^t\{e^{-(2p-2)B_s}\|\psi_s\|^{2p-2}\}e^{-2B_s}R_p(s)\Id s
\\
&\le\sup_{s\le t}\{e^{-2pB_s}\|\psi_s\|^{2p}\}^{1-\nf{1}{p}}\Big(\|\eta\|^2
+16p\int_0^te^{-2B_s}R_p(s)\Id s\Big),
\end{align*}
with $B_{\bullet}:=\int_0^\bullet(4\mho(\V{B}_s^{\V{q}})+1+V_-(\V{B}_s^{\V{q}}))\Id s$. 
Hence, we $\PP$-a.s. have
\begin{align}
\sup_{s\le t}\{e^{-pB_s}\|\psi_s\|^{p}\}&\label{pernille1a}
\le 2^{\nf{p}{2}}\Big\{\|\eta\|^{p}+4^p
\Big(p\int_0^te^{-2B_s}R_p(s)\Id s\Big)^{\nf{p}{2}}\Big\},\quad t\in[0,t_0].
\end{align}
By the assumptions on $\eta$, $\rho$, and $\V{R}$, and since $\Theta_s$ is uniformly bounded by
some $\ve$-dependent constant, \eqref{pernille1a} implies that the
expectation on the left hand side of \eqref{diane0} is finite for all $t\in[0,t_0]$; 
this a priori information is needed in the next step. 
 
{\em Step 5.}
By the remarks in the first step, $(\|\Theta_s\psi_s\|^{2p-2}\V{J}(s))_{s\in[0,t_0]}$
is an adapted continuous $\RR^\nu$-valued process, whence its
integral with respect to the Brownian motion in \eqref{spin81n} is a continuous local martingale
starting from zero.
Thanks to this and the remarks in Step 4, we may apply Lem.~\ref{lem-stoch-Gronwall} with 
$\gamma=1/2$, $\delta=1-1/p$, $Z_s=\|\Theta_s\psi_s\|^{2p}$, $b_s=2p(b_{p,\mu}(s)+1)$,
and $R_s=16pR_p(s)$ to infer \eqref{diane0} from \eqref{spin81n}.
(We also use that $c_{\eh,1-\nf{1}{p}}^{\nf{1}{p}}=2^\eh(\pi+1)<7$.)

{\em Step 6.} 
Finally, we verify \eqref{diane1}. 
We abbreviate $B_{p,t}:=t+\int_0^tb_{p,\mu}(s)\Id s$. Then \eqref{spin80} and
It\={o}'s formula $\PP$-a.s. imply
\begin{align*}
e^{-2B_{p,t}}\|\Theta_t\psi_t\|^2
&=\|\Theta_0\psi\|^2+\int_0^te^{-2B_{p,s}}\sJ(s)\Id s
\\
&\quad-2\int_0^t(b_{p,\mu}(s)+1)e^{-2B_{p,s}}\|\Theta_s\psi_s\|^2\Id s
+\int_0^te^{-2B_{p,s}}\V{J}(s)\Id\V{B}_s,
\end{align*} 
for all $t\in[0,t_0]$.
Next, we employ \eqref{diane99} with $\mu':=1/8$ and taking into account that
$\|T_{1,s}\|+\|\V{T}_s^+\|^2+\mu\|T_{2,s}\|\le f_{1,p}(s)\le1/8$ in the case $\vt=1+\Id\Gamma(\omega)$, by Lem.~\ref{lem-bd-T} and our assumption on $f_{p,\mu}$.
In all cases this $\PP$-a.s. yields
\begin{align}\nonumber
e^{-2B_{p,t}}&\|\Theta_t\psi_t\|^2
+\frac{1}{4}\int_0^te^{-2B_{p,s}}\|\Id\Gamma(\omega)^\eh\Theta_s\psi_s\|^2\Id s
\\\label{lola1}
&\le\|\Theta_0\psi_0\|^2+16\int_0^te^{-2B_{p,s}}R_{1}(s)\Id s
+\int_0^te^{-2B_{p,s}}\V{J}(s)\Id\V{B}_s,
\end{align}
for all $t\in[0,t_0]$.
Let us consider the previous bound in the case $\Theta_s=\id$ for the moment.
Then $\V{J}=\V{0}$ and, hence, the assumptions on $\eta=\psi_0$, $\rho$, and $\V{R}$ entail 
$\int_0^{t_0}e^{-2B_{p,s}}\|\Id\Gamma(\omega)^\eh\psi_s\|^2\Id s\in L^{\nf{p}{2}}(\PP)$.
Since $\|\Theta_s\|$ is uniformly bounded by some $\ve$-dependent
constant this ensures a priori that the left hand side of \eqref{diane1} is finite.

Returning to general $\Theta_s$ we raise \eqref{lola1} to the power
$p/2$ and employ the bound
$(a_1+\dots+a_n)^{q}\le n^{(q-1)\vee0}(a_1^q+\dots+a_n^q)$ afterwards, which leads to 
\begin{align}\nonumber
\EE\Big[\Big(\int_0^t&e^{-2B_{p,s}}
\|\Id\Gamma(\omega)^\eh\Theta_s\psi_s\|^2\Id s\Big)^{\nf{p}{2}}\Big]
\le 4^p\EE\big[\|\Theta_0\psi_0\|^{p}\big]
\\\label{diane2}
&+4^{2p}\EE\Big[\Big(\int_0^te^{-2B_{p,s}}R_1(s)\Id s\Big)^{\nf{p}{2}}\Big]
+4^{p}\EE\Big[\Big|\int_0^te^{-2B_{p,s}}\V{J}(s)\Id\V{B}_s\Big|^{\nf{p}{2}}\Big],
\end{align}
for all $t\in[0,t_0]$. Furthermore, there exists $c_p>0$ such that the
following special case of the Burkholder-Davis-Gundy inequality holds for every $t\in[0,t_0]$,
\begin{align}\label{Lpbd}
\EE\Big[\sup_{r\le t}\Big|\int_0^re^{-2B_{p,s}}\V{J}(s)\Id\V{B}_s\Big|^{\nf{p}{2}}\Big]
&\le c_p\EE\Big[\Big(\int_0^t|e^{-2B_{p,s}}\V{J}(s)|^2\Id s\Big)^{\nf{p}{4}}\Big];
\end{align}
see, e.g., \cite[Thm.~4.36]{daPrZa1992}.
In particular, the last term in \eqref{diane2} can be estimated as
\begin{align*}
\EE\Big[\Big|\int_0^te^{-2B_{p,s}}\V{J}(s)&\Id\V{B}_s\Big|^{\nf{p}{2}}\Big]
\le c_p\EE\big[\sup_{s\le t}\{e^{-pB_{p,s}}\|\Theta_s\psi_s\|^{p}\}\big]^\eh
\\
&\cdot\EE\Big[\Big(\int_0^te^{-2B_{p,s}}\big\|i\V{T}_s^+\vt^\eh\Theta_s\psi_s
+2\Theta_s\V{R}_s\big\|^2\Id s\Big)^{\nf{p}{2}}\Big]^\eh.
\end{align*}
Using that, by Lem.~\ref{lem-bd-T} and our assumption on $f_{p,\mu}$, 
$\vt=1+\Id\Gamma(\omega)$ implies $\|\V{T}_s^+\|\le1$, 
we readily infer from the previous inequality that, in all cases and for all $t\in[0,t_0]$,
\begin{align}\nonumber
\EE\Big[\Big|\int_0^te^{-2B_{p,s}}\V{J}(s)\Id\V{B}_s\Big|^{\nf{p}{2}}\Big]
\le \frac{1}{2\cdot 4^p}\EE\Big[\Big(\int_0^t
e^{-2B_{p,s}}\|\Id\Gamma(\omega)^\eh\Theta_s\psi_s\|^2\Id s\Big)^{\nf{p}{2}}\Big]\:
\\\nonumber
+c_p'\EE\big[\sup_{s\le t}\{e^{-pB_{p,s}}\|\Theta_s\psi_s\|^{p}\}\big]
+c_p'\EE\Big[\Big(\int_0^te^{-2B_{p,s}}\|\V{T}^+_s\|^2\|\Theta_s\psi_s\|^2\Id s\Big)^{\nf{p}{2}}\Big]\:
\\\label{diane4}
+c_p'\EE\Big[\Big(\int_0^te^{-2B_{p,s}}R_1(s)\Id s\Big)^{\nf{p}{2}}\Big],
\end{align}
for some $c_p'>0$ depending only on $p$.
Inserting \eqref{diane4} into \eqref{diane2}, solving the resulting inequality for 
the left hand side of \eqref{diane2} (which is finite as we observed above),
and applying \eqref{diane0} afterwards, we obtain \eqref{diane1}.
\end{proof}

In the rest of this subsection we apply Lem.~\ref{cor-spin1} to the process given by
Thm.~\ref{thm-WW} extending the bounds \eqref{diane0} and \eqref{diane1} to unbounded
weight functions ($\ve=0$) at the same time.
The reader might want to recall the definition of the weights $\Theta_{\ve,s}^{(\alpha)}$
and $\Xi_{\ve,s}^{(\delta)}$ in the beginning of Subsect.~\ref{ssec-def-weights} before
reading the next lemmas.

\begin{lem}\label{lem-spin4}
{\rm(1)} Let $t_0>0$, $|\alpha|\ge1/2$, and assume that the coefficient vector 
satisfies $\|\V{c}\|_{\circ,\infty}:=\sup_{\V{x}\in\RR^\nu}\|\V{c}_{\V{x}}\|_\circ<\infty$ 
with $\|\cdot\|_\circ$ given by the first line in \eqref{def-normcirc}.
Then, $\PP$-a.s., $\WW{t}{0}[\V{B}^{\V{q}}]$ maps $\dom({\Theta^{(\alpha)}_{0,0}})$
into $\dom(\Theta_{0,t}^{(\alpha)})$, for all $t\in[0,t_0]$, and
\begin{align}\label{pol-w}
\EE\big[\sup_{s\le t}\|\Theta_{0,s}^{(\alpha)}\WW{s}{0}[\V{B}^{\V{q}}]\eta\|^{p}\big]
&\le\big(7e^{(p\|\V{c}\|_{\circ,\infty}^2+4\|\mho\|_\infty+2)t}\big)^p
\EE\big[\|\Theta_{0,0}^{(\alpha)}\eta\|^{p}\big],
\end{align}
for all $p\in\NN$, $t\in[0,t_0]$, and $\fF_0$-measurable 
$(\V{q},\eta):\Omega\to\RR^\nu\times{\dom}(\Theta^{(\alpha)}_{0,0})$ 
with $\|\Theta_{0,0}^{(\alpha)}\eta\|\in L^{p}(\PP)$.

\smallskip

\noindent{\rm(2)}
Let $t_0>0$, $0<|\delta|\le1$, $p\in\NN$, and assume that 
$\|(q,\V{F})\|_{\circ,\infty}<\infty$ and $p\|\V{G}\|_{\circ,\infty}^2\le1/32$,
where $\|\cdot\|_\circ$ is given by the second line in \eqref{def-normcirc}. 
Then, $\PP$-a.s., $\WW{t}{0}[\V{B}^{\V{q}}]$ maps $\dom({\Xi^{(\delta)}_{0,0}})$
into $\dom(\Xi_{0,t}^{(\delta)})$, for all $t\in[0,t_0]$, and
\begin{align}\label{exp-w}
\EE\big[\sup_{s\le t}\|\Xi_{0,s}^{(\delta)}\WW{s}{0}\eta\|^{p}\big]
&\le\big(7e^{(8\|\V{c}\|_{\circ,\infty}^2+4\|\mho\|_\infty+2)t}\big)^p 
\EE\big[\|\Xi^{(\delta)}_{0,0}\eta\|^{p}\big],
\end{align}
for all $t\in[0,t_0]$ and $\fF_0$-measurable 
$(\V{q},\eta):\Omega\to\RR^\nu\times{\dom}(\Xi^{(\delta)}_{0,0})$ 
with $\|\Xi_{0,0}^{(\delta)}\eta\|\in L^{p}(\PP)$.
\end{lem}

\begin{rem}\label{rem-spin4}
Before we prove the lemma, let us clarify the purpose of the inverse weights in the polynomial 
case. In the situation of Lem.~\ref{lem-spin4} assume in addition that $\alpha\le-1/2$.
By Rem.~\ref{rem-weights} we have $\Theta_{0,t_0}^{(\alpha)}={\Theta_{0,0}^{(-\alpha)}}^{-1}$
and ${\Theta_{0,0}^{(\alpha)}}^{-1}={\Theta_{0,t_0}^{(-\alpha)}}$.
Substituting $\eta={\Theta_{0,0}^{(\alpha)}}^{-1}\zeta$ in \eqref{pol-w} we thus arrive at
\begin{align}\label{pol-w-inv}
\EE\big[\|{\Theta_{0,0}^{(-\alpha)}}^{-1}\WW{s}{0}[\V{B}^{\V{q}}]{\Theta_{0,t_0}^{(-\alpha)}}\zeta\|^{p}\big]
&\le\big(7e^{(p\|\V{c}\|_{\circ,\infty}^2+4\|\mho\|_\infty+2)t}\big)^p
\EE\big[\|\zeta\|^{p}\big],
\end{align}
for all $p\in\NN$ and $\fF_0$-measurable 
$(\V{q},\zeta):\Omega\to\RR^\nu\times{\dom}({\Theta_{0,t_0}^{(-\alpha)}})$ 
with $\|\zeta\|\in L^{p}(\PP)$. That is, the purpose of the inverse weights is to obtain bounds
where an unbounded weight stands to the right of $\WW{s}{0}[\V{B}^{\V{q}}]$.
\end{rem}

\begin{proof}
By virtue of Thm.~\ref{thm-WW} we know that Lem.~\ref{cor-spin1} applies to the semi-martingale
$(\WW{t}{0}[\V{B}^{\V{q}}]\eta)_{t\ge0}$, if we set $V$, $\rho$, and $\V{R}$ 
equal to zero in the statement of the lemma. Below we put $\eta_m:=1_{\{\|\eta\|\le m\}}\eta$ and
$\eta_{n,m}:=(1+n^{-1}\Id\Gamma(\omega))^{-1}\eta_m$, for
$n,m\in\NN$, so that each $\eta_{n,m}$ is a $\fF_0$-measurable 
$\dom(\Id\Gamma(\omega))$-valued initial condition. 

(1): For all $\ve\in(0,1]$ and $n,m\in\NN$, the bound 
\begin{align}\label{pol-w-eps}
\EE\big[\sup_{s\le t}\|\Theta_{\ve,s}^{(\alpha)}\WW{s}{0}[\V{B}^{\V{q}}]\eta_{n,m}\|^{p}\big]
&\le\big(7e^{(p\|\V{c}\|_{\circ,\infty}^2+4\|\mho\|_\infty+2)t}\big)^p
\EE\big[\|\Theta_{0,0}^{(\alpha)}\eta_{n,m}\|^{p}\big]
\end{align}
follows easily from 
\eqref{def-bpmu} and \eqref{diane0} with $\mu=1$. Thanks to \eqref{bd-WW}, we further know that
\begin{equation}\label{maria0}
\sup_{s\le t}\|\Theta_{\ve,s}^{(\alpha)}\WW{s}{0}[\V{B}^{\V{q}}]\eta_{n,m}\|
\le c_{\alpha,\ve} e^{\|\mho\|_\infty t}\|\eta_{m}\|\in L^{p}(\PP),\;\;\;\ve\in(0,1],\,m,n\in\NN,
\end{equation}
and the dominated convergence theorem implies the bounds
\begin{align}\label{pol-w-m}
\EE\big[\sup_{s\le t}\|\Theta_{\ve,s}^{(\alpha)}\WW{s}{0}[\V{B}^{\V{q}}]\eta_m\|^{2p}\big]
&\le\big(7e^{(p\|\V{c}\|_{\circ,\infty}^2+4\|\mho\|_\infty+2)t}\big)^p 
\EE\big[\|\Theta_{\ve,0}^{(\alpha)}\eta_m\|^{2p}\big],
\end{align}
for all $t\in[0,t_0]$, $\ve\in(0,1]$, and $m\in\NN$.

Consider the case $\alpha>0$. Then 
$\|\Theta_{\ve,0}^{(\alpha)}\eta_m\|\le\|\Theta_{0,0}^{(\alpha)}\eta\|\in L^{p}(\PP)$.
Employing the monotone convergence theorem we may thus pass to the limits $m\to\infty$ 
and $\ve\downarrow0$ on the left hand side of \eqref{pol-w-m}, which
proves \eqref{pol-w} in this case.

In the case $\alpha<0$, where 
$\Theta_{\ve,s}^{(\alpha)}\le2^{|\alpha|}$
we apply the dominated convergence theorem first, to pass to the limit $\ve\downarrow0$ on both
sides of \eqref{pol-w-m}. After that we estimate 
$\|\Theta_{0,0}^{(\alpha)}\eta_m\|\le\|\Theta_{0,0}^{(\alpha)}\eta\|\in L^{p}(\PP)$ on the right hand
side and, finally, we let $m$ go to infinity on the left hand side using the monotone convergence
theorem.

(2) can be proved analogously to (1), distinguishing the cases $\delta>0$ and $\delta<0$.
The only difference is that we choose $\mu$ in \eqref{def-bpmu}
such that $\mu\|(q,\V{F})\|_{\circ,\infty}=1/16$, for non-zero $(q,\V{F})$, 
when we apply \eqref{diane0}.
For then the condition $f_{p,\mu}\le1/8$ in Lem.~\ref{cor-spin1} is satisfied, since we are assuming
that $2p\|\V{G}\|_{\circ,\infty}^2\le1/16$.
\end{proof}

\begin{lem}\label{lem-spin42}
{\rm(1)} Let $|\alpha|\ge1/2$, $t_0>0$, $p\in\NN$, and let $\V{c}$, $\V{q}$, and $\eta$ 
fulfill the conditions in Lem.~\ref{lem-spin4}(1). Then, for $\Id t\otimes\PP$-a.e.
$(t,\vgamma)\in[0,t_0]\times\Omega$, 
$\WW{t}{0}[\V{B}^{\V{q}}](\vgamma)$ maps $\dom({\Theta^{(\alpha)}_{0,0}})$
into $\dom(\Id\Gamma(\omega)^\eh\Theta_{0,t}^{(\alpha)})$, and
\begin{align}\nonumber
\EE\Big[&\Big(\int_0^t\big\|\Id\Gamma(\omega)^\eh\Theta_{0,s}^{(\alpha)}
\WW{s}{0}[\V{B}^{\V{q}}]\eta\big\|^2\Id s\Big)^{\nf{p}{2}}\Big]
\\\label{cleo1}
&\le c_pe^{c(p\|\V{c}\|_{\circ,\infty}^2+\|\mho\|_\infty+1)pt}
\big(1+t^{\nf{p}{2}}\|\V{G}\|_{\circ,\infty}^p\big)\EE\big[\|\Theta_{0,0}^{(\alpha)}\eta\|^p\big],
\quad t\in[0,t_0],
\end{align}
for some universal constant $c>0$ and some $c_p>0$ depending only on $p$.

\smallskip

\noindent{\rm(2)} Let $0<|\delta|\le1$, $t_0>0$, $p\in\NN$, and let $\V{c}$, $\V{q}$, and $\eta$ 
fulfill the conditions in Lem.~\ref{lem-spin4}(2). Then, for $\Id t\otimes\PP$-a.e.
$(t,\vgamma)\in[0,t_0]\times\Omega$, $\WW{t}{0}[\V{B}^{\V{q}}](\vgamma)$ 
maps $\dom({\Xi^{(\delta)}_{0,0}})$ into $\dom(\Id\Gamma(\omega)^\eh\Xi_{0,t}^{(\delta)})$, and
\eqref{cleo1} holds true with $\Theta_{0,s}^{(\alpha)}$ replaced by $\Xi_{0,s}^{(\delta)}$.
\end{lem}

\begin{proof}
(1): Let $\eta_{n,m}$ and $\eta_m$
be defined as in the proof of Lem.~\ref{lem-spin4}. Then we can apply
\eqref{diane1} with $\psi_s=\WW{s}{0}[\V{B}^{\V{q}}]\eta_{n,m}$, $(\rho,\V{R})=0$, 
and \eqref{pol-w-eps} afterwards. In this way we easily arrive at
\begin{align}\nonumber
\EE\Big[&\Big(\int_0^t\big\|\Id\Gamma(\omega)^\eh\Theta_{\ve,s}^{(\alpha)}
\WW{s}{0}[\V{B}^{\V{q}}]\eta_{n,m}\big\|^2\Id s\Big)^{\nf{p}{2}}\Big]
\\\label{cleo2}
&\le c_pe^{c(p\|\V{c}\|_{\circ,\infty}^2+\|\mho\|_\infty+1)pt}
\big(1+t^{\nf{p}{2}}\|\V{G}\|_{\circ,\infty}^p\big)\EE\big[\|\Theta_{\ve,0}^{(\alpha)}\eta_{n,m}\|^p\big],
\end{align}
for all $t\in[0,t_0]$, $m,n,\in\NN$, and $\ve\in(0,1]$. 
To remove the regularization parameter $n$, we start by considering the expressions
\begin{align}\nonumber
\EE\Big[\Big(\int_0^t\big\|\Id\Gamma(\omega)^\eh(1+\ell^{-1}\Id\Gamma(\omega))^\mh
\Theta_{\ve,s}^{(\alpha)}\WW{s}{0}[\V{B}^{\V{q}}]&\eta_{n,m}\big\|^2\Id s\Big)^{\nf{p}{2}}\Big]
\\\label{cleo3}
&\le c_{p,t}\EE\big[\|\Theta_{\ve,0}^{(\alpha)}\eta_m\|^p\big],
\end{align}
where $c_{p,t}$ denotes the constant on the right hand side of \eqref{cleo2}.
Because of the bounds
$$
\big\|\Id\Gamma(\omega)^\eh(1+\ell^{-1}\Id\Gamma(\omega))^\mh
\Theta_{\ve,s}^{(\alpha)}\WW{s}{0}[\V{B}^{\V{q}}]\eta_{n}\big\|
\le c_{\alpha,\ve,\ell} e^{\|\mho\|_\infty t_0}\|\eta_m\|\in L^p(\PP),
$$
that hold pointwise on $[0,t_0]\times\Omega$, for $\ve\in(0,1]$ and $m,n\in\NN$, 
we may apply the dominated convergence
theorem to obtain \eqref{cleo3} with $\eta_{n,m}$ replaced by $\eta_m$ on its left hand side.
After that we can pass to the limits $\ve\downarrow0$ and $m\to\infty$ as in the proof of
Lem.~\ref{lem-spin4} distinguishing the cases $\alpha>0$ and $\alpha<0$. This yields
\eqref{cleo3} with $\eta_{n,m}$ and $\eta_m$ replaced by $\eta$ and for $\ve=0$.
Finally, we let $\ell$ go to infinity with the help of the monotone convergence theorem.

The proof of (2) is analogous.
\end{proof}


\subsection{An elementary algebraic identity}\label{ssec-i1i5}

\noindent
Before studying the behavior of $\WW{}{V}[\V{B}^{\V{q}}]$ under changes of $\V{q}$ and
perturbations of the coefficient vector, we derive the following algebraic lemma which is needed
in the proof of Lem.~\ref{lem-diff1} below.

\begin{lem}\label{lem-i1i5}
Let $\sK$ be a Hilbert space, $\V{v}=(v_1,\ldots,v_\nu)$ and 
$\tilde{\V{v}}=(\tilde{v}_1,\ldots,\tilde{v}_\nu)$ be tuples of self-adjoint operators acting in $\sK$,
and $\Theta=\Theta^*\in\LO(\sK)$ be such that $\Theta\dom(w)\subset\dom(w)$ and
$\Theta\dom(w^2)\subset\dom(w^2)$, if the operator
$w$ is any of the components of $\V{v}$ or $\tilde{\V{v}}$. Suppose that the vectors
$\eta$ and $\tilde{\eta}$ are elements of the domain
$\cC:=\bigcap_{j=1}^\nu\dom(v_j^2)\cap\dom(v_j\tilde{v}_j)
\cap\dom(\tilde{v}_jv_j)\cap\dom(\tilde{v}_j^2)$.  Then the following identity holds,
\begin{align}\nonumber
A&:=
-\Re\SPn{\eta-\tilde{\eta}}{\Theta^2\V{v}^2\eta-\Theta^2\tilde{\V{v}}^2\tilde\eta}
+\|\Theta\V{v}\eta-\Theta\tilde{\V{v}}\tilde\eta\|^2
\\\nonumber
&=\|\Theta(\V{v}-\tilde{\V{v}})\tilde\eta\|^2
+\tfrac{1}{2}\SPb{\eta-\tilde\eta}{\big[[\V{v},\Theta^2],\V{v}\big](\eta-\tilde\eta)}
\\\nonumber
&\quad-2\Re\SPb{\{\Theta[\V{v},\Theta]+[\V{v},\Theta]\Theta\}
(\eta-\tilde\eta)}{(\V{v}-\tilde{\V{v}})\tilde\eta}
\\\nonumber
&\quad+\Re\SPb{(\V{v}-\tilde{\V{v}})\Theta(\eta-\tilde\eta)}{\Theta(\V{v}-\tilde{\V{v}})\tilde\eta}
\\\nonumber
&\quad+\Re\SPb{[\V{v}-\tilde{\V{v}},\Theta]\Theta(\eta-\tilde\eta)}{(\V{v}-\tilde{\V{v}})\tilde\eta}
\\\label{i1i5}
&\quad-\Re\SPn{\Theta(\eta-\tilde\eta)}{\Theta[\V{v},\tilde{\V{v}}]\tilde\eta}.
\end{align}
\end{lem}

\begin{proof}
First we write 
$\Theta^2\V{v}^2=\V{v}\Theta^2\V{v}+[\Theta^2,\V{v}]\V{v}$ and analogously for $\tilde{\V{v}}$
and take the cancelations between the two terms in the first line
of \eqref{i1i5} into account to get
\begin{align}\nonumber
A&=\Re\SPb{\eta}{\V{v}\Theta^2(\V{v}-\tilde{\V{v}})\tilde\eta}
-\Re\SPb{\eta}{(\V{v}-\tilde{\V{v}})\Theta^2\tilde{\V{v}}\tilde\eta}
\\\label{isa1}
&\quad
-\Re\SPb{\eta-\tilde\eta}{[\Theta^2,\V{v}]\V{v}\eta
-[\Theta^2,\tilde{\V{v}}]\tilde{\V{v}}\tilde\eta}.
\end{align}
In both terms appearing on the right hand side of the first line in \eqref{isa1} and in the
right entry of scalar product in the second line of \eqref{isa1} we now write
$\eta=(\eta-\tilde\eta)+\tilde\eta$ and apply distributive laws to get
\begin{align}\nonumber
A&=\Re\SPb{\eta-\tilde\eta}{\V{v}\Theta^2(\V{v}-\tilde{\V{v}})\tilde\eta}
-\Re\SPb{\eta-\tilde\eta}{(\V{v}-\tilde{\V{v}})\Theta^2\tilde{\V{v}}\tilde\eta}
\\\nonumber
&\quad+\Re\SPb{\tilde\eta}{\V{v}\Theta^2(\V{v}-\tilde{\V{v}})\tilde\eta}
-\Re\SPb{\tilde\eta}{(\V{v}-\tilde{\V{v}})\Theta^2\tilde{\V{v}}\tilde\eta}
\\\nonumber
&\quad-\Re\SPb{\eta-\tilde\eta}{[\Theta^2,\V{v}]\V{v}(\eta-\tilde\eta)}
\\\label{isa2}
&\quad-\Re\SPb{\eta-\tilde\eta}{\big([\Theta^2,\V{v}]\V{v}
-[\Theta^2,\tilde{\V{v}}]\tilde{\V{v}}\big)\tilde\eta}.
\end{align}
Using $\Re\SPn{\phi}{\psi}=\Re\SPn{\psi}{\phi}$, we see that the term in the second
line of \eqref{isa2} is equal to the first term in the second line of \eqref{i1i5}.
Computing the real part we further see that the term in the third
line of \eqref{isa2} is equal to the second term in the second line of \eqref{i1i5}.
(Here we use $\Theta^2\dom(v_j^2)\subset\dom(v_j^2)$.)
The sum of the first and fourth line of \eqref{isa2} is equal to
$\Re\SPn{\eta-\tilde\eta}{B\tilde\eta}$, where $B$ is defined on the domain $\cC$ by
\begin{align}\nonumber
B&:=\V{v}\Theta^2(\V{v}-\tilde{\V{v}})-(\V{v}-\tilde{\V{v}})\Theta^2\tilde{\V{v}}
\\\label{isa3}
&\quad-[\Theta^2,\V{v}](\V{v}-\tilde{\V{v}})
-[\Theta^2,\V{v}-\tilde{\V{v}}]\tilde{\V{v}}.
\end{align}
Here we added and subtracted $[\Theta^2,\V{v}]\tilde{\V{v}}$ in the second line.
In the first line of \eqref{isa3} we now commute one factor $\Theta$ to the left,
and in the second line we use the Leibnitz rule $[\Theta^2,a]=\Theta[\Theta,a]+[\Theta,a]\Theta$.
This results in
\begin{align}\nonumber
B=&\;\Theta\big\{\V{v}\Theta(\V{v}-\tilde{\V{v}})
-(\V{v}-\tilde{\V{v}})\Theta\tilde{\V{v}}
-[\Theta,\V{v}](\V{v}-\tilde{\V{v}})
-[\Theta,\V{v}-\tilde{\V{v}}]\tilde{\V{v}}\big\}
\\\nonumber
&\quad+[\V{v},\Theta]\Theta(\V{v}-\tilde{\V{v}})
-[\V{v}-\tilde{\V{v}},\Theta]\Theta\tilde{\V{v}}
\\\nonumber
&\quad-[\Theta,\V{v}]\Theta(\V{v}-\tilde{\V{v}})
-[\Theta,\V{v}-\tilde{\V{v}}]\Theta\tilde{\V{v}}
\\\nonumber
=&\;
\Theta\big\{2[\V{v},\Theta]+\Theta(\V{v}-\tilde{\V{v}})\big\}(\V{v}-\tilde{\V{v}})
-\Theta^2[\V{v},\tilde{\V{v}}]
\\\nonumber
&\quad+2[\V{v},\Theta]\Theta(\V{v}-\tilde{\V{v}})\quad\text{on $\cC$.}
\end{align}
Since $\Theta[{v}_j,\Theta]+[{v}_j,\Theta]\Theta$ is skew symmetric on $\dom(v_j)$,
for every $j\in\{1,\ldots,\nu\}$, this yields the desired identity.
\end{proof}


\subsection{Continuity of the stochastic flow in weighted spaces}\label{ssec-flow-cont}

\noindent
In this subsection we prove the announced Lem.~\ref{lem-diff1} on the behavior of the
stochastic flow under perturbations of the initial condition $\V{q}$ and the coefficient vector.
We shall again employ Notation~\ref{not-Theta}, \ref{not-T}, and~\ref{not-normcirc}.
The following abbreviations will be useful as well:

\begin{notation}\label{not-diff}
Given two coeffcient vectors $\V{c}=(\V{G},q,\vsigma\cdot\V{F})$ and 
$\tilde{\V{c}}=(\wt{\V{G}},\tilde{q},\vsigma\cdot\wt{\V{F}})$
satisfying Hyp.~\ref{hyp-G} with the same $\omega$ and the same conjugation $C$
and two $\fF_0$-measurable initial conditions $\V{q},\tilde{\V{q}}:\Omega\to\RR^\nu$, we set
\begin{align}\label{def-vecTxy}
\V{D}^{\pm}_{s}&:=[{\vp(\V{G}_{\V{B}^{\V{q}}_s}-\wt{\V{G}}_{\V{B}_s^{\tilde{\V{q}}}})},
\Theta_{s}^{\pm1}]\Theta_{s}^{\mp1}\vt^\mh,
\\\label{def-T2xy}
D_s&:=\big[\tfrac{i}{2}\vp(q_{\V{B}_s^{\V{q}}}-\tilde{q}_{\V{B}_s^{\tilde{\V{q}}}})
+\vsigma\cdot\vp(\V{F}_{\V{B}_s^{\V{q}}}-\wt{\V{F}}_{\V{B}_s^{\tilde{\V{q}}}}),\Theta_s\big]
\Theta_s^{-1}\vt^\mh,
\end{align}
for all $s\in[0,t_0]$. For any vector $\V{v}$ with components in $\HP$ we further write
\begin{align*}
\|\V{v}\|_*&:=\|\V{v}\|_{\mathfrak{k}}+\|\V{v}\|_\circ.
\end{align*}
Finally, we abbreviate
\begin{equation}\label{short:theta}
\theta:=1+\Id\Gamma(\omega).
\end{equation}
\end{notation}

Again the operators in \eqref{def-vecTxy} and \eqref{def-T2xy} are well-defined a priori on 
$\dom(\Id\Gamma(\omega))$ and have bounded 
closures under suitable extra conditions on $\V{c}$. 
For we have the following analogue of Lem.~\ref{lem-bd-T}: 

\begin{lem}\label{lem-bd-Txy}
Let $s\in[0,t_0]$ and $\ve\in(0,1]$. Assume in addition that $\ve\le|\delta|$, if the
operators in \eqref{def-vecTxy} and \eqref{def-T2xy} are defined
by means of the exponential weights $\Theta_s={\Xi_{\ve,s}^{(\delta)}}$ and 
$\vt=1+\Id\Gamma(\omega)$. Then the following bounds hold on $\Omega$,
\begin{align}\label{tina-diff}
\|\V{D}^\pm_s\|&\le \|\V{G}_{\V{B}_s^{\V{q}}}-\wt{\V{G}}_{\V{B}_s^{\tilde{\V{q}}}}\|_\circ,
\quad
\|D_{s}\|\le \|(q,\vsigma\cdot\V{F})_{\V{B}_s^{\V{q}}}
-(\tilde{q},\vsigma\cdot\wt{\V{F}})_{\V{B}_s^{\tilde{\V{q}}}}\|_\circ.
\end{align}
\end{lem}

\begin{proof}
Ex.~\ref{ex-comm-omega} and Ex.~\ref{ex-comm-exp} immediately imply the asserted bounds,
if we replace $\V{G}_{\V{B}_s^{\V{q}}}$ by 
$\V{G}_{\V{B}_s^{\V{q}}}-\wt{\V{G}}_{\V{B}_s^{\tilde{\V{q}}}}$, etc., in the corresponding arguments.
\end{proof}

\begin{lem}\label{lem-diff1}
Let $p\in\NN$ and $V,\wt{V}:\RR^\nu\to\RR$ be bounded and continuous.
Let $\V{c}=(\V{G},q,\vsigma\cdot\V{F})$ and 
$\tilde{\V{c}}=(\wt{\V{G}},\tilde{q},\vsigma\cdot\wt{\V{F}})$ be two coefficient 
vectors satisfying Hyp.~\ref{hyp-G} with the same $\omega$ and the same conjugation $C$
such that $\|\V{c}_{\V{x}}\|_\circ,\|\tilde{\V{c}}_{\V{x}}\|_\circ<\infty$, for every $\V{x}\in\RR^\nu$.
Let $(\rho_t,\V{R}_t)_{t\in[0,t_0]}$ and $(\tilde{\rho}_t,\wt{\V{R}}_t)_{t\in[0,t_0]}$
be adapted $\LO(\RR^{1+\nu},\FHR)$-valued processes with 
continuous paths such that $\int_0^{t_0}\|(\theta^\mh\rho_s,\V{R}_s)\|^2\Id s$ and
$\int_0^{t_0}\|(\theta^\mh\tilde{\rho}_s,\wt{\V{R}}_s)\|^2\Id s$ belong to $L^{\nf{p}{2}}(\PP)$.
Define a family of closed operators 
$$
\wt{H}(\V{x}):=\tfrac{1}{2}\vp(\wt{\V{G}}_{\V{x}})^2-\tfrac{i}{2}\vp(\tilde{q}_{\V{x}})
-\vsigma\cdot\vp(\wt{\V{F}}_{\V{x}})+\Id\Gamma(\omega),\quad\V{x}\in\RR^\nu,
$$
on the constant domain $\dom(\Id\Gamma(\omega))$. 
Assume in addition that $f_{2p,\mu}(s)\le1/8$, for all $s\in[0,t_0]$, where $f_{2p,\mu}$ is
defined by \eqref{def-fpmu}.
Let $(\V{q},\tilde{\V{q}},\eta):\Omega\to\RR^\nu\times\RR^\nu\times\dom(\Id\Gamma(\omega))$ 
be $\fF_0$-measurable and
$(\psi_t)_{t\in[0,t_0]}$ and $(\tilde{\psi}_t)_{t\in[0,t_0]}$ be $\FHR$-valued semi-martingales
whose paths belong $\PP$-a.s. to $C([0,t_0],\dom(\Id\Gamma(\omega)))$ and which 
$\PP$-a.s. satisfy
\begin{align*}
\psi_t&=\eta-\int_0^t\big(\wh{H}(\V{B}_s^{\V{q}})+V(\V{B}_s^{\V{q}})\big)\psi_s\Id s
+\int_0^ti\vp(\V{G}_{\V{B}_s^{\V{q}}})\psi_s\Id\V{B}_s
\\
&\quad+\int_0^t\rho_s\Id s+\int_0^t\V{R}_s\Id\V{B}_s,
\\
\tilde{\psi}_t&=\eta-\int_0^t\big(\wt{H}(\V{B}_s^{\tilde{\V{q}}})
+\wt{V}(\V{B}_s^{\tilde{\V{q}}})\big)\tilde{\psi}_s\Id s
+\int_0^ti\vp(\wt{\V{G}}_{\V{B}_s^{\tilde{\V{q}}}})\tilde{\psi}_s\Id\V{B}_s
\\
&\quad+\int_0^t\tilde{\rho}_s\Id s+\int_0^t\wt{\V{R}}_s\Id\V{B}_s,
\end{align*}
for all $t\in[0,t_0]$. Abbreviate
\begin{align}\nonumber
\phi_s&:=\Theta_{s}(\psi_s-\tilde{\psi}_s),
\\\nonumber
d_{p}(s)&:=|V(\V{B}_s^{\V{q}})-\wt{V}(\V{B}_s^{\tilde{\V{q}}})|
+\|(q,\vsigma\cdot\V{F})_{\V{B}_s^{\V{q}}}
-(\tilde{q},\vsigma\cdot\wt{\V{F}})_{\V{B}_s^{\tilde{\V{q}}}}\|_*
\\\label{def-dps}
&\quad+\big(p+\|\V{G}_{\V{B}_s^{\V{q}}}-\wt{\V{G}}_{\V{B}_s^{\tilde{\V{q}}}}\|_{\mathfrak{k}}\big)
\|\V{G}_{\V{B}_s^{\V{q}}}-\wt{\V{G}}_{\V{B}_s^{\tilde{\V{q}}}}\|_*,
\\\nonumber
Q_p(s)&:=d_p(s)^2\|\theta^\eh\Theta_s\tilde\psi_s\|^2
\\\nonumber
&\quad+\big\|\theta^\mh\Theta_s(\rho_s-\tilde\rho_s)\big\|^2+(p+\|\V{G}_{\V{B}_s^{\V{q}}}\|_\circ^2)
\|\Theta_{s}(\V{R}_s-\wt{\V{R}}_s)\|^2,
\end{align}
for all $s\in[0,t_0]$. Then the following bounds hold, for all $t\in[0,t_0]$,
\begin{align}\nonumber
\EE\big[&\sup_{s\le t}e^{-p\int_0^tb_{2p,\mu}(s)\Id s}\|\phi_s\|^p\big]
\\\label{tom1}
&\le(ce^{t}p^\eh)^p\EE\Big[\Big(\int_0^te^{-2\int_0^sb_{2p,\mu}(r)\Id r}
Q_{p}(s)\Id s\Big)^{\nf{p}{2}}\Big],
\\\nonumber
\EE\Big[&\Big(\int_0^t\|\Id\Gamma(\omega)^\eh\phi_s\|^2\Id s\Big)^{\nf{p}{2}}\Big]
\le c_pe^{pt}\EE\Big[\Big(\int_0^te^{-2\int_0^sb_{2p,\mu}(r)\Id r}Q_{p}(s)\Id s\Big)^{\nf{p}{2}}\Big]
\\\label{tom2}
&\quad+c_pe^{pt}\EE\Big[\Big(\int_0^te^{-2\int_0^sb_{2p,\mu}(r)\Id r}
\|\V{G}_{\V{B}_s^{\V{q}}}\|_\circ^2\|\phi_s\|^2\Id s\Big)^{\nf{p}{2}}\Big].
\end{align}
Here $c>0$ is some universal constant and $c_p>0$ depends only on $p$. 
The process $b_{2p,\mu}$ is defined by \eqref{def-bpmu}.
\end{lem}

\begin{proof}
{\em Step 1.} Exactly as in the first step of the proof of Lem.~\ref{cor-spin1} we see that 
Lem.~\ref{lem-Ito-Theta} applies to both $(\psi_t)_{t\in[0,t_0]}$ and $(\tilde{\psi}_t)_{t\in[0,t_0]}$. 
As a consequence, $(\phi_t)_{t\in[0,t_0]}$ can $\PP$-a.s. be written as
\begin{align*}
\phi_t&=-\int_0^t\Theta_s\big(\tfrac{1}{2}\vp({\V{G}}_{\V{B}_s^{\V{q}}})^2-\Phi_s\big)\psi_s\Id s
-\int_0^t\Theta_s\big(\tfrac{1}{2}\vp(\wt{\V{G}}_{\V{B}_s^{\tilde{\V{q}}}})^2
-\wt{\Phi}_s\big)\tilde{\psi}_s\Id s
\\
&\quad-\int_0^t\big(\Id\Gamma(\omega)-\tfrac{\dot{\Theta}_s}{\Theta_s}\big)\phi_s\Id s
+\int_0^t\Theta_s(\rho_s-\tilde{\rho}_s)\Id s
\\
&\quad+\int_0^t\Theta_s\big(i\vp(\V{G}_{\V{B}_s^{\V{q}}})\psi_s-
i\vp(\wt{\V{G}}_{\V{B}_s^{\tilde{\V{q}}}})\tilde{\psi}_s\big)\Id\V{B}_s
+\int_0^t\Theta_s(\V{R}_s-\wt{\V{R}}_s)\Id\V{B}_s,
\end{align*}
for all $t\in[0,t_0]$, with
\begin{align*}
\Phi_s&:=\tfrac{i}{2}\vp(q_{\V{B}_s^{\V{q}}})+\vsigma\cdot\vp(\V{F}_{\V{B}_s^{\V{q}}})
-V(\V{B}_s^{\V{q}}),
\\
\wt{\Phi}_s&:=\tfrac{i}{2}\vp(\tilde{q}_{\V{B}_s^{\tilde{\V{q}}}})
+\vsigma\cdot\vp(\wt{\V{F}}_{\V{B}_s^{\tilde{\V{q}}}})-\wt{V}(\V{B}_s^{\tilde{\V{q}}}).
\end{align*}
Now It\={o}'s formula $\PP$-a.s. implies
\begin{align}\label{clara00}
\|\phi_t\|^{2}&=\sum_{\ell=1}^{7}\int_0^t\sI_\ell(s)\Id s
+\int_0^t\V{I}(s)\Id\V{B}_s,\quad t\in[0,t_0],
\end{align}
where 
\begin{align*}
\sI_1(s)&:=-\Re\SPb{\phi_s}{\Theta_{s}\big(\vp(\V{G}_{\V{B}_s^{\V{q}}})^2\psi_s
-\vp(\wt{\V{G}}_{\V{B}_s^{\tilde{\V{q}}}})^2\tilde{\psi}_s\big)},
\\
\sI_2(s)&:=-2\Re\SPb{\phi_s}{\big(\Id\Gamma(\omega)
-\tfrac{\dot{\Theta}_{s}}{\Theta_{s}}\big)\phi_s},
\\
\sI_3(s)&:=2\Re\SPb{\phi_s}{\Theta_s\big(\Phi_s\psi_s-\wt{\Phi}_s\tilde{\psi}_s\big)},
\\
\sI_4(s)&:=\big\|\Theta_{s}\big(\vp(\V{G}_{\V{B}_s^{\V{q}}})\psi_s
-\vp(\wt{\V{G}}_{\V{B}_s^{\tilde{\V{q}}}})\tilde{\psi}_s\big)\big\|^2,
\\
\sI_5(s)&:=2\Re\SPb{\phi_s}{\Theta_{s}(\rho_s-\tilde{\rho}_s)},
\\
\sI_6(s)&:=2\Re\SPb{\Theta_{s}(\V{R}_s-\wt{\V{R}}_s)}{
i\Theta_{s}\big(\vp(\V{G}_{\V{B}_s^{\V{q}}})\psi_s-
\vp(\wt{\V{G}}_{\V{B}_s^{\tilde{\V{q}}}})\tilde{\psi}_s\big)},
\\
\sI_7(s)&:=\|\Theta_{s}(\V{R}_s-\wt{\V{R}}_s)\|^2,
\\
\V{I}(s)&:=2\Re\SPb{\phi_s}{i\Theta_{s}\big(\vp(\V{G}_{\V{B}_s^{\V{q}}})\psi_s
-\vp(\wt{\V{G}}_{\V{B}_s^{\tilde{\V{q}}}})\tilde{\psi}_s\big)+\Theta_{s}(\V{R}_s-\wt{\V{R}}_s)}.
\end{align*}
Furthermore, 
the following identity \eqref{gundula} is a direct consequence of the algebraic Lem.~\ref{lem-i1i5}, 
the commutation relation 
$$
[\vp(\V{G}_{\V{B}_s^{\V{q}}}),\vp(\wt{\V{G}}_{\V{B}_s^{\tilde{\V{q}}}})]=2\Im
\SPn{\V{G}_{\V{B}_s^{\V{q}}}}{\wt{\V{G}}_{\V{B}_s^{\tilde{\V{q}}}}}\id=0,
\quad\text{on $\dom(\Id\Gamma(\omega))$,}
$$ 
and the definitions \eqref{eva1}, \eqref{eva4}, and \eqref{def-vecTxy}, 
\begin{align}\nonumber
\sI_1(s)+\sI_4(s)&=\SPn{\vt^\eh\phi_s}{T_{1,s}\vt^\eh\phi_s}
-2\Re\SPb{(\V{T}^+_s-\V{T}^-_s)\vt^\eh\phi_s}{\V{\zeta}_{s}}
\\\label{gundula}
&\quad+\Re\SPb{\vp(\V{G}_{\V{B}_s^{\V{q}}}-\wt{\V{G}}_{\V{B}_s^{\tilde{\V{q}}}})\phi_s}{
\V{\zeta}_{s}}-\Re\SPn{\V{D}_{s}^-\vt^\eh\phi_s}{\V{\zeta}_{s}}+\|\V{\zeta}_{s}\|^2.
\end{align}
Here we abbreviate 
\begin{align}\nonumber
\V{\zeta}_{s}&:=\Theta_{s}\vp(\V{G}_{\V{B}_s^{\V{q}}}-\wt{\V{G}}_{\V{B}_s^{\tilde{\V{q}}}})\tilde{\psi}_s
=\big(\vp(\V{G}_{\V{B}_s^{\V{q}}}-\wt{\V{G}}_{\V{B}_s^{\tilde{\V{q}}}})-\V{D}_{s}^+\vt^\eh\big)
\Theta_{s}\tilde{\psi}_s.
\end{align}
Using the shorthand \eqref{short:theta}, we further have
\begin{align}\label{clara77}
\|\V{\zeta}_{s}\|&\le c\|\V{G}_{\V{B}_s^{\V{q}}}-\wt{\V{G}}_{\V{B}_s^{\tilde{\V{q}}}}\|_*
\|\theta^{\eh}\Theta_{s}\tilde{\psi}_s\|,
\end{align}
for some universal constant $c>0$, where we used \eqref{tina-diff} and
\begin{align*}
\|\vp(\V{G}_{\V{B}_s^{\V{q}}}-\wt{\V{G}}_{\V{B}_s^{\tilde{\V{q}}}})\phi'\|
&\le 2^\eh\|\V{G}_{\V{B}_s^{\V{q}}}-\V{G}_{\V{B}_s^{\tilde{\V{q}}}}\|_{\mathfrak{k}}\|\theta^\eh\phi'\|,
\quad \phi'\in\dom(\theta^\eh).
\end{align*}
We readily infer from these remarks that
\begin{align}\nonumber
|\sI_1(s)+\sI_4(s)|&\le\|\phi_s\|^2/8+\|\Id\Gamma(\omega)^\eh\phi_s\|^2/8
+\|T_{1,s}\|\,\|\vt^\eh\phi_s\|^2
\\\label{clara1}
&\quad+c'd_1^{2}(s)\|\theta^{\eh}\Theta_{s}\tilde{\psi}_s\|^2,
\end{align}
with some universal constant $c'>0$.
Since $\Re\SPn{\phi_s}{\tfrac{i}{2}\vp(q_{\V{B}_s^{\V{q}}})\phi_s}=0$, we further have
\begin{align*}
\sI_3(s)&=-2\Re\SPn{\phi_s}{T_{2,s}\vt^\eh\phi_s}
+2\Re\SPb{\phi_s}{(\vsigma\cdot\vp(\V{F}_{\V{B}_s^{\V{q}}})-V(\V{B}_s^{\V{q}}))\phi_s}
+2\Re\SPn{\phi_s}{\xi_s},
\end{align*} 
with
\begin{align}\nonumber
\xi_s&:=\Theta_s(\Phi_s-\wt{\Phi}_s)\tilde\psi_s
=(\Phi_s-\wt{\Phi}_s)\Theta_s\tilde\psi_s-D_s\vt^\eh\Theta_s\tilde\psi_s,
\\\nonumber
\|\xi_s\|&\le c''d_1(s)\|\theta^{\eh}\Theta_{s}\tilde\psi_s\|.
\end{align}
Hence, using also \eqref{flora}, it is easy to see that
\begin{align}\nonumber
\sI_2(s)+\sI_3(s)&\le\big(8\lambda(\V{B}_s^{\V{q}})-2V(\V{B}_s^{\V{q}})
+\tfrac{1}{\mu}\|T_{2,s}\|+\tfrac{3}{2}\big)\|\phi_s\|^2-\|\Id\Gamma(\omega)^\eh\phi_s\|^2/2
\\\label{clara11}
&\quad+\mu\|T_{2,s}\|\|\vt^\eh\phi_s\|^2+c'''d_1(s)^2\|\theta^{\eh}\Theta_{s}\tilde\psi_s\|^2,
\end{align}
for all $\mu>0$. Moreover, we have, again for every $\mu'>0$,
\begin{align*}
|\sI_5(s)|&\le\mu'\|\theta^\eh\phi_s\|^2
+\tfrac{1}{\mu'}\big\|\theta^\mh\Theta_s(\rho_s-\tilde{\rho}_s)\big\|^2,
\\
|\sI_6(s)|&\le2\mu'\|\theta^\eh\phi_s\|^2
+\big(1+\tfrac{1}{\mu'}\|\vp(\V{G}_{\V{B}_s^{\V{q}}})\theta^\mh\|^2
+\tfrac{1}{\mu'}\|\V{T}_{s}^+\|^2\big)\sI_7(s)+\|\V{\zeta}_s\|^2.
\end{align*}
Summarizing the above bounds on $\sI_1,\ldots,\sI_7$, we obtain
\begin{align}\nonumber
\sum_{\ell=1}^7\sI_\ell(s)
&\le-\big(\tfrac{3}{8}-3\mu'\big)\|\Id\Gamma(\omega)^\eh\phi_s\|^2
+\big(\|T_{1,s}\|+\mu\|T_{2,s}\|\big)\|\vt^\eh\phi_s\|^2
\\\nonumber
&\quad+\big(8\lambda(\V{B}_s^{\V{q}})-2V(\V{B}_s^{\V{q}})
+\tfrac{1}{\mu}\|T_{2,s}\|+3\mu'+\tfrac{13}{8}\big)\|\phi_s\|^2
\\\nonumber
&\quad+cd_1(s)^2\|\theta^\eh\Theta_s\tilde{\psi}_s\|^2
\\\label{diane2000}
&\quad+\tfrac{1}{\mu'}\big\|\theta^\mh\Theta_s(\rho_s-\tilde{\rho}_s)\big\|^2+
c'\big(1+\tfrac{1}{\mu'}\|\V{G}_{\V{B}_s^{\V{q}}}\|_*^2\big)\sI_7(s),
\end{align}
for all $s\in[0,t_0]$ and $\mu,\mu'>0$.
Using $\Re\SPn{\phi_s}{i\vp(\V{G}_{\V{B}_s^{\V{q}}})\phi_s}=0$, we further have
\begin{align*}
\V{I}(s)&=-2\Re\SPb{\phi_s}{i\V{T}_{s}^+\vt^\eh\phi_s}+2\Re\SPb{\phi_s}{i\V{\zeta}_{s}}
+2\Re\SPb{\phi_s}{\Theta_{s}(\V{R}_s-\wt{\V{R}}_s)},
\end{align*}
which together with \eqref{clara77} permits to get
\begin{align}\label{diane1999}
\tfrac{1}{2}\V{I}(s)^2&\le\|\phi_s\|^2\big(4\|\V{T}_{s}^+\|^2\|\vt^\eh\phi_s\|^2
+c'd_1(s)^2\|\theta^\eh\Theta_s\tilde{\psi}_s\|^2+c''\sI_7(s)\big).
\end{align}

{\em Step 2.} 
Let $p\ge2$. Then, in view of \eqref{clara00}, an application of It\={o}'s formula $\PP$-a.s. yields
\begin{align}\nonumber
\|\phi_t\|^{2p}&=p\sum_{\ell=1}^{7}\int_0^t\|\phi_s\|^{2p-2}\sI_\ell(s)\Id s
+p\int_0^t\|\phi_s\|^{2p-2}\V{I}(s)\Id\V{B}_s
\\\label{herbie1}
&\quad+\frac{p(p-1)}{2}\int_0^t\|\phi_s\|^{2p-4}\V{I}(s)^2\Id s,\quad t\in[0,t_0].
\end{align}
Applying the bound \eqref{diane2000} with $\mu':=1/24$ and \eqref{diane1999}, 
we may $\PP$-a.s. deduce that
\begin{align}\nonumber
\|\phi_t\|^{2p}&\le\|\phi_t\|^{2p}
-\int_0^t\big(\tfrac{1}{4}-f_{2p,\mu}(s)\big)\|\phi_s\|^{2p-2}\|\Id\Gamma(\omega)^\eh\phi_s\|^2\Id s
\\\nonumber
&\le2p\int_0^t\big(b_{2p,\mu}(s)+1\big)\|\phi_s\|^{2p}\Id s
\\\nonumber
&\quad
+cp\int_0^t\|\phi_s\|^{2p-2}Q_p(s)\Id s+p\int_0^t\|\phi_s\|^{2p-2}\V{I}(s)\Id\V{B}_s,
\end{align}
for all $t\in[0,t_0]$. The previous bound holds actually true for $p=1$ as well;
see \eqref{clara00} and \eqref{diane2000}. Applying Lem.~\ref{lem-stoch-Gronwall} with
$\gamma=1/2$ and $\delta=1-1/p$ (exactly as in the proof of Lem.~\ref{cor-spin1}),
we arrive at \eqref{tom1}. In fact, we know that the a priori bound \eqref{inger1} is fulfilled
by the right hand side of \eqref{tom1} since the arguments of Step~4 in the proof of 
Lem.~\ref{cor-spin1} apply to both $(\psi_t)_{t\in[0,t_0]}$ and $(\tilde\psi_t)_{t\in[0,t_0]}$.

Finally, \eqref{tom2} is derived from \eqref{diane2000} (with $\mu':=1/24$)
exactly in the same fashion as \eqref{diane1}
was derived from \eqref{diane99} in Step~6 of the proof of Lem.~\ref{cor-spin1}.
\end{proof}


\begin{ex}\label{ex-antonio}
Let us apply Lem.~\ref{lem-diff1} with with constant $\eta=\phi\in\dom(\Id\Gamma(\omega))$
and with $\tilde{\V{c}}=\V{0}$, $V=\wt{V}=0$, $\rho=0$, 
$\V{R}=\wt{\V{R}}=\V{0}$, and $\tilde{\rho}=\Id\Gamma(\omega)\phi$. 
Then, up to indistinguishability, we simply have
$\tilde{\psi}_t=\phi$, $t\ge0$. Choosing the $s$-independent weight
$\Theta_s:=\theta_\ve^\mh:=(1+\ve(1+\Id\Gamma(\omega_\ve)))^\eh
(1+\Id\Gamma(\omega_\ve))^\mh$
with $\omega_\ve:=\omega(1+\ve\omega)^{-1}$. In this case we have $|\alpha|=1/2$, so that
the norm $\|\cdot\|_\circ$ is dominated by $\|\cdot\|_{\mathfrak{k}}$. We again write
$\theta=\theta_0$ and 
$\|\V{c}\|_{\mathfrak{k},\infty}:=\sup_{\V{x}\in\RR^\nu}\|\V{c}_{\V{x}}\|_{\mathfrak{k}}$. 
Employing \eqref{tom1} and the dominated convergence theorem, 
we then find a universal constant $c>0$ and $c_p>0$, depending only on $p\in\NN$, such that
\begin{align*}
&\EE\big[\sup_{s\le t}\big\|\theta^\mh(\WW{s}{0}[\V{B}^{\V{q}}]-\id)\phi\big\|^p\big]
\le\EE\big[\sup_{s\le t}\big\|\theta_\ve^\mh(\WW{s}{0}[\V{B}^{\V{q}}]-\id)\phi\big\|^p\big]
\\
&\le c_p(\|\V{c}\|_{\mathfrak{k},\infty}\vee\|\V{c}\|_{\mathfrak{k},\infty}^2)^pt^{\nf{p}{2}}
e^{(1+\|\V{c}\|_{\mathfrak{k},\infty}^2)p^2t}\big(\|\theta^\eh\theta_\ve^\mh \phi\|^2
+\|\theta^\mh\theta_\ve^\mh\Id\Gamma(\omega)\phi\|^2\big)^{\nf{p}{2}}
\\
&\xrightarrow{\;\;\ve\downarrow0\;\;}c_p
(\|\V{c}\|_{\mathfrak{k},\infty}\vee\|\V{c}\|_{\mathfrak{k},\infty}^2)^pt^{\nf{p}{2}}
e^{(1+\|\V{c}\|_{\mathfrak{k},\infty}^2)p^2t}\big(
\|\phi\|^2+\|\Id\Gamma(\omega)\theta^{-1}\phi\|^2\big)^{\nf{p}{2}}.
\end{align*}
Employing the dominated convergence
theorem once more to replace $\phi\in\dom(\Id\Gamma(\omega))$ by any arbitrary 
element of $\FHR$, we arrive at
\begin{align}\nonumber
\EE\big[\sup_{s\le t}&\big\|(1+\Id\Gamma(\omega))^\mh(\WW{s}{0}[\V{B}^{\V{q}}]-\id)\psi\big\|^p\big]
\\\label{antonio}
&\le c_p'(\|\V{c}\|_{\mathfrak{k},\infty}\vee\|\V{c}\|_{\mathfrak{k},\infty}^2)^pt^{\nf{p}{2}}
e^{(1+\|\V{c}\|_{\mathfrak{k},\infty}^2)p^2t}\|\psi\|^p,
\end{align}
for all $t\ge0$, $\psi\in\FHR$, and $p\in\NN$, where $c_p'>0$ depends only on $p$.
\end{ex}

Finally, we apply Lem.~\ref{lem-diff1} to two solution processes given by Thm.~\ref{thm-WW}
extending \eqref{tom1} to unbounded weights at the same time. 
We could also extend \eqref{tom2}, of course,
but refrain from doing so here, as \eqref{tom2} will only become important in our companion
paper \cite{Matte2015}.

\begin{lem}\label{lem-maria}
{\rm(1)} Let $t_0>0$, $|\alpha|\ge1/2$, and let $\V{c}$ and $\tilde{\V{c}}$ be two coefficient vectors
satisfying Hyp.~\ref{hyp-G} with the same $\omega$ and $C$ such that 
$\|\V{c}\|_{\circ,\infty}<\infty$ and $\|\tilde{\V{c}}\|_{\circ,\infty}<\infty$.
Let $V,\wt{V}:\RR^\nu\to\RR$ be bounded and continuous. Finally, let 
$(\V{q},\tilde{\V{q}},\eta):\Omega\to\RR^{2\nu}\times{\dom}(\Theta^{(\alpha)}_{0,0})$ 
be $\fF_0$-measurable with $\|\Theta_{0,0}^{(\alpha)}\eta\|\in L^{p}(\PP)$, and let
$\wt{\mathbb{W}}^{\wt{V}}[\V{B}^{\tilde{\V{q}}}]$ be the solution process given by
Thm.~\ref{thm-WW} applied to $\tilde{\V{c}}$, $\wt{V}$, and $\tilde{\V{q}}$.
Then there exist a universal constant $c>0$ and, for every $p\in\NN$, some
$c_p>0$, depending only on $p$, such that with
$$
c_{p,t}:=c_pe^{c(p\|\V{c}\|_{\circ,\infty}^2+\|\mho\|_\infty+\|V\|_\infty
+p\|\tilde{\V{c}}\|_{\circ,\infty}^2+\|\wt{\mho}\|_\infty+\|\wt{V}\|_\infty+1)pt}
$$
the following bound holds true,
\begin{align}\nonumber
\EE\Big[&\sup_{s\le t}\big\|\Theta_{0,s}^{(\alpha)}\big(\WW{s}{V}[\V{B}^{\V{q}}]
-\wt{\mathbb{W}}_s^{\wt{V}}[\V{B}^{\tilde{\V{q}}}]\big)\eta\big\|^p\Big]
\\\label{maria1}
&\le c_{p,t}\EE\big[\sup_{s\le t}d_p(s)^{2p}\big]^\eh
\big(1+t^{\nf{p}{2}}\|\wt{\V{G}}\|_{\circ,\infty}^p\big)
\EE\big[\|\Theta_{0,0}^{(\alpha)}\eta\|^{2p}\big]^\eh.
\end{align}
Here $d_p(s)$ is defined in \eqref{def-dps}.

\smallskip

\noindent{\rm(2)} The assertion in (1) holds true with $\alpha$ replaced by $0<|\delta|\le1$ and
$\Theta_{0,s}^{(\alpha)}$ replaced by $\Xi_{0,s}^{(\delta)}$ everywhere.
\end{lem}

\begin{proof}
(1): Let $\eta_{n,m}$ be defined as in the proof of Lem.~\ref{lem-spin4}.
By virtue of \eqref{tom1} we then obtain
\begin{align}\nonumber
&e^{-c(p\|\V{c}\|_{\circ,\infty}^2+\|\mho\|_\infty+1)pt}
\EE\Big[\sup_{s\le t}\big\|\Theta_{\ve,s}^{(\alpha)}\big(\WW{s}{V}[\V{B}^{\V{q}}]
-\wt{\mathbb{W}}_s^{\wt{V}}[\V{B}^{\tilde{\V{q}}}]\big)\eta_{n,m}\big\|^p\Big]
\\\nonumber
&\le c_{p}
\EE\Big[\Big(\int_0^td_p(s)^2\big\|\theta^\eh\Theta_{\ve,s}^{(\alpha)}
\wt{\mathbb{W}}_s^{\wt{V}}[\V{B}^{\tilde{\V{q}}}]\big)\eta_{n,m}\big\|^2\Id s\Big)^{\nf{p}{2}}\Big]
\\\nonumber
&\le c_{p}
\EE\big[\sup_{s\le t}d_p(s)^2\big]^\eh
\EE\Big[\Big(\int_0^t\big\|\theta^\eh\Theta_{\ve,s}^{(\alpha)}
\wt{\mathbb{W}}_s^{\wt{V}}[\V{B}^{\tilde{\V{q}}}]\eta_{n,m}\big\|^2\Id s\Big)^{p}\Big]^\eh
\\\label{maria2}
&\le c_{p,t}\EE\big[\sup_{s\le t}d_p(s)^{2p}\big]^\eh
\big(1+t^{\nf{p}{2}}\|\wt{\V{G}}\|_{\circ,\infty}^p\big)
\EE\big[\|\Theta_{\ve,0}^{(\alpha)}\eta_{m}\|^{2p}\big]^\eh.
\end{align}
Here we also applied \eqref{pol-w-eps} and \eqref{cleo2} in the last step.
Next, we let $n$ go to infinity on the left hand side of \eqref{maria2}
using the dominated convergence theorem with majorizing functions analogous
to the ones in \eqref{maria0}. After that the regularization parameters $m$ and $\ve$ are
removed in the same way as in the end of the proof of Lem.~\ref{lem-spin4}(1), distinguishing
the cases $\alpha>0$ and $\alpha<0$.

The proof of (2) is again completely analogous.
\end{proof}


\section{Weighted $L^p$-estimates and continuity in the potential}\label{sec-cont-V}

\noindent
The next theorem complements the $L^p$-estimates of Lem.~\ref{lem-T-bd}
by including unbounded multiplication operators in Fock space. 
The convergence result in the next theorem will imply strong convergence of the
semi-group when its Kato decomposable potential is approximated in the sense
of Lem.~\ref{lem-conny}; see Cor.~\ref{cor-SCV} below. 
It is also used to prove Thm.~\ref{thm-eq-cont} later on.
To shorten statements we introduce the following convention:

\begin{notation}\label{not-krass}
In the following table we introduce weight functions and corresponding norms used to state our
main theorems in this and the subsequent sections. These theorems hold for the weights in the
first column provided that the coefficient vector $\V{c}$ fulfills the corresponding hypothesis in the 
third column {\em in addition} to our standing hypothesis Hyp.~\ref{hyp-G}.
In the third column of the table we use the notation introduced in the second one.
We always assume that $\alpha\ge1/2$, $\delta\in(0,1]$, $t_*\ge2$, and that $\vo$ and $\vk$ 
are non-negative measurable functions on $\cM$ with $\vo\le\omega$. 
The constant $c_\alpha$ is the one appearing in 
Lem.~\ref{lem-bd-T}(1), $c$ is the one appearing in Lem.~\ref{lem-bd-T}(2).
Recall that $\sup_{\V{x}}\|\V{c}_{\V{x}}\|_{\mathfrak{k}}<\infty$ is part of Hyp.~\ref{hyp-G}.

\medskip

\begin{center}
\renewcommand{\arraystretch}{1.5}
\begin{tabular}{|l|l|l|l|}
\hline
&$\Upsilon_t$&$\|\V{v}\|_*:=\|\V{v}\|_{\mathfrak{k}}+\|\V{v}\|_{**}$&Hyp($\Upsilon$)\\
&  &\text{with $\|\V{v}\|_{**}:=$}&\\
\hline\hline
1.&$(1+\Id\Gamma(\vk)/2\alpha)^{\alpha}$&$c_{\alpha}\|\vk^\eh(1+{\vk})^{\alpha-\eh}\V{v}\|_{\HP}$
&$\sup_{\V{x}}\|{\vk}^\alpha\V{c}_{\V{x}}\|_{\HP}<\infty$\\
 \hline
 2.&$(1+(t+t\Id\Gamma(\vo))/2\alpha)^{\alpha}$&$c_{\alpha}\|\vo^\eh{(1+\vo)}^{\alpha-\eh}\V{v}\|_{\HP}$
 &$\sup_{\V{x}}\|{\vo}^\alpha\V{c}_{\V{x}}\|_{\HP}<\infty$\\
 \hline
 3.&$e^{\delta\Id\Gamma(\vk)}$&$c\delta^\eh\|(\vk\vee\tfrac{\vk^2}{\omega})^\eh 
 e^{\delta\vk}\V{v}\|_{\HP}$&$\sup_{\V{x}}\|(\V{F}_{\V{x}},q_{\V{x}})\|_{**}<\infty$\\
 &&&$\sup_{\V{x}}\|\V{G}_{\V{x}}\|_{**}\le\nf{1}{9}$\\
\hline
 4.&$e^{\delta(t\wedge t_*)(1+\Id\Gamma(\vo))/2}$&
 $ct_*\delta^\eh\|\vo^\eh e^{\delta t_*\vo/2}\V{v}\|_{\HP}$&
 $\sup_{\V{x}}\|(\V{F}_{\V{x}},q_{\V{x}})\|_{**}<\infty$\\
&&& $\sup_{\V{x}}\|\V{G}_{\V{x}}\|_{**}\le\nf{1}{9}$\\
\hline
\end{tabular}
\\
\medskip

\small{{\bf Table 1}: Weight functions with corresponding norm $\|\cdot\|_*$ and additional hypothesis. }
\end{center}

\medskip

\noindent
Furthermore, we abbreviate $\|\V{c}\|_{*,\infty}:=\sup_{\V{x}\in\RR^\nu}\|\V{c}_{\V{x}}\|_*$.
\end{notation}

\begin{thm}\label{thm-LNCV}
Let $V\in\cK_\pm(\RR^\nu)$, let $F:\RR^\nu\to\RR$ be globally Lipschitz
continuous, i.e., $|F(\V{x})-F(\V{y})|\le a|\V{x}-\V{y}|$, $\V{x},\V{y}\in\RR^\nu$, for some $a\ge0$,
and let $1\le p\le q\le\infty$. Let $(\Upsilon_t)_{t\ge0}$ and $\|\cdot\|_*$ be given by one of the lines 
in Table~1 and assume that $\V{c}$ fulfills {\rm Hyp($\Upsilon$)} given by the same line.
Then the following holds:

\smallskip

\noindent{\rm(1)} 
For all $t>0$, $T_t^{V}$ maps $\Upsilon_0^{-1}L^p(\RR^\nu,\FHR;e^{pF}\Id\V{x})$ 
into the domain of $\Upsilon_t$, considered as a densely defined operator in 
$L^q(\RR^\nu,\FHR;e^{qF}\Id\V{x})$, and
\begin{align}\label{norm-T-F-Theta}
\big\|e^F\Upsilon_tT_t^V{\Upsilon}_0^{-1}e^{-F}\big\|_{p,q}
\le c_{\nu,p,q}\,\frac{e^{8(1+\|\V{c}\|_{*,\infty}^2+a^2)t}}{t^{\nu(\nf{1}{p}-\nf{1}{q})/2}}
\sup_{\V{z}\in\RR^\nu}\EE\big[e^{8\int_0^tV_-(\V{B}_s^{\V{z}})\Id s}\big]^{\nf{1}{4}}.
\end{align}
Here the constant $c_{\nu,p,q}>0$ depends only on $p$, $q$, and $\nu$. 

\smallskip

\noindent{\rm(2)}
If $V_n\in\cK_\pm(\RR^\nu)$, $n\in\NN$, satisfy \eqref{approx1} and \eqref{approx2}, then 
\begin{align}\label{LNCV1}
\lim_{n\to\infty}\sup_{t\in[\tau_1,\tau_2]}\big\|1_Ke^{F}\Upsilon_t(T^{V_n}_t-T^V_t)
{\Upsilon}_0^{-1}e^{-F}\big\|_{p,q}&=0,
\\\label{LNCV2}
\lim_{n\to\infty}\sup_{t\in[\tau_1,\tau_2]}\big\|e^F\Upsilon_t(T^{V_n}_t-T^V_t)
{\Upsilon}_0^{-1}e^{-F}1_K\big\|_{p,q}&=0,
\end{align}
for all compact $K\subset\RR^\nu$ and $\tau_2>\tau_1>0$.
If $p=q$, then the choice $\tau_1=0$ is allowed for in \eqref{LNCV1} and \eqref{LNCV2} as well.
If $V-V_n\in\cK(\RR^\nu)$, for all $n\in\NN$, and \eqref{approx1} holds with $K$
replaced by $\RR^\nu$, then $K$ can be replaced by $\RR^\nu$ in
\eqref{LNCV1} and \eqref{LNCV2}, too. 
In all cases, the convergences \eqref{LNCV1} and \eqref{LNCV2} are uniform in all Lipschitz 
continuous $F$ with Lipschitz constant $\le a$ and all coefficient vectors $\V{c}$ satisfying
{\rm Hyp($\Upsilon$)} and $\|\V{c}\|_{*,\infty}\le A$, for some given $a,A\in(0,\infty)$.
\end{thm}

\begin{proof}
In the subequent five steps of
this proof we {\em fix} $t>0$, that plays the role of $t_0$ in Lem.~\ref{lem-spin4}.
The weight function $(\hat{\Upsilon}_s)_{s\in[0,t]}$ will either be
$({\Upsilon}_s)_{s\in[0,t]}$ or $({\Upsilon}_{t-s}^{-1})_{s\in[0,t]}$. We have to consider these two 
choices because of the duality arguments used below.
We shall make use of Rem.~\ref{rem-weights}(1) without further notice.
Moreover, we set $\wt{\Upsilon}_s:=\hat{\Upsilon}_{t-s}^{-1}$, $s\in[0,t]$.
One subtlety appears when $\Upsilon$ is given by Line~4 of Table~1:
If $t>t_*$, then $\hat{\Upsilon}$ (resp. $\wt{\Upsilon}$) is chosen to be equal to
$(e^{\delta_*s(1+\Id\Gamma(\vo))/2})_{s\in[0,t]}$ with $\delta_*:=\delta t^*/t$.

{\em Step 1.}
Let $p\in[2,\infty]$. Pick $\tau\in[0,t]$ and some measurable 
$\Psi:\RR^\nu\to\dom(\hat{\Upsilon}_{t-\tau}^{-1})$ with $\Psi\in L^p(\RR^\nu,\FHR)$
and $\hat{\Upsilon}_{t-\tau}^{-1}\Psi\in L^p(\RR^\nu,\FHR)$.
For all  $\V{x}\in\RR^\nu$ and $\phi\in\dom(\hat{\Upsilon}_t)$, we may then write
\begin{align*}
\SPb{\hat{\Upsilon}_{t}\phi}{(T_\tau^V&e^{-F}\hat{\Upsilon}_{t-\tau}^{-1}\Psi)(\V{x})}
\\
&=\EE\big[e^{-\int_0^\tau V(\V{B}_s^{\V{x}})\Id s}\SPb{\wt{\Upsilon}_\tau\WW{\tau}{0}[\V{B}^{\V{x}}]
\wt{\Upsilon}_0^{-1}\phi}{\Psi(\V{B}^{\V{x}}_\tau)}e^{-F(\V{B}^{\V{x}}_\tau)}\big],
\end{align*}
where we took Lem.~\ref{lem-spin4} into account.
An analogous formula holds true with $T^V_\tau$ replaced by $(T^{V_n}_\tau-T^V_\tau)$ 
on the left hand side and $e^{-\int_0^\tau V(\V{B}_s^{\V{x}})\Id s}$ replaced by 
$(e^{-\int_0^\tau V_n(\V{B}_s^{\V{x}})\Id s}-e^{-\int_0^\tau V(\V{B}_s^{\V{x}})\Id s})$ 
on the right hand side.
Define $q_1\in(1,\infty)$ by $q_1^{-1}=1-p^{-1}-4^{-1}-8^{-1}$, so that $q_1\le8$, and
let us agree that, for $p=\infty$, the symbol $\EE[\|\Psi(\V{B}^{\V{x}}_\tau)\|^{p}]^{\nf{1}{p}}$ should 
be read as $\|\Psi\|_\infty$ in what follows. Then the above remarks permit to get
\begin{align}\nonumber
\sup_{{\phi\in\dom(\hat{\Upsilon}_t)\atop\|\phi\|=1}}&e^{F(\V{x})}
\big|\SPb{\hat{\Upsilon}_{t}\phi}{(T_\tau^Ve^{-F}\hat{\Upsilon}_{t-\tau}^{-1}\Psi)(\V{x})}\big|
\\\nonumber
&\le\sup_{\V{z}\in\RR^\nu}\EE[e^{-8\int_0^\tau V(\V{B}^{\V{z}}_s)\Id s}]^{\nf{1}{8}}
\\\label{conny1}
&\quad\cdot\EE[e^{aq_1|\V{B}_\tau|}]^{\nf{1}{q_1}}
\sup_{\V{y}\in\RR^\nu}\sup_{{\phi\in\dom(\wt{\Upsilon}_0^{-1})\atop\|\phi\|=1}}
\EE[\|\wt{\Upsilon}_\tau\WW{\tau}{0}[\V{B}^{\V{y}}]\wt{\Upsilon}_0^{-1}\phi\|^{4}]^{\nf{1}{4}}
\EE[\|\Psi(\V{B}^{\V{x}}_\tau)\|^{p}]^{\nf{1}{p}},
\end{align}
and, for all $K\subset\RR^\nu$, $\V{x}\in K$, and $\tau_2\ge t$,
\begin{align}\nonumber
\sup_{{\phi\in\dom(\hat{\Upsilon}_t)\atop\|\phi\|=1}}&e^{F(\V{x})}
\big|\SPb{\hat{\Upsilon}_{t}\phi}{\big((T^{V_n}_\tau-T_\tau^V)
e^{-F}\hat{\Upsilon}_{t-\tau}^{-1}\Psi\big)(\V{x})}\big|
\\\nonumber
&\le\sup_{{\V{z}\in K\atop \tilde{t}\in[0,\tau_2]}}
\EE\big[|e^{-\int_0^{\tilde{t}}V_n(\V{B}_s^{\V{z}})\Id s}
-e^{-\int_0^{\tilde{t}}V(\V{B}_s^{\V{z}})\Id s}|^{8}\big]^{\nf{1}{8}}
\\\label{conny2}
&\quad\cdot\EE[e^{aq_1|\V{B}_\tau|}]^{\nf{1}{q_1}}
\sup_{\V{y}\in\RR^\nu}\sup_{{\phi\in\dom(\wt{\Upsilon}_0^{-1})\atop\|\phi\|=1}}
\EE[\|\wt{\Upsilon}_\tau\WW{\tau}{0}[\V{B}^{\V{y}}]\wt{\Upsilon}_0^{-1}\phi\|^{4}]^{\nf{1}{4}}
\EE[\|\Psi(\V{B}^{\V{x}}_\tau)\|^{p}]^{\nf{1}{p}}.
\end{align}
By virtue of Lem.~\ref{lem-spin4} (see also Rem.~\ref{rem-spin4}) and Hyp($\Upsilon$),
\begin{align}\label{aw}
\sup_{\V{y}\in\RR^\nu}
\EE[\|\wt{\Upsilon}_\tau\WW{\tau}{0}[\V{B}^{\V{y}}]\wt{\Upsilon}_0^{-1}\phi\|^{4}]^{\nf{1}{4}}
&\le7e^{4(1+\|\V{c}\|_{*,\infty}^2)\tau}\|\phi\|,\quad\phi\in\FHR,\,\tau\in[0,t].
\end{align}
Furthermore,
\begin{equation}\label{EEexp}
\EE[e^{\vr a|\V{B}_\tau|}]=\int_{\RR^\nu}p_1(\V{y},\V{0})e^{\vr a\tau^\eh|\V{y}|}\Id\V{y}
\le2^{\nf{\nu}{2}}e^{\vr^2a^2\tau},\quad\vr\ge0.
\end{equation}
In the case $p=q=\infty$, the term appearing in the last lines of \eqref{conny1} and \eqref{conny2}
 is thus bounded by $7\cdot2^{\nf{\nu}{2}}e^{8(1+\|\V{c}\|_{*,\infty}^2+a^2)\tau}\|\Psi\|_\infty$.
In the case $2\le p\le q<\infty$, we integrate
the $q$-th power of  \eqref{conny1} and \eqref{conny2} with respect to $\V{x}$ and use that
\begin{align}\nonumber
\int_{\RR^\nu}&\EE[\|\Psi(\V{B}^{\V{x}}_\tau)\|^{p}]^{\nf{q}{p}}\,\Id\V{x}
=\big\|e^{\tau\Delta/2}\|\Psi(\cdot)\|^p\big\|_{\nf{q}{p}}^{\nf{q}{p}}
\\\label{gandalf}
&\le c_{\nu,p,q}'\tau^{-q\nu(1-\nf{p}{q})/2p}\,\big\|\,\|\Psi(\cdot)\|^p\big\|_1^{\nf{q}{p}}
=c_{\nu,p,q}'\tau^{-q\nu(\nf{1}{p}-\nf{1}{q})/2}\,\|\Psi\|_p^{q},\quad\tau>0,
\end{align}
with a constant, $c_{\nu,p,q}'>0$, depending $\nu$, $p$, and $q$.
The case $2\le p<q=\infty$ is dealt with analogously, using
$$
\sup_{\V{x}\in\RR^\nu}\EE[\|\Psi(\V{B}^{\V{x}}_\tau)\|^{p}]=\sup_{\V{x}\in\RR^\nu}
\int_{\RR^\nu}p_\tau(\V{x},\V{y})\|\Psi(\V{y})\|^p\Id\V{y}\le(2\pi\tau)^{-\nf{\nu}{2}}\|\Psi\|_p^p,
\quad\tau>0.
$$
We may thus conclude that 
$(T_\tau^Ve^{-F}\hat{\Upsilon}_{t-\tau}^{-1}\Psi)(\V{x})\in\dom(\hat{\Upsilon}_t)$, 
for every $\V{x}\in\RR^\nu$, and obtain the following bounds, for all $2\le p\le q\le\infty$, 
\begin{align}\label{per1}
&\big\|e^F\hat{\Upsilon}_{t}T_\tau^V\hat{\Upsilon}_{t-\tau}^{-1}e^{-F}\big\|_{p,q}
\le c_{\nu,p,q}''\,\frac{e^{8(1+\|\V{c}\|_{*,\infty}^2+a^2)\tau}}{\tau^{\nu(\nf{1}{p}-\nf{1}{q})/2}}
\sup_{\V{z}\in\RR^\nu}\EE\big[e^{-8\int_0^\tau V(\V{B}_s^{\V{z}})\Id s}\big]^{\nf{1}{8}},
\\\nonumber
&\big\|1_Ke^{F}\hat{\Upsilon}_{t}(T^{V_n}_\tau-T^V_\tau)\hat{\Upsilon}_{t-\tau}^{-1}e^{-F}\big\|_{p,q}
\le c_{\nu,p,q}''\,\frac{e^{8(1+\|\V{c}\|_{*,\infty}^2+a^2)\tau}}{\tau^{\nu(\nf{1}{p}-\nf{1}{q})/2}}
\\\label{per2}
&\hspace{4.3cm}\cdot\sup_{{\V{z}\in K\atop \tilde{t}\in[0,\tau_2]}}\!\!
\EE\big[|e^{-\int_0^{\tilde{t}}V_n(\V{B}_s^{\V{z}})\Id s}-e^{-\int_0^{\tilde{t}}V(\V{B}_s^{\V{z}})\Id s}|^{8}
\big]^{\nf{1}{8}}.
\end{align}
{\em Step 2.} 
Again, let $2\le p\le q\le\infty$ and $\tau\in(0,t]$. 
By their definition, we may replace $\hat{\Upsilon}$ by $\wt{\Upsilon}$ in 
\eqref{per1} and \eqref{per2}.
Then, in the analogue of \eqref{per1}, the expression 
$e^F\wt{\Upsilon}_{t}T_\tau^V\wt{\Upsilon}_{t-\tau}^{-1}e^{-F}$ is a densely defined operator
in $L^p(\RR^\nu,\FHR)$ with domain 
$$
\big\{\Psi\in L^p(\RR^\nu,\FHR):\,\Psi\in\dom(\wt{\Upsilon}_{t-\tau}^{-1})\,\text{a.e.,}\,
\wt{\Upsilon}_{t-\tau}^{-1}\Psi\in L^p(\RR^\nu,\FHR)\big\};
$$
recall that $T^V_\tau$ was well-defined on $e^{-F}L^p(\RR^\nu,\FHR)$.
Hence, \eqref{per1} can be rephrased by saying that 
$\wt{\Upsilon}_{t}T_\tau^{V;(p,q),F}\wt{\Upsilon}_{t-\tau}^{-1}$ is a densely defined bounded
operator between the weighted spaces $L^p(\RR^\nu,\FHR;e^{pF}\Id\V{x})$ and 
$L^q(\RR^\nu,\FHR;e^{qF}\Id\V{x})$ with domain
$$
\sX:=\big\{\Psi\in L^p(\RR^\nu,\FHR;e^{pF}\Id\V{x}):\,
\Psi\in\dom(\wt{\Upsilon}_{t-\tau}^{-1})\,\text{a.e.,}\,
\wt{\Upsilon}_{t-\tau}^{-1}\Psi\in L^p(\RR^\nu,\FHR;e^{pF}\Id\V{x})\big\},
$$
whose norm is bounded by the right hand side of \eqref{per1};
recall the notation introduced for certain restrictions of $T_\tau^V$ in the parapraph
preceding Cor.~\ref{cor-SASG}. For all $\Psi\in\sX$ and 
$\Phi\in L^{q'}(\RR^\nu,\FHR;e^{-q'F}\Id\V{x})$ such that 
$\Phi\in\dom(\wt{\Upsilon}_t)$ a.e. and 
$\wt{\Upsilon}_t\Phi\in L^{q'}(\RR^\nu,\FHR;e^{-q'F}\Id\V{x})$,
we further infer from Cor.~\ref{cor-SASG} and Step~1 that
\begin{align*}
\Big|&\int_{\RR^\nu}\SPb{(T_\tau^{V;(q',p'),-F}\wt{\Upsilon}_t\Phi)(\V{x})}{
\wt{\Upsilon}_{t-\tau}^{-1}\Psi(\V{x})}\Id\V{x}\Big|
\\
&=
\Big|\int_{\RR^\nu}\SPb{\Phi(\V{x})}{
\big(\wt{\Upsilon}_{t}T_\tau^{V;(p,q),F}\wt{\Upsilon}_{t-\tau}^{-1}\Psi\big)(\V{x})}\Id\V{x}\Big|
\\
&\le \big\|e^F\wt{\Upsilon}_{t}T_\tau^V\wt{\Upsilon}_{t-\tau}^{-1}e^{-F}\big\|_{p,q}
\|\Phi\|_{L^{q'}(\RR^\nu,\FHR;e^{-q'F}\Id\V{x})}
\|\Psi\|_{L^{p}(\RR^\nu,\FHR;e^{pF}\Id\V{x})}.
\end{align*}
Observing also that $\hat{\Upsilon}_\tau=\wt{\Upsilon}_{t-\tau}^{-1}$, considered as a
densely defined operator in $L^{p'}(\RR^\nu,\FHR;e^{-p'F}\Id\V{x})$,
is the adjoint of itself, considered as a
densely defined operator in $L^{p}(\RR^\nu,\FHR;e^{pF}\Id\V{x})$,
we conclude that $T_\tau^{V;(q',p'),-F}\hat{\Upsilon}_{t-\tau}\Phi$ is in the domain of
$\hat{\Upsilon}_\tau$, acting in $L^{p'}(\RR^\nu,\FHR;e^{-p'F}\Id\V{x})$, and that
\begin{align*}
\big\|e^{-F}\hat{\Upsilon}_\tau T_\tau^{V}&\hat{\Upsilon}_{0}^{-1}e^F\big\|_{q',p'}
\le c_{\nu,p,q}''\,\frac{e^{8(1+\|\V{c}\|_{*,\infty}^2+a^2)\tau}}{\tau^{\nu(\nf{1}{p}-\nf{1}{q})/2}}
\sup_{\V{z}\in\RR^\nu}\EE\big[e^{-8\int_0^\tau V(\V{B}_s^{\V{z}})\Id s}\big]^{\nf{1}{8}},
\end{align*}
where $\nf{1}{p}-\nf{1}{q}=\nf{1}{q'}-\nf{1}{p'}$.
Altogether this proves \eqref{norm-T-F-Theta} in the case $1\le p\le q\le 2$.

{\em Step 3.} Let $1\le p<2< q\le\infty$. Then the semi-group property 
and the mapping properties of $T^V$ established in Steps~1 and~2 imply
\begin{align*}
\big\|e^F\Upsilon_tT_t^V{\Upsilon}_0^{-1}e^{-F}\big\|_{p,q}\le
\big\|e^F\Upsilon_tT^V_{\nf{t}{2}}{\Upsilon}_{\nf{t}{2}}^{-1}e^{-F}\big\|_{2,q}
\big\|e^F{\Upsilon}_{\nf{t}{2}}T_{\nf{t}{2}}^V{\Upsilon}_0^{-1}e^{-F}\big\|_{p,2}.
\end{align*}
Since Steps~1 and~2 also yield bounds on the first and second factor on the right hand side,
respectively, this completes the proof of Part~(1).

{\em Step 4.} Let $p\in[1,2)$ and $q\in[2,\infty]$. Analogously as in Step~3 we then obtain
\begin{align}\nonumber
\big\|1_Ke^F&\Upsilon_t(T^{V_n}_t-T^V_t){\Upsilon}_0^{-1}e^{-F}\big\|_{p,q}
\\\nonumber
&\le\big\|1_Ke^F\Upsilon_t(T^{V_n}_{\nf{t}{2}}-T^V_{\nf{t}{2}})
{\Upsilon}_{\nf{1}{2}}^{-1}e^{-F}\big\|_{2,q}
\big\|e^F{\Upsilon}_{\nf{t}{2}}T^{V}_{\nf{t}{2}}{\Upsilon}_0^{-1}e^{-F}\big\|_{p,2}
\\\label{conny7}
&\quad+\big\|1_Ke^{F-D}\Upsilon_tT^{V_n}_{\nf{t}{2}}{\Upsilon}_{\nf{t}{2}}^{-1}e^{-F+D}\big\|_{2,q}
\big\|e^{F-D}{\Upsilon}_{\nf{t}{2}}(T^{V_n}_{\nf{t}{2}}-T^V_{\nf{t}{2}})
{\Upsilon}_0^{-1}e^{-F}\big\|_{p,2},
\end{align}
where $D(\V{x}):=\dist(\V{x},K)$, $\V{x}\in\RR^\nu$, is globally Lipschitz continuous.
According to the first two steps, the term in the second line of \eqref{conny7} 
can be estimated using \eqref{per1} and \eqref{per2}. The duality arguments of Step~2
and the trivial bound $\|e^{-D}\|_{2,2}\le1$ further imply
\begin{align}\nonumber
\big\|e^{F-D}{\Upsilon}_{\nf{t}{2}}(T^{V_n}_{\nf{t}{2}}-T^V_{\nf{t}{2}})
{\Upsilon}_0^{-1}e^{-F}\big\|_{p,2}
&\le\big\|1_{K_R}e^{D-F}\wt{\Upsilon}_t
(T^{V_n}_{\nf{t}{2}}-T^V_{\nf{t}{2}})\wt{\Upsilon}_{\nf{t}{2}}^{-1}e^{F-D}\big\|_{2,p'}
\\\label{per3}
&+e^{-R}\big\|e^{-F+D}\wt{\Upsilon}_t(T^{V_n}_{\nf{t}{2}}-T^V_{\nf{t}{2}})
\wt{\Upsilon}_{\nf{t}{2}}^{-1}e^{F-D}\big\|_{2,p'},
\end{align}
where $K_R:=D^{-1}([0,R])$, $R\ge1$, so that $\|e^{-D}1_{K_R^c}\|_{p',p'}\le e^{-R}$.
By \eqref{approx3a} and \eqref{per1}, the norm in the second line of \eqref{per3} is
bounded uniformly in $n\in\NN$, all $F$ with Lipschitz constant $\le a$, and all coefficient vectors 
$\V{c}$ with $\|\V{c}\|_{*,\infty}\le A$. Therefore, the whole
term in the second line can be made arbitrarily small by choosing $R$ large.
Moreover, \eqref{per2} applies to the norm in the first line of the right hand side
of \eqref{per3} after we have fixed some large value of $R$.
Finally, we note that the first norm in the last line of \eqref{conny7} is bounded uniformly in 
$n\in\NN$ by \eqref{approx3a} and \eqref{per1} as well.

{\em Step 5.} We have already proved \eqref{LNCV1} in all cases except for
$1\le p\le q<2$. By duality, we have also proved \eqref{LNCV2} in all cases
except for $2<p\le q\le \infty$. 

We will now treat \eqref{LNCV2} in these missing cases. So, let $2<p\le q\le \infty$.
Define $K_R$ and $D$ as in Step~4.
On account of the already proven cases of \eqref{LNCV1} and $1_K=e^{-D}1_K$,
it is enough to show that
\begin{align*}
\sup_{{n\in\NN}}\sup_{t\in[\tau_1,\tau_2]}\big\|e^{-D}1_{K_R^c}e^{F+D}
\Upsilon_t(T_t^{V_n}-T^V_t){\Upsilon}_0^{-1}e^{-F-D}1_K\big\|_{p,q}\xrightarrow{\;\;R\to\infty\;\;}0,
\end{align*}
where the convergence should in addition be uniform in $F$ and $\V{c}$ as in the last assertion
of Part~(2). This is, however, obvious from 
$\|e^{-D}1_{K_R^c}\|\le e^{-R}$, \eqref{approx3a}, and \eqref{norm-T-F-Theta}.
Again by duality, these arguments can also be used to cover \eqref{LNCV1} 
in the remaining cases$1\le p\le q<2$.
\end{proof}

\begin{cor}\label{cor-WEK}
Let $V\in\cK_\pm(\RR^\nu)$. Let $(\Upsilon_t)_{t\ge0}$ and $\|\cdot\|_*$ be given by one of the 
lines in Table~1 and assume that $\V{c}$ fulfills {\rm Hyp($\Upsilon$)} given by the same line.
Then there exist a universal constant $c>0$ such that 
$\Ran(T^V_t(\V{x},\V{y}))\subset\dom(\Upsilon_{\nf{t}{12}})$ with
\begin{align*}
\|\Upsilon_{\nf{t}{12}}T^V_t(\V{x},\V{y})\|&\le c_\nu t^{-\nf{\nu}{2}}
e^{c\|\V{c}\|_{*,\infty}^2t-|\V{x}-\V{y}|^2/ct}\sup_{\V{z}\in\RR^\nu}
\EE\big[e^{8\int_0^tV_-(\V{B}_r^{\V{z}})\Id r}\big]^{\nf{3}{8}},
\end{align*}
for all $t>0$ and $\V{x},\V{y}\in\RR^\nu$.
\end{cor}

\begin{proof}
Set $F_{\V{x}}(\V{z}):=a|\V{z}-\V{x}|$, $\V{z}\in\RR^\nu$,
for some $a>0$. By virtue of Prop.~\ref{prop-CK} we can then write, for any $s\in(0,t)$,
\begin{align*}
&\|\Upsilon_sT^V_t(\V{x},\V{y})\|
=\sup_{\|\psi\|=1}\big\|(e^{-F_{\V{x}}}\Upsilon_tT^V_{s})(T_{t-s}^V(\cdot,\V{y})\psi)(\V{x})\big\|
\\
&\le\|e^{-F_{\V{\V{x}}}}\Upsilon_sT^V_{s}e^{F_{\V{x}}}\|_{1,\infty}
\int_{\RR^\nu}e^{-F_{\V{x}}(\V{z})}\|T_{t-s}^V(\V{z},\V{y})\|\Id\V{z}
\\
&\le c_\nu s^{-\nf{\nu}{2}}{e^{8(1+\|\V{c}\|_{*,\infty}^2+a^2)s}}\sup_{\V{z}\in\RR^\nu}
\EE\big[e^{-8\int_0^sV(\V{B}_r^{\V{z}})\Id r}\big]^{\nf{1}{4}}
\int_{\RR^\nu}e^{-F_{\V{x}}(\V{z})}\|T_{t-s}^V(\V{z},\V{y})\|\Id\V{z},
\end{align*}
where we used \eqref{norm-T-F-Theta} in the last estimate. Furthermore, \eqref{T(x,y)-sym} yields
\begin{align*}
&\int_{\RR^\nu}e^{-F_{\V{x}}(\V{z})}\|T_{t-s}^V(\V{z},\V{y})\|\Id\V{z}
\\
&=\int_{\RR^\nu}e^{-a|\V{x}-\V{z}|+\|\mho\|_\infty(t-s)}p_{t-s}(\V{y},\V{z})
\EE\big[e^{-\int_0^{t-s}V(\V{b}_r^{t-s;\V{y},\V{z}})\Id r}\big]\Id\V{z}
\\
&\le e^{\|\mho\|_\infty t}\Big(\int_{\RR^\nu}\!\!e^{-2F_{\V{x}}(\V{z})}
p_{t-s}(\V{y},\V{z})\Id\V{z}\Big)^{\frac{1}{2}}
\Big(\int_{\RR^\nu}\!p_{t-s}(\V{y},\V{z})\EE\big[e^{-\int_0^{t-s}V(\V{b}_r^{t-s;\V{y},\V{z}})\Id r}
\big]^2\Id\V{z}\Big)^{\frac{1}{2}}
\\
&\le e^{\|\mho\|_\infty t}\Big((2\pi(t-s))^{-\nf{\nu}{2}}\int_{\RR^\nu}\!
e^{-2a|\V{z}|-\frac{|\V{z}+\V{x}-\V{y}|^2}{2(t-s)}}\Id\V{z}\Big)^{\frac{1}{2}}
\Big(\int_{\RR^\nu}S_{t-s}^{2V}(\V{z},\V{y})\Id\V{z}\Big)^{\frac{1}{2}}
\end{align*}
On account of $S_{t-s}^{2V}(\V{z},\V{y})=S_{t-s}^{2V}(\V{y},\V{z})$ and \eqref{def-StV(x,y)1}
we further observe that
\begin{align*}
\int_{\RR^\nu}S_{t-s}^{2V}(\V{z},\V{y})\Id\V{z}=(S_{t-s}^{2V}1)(\V{y})
\le\EE\big[e^{8\int_0^tV_-(\V{B}_r^{\V{y}})\Id r}\big]^{\nf{1}{4}}.
\end{align*}
Choosing $a:=|\V{x}-\V{y}|/2(t-s)$, we finally have
$$
2a|\V{z}|+\frac{|\V{z}+
\V{x}-\V{y}|^2}{2(t-s)}\ge\frac{|\V{x}-\V{y}||\V{z}|}{(t-s)}+\frac{(|\V{z}|-|\V{x}-\V{y}|)^2}{2(t-s)}
=\frac{|\V{z}|^2}{2(t-s)}+\frac{|\V{x}-\V{y}|^2}{2(t-s)}.
$$
Inserting this bound and choosing $s:=t/8$ we arrive at the asserted estimate.
\end{proof}

\begin{cor}\label{cor-SCV}
Let $V\in\cK_{\pm}(\RR^\nu)$, let $V_n\in\cK(\RR^\nu)$, $n\in\NN$, satisfy \eqref{approx1}
and \eqref{approx2}, and let $F:\RR^\nu\to\RR$ be globally Lipschitz continuous. 
Let $(\Upsilon_t)_{t\ge0}$ and $\|\cdot\|_*$ be given by one of the lines 
in Table~1 and assume that $\V{c}$ fulfills {\rm Hyp($\Upsilon$)} given by the same line.
Furthermore, let $p\in[1,\infty)$ be finite and $q\in[p,\infty]$. Then
\begin{align}\label{stefan}
\lim_{n\to\infty}\sup_{t\in[\tau_1,\tau_2]}
\big\|e^F\Upsilon_t(T^{V_n}_t-T^V_t){\Upsilon}_0^{-1}e^{-F}\Psi\big\|_{q}=0,
\end{align}
for all $\tau_2>\tau_1>0$ and $\Psi\in L^p(\RR^\nu,\FHR)$. 
In the case $p=q$, the value $\tau_1=0$ is allowed for in \eqref{stefan} as well.
The convergence \eqref{stefan} is in fact uniform in all Lipschitz 
continuous $F$ with Lipschitz constant $\le a$ and all coefficient vectors $\V{c}$ satisfying
{\rm Hyp($\Upsilon$)} and $\|\V{c}\|_{*,\infty}\le A$, for some given $a,A\in(0,\infty)$.
\end{cor}

\begin{proof}
Set $\ol{B}_R:=\{\V{x}\in\RR^\nu:|\V{x}|\le R\}$, $R>0$. Then $\|1_{\ol{B}_R^c}\Psi\|_p\to0$,
$R\to\infty$, and the claim follows easily from \eqref{norm-T-F-Theta} and \eqref{LNCV2} 
with $K=\ol{B}_R$.
\end{proof}


\section{Continuity in the coupling functions}\label{sec-coup}

\noindent
In the following theorem we complement the discussion of the previous section by considering 
the behavior of the semi-group under perturbations of the coefficient vector $\V{c}$. 
Thanks to its last assertion, the following theorem can be combined with 
Thm.~\ref{thm-LNCV}(2), where the convergences are uniform in coefficient vectors
as in the statement of Thm.~\ref{thm-coup}.

\begin{thm}\label{thm-coup}
Let $V\in\cK_\pm(\RR^\nu)$, $F:\RR^\nu\to\RR$ be globally Lipschitz continuous with Lipschitz
constant $a\ge0$, and $1\le p\le q\le\infty$. 
Let $(\Upsilon_t)_{t\ge0}$ and $\|\cdot\|_*$ be given by one of the lines in Table~1. 
Assume that $\V{c}=(\V{G},q,\V{F})$ and $\V{c}_n=({\V{G}}_n,{q}_n,{\V{F}}_n)$, $n\in\NN$, 
are coefficient vectors satisfying
Hyp.~\ref{hyp-G} with the same $\omega$ and the same conjugation $C$ and {\rm Hyp($\Upsilon$)}.
Assume further that $\|\V{c}_{\V{x}}\|,\|\V{c}_{n,\V{x}}\|_*\le A$, for some $A\in[1,\infty)$, 
all $\V{x}\in\RR^\nu$, and all $n\in\NN$, and that
\begin{align}\label{marie1}
\lim_{n\to\infty}\sup_{\V{x}\in K}\|\V{c}_{n,\V{x}}-\V{c}_{\V{x}}\|_*&=0,
\end{align}
for all compact $K\subset\RR^\nu$.
Let $({T}_t^{V,n})_{t\ge0}$ denote the semi-group defined by means of ${\V{c}}_n$. Then
\begin{align}\label{marie3}
\lim_{n\to\infty}\sup_{t\in[\tau_1,\tau_2]}
\big\|1_Ke^F\Upsilon_t(T^V_t-{T}^{V,n}_t){\Upsilon}_0^{-1}e^{-F}\big\|_{p,q}&=0,
\\\label{marie4}
\lim_{n\to\infty}\sup_{t\in[\tau_1,\tau_2]}
\big\|e^F\Upsilon_t(T^V_t-{T}^{V,n}_t){\Upsilon}_0^{-1}e^{-F}1_K\big\|_{p,q}&=0,
\end{align}
for all  $1\le p\le q\le\infty$, $0<\tau_1\le\tau_2$, and
compact $K\subset\RR^\nu$. If $K$ can be replaced by $\RR^\nu$ in \eqref{marie1},
then it may also be replaced by $\RR^\nu$ in \eqref{marie3} and \eqref{marie4}.
If $p=q$, then the choice $\tau_1=0$ is allowed for in \eqref{marie3} and \eqref{marie4} as well.
In all cases, the convergences in \eqref{marie3} and \eqref{marie4} are uniform
in all globally Lipschitz continuous $F:\RR^\nu\to\RR$ with Lipschitz constant $\le a$ and all
$V\in\cK_\pm(\RR^\nu)$ satisfying 
$\sup_{t\in[\tau_1,\tau_2]}\sup_{\V{x}\in\RR^\nu}\EE[e^{-\int_0^{t}V(\V{B}_s^{\V{x}})\Id s}]\le D$,
for some fixed $a,D\in(0,\infty)$. 
\end{thm}

\begin{proof}
We shall only explain a substitute of Step~1 of the proof of Thm.~\ref{thm-LNCV},
as the remaining steps of this proof are completely analogous to the proof of \eqref{LNCV1} 
and \eqref{LNCV2}. 
To this end, we let $({\mathbb{W}}_\tau^{0,n}[\V{B}^{\V{x}}])_{\tau\ge0}$, $\V{x}\in\RR^\nu$, 
denote the solution processes of Thm.~\ref{thm-WW} defined by means of ${\V{c}}_n$.
Again, we fix $t>0$ and define $(\hat{\Upsilon}_s)_{s\in[0,t]}$ and 
$\wt{\Upsilon}_s:=\hat{\Upsilon}_{t-s}^{-1}$, 
$s\in[0,t]$, in the same way as in the beginning of the proof of Thm.~\ref{thm-LNCV}.

Assume first that $p\ge2$. Define $q_1>0$ by $q_1^{-1}=1-p^{-1}-4^{-1}-8^{-1}$.
For all $\V{x}\in K\subset\RR^\nu$ and $\Psi\in L^p(\RR^\nu,\FHR)$
with $\Psi(\V{z})\in\dom(\hat{\Upsilon}_{t-\tau}^{-1})$, a.e. $\V{z}$, and
$\hat{\Upsilon}_{t-\tau}^{-1}\Psi\in L^p(\RR^\nu,\FHR)$, we then obtain
\begin{align}\nonumber
\big\|\big(e^F\hat{\Upsilon}_t(T^V_\tau&-{T}^{V,n}_\tau)
\hat{\Upsilon}_{t-\tau}^{-1}e^{-F}\Psi\big)(\V{x})\big\|
\\\nonumber
\le\sup_{{\phi\in\dom(\wt{\Upsilon}_0^{-1})\atop\|\phi\|=1}}&
\EE\big[e^{F(\V{x})-F(\V{B}_t^{\V{x}})-\int_0^tV(\V{B}_s^{\V{x}})\Id s}
\|\wt{\Upsilon}_\tau(\WW{\tau}{0}[\V{B}^{\V{x}}]-{\mathbb{W}}_\tau^{0,n}[\V{B}^{\V{x}}])
\wt{\Upsilon}_0^{-1}\phi\|\|\Psi(\V{B}_\tau^{\V{x}})\|\big]
\\\nonumber
\le\sup_{{\phi\in\dom(\wt{\Upsilon}_0^{-1})\atop\|\phi\|=1}}
&\sup_{\V{y}\in K}\EE\big[\|\wt{\Upsilon}_\tau(\WW{\tau}{0}[\V{B}^{\V{y}}]
-{\mathbb{W}}_\tau^{0,n}[\V{B}^{\V{y}}])\wt{\Upsilon}_0^{-1}\phi\|^4\big]^{\nf{1}{4}}
\\\label{per6}
&\cdot
\EE[e^{aq_1|\V{B}_\tau|}]^{\nf{1}{q_1}}\sup_{\V{y}\in\RR^\nu}
\EE[e^{-8\int_0^\tau V(\V{B}_s^{\V{y}})\Id s}]^{\nf{1}{8}}
\EE[\|\Psi(\V{B}_\tau^{\V{x}})\|^p]^{\nf{1}{p}},
\end{align}
with $\EE[\|\Psi(\V{B}_\tau^{\V{x}})\|^\infty]^{\nf{1}{\infty}}:=\|\Psi\|_\infty$ in the case $p=\infty$.
Applying Lem.~\ref{lem-maria} with $\tilde{\V{c}}:=\V{c}_n$, we further obtain
\begin{align}\nonumber
\EE\big[&\|\wt{\Upsilon}_\tau(\WW{\tau}{0}[\V{B}^{\V{y}}]
-{\mathbb{W}}_\tau^{0,n}[\V{B}^{\V{y}}])\wt{\Upsilon}_0^{-1}\phi\|^{4}\big]
\\\label{per7a}
&\le cA^4e^{cA^2\tau}\EE\big[\sup_{s\le\tau}\|\V{c}_{n,\V{B}_s^{\V{y}}}
-\V{c}_{\V{B}_s^{\V{y}}}\|_*^4\big]^\eh\|\phi\|^4,
\end{align}
for all $\V{y}\in\RR^\nu$, $\phi\in\dom(\wt{\Upsilon}_0^{-1})$, and some universal constant $c>0$.
Since $\{\V{B}_s^{\V{y}}(\vgamma):\,s\in[0,t],\V{y}\in K\}\subset\RR^\nu$ is compact,
for all $\vgamma\in\Omega$ and all compact $K\subset\RR^\nu$, the dominated convergence
theorem and the validity of \eqref{marie1} for all compact sets imply
\begin{align}\label{per7}
\sup_{\V{y}\in K}\EE\big[\sup_{s\le\tau}\|\V{c}_{n,\V{B}_s^{\V{y}}}
-\V{c}_{\V{B}_s^{\V{y}}}\|_*^{4}\big]
&\le\EE\big[\sup_{{s\le \tau\atop\V{y}\in K}}\|\V{c}_{n,\V{B}_s^{\V{y}}}
-\V{c}_{\V{B}_s^{\V{y}}}\|_*^{4}\big]\xrightarrow{\;\;n\to\infty\;\;}0,\quad\tau>0.
\end{align}
Therefore, the term in the third line of \eqref{per6} converges to zero as $n\to\infty$,
locally uniformly in $\tau\ge0$ and for all compact $K$. 
If $K$ can be replaced by $\RR^\nu$ in \eqref{marie1}, then it can be replaced by $\RR^\nu$
in \eqref{per7} as well, which then shows that the term in the third line of \eqref{per6} with
$K=\RR^\nu$ tends to zero as $n\to\infty$. We may now proceed along the lines of the proof
of \eqref{LNCV1} and \eqref{LNCV2} to conclude.
\end{proof}


\section{Strong continuity in the time parameter}\label{sec-SG}

\noindent Next, we consider the strong continuity with respect to $t\ge0$ of our semi-groups
in the $L^p$-spaces with a finite $p$.
The somewhat technical statement of Part~(1) of the next theorem, which implies its second part,
will be needed in the proof of Thm.~\ref{thm-kern-cont} on the {\em operator norm}-continuity of the 
integral kernel. It is perhaps needless to recall that a set, $\cM$, 
of functions from $\RR^\nu$ to $\FHR$
is called {\em uniformly equicontinuous on a subset $\cQ\subset\RR^\nu$}, iff
$$
\lim_{r\downarrow0}\sup_{{\V{x},\V{y}\in\cQ\atop|\V{x}-\V{y}|<r}}\sup_{\Psi\in\cM}
\|\Psi(\V{x})-\Psi(\V{y})\|=0.
$$
In Part~(2) of Thm.~\ref{thm-SC} we again use the notation for restrictions of $T^V_t$ introduced
in the paragraph preceding Cor.~\ref{cor-SASG}.

\begin{thm}\label{thm-SC}
Let $V\in\cK_\pm(\RR^\nu)$ and $p\in[1,\infty)$. Then the following holds:

\smallskip

\noindent{\rm(1)}
Let  $\cM\subset L^p(\RR^\nu,\FHR)\cap L^\infty(\RR^\nu,\FHR)$ be such that
\begin{enumerate}
\item[{\rm(a)}] $\cM$ is uniformly equicontinuous on every compact subset of $\RR^\nu$;
\item[{\rm(b)}]  $\sup_{\Psi\in\cM}\|\Psi\|_\infty<\infty$ and, for every $\Psi\in\cM$, there exists
${\V{a}}_\Psi\in\RR^\nu$ such that
$\sup_{\Psi\in\cM}\|\Psi-\Psi_R\|_p\to0$, $R\to\infty$, where
$\Psi_R(\V{x}):=1_{|\V{x}-\V{a}_\Psi|<R}\Psi(\V{x})$;
if $V\notin\cK(\RR^\nu)$, then we assume that the points $\V{a}_\Psi$, $\Psi\in\cM$, 
are contained in a compact set;
\item[{\rm(c)}] for every $\Psi\in\cM$, we have $\Psi(\V{x})\in\dom(\Id\Gamma(\omega)^\eh)$, a.e. 
$\V{x}\in\RR^\nu$, and
$\sup_{\Psi\in\cM}\|f_\Psi\|_p<\infty$, where $f_\Psi(\V{x}):=\|(1+\Id\Gamma(\omega))^\eh\Psi(\V{x})\|$.
\end{enumerate}
Furthermore, let $F:\RR^\nu\to\RR$ be globally Lipschitz continuous and set
$F_\Psi(\V{x}):=F(\V{x}-\V{a}_\Psi)$, $\V{x}\in\RR^\nu$. 
Then
\begin{align*}
\lim_{t\downarrow0}\sup_{\Psi\in\cM}\|e^{-F_\Psi}(T_t^{V}-\id)e^{F_\Psi}\Psi\|_p=0.
\end{align*}
{\rm(2)}
The semi-group $(T_t^{V;(p,p),0})_{t\ge0}$ is strongly continuous.
\end{thm}

\begin{proof}
(1): Let $\chi\in C(\RR,[0,1])$ satisfy $\chi_0=1$ on $(-\infty,0]$ and $\chi=0$ on $[1,\infty)$ and set
$\Psi_R'(\V{x}):=\chi(|\V{x}-\V{a}_\Psi|-R)\Psi(\V{x})$, $\V{x}\in\RR^\nu$, $R\ge1$. On account of 
$$
\sup_{\Psi\in\cM}\sup_{0\le t\le 1}\|e^{-F_\Psi}T_t^Ve^{F_\Psi}\|_{p,p}<\infty,
$$ 
Condition (b), and $\|\Psi-\Psi_R'\|_p\le\|\Psi-\Psi_R\|_p$ it suffices to show that
\begin{align*}
\lim_{t\downarrow0}\sup_{\Psi\in\cM}\|e^{-F_\Psi}(T_t^V-\id)e^{F_\Psi}\Psi_R'\|_p=0,
\end{align*}
for every $R\ge2$. To do so we estimate, using 
$\|\WW{t}{0}[\V{B}^{\V{x}}]\|\le e^{\|\mho\|_\infty t}$, 
\begin{align*}
&\|e^{-F_\Psi}(T_t^V-\id)e^{F_\Psi}\Psi_R'\|_p^p
\\
&=\int_{\RR^\nu}\big\|\EE\big[
\WW{t}{0}[\V{B}^{\V{x}}]^*\big(e^{-\int_0^tV(\V{B}_s^{\V{x}})\Id s
-F_\Psi(\V{x})+F_\Psi(\V{B}^{\V{x}}_t)}
\Psi_R'(\V{B}^{\V{x}}_t)-\Psi_R'(\V{x})\big)
\\
&\hspace{6cm}
+(\WW{t}{0}[\V{B}^{\V{x}}]^*-\id)\Psi_R'(\V{x})\big]\big\|^p\Id\V{x}
\\
&\le 2^{p-1}e^{p\|\mho\|_\infty t}\big(I_{1,1}(\Psi,R,t)+I_{1,2}(\Psi,R,t)\big)+2^{p-1}\,I_2(\Psi,t)\,,
\end{align*}
for all $\Psi\in\cM$ and $R\ge2$, where
\begin{align*}
I_{1,1}(\Psi,R,t)&:=\int\limits_{|\V{x}-\V{a}_\Psi|<2R}\EE\big[\|e^{-\int_0^tV(\V{B}^{\V{x}}_s)\Id s
-F_\Psi(\V{x})+F_\Psi(\V{B}^{\V{x}}_t)}\Psi_R'(\V{B}_t^{\V{x}})-\Psi_R'(\V{x})\|^p\big]\Id\V{x},
\\
I_{1,2}(\Psi,R,t)&:=\int\limits_{|\V{x}-\V{a}_\Psi|\ge2R}\EE\big[e^{ap|\V{B}_t|}
\|e^{-\int_0^tV(\V{B}^{\V{x}}_s)\Id s}\Psi_R'(\V{B}_t^{\V{x}})\|^p\big]\Id\V{x},
\\
I_2(\Psi,t)&:=\int_{\RR^\nu}\big\|
\EE\big[(\WW{t}{0}[\V{B}^{\V{x}}]^*-\id)\Psi(\V{x})\big]\big\|^p\Id\V{x}.
\end{align*}
The first integral can be estimated as
\begin{align}\nonumber
\sup_{\Psi\in\cM}&I_{1,1}(\Psi,R,t)
\\\nonumber
&\le 
3^{p-1}c_\nu R^\nu\sup_{\Psi\in\cM}\sup_{|\V{x}-\V{a}_\Psi|<2R}
\EE\big[|e^{F_\Psi(\V{B}^{\V{x}}_t)-F_\Psi(\V{x})}-1|^{p}
e^{-p\int_0^tV(\V{B}^{\V{x}}_s)\Id s}\big]\|\Psi\|_\infty^p
\\\nonumber
&\quad +3^{p-1}c_{\nu} R^\nu\sup_{\Psi\in\cM}\sup_{|\V{x}-\V{a}_\Psi|\le2R}
\EE\big[|e^{-\int_0^tV(\V{B}_s^{\V{x}})\Id s}-1|^p\big]
\,\|\Psi\|_\infty^p
\\\label{regina2}
&\quad+3^{p-1}\sup_{\Psi\in\cM}
\int_{|\V{x}-\V{a}_\Psi|<2R}\int_{\RR^\nu}p_t(\V{x},\V{y})
\|\Psi_R'(\V{y})-\Psi_R'(\V{x})\|^p\Id\V{y}\,\Id\V{x}.
\end{align}
Here the term in the first line of the right hand side goes to zero, as $t\downarrow0$,
due to \eqref{dirk0} (with $p$ replaced by $2p$) and 
\begin{align*}
\EE\big[|e^{F_\Psi(\V{B}^{\V{x}}_t)-F_\Psi(\V{x})}-1|^{2p}\big]
&\le \EE\big[(a|\V{B}_t|e^{a|\V{B}_t|})^{2p}\big]
=\int_{\RR^\nu}p_1(\V{y},\V{0})(at^\eh|\V{y}|e^{at^\eh|\V{y}|})^{2p}\Id\V{y},
\end{align*}
where $a$ denotes a Lipschitz constant for $F$.
The term in the second line of the right hand side vanishes in the limit $t\downarrow0$
by virtue of Lem.~\ref{lem-dirk}(2) and Condition~(b). The term in the last line of \eqref{regina2}
goes to zero as a consequence of Condition~(a). Furthermore,
\begin{align*}
I_{1,2}(\Psi,R,t)&\le\|\Psi\|_\infty^p\sup_{\V{x}\in\RR^\nu}
\EE\big[e^{-2p\int_0^tV(\V{B}_s^{\V{x}})\Id s}\big]^\eh\EE[e^{2ap|\V{B}_t|}]^\eh
\\
&\qquad\cdot
\int_{|\V{x}|\ge2R}\int_{|\V{y}|<R+1}p_t(\V{x},\V{y})\Id\V{y}\,\Id\V{x},
\end{align*}
so that $\sup_{\Psi\in\cM}I_{1,2}(\Psi,R,t)\to0$, $t\downarrow0$, by Condition~(b),
\eqref{dirk0}, and \eqref{EEexp}.
Here we also used that $2R-(R+1)\ge1$, for $R\ge2$. Finally,  
\begin{align*}
\big\|\EE\big[(\WW{t}{0}[\V{B}^{\V{x}}]^*-\id)\Psi(\V{x})\big]\big\|
&=\sup_{\|\phi\|=1}\big\|(1+\Id\Gamma(\omega))^\mh
\EE\big[(\WW{t}{0}[\V{B}^{\V{x}}]-\id)\phi\big]\big\|\,f_\Psi(\V{x}),
\end{align*}
where the supremum runs over normalized elements of $\FHR$,
which together with \eqref{antonio} leads to the bound $I_2(\Psi,t)\le \cO(t^{\nf{p}{4}})\|f_\Psi\|_p^p$.
Invoking Condition~(c) we see that $\sup_{\Psi\in\cM}I_2(\Psi,t)\to0$, $t\downarrow0$.

(2): By the semi-group property it suffices to prove the strong continuity at zero only. 
Since $C_0(\RR^\nu,\FHR)$ is dense in $L^p(\RR^\nu\FHR)$ with
$p\in[1,\infty)$, $\|T_t^V\|_{p,p}$ is bounded uniformly in $t\in[0,1]$, and
$\Psi_\ve:=(1+\ve\Id\Gamma(\omega))^\mh\Psi\to\Psi$, $\ve\downarrow0$, 
in $L^p(\RR^\nu,\FHR)$, it is even sufficient to show that
$T_t^{V}\Psi_\ve\to\Psi_\ve$, $t\downarrow0$, in $L^p(\RR^\nu,\FHR)$, 
for every $\Psi\in C_0(\RR^\nu,\FHR)$ and $\ve>0$.
This follows, however, immediately from (2) upon choosing $\cM=\{\Psi_\ve\}$.
\end{proof}


\section{Equicontinuity in the image of the semi-group}\label{sec-eq-cont}

\noindent
Our next theorem, implying that $T^V_t$ with $t>0$ maps bounded sets in $L^p$
into equicontinuous ones, will be needed to prove the joint continity of the integral kernel
in Thm.~\ref{thm-kern-cont} as well. In the succeeding corollary we combine
the next theorem with Thms.~\ref{thm-LNCV} and~\ref{thm-coup}.

\begin{thm}\label{thm-eq-cont}
Let $V\in\cK_\pm(\RR^\nu)$, $0<\tau_1\le\tau_2<\infty$, $p\in[1,\infty]$, and 
let $(\Upsilon_t)_{t\ge0}$ and $\|\cdot\|_*$ be given by either Line~2 or Line~4 of Table~1. 
Assume that $\RR^\nu\ni\V{x}\mapsto\V{c}_{\V{x}}$ is bounded and continuous
with respect to the norm $\|\cdot\|_*$. Then the following set of $\FHR$-valued functions,
\begin{align}\label{thyra}
&\big\{\Upsilon_{\nf{t}{2}}T_t^V\Psi:\,\Psi\in L^p(\RR^\nu,\FHR),\,\|\Psi\|_p\le1,\,t\in[\tau_1,\tau_2]\big\},
\end{align}
is uniformly equicontinuous on every compact subset of $\RR^\nu$.
If $V\in\cK(\RR^\nu)$ and if $\RR^\nu\ni\V{x}\mapsto\V{c}_{\V{x}}$ is bounded and 
uniformly continuous with respect to $\|\cdot\|_*$, then the set in \eqref{thyra} is 
uniformly equicontinuous on $\RR^\nu$.
\end{thm}

\begin{proof}
{\em Step 1.} First, we consider the case $\Upsilon=\id$. To this end
we shall adapt an argument by Carmona~\cite{Carmona1979}; see also Sect.~4 in \cite{BHL2000}.
On account of the semi-group property and
\eqref{norm-T-F-Theta} it suffices to treat the case $p=\infty$.
In view of \eqref{norm-T-F-Theta} and the well-known properties
of the semi-group of the free Laplacian we know that, for fixed $0<\tau<\tau_1$, the set
$$
\big\{e^{\tau\Delta/2}T_{t-\tau}^V\Psi:\,
\Psi\in L^\infty(\RR^\nu,\FHR),\,\|\Psi\|_\infty\le1,\,t\in[\tau_1,\tau_2]\big\}
$$
is uniformly equicontinuous on $\RR^\nu$. Let $K\subset\RR^\nu$ be compact, 
if $V\in\cK_\pm(\RR^\nu)\setminus\cK(\RR^\nu)$, and $K=\RR^\nu$, if $V\in\cK(\RR^\nu)$.
Then it suffices to show that
$$
\lim_{\tau\downarrow0}\sup_{t\in[\tau_1,\tau_2]}\sup_{\V{x}\in K}\sup_{\|\Psi\|_\infty\le1}
\|(D_{t,\tau}\Psi)(\V{x})\|=0,
$$
where
\begin{align*}
D_{t,\tau}&:=e^{\tau\Delta/2}T_{t-\tau}^V-T_t^V=(e^{\tau\Delta/2}-T_\tau^V)T_{t-\tau}^V.
\end{align*}
Again by \eqref{norm-T-F-Theta} we know that
$\sup_{0\le\tau\le\tau_1/2}\sup_{t\in[\tau_1,\tau_2]}\|\theta^\eh T_{t-\tau}^V\|_{\infty,\infty}<\infty$,
where $\theta:=1+\Id\Gamma(\omega)$, and it remains to show that
\begin{equation}\label{anna1}
\lim_{\tau\downarrow0}\sup_{\V{x}\in K}\sup_{\|\Psi\|_\infty\le1}
\big\|\EE\big[(\id-\WW{\tau}{V}[\V{B}^{\V{x}}]^*)\theta^\mh\Psi(\V{B}_\tau^{\V{x}})\big]\big\|=0,
\end{equation}
for every compact $K\subset\RR^\nu$. The convergence \eqref{anna1} follows, however, from
\begin{align*}
\big\|\EE\big[\WW{\tau}{0}[\V{B}^{\V{x}}]^*(1-e^{-\int_0^\tau V(\V{B}_s^{\V{x}})\Id s}&)
\theta^\mh\Psi(\V{B}_\tau^{\V{x}})\big]\big\|
\\
&\le e^{\|\mho\|_\infty\tau}\EE\big[|1-e^{-\int_0^\tau V(\V{B}_s^{\V{x}})\Id s}|\big]\|\Psi\|_\infty
\end{align*}
in combination with \eqref{dirk2a} and from
\begin{align*}
\big\|\EE\big[(\WW{\tau}{0}[\V{B}^{\V{x}}]^*-\id)\theta^\mh
\Psi(\V{B}_\tau^{\V{x}})\big]\big\|
&=\sup_{\|\phi\|=1}\big|\EE\big[
\SPn{\theta^\mh(\WW{\tau}{0}[\V{B}^{\V{x}}]-\id)\phi}{\Psi(\V{B}_\tau^{\V{x}})}\big]\big|
\\
&\le\sup_{\|\phi\|=1}
\EE\big[\|\theta^\mh(\WW{\tau}{0}[\V{B}^{\V{x}}]-\id)\phi\|\big]\|\Psi\|_\infty
\end{align*}
in combination with \eqref{antonio}. In the case $\Upsilon_t=\id$, these
remarks already prove the theorem.

{\em Step 2.} Next, we consider arbitrary $\Upsilon$ as in the statement but restrict our attention
to $V\in C_0(\RR^\nu)$. Let $\cE$ be any uniformly bounded set of functions
from $\RR^\nu$ to $\FHR$ which is uniformly equicontinuous on $\RR^\nu$.
Then, by the semi-group property, by Step 1, and by the obvious inclusion 
$C_0(\RR^\nu)\subset\cK(\RR^\nu)$, it suffices to show that the set
$\{\Upsilon_{t}T^V_{t}\Psi:\Psi\in\cE,\,t\in[\tau_1/2,\tau_2/2]\}$ is uniformly equicontinuous as well. 

To do so we fix $t\in [\tau_1/2,\tau_2/2]$ and set $\wt{\Upsilon}_s:=\Upsilon_{t-s}^{-1}$, $s\in[0,t]$.
Then we observe that, for all $\V{x},\V{y}\in\RR^\nu$ and $\Psi\in\cE$,
\begin{align}\nonumber
\|\Upsilon_t(T^V_t\Psi(\V{x})-T^V_t\Psi(\V{y}))\|
&\le\big\|\Upsilon_t\EE\big[(\WW{t}{V}[\V{B}^{\V{x}}]-\WW{t}{V}[\V{B}^{\V{y}}])^*
\Psi(\V{B}_t^{\V{x}})\big]\big\|
\\\nonumber
&\qquad
+\big\|\Upsilon_t\EE\big[\WW{t}{V}[\V{B}^{\V{y}}]^*(\Psi(\V{B}_t^{\V{x}})
-\Psi(\V{B}_t^{\V{y}}))\big]\big\|
\\\nonumber
&\le\sup_{{\phi\in\dom(\Upsilon_t)\atop\|\phi\|=1}}
\EE\big[\|(\WW{t}{V}[\V{B}^{\V{x}}]-\WW{t}{V}[\V{B}^{\V{y}}])\Upsilon_t\phi\|\big]\,\|\Psi\|_\infty
\\\label{thyra0}
&\quad+\sup_{{\phi\in\dom(\Upsilon_t)\atop\|\phi\|=1}}
\EE\big[\|\WW{t}{V}[\V{B}^{\V{y}}]\Upsilon_t\phi\|\big]\!\!
\sup_{{\tilde{\V{x}},\tilde{\V{y}}\in\RR^\nu\atop|\tilde{\V{x}}-\tilde{\V{y}}|
=|\V{x}-\V{y}|}}\!\!\|\Psi(\tilde{\V{x}})-\Psi(\tilde{\V{y}})\|.
\end{align}
Next, we employ Lem.~\ref{lem-maria} with $\tilde{\V{c}}:=\V{c}$, $V=\wt{V}$, $\V{q}=\V{x}$, and
$\tilde{\V{q}}=\V{y}$. This yields
\begin{align*}
\sup_{{\phi\in\dom(\Upsilon_t)\atop\|\phi\|=1}}
\EE\big[&\|\wt{\Upsilon}_t(\WW{t}{V}[\V{B}^{\V{x}}]
-{\mathbb{W}}_t^{V,n}[\V{B}^{\V{y}}])\wt{\Upsilon}_0^{-1}\phi\|\big]
\le cA^2e^{cA^2t}\EE\big[\sup_{s\le t}f(s,\V{x},\V{y})^2\big]^\eh,
\end{align*}
for all $\V{x},\V{y}\in\RR^\nu$, with a universal constant $c>0$ and random variables
\begin{align*}
f(s,\V{x},\V{y})&:=|V(\V{B}_s^{\V{x}})-V(\V{B}_s^{\V{y}})|+
\|\V{c}_{\V{B}_s^{\V{x}}}-\V{c}_{\V{B}_s^{\V{y}}}\|_*.
\end{align*}
Since $\{\V{B}_s^{\V{z}}(\vgamma):\,s\in[0,t],\mathrm{dist}(\V{z},K)\le1\}\subset\RR^\nu$ is compact,
for all $\vgamma\in\Omega$ and all compact $K\subset\RR^\nu$, since $V$ is uniformly 
continuous, and since $\V{x}\mapsto\V{c}_{\V{x}}$ is uniformly continuous on every 
compact subset of $\RR^\nu$ 
with respect to the norm $\|\cdot\|_*$, the dominated convergence theorem implies
\begin{align}\label{per700}
\lim_{r\downarrow0}\sup_{{\V{x},\V{y}\in K\atop|\V{x}-\V{y}|<r}}
\EE\big[\sup_{s\le\tau}f(s,\V{x},\V{y})^2\big]&=0,\quad\tau>0.
\end{align}
If $\V{x}\mapsto\V{c}_{\V{x}}$ is uniformly continuous on $\RR^\nu$ 
with respect to $\|\cdot\|_*$, then $K$ can be replaced by $\RR^\nu$ in \eqref{per700}.
Taking also \eqref{thyra0}, the uniform boundedness of $\cE$, and the uniform equicontinuity of 
$\cE$ on $\RR^\nu$ into account, we conclude that
$$
\lim_{r\downarrow0}\sup_{t\in[2\tau_1,2\tau_2]}
\sup_{{\V{x},\V{y}\in\RR^\nu\atop|\V{x}-\V{y}|<r}}\sup_{\Psi\in\cE}
\|\Upsilon_t(T^V_t\Psi(\V{x})-T^V_t\Psi(\V{y}))\|=0.
$$

{\em Step 3.}
Now let $V\in\cK_\pm(\RR^\nu)$ be arbitrary and $\Upsilon_{\nf{t}{2}}$ as in the statement. 
Let $p\in[1,\infty]$. We pick $V_n\in C_0^\infty(\RR^\nu)$, $n\in\NN$, approximating
$V$ in the sense made precise in Lem.~\ref{lem-conny}(1). For each $n\in\NN$, 
we then know from Step 2 that the set 
$\{\Theta_{\nf{t}{2}}T^{V_n}_t\Psi:\,\Psi\in L^p(\RR^\nu,\FHR),\,\|\Psi\|\le1,\,t\in[\tau_1,\tau_2]\}$ 
is uniformly equicontinuous on $\RR^\nu$.
On account of \eqref{LNCV1} (with $q=\infty$) and the boundedness of $\cE$ we further have
\begin{align}\label{hans}
\lim_{n\to\infty}\sup_{t\in[0,\tau_2]}\sup_{\|\Psi\|_p\le1}\sup_{\V{x}\in K}
\|\Upsilon_{\nf{t}{2}}(T^{V_n}_t\Psi-\Theta_tT^{V}_t\Psi)(\V{x})\|=0,
\end{align}
for every compact $K\subset\RR^\nu$. In the case $V\in\cK(\RR^\nu)$ we may even
choose $K=\RR^\nu$ in \eqref{hans} according to Thm.~\ref{thm-LNCV}(2).
Altogether this proves the theorem in full generality.
\end{proof}

\begin{cor}\label{cor-eq-cont}
Let  $t>0$, $p\in[1,\infty]$, and let
$V,V_n\in\cK_\pm(\RR^\nu)$, $n\in\NN$, satisfy \eqref{approx1} and \eqref{approx2}.
Let $(\Upsilon_t)_{t\ge0}$ and $\|\cdot\|_*$ be given by either Line~2 or Line~4 of Table~1. 
Suppose that $\V{c},\V{c}_n$, $n\in\NN$, are coefficient vectors satisfying Hyp.~\ref{hyp-G} 
with the same $\omega$ and $C$ such that the maps
$\RR^\nu\ni\V{x}\mapsto\V{c}_{\V{x}}$ and $\RR^\nu\ni\V{x}\mapsto\V{c}_{n,\V{x}}$, $n\in\NN$,
are continuous and uniformly bounded with respect to the norm $\|\cdot\|_*$.
Assume that the $\V{c}_n$ converge to $\V{c}$ in the sense that \eqref{marie1} holds,
for all compact $K\subset\RR^\nu$. Finally, let $\{\Psi_n\}_{n\in\NN}$ be a converging
sequence in $L^p(\RR^\nu,\FHR)$ with limit $\Psi$, $\{\V{x}_n\}_{n\in\NN}$ be a converging
sequence in $\RR^\nu$ with limit $\V{x}$, and denote the semi-group defined by
means of $V_n$ and $\V{c}_n$ by $(T_t^{V_n,n})_{t\ge0}$. Then
\begin{align}\label{alfred1}
\lim_{n\to\infty}\Upsilon_{\nf{t}{2}}(T_t^{V_n,n}\Psi_n)(\V{x}_n)=\Upsilon_{\nf{t}{2}}(T_t^V\Psi)(\V{x}).
\end{align}
\end{cor}


\begin{proof}
Combining Thms.~\ref{thm-LNCV} and~\ref{thm-coup} reveals that
\begin{align*}
\sup_{n\in\NN}\|\Upsilon_{\nf{t}{2}}T^{V_n,n}\|_{p,\infty}<\infty,\quad
\sup_{\V{y}\in K}
\big\|\Upsilon_{\nf{t}{2}}(T_t^{V_n,n}\Psi)(\V{y})-\Upsilon_{\nf{t}{2}}(T_t^V\Psi)(\V{y})\big\|
\xrightarrow{\;\;n\to\infty\;\;}0,
\end{align*}
for every compact $K\subset\RR^\nu$. We conclude by choosing $K$ such that it contains the 
image of the sequence $\{\V{x}_n\}_{n\in\NN}$, writing
\begin{align*}
(T_t^{V_n,n}\Psi_n)(\V{x}_n)&-(T_t^V\Psi)(\V{x})
=(T_t^{V}\Psi)(\V{x}_n)-(T_t^V\Psi)(\V{x})
\\
&+\big((T_t^{V_n,n}-T_t^V)\Psi\big)(\V{x}_n)
+\big(T_t^{V_n,n}(\Psi_n-\Psi)\big)(\V{x}_n),
\end{align*} 
and employing the continuity of $\Upsilon_{\nf{t}{2}}T_t^V\Psi$ guaranteed by 
Thm.~\ref{thm-eq-cont}.
\end{proof}

The preceding corollary will be applied in a more specific situation in Sect.~\ref{sec-cont-GS}. 

Let us consider the standard model of non-relativistic QED for $N$ spin one-half
electrons to illustrate  Thms.~\ref{thm-LNCV}, \ref{thm-coup}, and~\ref{thm-eq-cont}:

\begin{ex}\label{ex-QED-exp}
Assume that $\V{G}=\V{G}^{\chi,N}$ and $\V{F}=\V{F}^{\chi,N}$ 
are as in Ex.~\ref{ex-NRQED} and that $V=V_{\V{Z},\ul{\V{R}}}^{N,\mathfrak{e}}$ 
is the many-body Coulomb
potential defined in Ex.~\ref{ex-QED-V} with atomic numbers $\V{Z}\in[0,\infty)^K$. 
(Of course, the component $Z_\vk$ of $\V{Z}$ being zero means that the $\vk$-th nucleus is 
absent.) Let $E_{\V{Z},\ul{\V{R}}}^{\chi,N,\mathfrak{e}}$ denote the infimum of the spectrum 
of the physical Hamiltonian
$H_{\V{Z},\ul{\V{R}}}^{\chi,N,\mathfrak{e}}\!\!\upharpoonright_{\HR^N_{\mathrm{phys}}}$ and let
$$
\Sigma_{\V{Z},\ul{\V{R}}}^{\chi,N,\mathfrak{e}}:=
\min_{M=1,...,N}\big\{E_{\V{Z},\ul{\V{R}}}^{\chi,M,\mathfrak{e}}
+E_{\V{0}}^{\chi,N-M,\mathfrak{e}}\big\}
$$ 
denote the {\em ionization threshold}. Physically,
this is the minimal energy required for removing at least one electron from the
confining potential of the molecule modeled by the first term in \eqref{def-VN}.
Let $\{P_s\}_{s\in\RR}$ denote the spectral family of 
$H_{\V{Z},\ul{\V{R}}}^{\Lambda,N,\mathfrak{e}}\!\!\upharpoonright_{\HR^N_{\mathrm{phys}}}$.
Then it is known \cite{Griesemer2004} that the range of $P_s$ is contained in 
$\dom(e^{a|\cdot|})$, provided that $s$ and $a$ satisfy
\begin{align}\label{As}
s+a^2/2<\Sigma_{\V{Z},\ul{\V{R}}}^{\chi,N,\mathfrak{e}}.
\end{align}
(1) Assume that $a$ and $s$ satisfy \eqref{As}.
Then $e^{a|\cdot|}P_s e^{2H^{\Lambda,N,\mathfrak{e}}_{\V{Z},\ul{\V{R}}}}P_s\in\LO(\HR)$ 
and the relation $P_s=P_s^2$ entails
\begin{equation}\label{wayne}
P_s=\big(e^{-2H^{\chi,N,\mathfrak{e}}_{\V{Z},\ul{\V{R}}}}e^{-a|\cdot|}\big)e^{a|\cdot|}
P_se^{2H^{\chi,N,\mathfrak{e}}_{\V{Z},\ul{\V{R}}}}P_s.
\end{equation}
Thanks to Thm.~\ref{thm-eq-cont} we know that every $\Psi_s\in\Ran(P_s)$ with $s\in\RR$ 
admits a unique uniformly continuous representative that we shall consider in the following.
We shall also assume that $\|\Psi_s\|=1$.

Assume that that the cut-off function
$\chi$ appearing in Ex.~\ref{ex-NRQED} satisfies $\omega e^{\delta_0\omega}\chi\in\HP$, 
for some $\delta_0>0$. For sufficiently small $\delta>0$, 
we then infer the following pointwise bound from the above remarks and \eqref{norm-T-F-Theta},
$$
\|e^{\delta\Id\Gamma(\omega)}\Psi_s(\ul{\V{x}})\|_{\FHR}
\le e^{-a|\ul{\V{x}}|}\big\|e^{\delta\Id\Gamma(\omega)+a|\cdot|}
e^{-2H^{\chi,N,\mathfrak{e}}_{\V{Z},\ul{\V{R}}}}e^{-a|\cdot|}\big\|_{2,\infty}
\big\|e^{a|\cdot|}\chi_s e^{2H^{\chi,N,\mathfrak{e}}_{\V{Z},\ul{\V{\V{R}}}}}P_s\big\|_{2,2},
$$ 
for all $\ul{\V{x}}\in\RR^{3N}$.
Here $\delta>0$ depends only on the coefficient vector, in this case
$\delta_0$, $\chi$, and the number of electrons $N$; see \eqref{hyp-aw-QED}.
For every $r>0$, put $\omega_r:=\omega 1_{\{\omega\ge r\}}$ and consider the weight 
given by Line~4 of Table~1 with $t_*:=2$ and $\vo:=\omega_r$. Then can we fix $r$ large 
enough such that the corresponding condition Hyp$(\Upsilon)$ in Line~4 of the table is fulfilled.
As above we now obtain
$$
\|e^{\delta_0\Id\Gamma(\omega_r)}\Psi_s(\ul{\V{x}})\|_{\FHR}\le 
c_{r,a,s,N,\chi,\V{Z},\V{R},\mathfrak{e}} 
e^{-a|\ul{\V{x}}|},\quad\ul{\V{x}}\in\RR^{3N}.
$$
Let $\Psi_s=(\Psi_s^{(n)})_{n=0}^\infty$ the representation of $\Psi_s$ as a sequence indexed by
the boson number according to the isomorphism 
$L^2(\RR^{3N},\FHR)=\bigoplus_{n=0}^\infty L^{2}(\RR^{3N},\CC^{2^N}\otimes\sF^{(n)})$. 
Then the previous bound implies
$$
\|e^{\delta_0\Id\Gamma^{(n)}(\omega)}\Psi_s^{(n)}(\ul{\V{x}})\|_{\CC^{2^N}\otimes\sF^{(n)}}\le 
c_{r,a,s,N,\chi,\V{Z},\V{R},\mathfrak{e}}e^{\delta_0nr}
e^{-a|\ul{\V{x}}|},\quad\ul{\V{x}}\in\RR^{3N},\,n\in\NN.
$$
This shows that at least any fixed $n$-boson function $\Psi_s^{(n)}$ has an
exponential decay with respect to the photon momentum variables
in the $L^2$-sense with a rate no less than $\delta_0$.

Likewise, if we suppose instead that $\omega^{\alpha+\eh}\chi\in\HP$, for some $\alpha\ge1$, then
$$
\|(1+\Id\Gamma(\omega))^\alpha\Psi_s(\ul{\V{x}})\|_{\FHR}\le 
c_{\alpha,a,s,N,\chi,\V{Z},\V{R},\mathfrak{e}} 
e^{-a|\ul{\V{x}}|},\quad\ul{\V{x}}\in\RR^{3N}.
$$

For general elements in spectral subspaces that are not eigenvectors, these pointwise 
bounds are new. Moreover, in the case of the Coulomb potential,
the relation \eqref{As} gives a more explicit and probably
better bound on the decay rate $a$ as compared to earlier pointwise decay estimates
on ground state eigenvectors \cite{HiHi2010,Hiroshima2003,LHB2011}.
We should also mention at this point that ground state eigenvectors in non-relativistic QED
are in the domain of all powers of the number operator \cite{Hiroshima2003}. 
In order to show this one has to exploit the corresponding eigenvalue equation.

\smallskip

\noindent
(2) Let $s\in\RR$ be arbitrary, $\Psi_s\in\Ran(P_s)$, and assume
that $\omega^{\alpha+\eh}\chi\in\HP$, for some $\alpha\ge1$. 
Then \eqref{wayne} with $a=0$ and Thm.~\ref{thm-LNCV} imply, for all 
$n\in\NN$ and $\ul{\V{x}}\in\RR^{3N}$,
$$
\|(1+\Id\Gamma^{(n)}(\omega))^\alpha\Psi_s^{(n)}(\ul{\V{x}})\|_{\CC^{2^N}\otimes\sF^{(n)}}
\le\|(1+\Id\Gamma(\omega))^\alpha\Psi_s(\ul{\V{x}})\|_{\FHR}
\le c_{\alpha,s,N,\chi,\V{Z},\V{R}}\|\Psi_s\|.
$$
Notice that, since $\omega(k)=|\V{k}|$, the previous bound can be read as an estimate on 
the norm of $\Psi_s^{(n)}(\ul{\V{x}})$ in the Sobolev space $H^{\alpha}(\RR^{3n},\CC^{2^{N+n}})$. 
(Here $\CC^{2^{N+n}}$ accounts for the spin degrees of freedom of the $N$ electrons and the 
polarizations of the $n$ photons.) 
Under the present assumption on $\chi$ we further know that the maps 
$\RR^{3N}\ni\ul{\V{x}}\mapsto\omega^\gamma(\V{G}_{\ul{\V{x}}}^{\chi,N},\V{F}_{\ul{\V{x}}}^{\chi,N})
\in\HP$ are bounded and continuous, for all $\gamma\in[\mh,\alpha]$. 
Hence, Thm.~\ref{thm-eq-cont} implies continuity of the maps 
$\RR^{3N}\ni\ul{\V{x}}\mapsto\Psi_s^{(n)}(\ul{\V{x}})\in H^{\alpha}(\RR^{3n},\CC^{2^{N+n}})$.
For fixed $\ul{\V{x}}$,
let $(\Psi_s^{(n)}(\ul{\V{x}}))^\wedge(\ul{\V{y}})$ denote the Fourier transform of
$\Psi_s^{(n)}(\ul{\V{x}})$ evaluated at $\ul{\V{y}}\in\RR^{3n}$.
If $\omega^{\alpha}\chi\in\HP$, for {\em all} $\alpha\ge1$, then we may employ the Sobolev 
embedding theorem to argue that every partial derivative
$\partial_{\ul{\V{y}}}^{\beta}(\Psi_s^{(n)}(\ul{\V{x}}))^\wedge(\ul{\V{y}})$ with $\beta\in\NN_0^{3n}$
is bounded and jointly continuous in $(\ul{\V{x}},\ul{\V{y}})\in\RR^{3(N+n)}$.
\end{ex}


\section{Continuity of the integral kernel}\label{sec-cont-ker}

\noindent
Next, we turn our attention to the operator-valued integral kernel $T_t^V(\V{x},\V{y})$.
The first aim of this section is to show that it
is jointly continuous in the variables $t>0$, $\V{x}$, and $\V{y}$; see the subsequent
Thm.~\ref{thm-kern-cont} which is obtained by extending the arguments applied to
Schr\"odinger semi-groups in, e.g.,  \cite{Sznitman1998}.
In fact, the results of our Sects.~\ref{sec-cont-V}, \ref{sec-SG}, and \ref{sec-eq-cont} 
are needed in its proof. At this point we should also mention the article \cite{BHL2000}
where possibly very singular classical magnetic fields are treated, as well as
\cite{Gueneysu2011}, where the analysis of \cite{BHL2000} is pushed forward to a 
matrix-valued case.
In the standard reference for Schr\"odinger semi-groups \cite{Simon1982}, 
the continuity of the semi-group kernel is investigated in the absence of magnetic fields
by methods somewhat different from \cite{Sznitman1998}.

Proceeding as in \cite{Simon1982}, it is actually possible to combine the next theorem with the 
Feynman-Kac formula \eqref{feyn} and \eqref{intK} and to argue that resolvents of the
operator $H^V$ defined in Thm.~\ref{thm-FK} have $\LO(\FHR)$-valued integral kernels which are 
jointly continuous in their arguments $\V{x}$ and $\V{y}$. Using this result we could further argue 
that, for every bounded Borel function $f:\RR\to\CC$, the operator $f(H^V)$, defined by spectral 
calculus, has a jointly continuous $\LO(\FHR)$-valued integral kernel. 

The next theorem can also be combined with Thm.~\ref{thm-kern-para} further below to prove
joint continuity of the semi-group kernel in $t>0$, $\V{x}$, $\V{y}$, and additional model
parameters.

\begin{thm}\label{thm-kern-cont}
Let $V\in\cK_\pm(\RR^\nu)$, $\tau_2>\tau_1>0$, and 
let $(\Upsilon_t)_{t\ge0}$ and $\|\cdot\|_*$ be given by either Line~2 or Line~4 of Table~1. 
Assume that the map 
$\RR^\nu\ni\V{x}\mapsto\V{c}_{\V{x}}$ is bounded and continuous with respect to the
norm $\|\cdot\|_*$. Then the map 
\begin{equation}\label{marah}
[\tau_1,\tau_2]\times\RR^\nu\times\RR^\nu\ni(t,\V{x},\V{y})\mapsto
\Upsilon_{\nf{\tau_1}{8}}T^V_t(\V{x},\V{y})\in\LO(\FHR)
\end{equation}
is continuous with respect to the norm topology on $\LO({\FHR})$.
If $V\in\cK(\RR^\nu)$ and $\RR^\nu\ni\V{x}\mapsto\V{c}_{\V{x}}$ is bounded and 
uniformly continuous with respect to $\|\cdot\|_*$, then \eqref{marah} is uniformly continuous.
\end{thm}

\begin{proof}
We split the proof into three steps.
In Steps 1, 2 and 3 we verify uniform continuity in $\V{x}$, $\V{y}$, and $t$, respectively, as the
other two parameters are varying in suitable sets. To this end we pick $\tau_2>\tau_1>0$ and 
set $s:=\tau_1/2$.

{\em Step 1.}
As a consequence of the inequality in \eqref{T(x,y)-sym} the $L^1(\RR^\nu,\FHR)$-norm 
of the functions $T^V_s(\cdot,\V{y})\psi$  with $\V{y}\in\RR^\nu$
and $\psi\in\FHR$, $\|\psi\|=1$, is uniformly bounded. In view of this fact and Thm.~\ref{thm-eq-cont},
the set of $\FHR$-valued functions
$$
\big\{\Upsilon_{\nf{s}{4}}T^V_{t-s}(T^V_s(\cdot,\V{y})\psi):\,\psi\in\FHR,\,\|\psi\|=1,\,\V{y}\in\RR^\nu,\,
t\in[\tau_1,\tau_2]\big\}
$$
is uniformly equicontinuous on every compact subset of $\RR^\nu$, and uniformly equicontinuous
on the whole $\RR^\nu$ provided that $V\in\cK(\RR^\nu)$ and 
$\RR^\nu\ni\V{x}\mapsto\V{c}_{\V{x}}$ 
is bounded and uniformly continuous with respect to $\|\cdot\|_*$. 
Let $K=\RR^\nu$, if the latter conditions are fulfilled, and let
$K\subset\RR^\nu$ be compact, if not.
Combining the above observation with Prop.~\ref{prop-CK} we then deduce that the expression
\begin{align*}
&\sup_{{\V{x},\tilde{\V{x}}\in K\atop|\V{x}-\tilde{\V{x}}|<r}}
\sup_{t\in[\tau_1,\tau_2]}\sup_{\V{y}\in\RR^\nu}
\|\Upsilon_{\nf{s}{4}} (T^V_t(\V{x},\V{y})-T^V_t(\tilde{\V{x}},\V{y}))\|
\\
&=\sup_{{\V{x},\tilde{\V{x}}\in K\atop|\V{x}-\tilde{\V{x}}|<r}}
\sup_{t\in[\tau_1,\tau_2]}\sup_{\V{y}\in\RR^\nu}\sup_{\|\psi\|=1}
\big\|\Upsilon_{\nf{s}{4}} T^V_{t-s}(T^V_s(\cdot,\V{y})\psi)(\V{x})
-\Upsilon_{\nf{s}{4}} T^V_{t-s}(T^V_s(\cdot,\V{y})\psi)(\tilde{\V{x}})\big\|
\end{align*}
converges to zero in the limit $r\downarrow0$.
 
{\em Step 2.} Likewise, the $L^1(\RR^\nu,\FHR)$-norm of the functions
$T^V_s(\cdot,\V{x})\Theta_{\nf{s}{4}}\psi$  with $\V{x}\in\RR^\nu$
and $\psi\in\dom(\Theta_{\nf{s}{4}})$, $\|\psi\|=1$, is uniformly bounded; in fact, 
setting $F(\V{z}):=|\V{z}-\V{x}|$, $\V{z}\in\RR^\nu$, 
we infer from \eqref{T(x,y)-sym} and \eqref{CK1} that
\begin{align*}
\int_{\RR^\nu}\|T^V_s(\V{z},\V{x})\Upsilon_{\nf{s}{4}}&\psi\|\Id\V{z}\le
\int_{\RR^\nu}\sup_{\|\phi\|=1}\|\Theta_{\nf{s}{4}} T^V_s(\V{x},\V{z})\phi\|\Id\V{z}
\\
&=\int_{\RR^\nu}\sup_{\|\phi\|=1}\big\|(\Upsilon_{\nf{s}{4}} 
e^{-F}T^V_{\nf{s}{2}}e^F)(e^{-F}T^V_{\nf{s}{2}}(\cdot,\V{z})\phi)(\V{x})\big\|\Id\V{z}
\\
&\le\|\Upsilon_{\nf{s}{4}} 
e^{-F}T^V_{\nf{s}{2}}e^F\|_{1,\infty}\int_{\RR^\nu}\sup_{\|\phi\|=1}
\int_{\RR^\nu}\|e^{-F(\tilde{\V{x}})}T^V_{\nf{s}{2}}(\tilde{\V{x}},\V{z})\phi\|\Id\tilde{\V{x}}\Id\V{z}
\\
&\le c_{s,V,\V{c}}\int_{\RR^\nu}e^{-|\tilde{\V{x}}|}\Id\tilde{\V{x}}\int_{\RR^\nu}e^{-|\V{z}|^2/4s}\Id\V{z},
\end{align*}
for all $\V{x}\in\RR^\nu$ and normalized $\psi\in\dom(\Upsilon_{\nf{s}{4}})$.
Employing Thm.~\ref{thm-eq-cont} once more, we see that
the set of $\FHR$-valued functions
$$
\big\{T^V_{t-s}(T^V_s(\cdot,\V{x})\Upsilon_{\nf{s}{4}}\psi):\,\psi\in\dom(\Upsilon_{\nf{s}{4}}),
\,\|\psi\|=1,\,\V{x}\in\RR^\nu,\,t\in[\tau_1,\tau_2]\big\}
$$
is uniformly equicontinuous on every compact subset of $\RR^\nu$, and uniformly equicontinuous
on all of $\RR^\nu$, if the additional conditions in the last assertion of the statement are fulfilled. 
Let $K\subset\RR^\nu$ be given as in Step~1. Using 
\begin{align*}
\|\Upsilon_{\nf{s}{4}}(T^V_t(\V{x},\V{y})-T^V_t(\V{x},\tilde{\V{y}}))\|
&=\sup_{{\psi\in\dom(\Upsilon_{\nf{s}{4}})\atop\|\psi\|=1}}
\sup_{\|\phi\|=1}\big|\SPb{\psi}{\Upsilon_{\nf{s}{4}}(T^V_t(\V{x},\V{y})
-T^V_t(\V{x},\tilde{\V{y}}))\phi}\big|
\\
&=\sup_{{\psi\in\dom(\Upsilon_{\nf{s}{4}})\atop\|\psi\|=1}}
\|(T^V_t(\V{y},\V{x})-T^V_t(\tilde{\V{y}},\V{x}))\Upsilon_{\nf{s}{4}}\psi\|
\end{align*}
and invoking Prop.~\ref{prop-CK} once more we conclude that the expression
\begin{align*}
&\sup_{{\V{y},\tilde{\V{y}}\in K\atop|\V{y}-\tilde{\V{y}}|<r}}
\sup_{t\in[\tau_1,\tau_2]}\sup_{\V{x}\in\RR^\nu}
\|\Upsilon_{\nf{s}{4}}(T^V_t(\V{x},\V{y})-T^V_t(\V{x},\tilde{\V{y}}))\|
\\
&=\sup_{{\V{y},\tilde{\V{y}}\in K\atop|\V{y}-\tilde{\V{y}}|<r}}
\sup_{t\in[\tau_1,\tau_2]}\sup_{\V{x}\in\RR^\nu}\sup_{\|\psi\|=1}
\big\|T_{t-s}^V(T^V_s(\cdot,\V{x})\Upsilon_{\nf{s}{4}}\psi)(\V{y})
-T_{t-s}^V(T^V_s(\cdot,\V{x})\Upsilon_{\nf{s}{4}}\psi)(\tilde{\V{y}})\big\|
\end{align*}
goes to zero in the limit $r\downarrow0$.

{\em Step 3.}
For all $t\ge\tau_1$ and $\psi\in\FHR$, Prop.~\ref{prop-CK} further implies
\begin{align*}
{T^V_t(\V{x},\V{y})\psi}
&=\int_{\RR^\nu}{T^V_{\nf{s}{2}}(\V{x},\V{z})}T^V_{t-\nf{s}{2}}(\V{z},\V{y})\psi\Id\V{z}
=\int_{\RR^\nu}T_{\nf{s}{4}}^V\big(T_{\nf{s}{4}}^V(\cdot,\V{z})\Psi_{s,t,\V{y}}(\V{z})\big)(\V{x})\Id\V{z},
\end{align*}
with $\Psi_{s,t,\V{y}}:=T_{t-s}^V(T_{\nf{s}{2}}^V(\cdot,\V{y})\psi)$, which yields,
again with $F_{\V{x}}(\V{z}):=|\V{z}-\V{x}|$,
\begin{align*}
&\big\|\Upsilon_{\nf{s}{4}}(T_t^V(\V{x},\V{y})-T^V_{\tilde{t}}(\V{x},{\V{y}}))\psi\big\|
\\
&\le\|e^{-F_{\V{x}}}\Upsilon_{\nf{s}{4}} T^V_{\nf{s}{4}}e^{F_{\V{x}}}\|_{1,\infty}
\int_{\RR^\nu}\!\int_{\RR^\nu}\big\|e^{-F_{\V{x}}(\tilde{\V{x}})}T_{\nf{s}{4}}^V(\tilde{\V{x}},\V{z})
\big(\Psi_{s,t,\V{y}}(\V{z})-\Psi_{s,\tilde{t},{\V{y}}}(\V{z})\big)\big\|\Id\tilde{\V{x}}\,\Id\V{z}
\\
&\le\|e^{-F_{\V{x}}}\Upsilon_{\nf{s}{4}} T^V_{\nf{s}{4}}e^{F_{\V{x}}}\|_{1,\infty}
\sup_{\tilde{\V{z}}\in\RR^\nu}\int_{\RR^\nu}e^{|\tilde{\V{x}}-\tilde{\V{z}}|}
\|T_{\nf{s}{4}}^V(\tilde{\V{x}},\tilde{\V{z}})\|\Id\tilde{\V{x}}\,
\big\|e^{-F_{\V{x}}}(\Psi_{s,t,\V{y}}-\Psi_{s,\tilde{t},{\V{y}}})\big\|_{1},
\end{align*}
for all $t,\tilde{t}\ge\tau_1$. Employing the bound in \eqref{T(x,y)-sym} and \eqref{norm-T-F-Theta}, 
we find some $c_s>0$ such that
\begin{align}\nonumber
\big\|\Upsilon_{\nf{s}{4}}(T_t^V(\V{x},\V{y})-T^V_{\tilde{t}}(\V{x},{\V{y}}))\big\|
&\le c_s\sup_{\|\psi\|=1}\|e^{-F_{\V{x}}}(T_{t-s}^V-T_{\tilde{t}-s}^V)\Phi_{s,{\V{y}},\psi}\|_{1}
\\\label{vincius1}
&=c_s\sup_{\|\psi\|=1}\|e^{-F_{\V{x}}}(T_{|t-\tilde{t}|}^V-\id)T^V_{t\wedge\tilde{t}-s}
\Phi_{s,\V{y},\psi}\|_1,
\end{align}
with $\Phi_{s,\V{y},\psi}:=T_{\nf{s}{2}}^V(\cdot,\V{y})\psi$. 
Next, we observe that \eqref{T(x,y)-sym} implies that
\begin{align}
\sup_{\V{y}\in\RR^\nu}\sup_{\|\psi\|=1}\|\Phi_{s,\V{y},\psi}\|_1
&\le\sup_{\V{y}\in\RR^\nu}\int_{\RR^\nu}\|T_{\nf{s}{2}}^V(\V{z},\V{y})\|\Id\V{z}\label{vincius4}
\le c_s\int_{\RR^\nu}e^{-|\V{z}|^2/4s}\Id\V{z}<\infty.
\end{align} 
As above, let $K$ be equal to $\RR^\nu$, if $V\in\cK(\RR^\nu)$ and 
$\V{x}\mapsto\V{c}_{\V{x}}$ is bounded and uniformly continuous with respect to $\|\cdot\|_*$, and 
let $K$ be a compact subset of $\RR^\nu$ otherwise. Defining
$$
\cM:=\big\{e^{-F_{\V{x}}}T^V_{\tau-s}\Phi_{s,\V{y},\psi}:\,\psi\in\FHR,\,\|\psi\|=1,\,
\V{x}\in K,\,\V{y}\in\RR^\nu,\,\tau\in[\tau_1,\tau_2]\big\}
$$
we then infer that $\cM\subset L^1(\RR^\nu,\FHR)\cap L^\infty(\RR^\nu,\FHR)$ 
by Lem.~\ref{lem-T-bd} and
\begin{enumerate}
\item[(a)] $\cM$ is uniformly equicontinuous on every compact subset of $\RR^\nu$ 
by Thm.~\ref{thm-eq-cont} applied to the set 
$\{T^V_{\tau-s}\Phi_{s,\V{y},\psi}:\,\|\psi\|=1,\,\V{y}\in\RR^\nu,\,\tau\in[\tau_1,\tau_2]\}$
and the bounds \eqref{vincius4} and
\begin{align}\label{vincius5}
&m:=\sup_{\V{y}\in\RR^\nu}\sup_{\|\psi\|=1}\sup_{\tau\in[\tau_1,\tau_2]}
\|T^V_{\tau-s}\Phi_{s,\V{y},\psi}\|_\infty<\infty,
\\\nonumber
&|e^{-F_{\V{x}}(\V{z})}-e^{-F_{\V{x}}(\tilde{\V{z}})}|\le|\V{z}-\tilde{\V{z}}|,
\quad\V{x},\V{z},\tilde{\V{z}}\in\RR^\nu;
\end{align}
here the first one follows from
$\sup_{\tau\in[\tau_1,\tau_2]}\|T^V_{\tau-s}\|_{1,\infty}<\infty$ and \eqref{vincius4}.
If the additional conditions in the last assertion of the statement are fulfilled, then
$\cM$ is uniformly equicontinuous on the whole $\RR^\nu$.
\item[(b)] $\sup_{\Psi\in\cM}\|\Psi\|_\infty<\infty$ by \eqref{vincius5}. Moreover,
setting $\V{a}_\Psi:=\V{x}$, if $\Psi\in\cM$ is given as 
$\Psi=e^{-F_{\V{x}}}T^V_{\tau-s}\Phi_{s,\V{y},\psi}$, and abbreviating
$\Psi_R(\V{z}):=1_{|\V{z}-\V{a}_\Psi|<R}\Psi(\V{z})$, we verify that
$\|\Psi-\Psi_R\|_1\le e^{-R}m$.
\item[(c)] If $f_\Psi$ is defined as in Condition~(c) of Thm.~\ref{thm-SC}(1),
then \eqref{norm-T-F-Theta} implies
\begin{align*}
\sup_{\Psi\in\cM}\|f_\Psi\|_1&\le\sup_{\tau\in[\tau_1,\tau_2]}\sup_{\V{y}\in\RR^\nu}
\sup_{\|\psi\|=1}\|(1+\Id\Gamma(\omega))^\eh T_{\tau-s}^V\Phi_{s,\V{y},\psi}\|_1
\\
&\le\sup_{\tau\in[\tau_1,\tau_2]}\|(1+\Id\Gamma(\omega))^\eh T_{\tau-s}^V\|_{1,1}
\sup_{\V{y}\in\RR^\nu}\sup_{\|\psi\|=1}\|\Phi_{s,\V{y},\psi}\|_1<\infty.
\end{align*}
\end{enumerate}
We may hence apply Thm.~\ref{thm-SC}(1) to deduce that
\begin{align*}
\lim_{r\downarrow0}\sup_{{t,\tilde{t}\in[\tau_1,\tau_2]\atop|t-\tilde{t}|<r}}\sup_{\V{x}\in K}
\sup_{\V{y}\in\RR^\nu}\sup_{\|\psi\|=1}
&\|e^{-F_{\V{x}}}(T_{|t-\tilde{t}|}^V-\id)T^V_{t\wedge\tilde{t}-s}\Phi_{s,\V{y},\psi}\|_1
\\
&\le\lim_{r\downarrow0}\sup_{\Psi\in\cM}\|e^{-F_{\Psi}}(T^V_r-\id)e^{F_{\Psi}}\Psi\|=0,
\end{align*}
where $F_\Psi:=F_{\V{a}_\Psi}$. Combining this with \eqref{vincius1} we arrive at
\begin{align*}
\lim_{r\downarrow0}
\sup_{{t,\tilde{t}\in[\tau_1,\tau_2]\atop|t-\tilde{t}|<r}}\sup_{\V{x}\in K}\sup_{\V{y}\in\RR^\nu}
\big\|\Upsilon_{\nf{s}{4}}(T_t^V(\V{x},\V{y})-T^V_{\tilde{t}}(\V{x},{\V{y}}))\big\|&=0,
\end{align*}
and we conclude.
\end{proof}

\begin{thm}\label{thm-kern-para}
Let $\tau_2>\tau_1>0$ and $V,V_n\in\cK_\pm(\RR^\nu)$, $n\in\NN$, such that
\eqref{approx1} and \eqref{approx2} are satisfied. Let $\V{c},\V{c}_n$, $n\in\NN$ be coefficient
vectors satisfying Hyp.~\ref{hyp-G} with the same $\omega$ and $C$ and assume that
$\|\V{c}_{\V{x}}\|_{\mathfrak{k}},\|\V{c}_{n,\V{x}}\|_{\mathfrak{k}}\le A$, for all $\V{x}\in\RR^\nu$,
$n\in\NN$, and some $A\in(0,\infty)$. Assume further that \eqref{marie1} holds true, for all
compact $K\subset\RR^\nu$. Then
\begin{align}\label{clementine1}
\lim_{n\to\infty}\sup_{t\in[\tau_1,\tau_2]}\sup_{\V{x},\V{y}\in K}
\big\|\Upsilon_{\nf{\tau_1}{8}}\big(T^{V_n,n}_t(\V{x},\V{y})-T_t^V(\V{x},\V{y})\big)\big\|&=0,
\end{align}
for all compact $K\subset\RR^\nu$.
\end{thm}

\begin{proof}
We pick some $s\in[\nf{\tau_1}{4},\nf{\tau_1}{2}]$ and write
\begin{align*}
T^{V_n,n}_t(\V{x},\V{y})-T_t^V(\V{x},\V{y})
&=\big((T_{t-s}^{V_n,n}-T_{t-s}^V)(T_s^{V_n,n}(\cdot,\V{y})\psi)\big)(\V{x})
\\
&\quad
+\big(T_{t-s}^{V}(T_s^{V_n,n}(\cdot,\V{y})\psi-T_s^{V}(\cdot,\V{y})\psi)\big)(\V{x}),
\end{align*}
for a given $\psi\in\FHR$. Next, we observe that
\begin{align*}
\sup_{t\in[\tau_1,\tau_2]}\sup_{\V{x},\V{y}\in K}\sup_{\|\psi\|=1}
\big\|\Upsilon_{\nf{\tau_1}{8}}\big((T_{t-s}^{V_n,n}-T_{t-s}^V)(T_s^{V_n,n}
(\cdot,\V{y})\psi)\big)(\V{x})\big\|\xrightarrow{\;\;n\to\infty\;\;}0,
\end{align*}
because, on the one hand, $\{T_s^{V_n,n}(\cdot,\V{y})\psi:\,n\in\NN,\V{y}\in K,\|\psi\|=1\}$ 
is a bounded subset of $L^2(\RR^\nu,\FHR)$ in view of Cor.~\ref{cor-WEK} and 
\eqref{approx3a}, and, on the other hand,
$1_K\Upsilon_{\nf{\tau_1}{8}}(T_{t-s}^{V_n,n}-T_{t-s}^V)$ converges to zero 
in $\LO(L^2(\RR^\nu,\FHR),L^\infty(\RR^\nu,\FHR))$, 
uniformly in $t\in[\tau_1,\tau_2]$, according to Thms.~\ref{thm-LNCV} and~\ref{thm-coup}.
Since $\Upsilon_{\nf{\tau_1}{8}}T_{t-s}^{V}$ is bounded from 
$L^1(\RR^\nu,\FHR)$ to $L^\infty(\RR^\nu,\FHR)$, uniformly in $t\in[\tau_1,\tau_2]$,
it thus remains to show that
\begin{align}\label{wilmar}
\sup_{\|\psi\|=1}\sup_{\V{y}\in K}\int_{\RR^\nu}
\big\|T_s^{V_n,n}(\V{z},\V{y})\psi-T_s^{V}(\V{z},\V{y})\psi\big\|\Id\V{z}\xrightarrow{\;\;n\to\infty\;\;}0.
\end{align}
To this end we write
\begin{align}\nonumber
T_s^{V_n,n}(\V{z},\V{y})\psi-T_s^{V}(\V{z},\V{y})\psi&=
\big(T_s^{V_n,n}(\V{z},\V{y})-T_s^{V,n}(\V{z},\V{y})\big)\psi
\\\label{wilmar2}
&\quad+\big(T_s^{V,n}(\V{z},\V{y})-T_s^{V}(\V{z},\V{y})\big)\psi,
\end{align}
and estimate, using \eqref{bd-WW},
\begin{align}\nonumber
\big\|\big(&T_s^{V_n,n}(\V{z},\V{y})-T_s^{V,n}(\V{z},\V{y})\big)\psi\big\|
\\\nonumber
&\le p_s(\V{z},\V{y})\EE\big[\|\WW{s}{V_n,n}[\V{b}^{s;\V{y},\V{z}}]\psi-
\WW{s}{V,n}[\V{b}^{s;\V{y},\V{z}}]\psi\|\big]
\\\label{wilmar3}
&\le p_s(\V{z},\V{y})
\EE\big[|e^{-\int_0^sV_n(\V{b}_r^{s;\V{y},\V{z}})\Id r}
-e^{-\int_0^sV(\V{b}_r^{s;\V{y},\V{z}})\Id r}|\big]e^{\|\mho_n\|_\infty s}\|\psi\|.
\end{align}
Here, the terms $\|\mho_n\|_\infty$, which are defined by \eqref{def-lambda}
with $\V{F}$ replaced by $\V{F}_n$, are bounded uniformly in $n\in\NN$ 
by our assumptions on $\V{c}_n$.
The relation between the laws of $\V{b}^{s;\V{y},\V{z}}$ and $\V{B}^{\V{y}}$ further yields
\begin{align}\nonumber
\int_{\RR^\nu}&p_s(\V{z},\V{y})
\EE\big[|e^{-\int_0^{\sigma}V_n(\V{b}_r^{s;\V{y},\V{z}})\Id r}
-e^{-\int_0^{\sigma}V(\V{b}_r^{s;\V{y},\V{z}})\Id r}|\big]\Id\V{z}
\\\nonumber
&=\EE\Big[\int_{\RR^\nu}p_{s-\sigma}(\V{B}_\sigma^{\V{y}},\V{z})\Id\V{z}\,
\big|e^{-\int_0^{\sigma}V_n(\V{B}_r^{\V{y}})\Id r}
-e^{-\int_0^{\sigma}V(\V{B}_r^{\V{y}})\Id r}\big|\Big]
\\\label{clemens}
&=\EE\big[|e^{-\int_0^{\sigma}V_n(\V{B}_r^{\V{y}})\Id r}
-e^{-\int_0^{\sigma}V(\V{B}_r^{\V{y}})\Id r}|\big]
\end{align}
for all $\sigma\in(0,s)$, and, by dominated convergence, the equality between the right and left 
hand sides of \eqref{clemens} extends to $\sigma=s$. 
With the help of the dominated convergence theorem and \eqref{bridget72} we further obtain
\begin{align*}
\big(T_s^{V,n}(\V{z},\V{y})-T_s^{V}(\V{z},\V{y})\big)\psi
&=\lim_{\sigma\uparrow s}p_{s}(\V{z},\V{y})
\EE\big[\WW{\sigma}{V,n}[\V{b}^{s;\V{y},\V{z}}]\psi
-\WW{\sigma}{V}[\V{b}^{s;\V{y},\V{z}}]\psi\big]
\\
&=\lim_{\sigma\uparrow s}\EE\big[p_{s-\sigma}(\V{B}_\sigma^{\V{y}},\V{z})\big(
\WW{\sigma}{V,n}[\V{B}^{\V{y}}]-\WW{\sigma}{V}[\V{B}^{\V{y}}]\big)\psi\big].
\end{align*}
From the latter relation and Fatou's lemma we infer that
\begin{align}\nonumber
\int_{\RR^\nu}&\big\|\big(T_s^{V,n}(\V{z},\V{y})-T_s^{V}(\V{z},\V{y})\big)\psi\big\|\Id\V{z}
\\\nonumber
&\le\liminf_{\sigma\uparrow s}\int_{\RR^\nu}
\EE\Big[p_{s-\sigma}(\V{B}_\sigma^{\V{y}},\V{z})\big\|\big(
\WW{\sigma}{V,n}[\V{B}^{\V{y}}]-\WW{\sigma}{V}[\V{B}^{\V{y}}]\big)\psi\big\|\Big]\Id\V{z}
\\\nonumber
&=\EE\big[\|\WW{s}{V,n}[\V{B}^{\V{y}}]\psi-\WW{s}{V}[\V{B}^{\V{y}}]\psi\|\big]
\\\label{wilmar4}
&\le\EE\big[e^{-\int_0^sV(\V{B}_r^{\V{y}})\Id r}\big]^\eh\EE\big[
\|\WW{s}{0,n}[\V{B}^{\V{y}}]\psi-\WW{s}{0}[\V{B}^{\V{y}}]\psi\|^2\big]^\eh,
\end{align}
where we also applied Fubini's theorem and the dominated convergence theorem in the second 
step. The convergence \eqref{wilmar} is now implied by \eqref{dirk1}, \eqref{approx3},
\eqref{per7a}, \eqref{per7}, \eqref{wilmar2}, \eqref{wilmar3}, 
the extension of \eqref{clemens} to $\sigma=s$, and \eqref{wilmar4}
\end{proof}


\section{Positivity improvement by the semi-group kernel in the scalar case}\label{sec-pos}

\noindent
In this section we complement the discussion of the semi-group kernel by 
showing that, in the scalar case $L=1$ with either $\V{F}=0$ or $\V{G}=\V{0}$, $T^V_t(\V{x},\V{y})$
is positivity improving with respect to a suitable positive cone in the Fock space, for all $t>0$ and
$\V{x},\V{y}\in\RR^\nu$.
As an immediate corollary we re-obtain a theorem, due to Hiroshima \cite{Hiroshima2000}, 
stating that the semi-group in non-relativistic QED is positivity improving, 
if spin is neglected and the Pauli principle is discarded. 
We will apply this result in Sect.~\ref{sec-cont-GS} to discuss the continuous dependence
on model parameters of strictly positive ground state eigenvectors of scalar Hamiltonians.

As usual, the proof of the aforementioned results is an application of Perron-Frobenius type 
theorems \cite{Faris1972,ReedSimonIV}.
Let us point out, however, that our proof of Thm.~\ref{thm-pos-ker} is
based on a novel factorization of the Feynman-Kac integrand found in \cite{GMM2014}, from which
the positivity improvement by the integrand can be read off more easily.
In fact, the factorization used in Hiroshima's proof 
involved unbounded operators that lead to additional technical difficulties. 

Let us introduce the notation $\vo(g):=\vp(ig)$ and $W(g):=e^{-i\vo(g)}$, $g\in\HP$, for
the conjugate field and the associated Weyl operator, respectively. As usual,
$\Omega=(1,0,0,\ldots\,)\in\sF$ will denote the vacuum vector. There will be no danger of
confusing it with the underlying probability space that is denoted by the same symbol.

A suitable self-dual convex cone in the Fock space $\sF$
is now given by $\sP:=\ol{\mr{\sP}}$, the closure of the set
$$
\mr{\sP}:=\big\{F(\vo(\V{g}))\Omega:\,F\in\sS(\RR^n),\,F\ge0,\,\V{g}\in\HP_C^n,\,n\in\NN\big\}.
$$
Here $\sS(\RR^n)$ is the set of Schwartz test functions on $\RR^n$ and, 
for every $F\in\sS(\RR^n)$,
the bounded operator $F(\vo(\V{g})):=F(\vo(g_1),\ldots,\vo(g_n))$ is defined by the spectral calculus
for finitely many commuting self-adjoint operators.  For later reference, we observe that the latter 
operator is equal to the Bochner-Lebesgue integral
\begin{align}\label{kent1}
F(\vo(\V{g}))&=(2\pi)^{-n}\int_{\RR^n}\hat{F}(-\vxi)W(\vxi\cdot\V{g})\Id\vxi.
\end{align}
In fact, there exists a probability space $(\cQ,\mathfrak{Q},\eta)$ and a
unitary map $\sU:\sF\to L^2(\cQ,\mathfrak{Q},\eta)$, turning each conjugate
field operator $\vo(g)$, $g\in\HP_C$, into a multiplication operator with a Gaussian random variable,
such that $\sP$ is the pre-image under $\sU$ of all non-negative elements of 
$L^2(\cQ,\mathfrak{Q},\eta)$. Furthermore, $\sU\Omega=1$, and in particular $\Omega$ is
strictly positive with respect to $\sP$. See, e.g., \cite{LHB2011,Simon1974} for the
construction of $\eta$ and a detailed discussion of related matters.

Notice that $\psi\in\sF$ belongs to the
completely real subspace $\sF_C:=\{\phi\in\sF:\Gamma(-C)\phi=\phi\}$, if and only if
there exist $\psi_+,\psi_-\in\sP$ with $\psi_+\bot\psi_-$ and $\psi=\psi_+-\psi_-$.
In fact, $F\in\sS(\RR^n)$ is real-valued, if and only if $\ol{\hat{F}(-\vxi)}=\hat{F}(\vxi)$,
and we further have $\Gamma(-C)W(g)\Omega=W(-g)\Omega$, $g\in\HP_C$. 
Thus, \eqref{kent1} implies
$F(\vo(\V{g}))\Omega\in\sF_C$, for all real-valued $F\in\sS(\RR^n)$ and $\V{g}\in\HP_C^n$.

\begin{thm}\label{thm-pos-ker}
Consider the scalar case $L=1$ with $\V{F}={0}$. Let $t>0$ and $\V{x},\V{y}\in\RR^\nu$.
Then $T^V_t(\V{x},\V{y})$ maps $\sF_C$ into itself and it is positivity improving with respect to $\sP$.
\end{thm}

\begin{proof}
According to \cite[Rem.~F.2(1)]{GMM2014} we have the factorization
\begin{align}\label{fact}
\WW{t}{V}[\V{b}^{t;\V{x},\V{y}}]&=e^{-u^V_t}A_{\nf{t}{3}}[U^+_t]e^{-t\Id\Gamma(\omega)/3}
A_{\nf{t}{3}}[-U^-_t]^*,
\\\nonumber
A_s[f]&:=\sum_{n=0}^\infty\frac{i^n}{n!}\ad(f)^ne^{-s\Id\Gamma(\omega)},\quad
f\in\HP_C,\,s>0,
\end{align}
where $(u_s^V)_{s\in[0,t]}$ is a real-valued semi-martingale and $(U^\pm_s)_{s\in[0,t]}$
are $\HP_C$-valued semi-martingales, whose dependence on $(t,\V{x},\V{y})$ has been dropped
in the notation.
The sum defining $A_s[f]\in\LO(\sF)$ converges absolutely 
in the operator norm \cite[Lem.~F.1]{GMM2014}. Standard arguments \cite{Simon1974}
now show that $A_s[f]^*$ is 
positivity preserving: To see this, let $g\in\HP_C$ and $W(g)=e^{-i\vo(g)}$ be the corresponding 
Weyl operator. Since $W(g)\Omega$ is an analytic vector for $a(f)$, we easily find
\begin{align}\nonumber
A_s[f]^*W(g)\Omega&=e^{-s\Id\Gamma(\omega)}\sum_{n=0}^\infty\frac{(-i)^n}{n!}a(f)^nW(g)\Omega
\\\label{kent2}
&=e^{-\|g\|^2/2+\|e^{-s\omega}g\|^2/2-i\SPn{f}{g}}W(g),\quad f,g\in\HP_C.
\end{align}
Let $g_1,\ldots,g_n\in\HP_C$ and $F\in\sS(\RR^n)$.
Since the integration in \eqref{kent1} commutes with the bounded operator $A_s[f]^*$, 
we may employ \eqref{kent2} to get
\begin{align*}
A_s[f]^*F(\vp(\V{g}))\Omega&={(2\pi)^{-n}}\int_{\RR^n}
e^{-\|(1-e^{-2s\omega})\vxi\cdot\V{g}\|^2/2-i\vxi\cdot\SPn{f}{\V{g}}}
\hat{F}(-\vxi)W(\vxi\cdot\V{g})\Omega\,{\Id\vxi}.
\end{align*}
Another use of \eqref{kent1} yields
$A_s[f]^*F(\vp(\V{g}))\Omega=\wt{F}(\vo(\V{g}))\Omega$,
where $\wt{F}$ is the inverse Fourier transform of the Schwartz function
$\vxi\mapsto e^{-\|(1-e^{-2s\omega})\vxi\cdot\V{g}\|^2/2-i\vxi\cdot\SPn{f}{\V{g}}}\hat{F}(-\vxi)$.
If $F$ is non-negative, then $\wt{F}$ is non-negative as well.
We have thus shown that, for all $f\in\HP_C$ and $s>0$, 
$A_s[f]^*$ maps $\mr{\sP}$ into itself. Since it is bounded, it also maps $\sP$ into itself,
and this holds in particular for $e^{-s\Id\Gamma(\omega)}=A_s[0]^*$. It also follows that 
$A_s[f]=A_s[f]^{**}$ is positivity preserving. Since $\omega>0$, $\mu$-a.e., we further know that
$1$ is a non-degenerate eigenvalue of $e^{-s\Id\Gamma(\omega)}$, $s>0$, 
with strictly positive eigenvector $\Omega$. Therefore, the Perron-Frobenius type theorem
\cite[Thm.XIII.44(a)$\Rightarrow$(e)]{ReedSimonIV} (applied to $\sU e^{-s\Id\Gamma(\omega)}\sU^*$)
ensures that $e^{-s\Id\Gamma(\omega)}$ with $s>0$ is actually positivity improving. As a 
consequence, $\WW{t}{V}[\V{b}^{t;\V{x},\V{y}}]$ is positivity improving, pointwise on the underlying 
probability space, as a composition of positivity preserving operators 
one of which improves positivity. It is now clear that
$T^V_t(\V{x},\V{y})=p_t(\V{x},\V{y})\EE\big[\WW{t}{V}[\V{b}^{t;\V{x},\V{y}}]\big]$ is positivity
improving, too.
\end{proof}

\begin{cor}\label{cor-pos}
Consider the scalar case $L=1$ with $\V{F}={0}$. Let $t>0$.
Then $e^{-tH^V}$ is positivity improving with respect to the self-dual convex cone in 
$L^2(\RR^\nu,\FHR)$ given by
\begin{equation}\label{def-int-cone}
\int_{\RR^\nu}^\oplus\sP\Id\V{x}:=
\big\{\Psi\in L^2(\RR^\nu,\FHR):\,\Psi(\V{x})\in\sP,\,\text{a.e. $\V{x}$}\big\}.
\end{equation}
In particular, if $\inf\sigma(H^V)$ is an eigenvalue, then it is non-degenerate and
the corresponding eigenvector can be chosen strictly positive.
\end{cor}

\begin{proof}
The first assertion is obvious from Thm.~\ref{thm-pos-ker}. The last statement follows from Faris' 
Perron-Frobenius theorem \cite[Cor.~1.2]{Faris1972}.
\end{proof}

If $\V{G}$ vanishes, instead of $\V{F}$, then one can use the cone 
$\sP_i:=\ol{\mr{\sP}_i}=\Gamma(i)\sP$ with
\begin{align}\label{kent7}
\mr{\sP}_i&:=\big\{F(\vp(\V{g}))\Omega:\,F\in\sS(\RR^n),\,F\ge0,\,\V{g}\in\HP_C^n,\,n\in\NN\big\}
=\Gamma(i)\mr{\sP}.
\end{align}
Here the conjugate fields $\vo(g)$ have been replaced by the fields $\vp(g)$ in the definition of
$\mr{\sP}_i$, and we used the relations $\Gamma(i)W(g)\Gamma(-i)=W(-ig)$ and \eqref{kent1}
to get the second equality in \eqref{kent7}; also recall $\vo(-ig)=\vp(g)$.
Hence, $\sP_i$ is obtained upon replacing
$\HP_C$ by the completely real subspace $i\HP_C$ of $\HP$ associated
with the conjugation $-C$. Likewise, $\sF_{-C}:=\{\psi_+-\psi_-:\psi_+,\psi_-\in\sP_i\}$ is the 
completely real subspace of $\sF$ associated with the conjugation $\Gamma(C)$. 

\begin{thm}\label{thm-pos-Nelson}
Consider the scalar case $L=1$ with $\V{G}=\V{0}$. Let $t>0$ and $\V{x},\V{y}\in\RR^\nu$.
Then $T^V_t(\V{x},\V{y})$ maps $\sF_{-C}$ into itself and it is positivity improving with respect to
$\sP_i$. In particular, $e^{-tH^V}$ is positivity improving with respect to 
$\int_{\RR^\nu}\sP_i\Id\V{x}$ (defined as in \eqref{def-int-cone}).
\end{thm}

\begin{proof}
The only difference to the proofs of Thm.~\ref{thm-pos-ker} and Cor.~\ref{cor-pos} is that the
semi-martingales $(U^\pm_s)_{s\in[0,t]}$ appearing in the formula \eqref{fact} are now
$i\HP_C$-valued. The proof of Thm.~\ref{thm-pos-ker} thus shows that the transformed kernel 
$\Gamma(i)T^V_t(\V{x},\V{y})\Gamma(-i)$ is positivity improving with respect to $\sP$.
\end{proof}


\section{Continuity properties of ground states in the scalar case}\label{sec-cont-GS}

\noindent
In this section we discuss the joint continuity of ground state eigenvectors in non-relativistic QED
with respect to position coordinates and model parameters. Thanks to Cor.~\ref{cor-eq-cont}, all
what is left to do is to prove $L^2$-continuity with respect to model parameters of the ground state 
eigenvectors. This is, however, still a nontrivial task as the ground state eigenvalue typically is
embedded in the continuous spectrum so that the standard methods of analytic perturbation theory
are not available. 

Besides providing an example for an application of Cor.~\ref{cor-eq-cont}, the aim
of the present section, thus, is to demonstrate that a certain compactness argument usually applied 
to prove the existence of ground states \cite{GLL2001} can also be employed to prove continuous
dependence on model parameters of the ground state eigenvectors. Furthermore, we shall present
a simplified version of the compactness argument of \cite{GLL2001} that allows to work with more
general assumptions on the coefficient vector $\V{c}$ and thus reduces the amount of technicalities.
In particular, we do not have to specify a choice of the polarization vectors and we do not need a
certain photon derivative bound \cite{GLL2001}. The idea is, roughly, to apply the standard 
characterization of relatively compact sets in $L^2(\RR^d)$, instead of the Rellich-Kondrachov
theorem used in \cite{GLL2001}, in combination with the usual formula for $a(k)$ applied to a 
ground state eigenvector.

To implement these ideas we shall consider a simplified situation where only one electron is present
(no Pauli principle) and spin is neglected. The main reason for this is that we require
at least a constant, finite degeneracy of the ground state eigenvalues to infer the continuity from
the compactness result. (If the geometric multiplicity
of the ground state eigenvalue is bigger than one, then one could still try to prove norm continuity
of the ground state eigenprojections.) Besides the case of one electron with spin neglected, where
the methods discussed in Sect.~\ref{sec-pos} ensure non-degeneracy of ground state eigenvalues
and where we may talk about distinguished {\em positive} eigenvectors,
we are not aware of any non-perturbative method in non-relativistic QED that controls the 
degeneracy in a more general situation.

Since a detailed workout of all construction steps in \cite{Griesemer2004,GLL2001} required 
here would be too space consuming and too repetitive at the same time, we shall concentrate on 
those things that are changed and otherwise be rather sketchy in this section. 

In the whole section we shall consider the following situation:

\begin{hyp}\label{hyp-cont-GS}
{\rm(a)} The one-boson space is given by $\HP=L^2(\RR^3\times\{0,1\})$, the $L^2$-space 
associated with the product of the Lebesgue measure on $\RR^3$ and the counting measure on 
$\{0,1\}$. We consider the scalar case where $L=1$ and set $\nu=3$, thus $\HR=L^2(\RR^3,\sF)$. 
We choose $\omega(\V{k},\lambda):=|\V{k}|$, $\V{k}\in\RR^3$, $\lambda\in\{0,1\}$, and
$\V{c}=(\V{G},q,\V{0})$, $\V{c}_n=(\V{G}_n,q_n,\V{0})$, $n\in\NN$, are coefficient vectors
fufilling Hyp.~\ref{hyp-G} together with $\omega$ and some fixed conjugation $C$ on $\HP$ that
commutes with $\omega$. We assume that 
\begin{equation}
\lim_{n\to\infty}\sup_{\V{x}\in K}\|\V{c}_{n,\V{x}}-\V{c}_{\V{x}}\|_{\mathfrak{k}}=0,
\end{equation}
for all compact $K\subset\RR^3$. We further set $\wt{\V{G}}_{\V{x}}:=\V{G}_{\V{x}}-\V{G}_{\V{0}}$
and $\wt{\V{G}}_{n,\V{x}}:=\V{G}_{n,\V{x}}-\V{G}_{n,\V{0}}$, $n\in\NN$,
and assume that the maps
$\RR^3\ni\V{x}\mapsto(\omega^{-1}q_{\V{x}},\omega^{-1}\wt{\V{G}}_{\V{x}})\in\HP^4$
and $\RR^3\ni\V{x}\mapsto(\omega^{-1}q_{n,\V{x}},\omega^{-1}\wt{\V{G}}_{n,\V{x}})\in\HP^4$
are well-defined, bounded, and continuous with
\begin{align*}
\lim_{n\to\infty}\sup_{\V{x}\in K}\big\|(\omega^{-1}q_{\V{x}},\omega^{-1}\wt{\V{G}}_{\V{x}})
-(\omega^{-1}q_{n,\V{x}},\omega^{-1}\wt{\V{G}}_{n,\V{x}})\big\|_{\HP}&=0,
\end{align*}
for all compact $K\subset\RR^3$.

\noindent
{\rm(b)} $V,V_n\in\cK_\pm$, $n\in\NN$, are such that $V_n-V\in\cK(\RR^3)$, for all $n\in\NN$, and
\eqref{approx1} and \eqref{approx2} hold true.

\smallskip

\noindent
{\rm(c)} We denote by $H$ the Hamiltonian defined by means of $(\omega,\V{c},V)$ 
(according to Thm.~\ref{thm-FK}). For every $n\in\NN$, the Hamiltonian defined by means of 
$(\omega,\V{c}_n,V_n)$ is denoted by $H_n$. We further define
$$
E:=\inf\sigma(H),\;\;E_n:=\inf\sigma(H_n),\;\; 
\Sigma:=\lim_{R\to\infty}\Sigma_R,\;\;\Sigma_n:=\lim_{R\to\infty}\Sigma_{n,R},
$$ 
with
\begin{align*}
\Sigma_R&:=\inf\big\{\SPn{\Phi}{H\Phi}\,\big|\,\Phi\in\fdom(H)\subset L^2(\RR^3,\sF),\,
\|\Phi\|=1,\,\supp(\Phi)\subset\{|\V{x}|\ge R\}\big\},
\end{align*}
for all $R\ge1$, and analogous definitions of $\Sigma_{n,R}$. We assume that the
following {\em binding condition} is satisfied,
\begin{equation}\label{bc}
\Sigma>E.
\end{equation}
\end{hyp}

\begin{ex}
(1) In the situation of Ex.~\ref{ex-NRQED} with $N=1$, let $\{\chi_n\}_{n\in\NN}$ be a
sequence of even, non-negative functions on $\RR^3$ converging to $\chi$
in $L^2(\RR^3,(\omega^{-1}+\omega)\Id\V{k})$. Then the coupling functions
$\V{G}^\chi$, $\V{G}^{\chi_n}$, $n\in\NN$, fulfill Hyp.~\ref{hyp-cont-GS}(a)
(with $q=q_n=0$).

Notice that in this example the functions $\omega^{-1}\V{G}_{\V{x}}^{\chi_n}$ do {\em not} belong
to $\HP^3$, if $\chi_n$ is constant in a neighborhood of $0$, which is the reason why the
function $\wt{\V{G}}_{\V{x}}$ is introduced \cite{BFS1999}. The pre-factor $\omega^{-1}$
appears upon estimating the resolvents in the key relation \eqref{aki1} below as 
$\|R(k)\|\le\omega(k)^{-1}$.

\smallskip

\noindent
(2) Let us again consider $\V{G}^\chi$ as in Ex.~\ref{ex-NRQED}.
Assume that $V(\V{x})\to0$, $|\V{x}|\to\infty$, and that
$e_V:=\inf\sigma(-\tfrac{1}{2}\Delta\dot{+}V)$ is a strictly negative
discrete eigenvalue of the Schr\"odinger
operator $-\tfrac{1}{2}\Delta\dot{+}V$. (Here the dot on the $+$ indicates that the Schr\"odinger
operator is defined via quadratic forms.) 
Then it follows from the arguments in \cite{GLL2001} that $\Sigma-E\ge -e_V$ and in particular
that the binding condition \eqref{bc} is fulfilled.

There are also certain short range potentials $V$ where 
$e_V$ is equal to the lower end of the essential spectrum of the Schr\"odinger operator and 
nevertheless \eqref{bc} is satisfied; see \cite{KoenenbergMatte2013a} for a non-perturbative 
discussion of this {\em enhanced binding} effect and for references to earlier perturbative studies.

\smallskip

\noindent(3) The binding condition \eqref{bc} is clearly fulfilled if $V(\V{x})\to\infty$, $|\V{x}|\to\infty$.
In this case the {\em ionization threshold} $\Sigma$ is infinite.
\end{ex}

Straightforward variational arguments employing \eqref{rb-a}--\eqref{qfb-vp}
and Lem.~\ref{lem-Kato-qfb} reveal that (in $(-\infty,\infty]$)
\begin{align}\label{conv-En}
&\lim_{n\to\infty}E_n=E,\quad\lim_{n\to\infty}\Sigma_{n}=\Sigma,\quad
\lim_{n\to\infty}\Sigma_{n,R}=\Sigma_R,\;R\ge1.
\end{align}

Although the existence part of the following proposition is not formally covered by the existing 
literature in its full generality, we shall treat it as well-known, too. 

In this proposition and henceforth, strict positivity of a vector in $\HR$ is understood with respect 
to the cone $\int_{\RR^3}^\oplus\sP\Id\V{x}$ introduced in Cor.~\ref{cor-pos}.

\begin{prop}\label{prop-ex-GS}
In the situation of Hyp.~\ref{hyp-cont-GS}, $E$ is a non-degenerate eigenvalue of
$H$ and there exists a unique normalized, strictly positive, corresponding eigenvector of $H$
that we denote by $\Psi$. Furthermore, there exists $n_0\in\NN$
such that, for all $n\ge n_0$, $E_n$ is a non-degenerate eigenvalue of $H_n$ to which there
corresponds a unique normalized, strictly positive eigenvector $\Psi_n$.
\end{prop}

\begin{proof}
The non-degeneracy and strict positivity follow immediately from Cor.~\ref{cor-pos} 
as soon as we have pegged $E$ and $E_n$, $n\ge n_0$, as eigenvalues. Here the number $n_0$
is identical to the one found in Prop.~\ref{prop-loc} below, which ensures that the binding
condition $\Sigma_n>E_n$ holds, for $n\ge n_0$, as well. To prove that $E$ and $E_n$, 
$n\ge n_0$, are eigenvalues, one would then first replace $\omega$ by
$\omega_{\tilde{n}}(k):=(\V{k}^2+\tilde{n}^{-2})^\eh$, $\tilde{n}\in\NN$, and show that the 
Hamiltonians, thus modified,
have an eigenvalue at the bottom of their spectrum. This is by now fairly standard and works
under the assumptions of Hyp.~\ref{hyp-cont-GS} \cite{GLL2001}.
Let $\Psi^{(\tilde{n})}$ and $\Psi_n^{(\tilde{n})}$, $n\ge n_0$, be corresponding normalized 
eigenvectors. Then, in the next step, one removes the artificial photon mass $\tilde{n}^{-1}$ 
by a compactness argument implying that $\{\Psi^{(\tilde{n})}\}_{\tilde{n}\in\NN}$ and 
$\{\Psi_n^{(\tilde{n})}\}_{\tilde{n}\in\NN}$ contain norm convergent subsequences.
The arguments of \cite{GLL2001} do not apply here directly, because we work with less specific
assumptions on the coefficient vectors. To remedy this one can, however, adapt the
compactness argument presented in this section: just consider constant sequences of
potentials $V$ and coefficient vectors $\V{c}$, and the non-constant sequence
$\omega_{\tilde{n}}$, $\tilde{n}\in\NN$, instead.
\end{proof}

\begin{thm}\label{thm-cont-GS}
In the setting of Hyp.~\ref{hyp-cont-GS} and using the notation of Prop.~\ref{prop-ex-GS},
we may conclude that the following holds true:

\smallskip

\noindent{\rm(1)} $\Psi_n\to\Psi$, $n\to\infty$, in $L^2(\RR^3,\sF)$.

\smallskip

\noindent{\rm(2)}
Let $\{\V{x}_n\}_{n\in\NN}$ be a converging sequence in $\RR^3$ with limit $\V{x}$ and let us
consider the unique continuous representatives of the ground state eigenvectors. Then
$\Psi_n(\V{x}_n)\to\Psi(\V{x})$, $n\to\infty$, in $\sF$. If $t>0$ and
$(\Upsilon_t)_{t\ge0}$ and $\|\cdot\|_*$ are given by either Line~2 or Line~4 of Table~1, and if
we additionally assume that the maps
$\RR^\nu\ni\V{x}\mapsto\V{c}_{\V{x}}$ and $\RR^\nu\ni\V{x}\mapsto\V{c}_{n,\V{x}}$, $n\in\NN$,
are continuous and uniformly bounded with respect to the norm $\|\cdot\|_*$, then
$\Upsilon_t\Psi_n(\V{x}_n)\to\Upsilon_t\Psi(\V{x})$, $n\to\infty$, in $\sF$.
\end{thm}

\begin{proof}
(1): Let $\{\Psi_{n_j}\}_{j\in\NN}$ be any subsequence of $\{\Psi_n\}_{n\ge n_0}$.
By Prop.~\ref{prop-GS} below this subsequence contains another subsequence,
say $\{\Psi_{\ell_j}\}_{j\in\NN}$, such that $\Psi_{\ell_j}\to\Psi_\infty$, $j\to\infty$,
for some normalized $\Psi_\infty\in\HR$. Since each $\Psi_{\ell_j}$ is strictly positive,
$\Psi_\infty$ is non-negative. Furthermore, Thm.~\ref{thm-LNCV}, Thm.~\ref{thm-coup},
and Hyp.~\ref{hyp-cont-GS} imply that $H_n$ converges to $H$ in the strong resolvent sense;
recall the proof of Cor.~\ref{cor-SCV} and that strong convergence of the semi-group implies
strong resolvent convergence. Taking also \eqref{conv-En} into account, we deduce that
$$
\Psi_\infty=\lim_{j\to\infty}\Psi_{\ell_j}=\lim_{j\to\infty}(H_{\ell_j}-E_{\ell_j}+1)^{-1}\Psi_{\ell_j}
=(H-E+1)^{-1}\Psi_\infty,
$$
whence $\Psi_\infty\in\dom(H)$ and $H\Psi_\infty=E\Psi_\infty$. Since $\Psi_\infty$ is normalized
and non-negative, we must have that $\Psi_\infty=\Psi$. Hence, every subsequence of
$\{\Psi_n\}_{n\ge n_0}$ contains a subsequence that converges to $\Psi$. This is, however, 
only possible, if $\Psi_n\to\Psi$, $n\to\infty$, in $\HR$.

(2): Set $\tau:=2t$.
By \eqref{conv-En} and Part~(1) the vectors $\Psi_n':=e^{\tau E_n}\Psi_n$, $n\ge n_0$,
converge in $L^2(\RR^3,\sF)$ to $\Psi':=e^{\tau E}\Psi$, as $n$ goes to infinity. Furthermore,
$\Psi=e^{-\tau H}\Psi'=T_\tau^V\Psi'$ and similarly $\Psi_n=T_\tau^{V_n,n}\Psi_n'$, $n\ge n_0$, 
where we use some notation of Cor.~\ref{cor-eq-cont}. 
The latter now implies all statements under (2).
\end{proof}

\begin{prop}\label{prop-GS}
Every subsequence of the eigenvector sequence $\{\Psi_n\}_{n\ge n_0}$ found in
Prop.~\ref{prop-ex-GS} contains another subsequence that converges in $\HR$.
\end{prop}

The previous proposition is proved at the end of this section. The remainder of this section
serves as a preparation for that proof.

The compactness argument carried through below (see in particular Lem.~\ref{lem-cpt1})
requires some information on the uniform
spatial localization of the ground state eigenvectors $\Psi_n$. This is provided by the next 
proposition which is a corollary of a result in \cite{Griesemer2004}.

In what follows, the symbol $F(\hat{\V{x}})$ will be a convenient notation for
the multiplication operator associated with a measurable function 
$\RR^3\ni\V{x}\mapsto F(\V{x})$. 

\begin{prop}\label{prop-loc}
In the situation of Hyp.~\ref{hyp-cont-GS},
pick $a,b>0$ and $\delta\ge0$ with $E+\delta+b+a^2/2<\Sigma$.
Then we find some $n_0\in\NN$ such that 
\begin{align}\label{tizian0}
\Sigma_n\ge E_n+\delta+b+a^2/2,\quad n\ge n_0.
\end{align}
We further find $c,c',c''>0$ such that, for every $n\ge n_0$ and every normalized element 
$\Phi_n$ of the range of the spectral projection $1_{(-\infty,E_n+\delta]}(H_n)$,
the following bounds hold true,
\begin{align}\label{tizian1}
\|e^{a|\hat{\V{x}}|}\Phi_n\|&\le c,
\\\label{tizian2}
\int_{\RR^3}e^{2a|\V{x}|}\big\|(-i\nabla-\vp(\V{G}_{n,\V{x}}))\Phi_n(\V{x})\big\|^2\Id\V{x}&\le c',
\\\label{tizian5}
\int_{\RR^3}e^{2a|\V{x}|}\big\|\nabla\Phi_n(\V{x})\big\|^2\Id\V{x}&\le c''.
\end{align}
Furthermore, if $\sup\{\|\omega^\alpha\V{c}_{n,\V{x}}\|:n\in\NN,\V{x}\in\RR^3\}<\infty$,
for some $\alpha\ge1$, then there exists $c'''>0$ such that, for every $n\ge n_0$, the
unique continuous representative of $\Psi_n$ satisfies
\begin{align}\label{tizian10}
\|(1+\Id\Gamma(\omega))^\alpha\Phi_n(\V{x})\|\le c'''e^{-a|\V{x}|},\;\;\V{x}\in\RR^\nu.
\end{align}
If $\sup\{\|e^{\delta_0\omega}\V{c}_{n,\V{x}}\|:n\in\NN,\V{x}\in\RR^3\}<\infty$,
for some $\delta_0>0$, then there exist $c''''>0$ and $\delta\in(0,\delta_0]$ such that, 
for all $n\ge n_0$,
\begin{align}\label{tizian11}
\|e^{\delta\Id\Gamma(\omega)}\Phi_n(\V{x})\|\le c''''e^{-a|\V{x}|},\;\;\V{x}\in\RR^\nu.
\end{align}
\end{prop}

\begin{rem}\label{rem-loc}
Again the pointwise bounds \eqref{tizian10} and \eqref{tizian11} are new
for general elements in spectral subspaces that are not eigenvectors. 
If $V(\V{x})$ goes to infinity, as $|\V{x}|\to\infty$, which clearly entails
$\Sigma=\infty$, then one can actually observe a much faster, possibly super-exponential decay
of ground state eigenfunctions \cite{HiHi2010,Hiroshima2003,LHB2011}.
The bounds of Thm.~\ref{prop-loc} are, however, more than sufficient for the present purpose. 
\end{rem}

\begin{proof}[Proof of Prop.~\ref{prop-loc}]
Combining the estimates in the proof of \cite[Thm.~1]{Griesemer2004} and analyzing the
integral involving the almost analytic function appearing there, it is easy to see that
there exists a universal constant $C>0$ such that
\begin{align}\label{tizian4}
\|e^{a|\hat{\V{x}}|}\Phi_n\|&\le Ce^{2aR}\big(\Sigma_{n,R}-E_n\big)
(\delta+b)\max\{b^{-1},b^{-3}\},\quad n\ge n_0,
\end{align}
provided that we can choose $R\ge1$ and $n_0\in\NN$ so large that $16/R^2\le b$
and $\Sigma_{n,R}\ge E_n+\delta+3b/4+a^2/2$, for all $n\ge n_0$.
That this is actually possible can be seen as follows: First, we fix some $R\ge4/b^\eh$ such that 
$\Sigma_R\ge E+\delta+7b/8+a^2/2$.
After that we avail ourselves of \eqref{conv-En} and pick $n_0$ so large that 
$|\Sigma_{n,R}-\Sigma_R|\le b/16$ and $|E_n-E|\le b/16$, for all $n\ge n_0$.
We conclude that \eqref{tizian4} implies a bound on its left hand side that is uniform in $n\ge n_0$.
This proves \eqref{tizian1}. Next, we set $F(\V{x}):=a|\V{x}|$, $\V{x}\in\RR^\nu$,
and notice that the validity of \eqref{tizian1} is equivalent to the bound 
$\|e^F1_{(-\infty,E_n+\delta]}(H_n)\|\le c$.
Writing $\Phi_n=e^{-H_n}1_{(-\infty,E_n+\delta]}(H_n)e^{H_n}\Phi_n$ 
and denoting the semi-group associated with $(V_n,\V{c}_n)$ by $(T^{V_n,n}_t)_{t\ge0}$,
we thus obtain
\begin{align*}
\|\Phi_n(\V{x})\|\le e^{-F(\V{x})+E_n+\delta}\big\|e^FT^{V_n,n}_1e^{-F}\big\|_{2,\infty}\|e^F\Phi_n\|,
\quad\|e^F\Phi_n\|\le c,
\end{align*}
whence \eqref{tizian10} and \eqref{tizian11} follow from \eqref{norm-T-F-Theta} and the fact
that $\Sigma_R$ is an upper bound on each $E_n+\delta$ with $n\ge n_0$.

Let $\chi\in C^\infty(\RR^\nu)$ be bounded with bounded first order partial derivatives. Then it is 
easy to see that $\chi\fdom(H)\subset\fdom(H)\subset\dom(-i\partial_{x_j}-\vp(G_{n,j}))$,
$j\in\{1,2,3\}$, and we obtain
\begin{align}\label{tizian50}
\big\|\chi(-i\nabla-\vp(\V{G}_n))\Phi_n\big\|&\le
\big\|(-i\nabla-\vp(\V{G}_n))\chi\Phi_n\big\|+\|(\nabla\chi)\Phi_n\|,\quad n\in\NN.
\end{align}
Denoting the quadratic form associated with $H_n$ by $\mathfrak{q}_n$, we further have, 
for some $\beta>0$ and all $n\ge n_0$,
\begin{align}\nonumber
\big\|(&-i\nabla-\vp(\V{G}_n))\chi\Phi_n\big\|^2\le4\mathfrak{q}_n[\chi\Phi_n]+2\beta\|\chi\Phi_n\|^2
\\\nonumber
&=4\Re\,\mathfrak{q}_n[\chi^2\Phi_n,\Phi_n]+4\big\|(\nabla\chi)\Phi_n\big\|^2+2\beta\|\chi\Phi_n\|^2
\\\label{tizian51}
&=4\SPb{\chi\Phi_n}{\chi1_{(-\infty,E_n+\delta]}(H_n)H_n\Phi_n}
+4\big\|(\nabla\chi)\Phi_n\big\|^2+2\beta\|\chi\Phi_n\|^2,
\end{align}
where we used a well-known identity for the form $\mathfrak{q}_n$ in the second step;
see, e.g., \cite[p.~324]{Griesemer2004}. We now infer from \eqref{tizian50} and \eqref{tizian51} that,
if $\chi$ is of the form $\chi=e^F$ with a bounded smooth $F$ satisfying $|\nabla F|\le a$ 
and $F(\V{0})=0$,
\begin{align*}
\frac{1}{2}\big\|e^F(-i\nabla-\vp(\V{G}_n))\Phi_n\big\|^2&\le
(4\Sigma_R+2\beta+5a^2)c^2.
\end{align*}
Here we also used that $\|H_n\Phi_n\|\le E_n+\delta\le\Sigma_R$
and $\|e^F1_{(-\infty,E_n+\delta]}(H_n)\|\le c$, which is equivalent to the
validity of the first bound in \eqref{tizian1}.
Inserting $F_\ve(\V{x}):=a\langle\V{x}\rangle/(1+\ve\langle\V{x}\rangle)-a/(1+\ve)$, 
$\V{x}\in\RR^\nu$, we obtain
\eqref{tizian2} in the limit $\ve\downarrow0$ with the help of the monotone convergence theorem.
Finally, \eqref{tizian5} follows from \eqref{rb-a}, \eqref{rb-ad}, \eqref{tizian2}, and \eqref{tizian10} 
with $\alpha=1$ and a slightly larger $a$.
\end{proof}

 The (essentially well-known) relation \eqref{aki1} asserted in the next proposition is the key identity
in the compactness argument. It allows to exploit the convergence properties of the
coefficient vectors $\V{c}_n$ in order to prove relative compactness of $\{\Psi_n:n\in\NN\}$
by means of the standard characterization of relatively compact subsets of $L^2(\RR^d)$.

In what follows it will sometimes be more convenient to write $(a\Phi)(k)$ instead of using
the more customary notation $a(k)\Phi$ for the pointwise annihilation operator.

\begin{prop}\label{prop-IR}
In the situation described by Hyp.~\ref{hyp-cont-GS}, let $n\in\NN$ and write
\begin{align}\nonumber
R_n(k)&:=(H_n-E_n+\omega(k))^{-1},\quad
B_n(k):=(H_n-E_n)(H_n-E_n+\omega(k))^{-1}.
\end{align}
Let $\Psi_n$ be the ground state eigenvector found in Prop.~\ref{prop-ex-GS}.
Then the following identity holds, for a.e. 
$k=(\V{k},\lambda)\in(\RR^3\setminus\{\V{0}\})\times\{0,1\}$,
\begin{align}\nonumber
(a\Psi_n)(k)&=B_n(k)\V{G}_{n,\V{0}}(k)\cdot\hat{\V{x}}\Psi_n
+R_n(k)\tfrac{i}{2}q_{n,\hat{\V{x}}}(k)\Psi_n
\\\label{aki1}
&\quad+R_n(k)\wt{\V{G}}_{n,\hat{\V{x}}}(k)\cdot(-i\nabla-\vp(\V{G}_{n}))\Psi_n.
\end{align}
\end{prop}

\begin{proof}
The proof follows a standard procedure 
(see, e.g., \cite{BFS1999,GLL2001,Hiroshima2003,Spohn2004}), 
whence we drop it here.
We just remark that the technical details will be facilitated by exploiting that $H_n$
is essentially self-adjoint on the complex linear span of the functions $f\phi$, where
$f$ is in the domain of the Schr\"{o}dinger operator $-\tfrac{1}{2}\Delta\dot{+}V$ and
$\phi$ is, e.g., an analytic vector for $\Id\Gamma(\omega)$; 
in the situation of Hyp.~\ref{hyp-cont-GS}
the latter result does not yet follow from the existing literature, but it will be shown in the
forthcoming note \cite{Matte-esa}.
\end{proof}

\begin{cor}\label{cor-IRa}
In the situation of Prop.~\ref{prop-IR}, $\sup_{n\in\NN}\|\Id\Gamma(1)^\eh\Psi_n\|<\infty$.
\end{cor}

\begin{proof}
As usual, this follows easily from Hyp.~\ref{hyp-cont-GS}, \eqref{aki1}, and the fact that
$\Phi\in\fdom(\Id\Gamma(1))$, if and only if the right hand side of the formula
$\|\Id\Gamma(1)^\eh\Phi\|^2=\sum_{\lambda=0}^1\int_{\RR^3}\|a(\V{k},\lambda)\Phi\|^2\Id\V{k}$
is finite, the latter being valid in the affirmative case.
\end{proof}

\begin{cor}\label{cor-IR}
In the situation of Prop.~\ref{prop-IR},
let $\wt{\chi}\in C^\infty(\RR,\RR)$ be such that $0\le\wt{\chi}\le1$, $\wt{\chi}=0$ on $(-\infty,1]$
and $\wt{\chi}=1$ on $[2,\infty)$. 
Set $\chi_\delta(\V{k}):=\wt{\chi}(|\V{k}|/\delta)$, for all $\V{k}\in\RR^3$ and $\delta\in(0,1]$.
For all $0\le r_0<r_1\le\infty$, $\delta\in(0,1]$, and $\V{h}\in\RR^3$, we further define
\begin{align*}
F(r_0,r_1)&:=\sup_{n\in\NN}\sum_{\lambda=0}^1\int_{\RR^3}1_{\{r_0\le|\V{k}|\le r_1\}}
\|(a\Psi_n)(\V{k},\lambda)\|^2\Id\V{k},
\\
\triangle_\delta(\V{h})&:=\sup_{n\in\NN}\sum_{\lambda=0}^1\int_{\RR^3}
\|(\chi_\delta a\Psi_n)(\V{k},\lambda)-(\chi_\delta a\Psi_n)(\V{k}+\V{h},\lambda)\|^2\Id\V{k}.
\end{align*}
Then $F(r_0,\infty)\to0$, as $r_0\to\infty$, $F(0,r_1)\to0$, as $r_1\downarrow0$,
and $\triangle_\delta(\V{h})\to0$, as $\V{h}\to\V{0}$, for every $\delta\in(0,1]$.
\end{cor}

\begin{proof}
To start with we observe that $\|B_n(k)\|\le1$, $\|R_n(k)\|\le1/\omega(k)$,
as well as
$\|1_{K_R^c}(\hat{\V{x}})\langle\hat{\V{x}}\rangle^{-1}\|\le1/R$, $R\ge1$, and that,
by \eqref{tizian1} and \eqref{tizian2}, the norms of $\langle\hat{\V{x}}\rangle^2\Psi_n$ and 
$\langle\hat{\V{x}}\rangle(-i\nabla-\vp(\V{G}_{n}))\Psi_n$ are bounded uniformly in $n$.

Let $\ve>0$ and define
\begin{align}\nonumber
\Phi_n^R(k)&:=B_n(k)\V{G}_{n,\V{0}}(k)\cdot1_{K_R}(\hat{\V{x}})\hat{\V{x}}\Psi_n
+R_n(k)\tfrac{i}{2}q_{n,\hat{\V{x}}}(k)1_{K_R}(\hat{\V{x}})\Psi_n
\\\label{aki81}
&\quad+R_n(k)\wt{\V{G}}_{n,\hat{\V{x}}}(k)\cdot 1_{K_R}(\hat{\V{x}})(-i\nabla-\vp(\V{G}_{n}))\Psi_n.
\end{align}
By virtue of Hyp.~\ref{hyp-cont-GS} and the preceding remarks, 
we may fix $R\ge1$ such that
$\|a\Psi_n-\Phi_n^R\|_{L^2(\RR^3\times\{0,1\},\HR)}<\ve$, for all $n\in\NN$.
Moreover, it follows from \eqref{aki81} that
\begin{align*}
\sup_{n\in\NN}&\sum_{\lambda=0}^1\int_{\RR^3}1_{\{r_0\le|\V{k}|\le r_1\}}
\|\Phi_n^R(\V{k},\lambda)\|^2\Id\V{k}
\\
&\le c^2\sup_{n\in\NN}\sup_{\V{x}\in K_R}\|1_{\{r_0\le\omega\le r_1\}}
\mathfrak{g}_{n,\V{x}}\|_{\HP}^2=:\wt{F}_R(r_0,r_1),
\end{align*}
for some $(n,r_0,r_1)$-independent $c>0$, where
$\mathfrak{g}_{n,\V{x}}:=(\V{G}_{n,\V{0}},\omega^{-1}q_{n,\V{x}},
\omega^{-1}\wt{\V{G}}_{n,\V{x}})$. By Hyp.~\ref{hyp-cont-GS} the sequence of functions 
$\{\mathfrak{g}_{n}\}_n$ converges in $C(K_R,\HP^7)$. Therefore, if we turn
$\ol{\NN}:=\NN\cup\{\infty\}$ into a compact topological space by demanding that the map 
$n\mapsto n^{-1}$ is a homeomorphism onto $\{0\}\cup\{n^{-1}:n\in\NN\}\subset\RR$, then the map 
${\NN}\times K_R\ni(n,\V{x})\mapsto\mathfrak{g}_{n,\V{x}}\in\HP^7$ has a jointly continuous 
extension to $\ol{\NN}\times K_R$ and in particular its image is relatively compact in $\HP^7$.
We can thus find finitely many
functions in $\HP^{7}$ with compact supports in $(\RR^3\setminus\{\V{0}\})\times\{0,1\}$,
such that the set $\{\mathfrak{g}_{n,\V{x}}:n\in\NN,\V{x}\in K_R\}$ is contained in the union of 
the open balls of radius $\ve$ about
those compactly supported functions. If $r_0$ is large enough  (resp. $r_1$ small enough),
then the supports of all these finitely many compactly supported functions 
are disjoint from $\{r_0\le\omega\}$ (resp. $\{\omega\le r_1\}$).
Then it follows that $\limsup_{r_0\to\infty}\wt{F}_R(r_0,\infty)^\eh\le c\ve$, 
thus $\limsup_{r_0\to\infty}F(r_0,\infty)^\eh\le(1+c)\ve$,
and similarly $\limsup_{r_1\downarrow0}{F}(0,r_1)^\eh\le (1+c)\ve$.

We may further verify that the operator-valued functions 
$\RR^d\setminus\{\V{0}\}\ni\V{k}\mapsto\chi_\delta(\V{k})B_n(\V{k},\lambda)$ and 
$\RR^d\setminus\{\V{0}\}\ni\V{k}\mapsto\chi_\delta(\V{k})R_n(\V{k},\lambda)$ 
with $\lambda\in\{0,1\}$ are continuously differentiable with 
$\|\nabla_{\V{k}}\{\chi_\delta(\V{k})B_n(\V{k},\lambda)\}\|\le c'/\delta^2$
and $\|\nabla_{\V{k}}\{\chi_\delta(\V{k})R_n(\V{k},\lambda)\}\|\le c'/\delta^2$,
for some universal constant $c'>0$ and all $\delta\in(0,1]$. 
Likewise, $|\nabla_{\V{k}}(\chi_\delta/\omega)|\le c'/\delta^2$, $\delta\in(0,1]$.
Combining these remarks with the observations made in the first paragraph of this proof,
we easily find $c'',c'''>0$ such that, for all $\delta\in(0,1]$ and $\V{h}\in\RR^d$,
\begin{align}\nonumber
\sup_{n\in\NN}&\sum_{\lambda=0,1}\int_{\RR^3}
\|(\chi_\delta \Phi_n^R)(\V{k},\lambda)-(\chi_\delta\Phi_n^R)(\V{k}+\V{h},\lambda)\|^2\Id\V{k}
\\\nonumber
&\le\frac{c''}{\delta^2}\,|\V{h}|^2\sup_{n\in\NN}\sup_{\V{x}\in K_R}\sum_{\lambda=0}^1
\int_{\RR^d}|\omega\mathfrak{g}_{n,\V{x}}|^2(\V{k}+\V{h},\lambda)\Id\V{k}
\\\label{pat2}
&\quad+c'''\sup_{n\in\NN}\sup_{\V{x}\in K_R}\sum_{\lambda=0}^1
\int_{\RR^d}\big|(\omega\mathfrak{g}_{n,\V{x}})(\V{k}+\V{h},\lambda)
-(\omega\mathfrak{g}_{n,\V{x}})(\V{k},\lambda)\big|^2\Id\V{k}.
\end{align}
The supremum in the first line of the right hand side of \eqref{pat2} is obviously 
$\V{h}$-independent; it is finite, because $\{\mathfrak{g}_n\}_n$ converges in 
$C(K_R,\HP^{7})$. The supremum in the
second line of \eqref{pat2} goes to $0$ as $\V{h}\to\V{0}$ since, by the above arguments, the set 
$\{\omega\mathfrak{g}_{n,\V{x}}:n\in\NN,\V{x}\in K_R\}$ is relatively compact in $\HP^7$ and 
in particular the shift operation applied to its elements is norm-continuous uniformly in  
$(n,\V{x})\in\NN\times K_R$. Altogether this implies 
$\limsup_{\V{h}\to\V{0}}\triangle_\delta(\V{h})^\eh\le2\ve$ and we conclude recalling that $\ve>0$
was arbitrary.
\end{proof}

Employing the canonical isomorphism $\HR=\bigoplus_{m=0}^\infty L^2(\RR^3,\sF^{(m)})$, 
we represent the ground state eigenvectors found in Prop.~\ref{prop-ex-GS}
as sequences $\Psi_n=(\Psi_n^{(m)})_{m\in\NN_0}$
with $\Psi_n^{(m)}\in\sF^{(m)}$ in the next lemma.

\begin{lem}\label{lem-cpt1}
Let $m\in\NN$, $\ul{\lambda}=(\lambda_1,\ldots,\lambda_m)\in\{0,1\}^m$, and $\delta\in(0,1]$. Put
\begin{align*}
\Psi_{\delta,n}^{(m,\ul{\lambda})}(\V{x},\V{k}_1,\ldots,\V{k}_m)
:=\Big(\prod_{j=1}^m\chi_\delta(\V{k}_j)\Big)
\Psi_{n}^{(m)}(\V{x},\V{k}_1,\lambda_1,\ldots,\V{k}_m,\lambda_m),
\end{align*}
where $\chi_\delta$ is defined in the statement of Cor.~\ref{cor-IR}. Then the set 
$$
\cK_\delta^{(m,\ul{\lambda})}:=\{\Psi_{\delta,n}^{(m,\ul{\lambda})}:n\in\NN\}
$$ 
is relatively compact in $L^2(\RR^{3(m+1)})$.
\end{lem}

\begin{proof}
For a start, it is clear that $\cK_\delta^{(m,\ul{\lambda})}$ is contained in the closed unit ball 
about the origin in $L^2(\RR^{3(m+1)})$.
Furthermore, writing $\V{k}_{[m]}:=(\V{k}_1,\ldots,\V{k}_m)$, etc., 
we observe the following pointwise bounds between functions 
defined on $\RR^{3(m+1)}$,
\begin{align*}
1_{\{|(\V{x},\V{k}_{[m]})|\ge R\}}
&\le1_{\{|\V{x}|\ge R/(m+1)\}}+\sum_{j=1}^m1_{\{|\V{k}_j|\ge R/(m+1)\}},
\quad R>0.
\end{align*}
Then the permutation symmetry of $\Psi_{n}^{(m)}$ in its last $n$ variable pairs
$(\V{k}_j,\lambda_j)$ implies, for all $\delta\in(0,1]$ and with $a$ and $c$ as in \eqref{tizian1},
\begin{align*}
\int_{\RR^{3(m+1)}}&1_{\{|(\V{x},\V{k}_{[m]})|\ge R\}}|\Psi_{\delta,n}^{(m,\ul{\lambda})}
(\V{x},\V{k}_{[m]})|^2\Id(\V{x},\V{k}_{[m]})
\\
&\le\int_{\RR^{3(m+1)}}1_{\{|\V{x}|\ge R/(m+1)\}}|\Psi_{n}^{(m,\ul{\lambda})}(\V{x},\V{k}_{[m]})|^2
\Id(\V{x},\V{k}_{[m]})
\\
&\quad
+m\sum_{\ul{\lambda}'\in\{0,1\}^m}\int_{\RR^{3(m+1)}}1_{\{|\V{k}_1|\ge R/(m+1)\}}
|\Psi_{n}^{(m,\ul{\lambda}')}(\V{x},\V{k}_{[m]})|^2\Id(\V{x},\V{k}_{[m]})
\\
&\le\|1_{\{|\hat{\V{x}}|\ge R/(m+1)\}}\Psi_n\|^2+\sum_{\lambda=0}^1
\int_{\RR^3}1_{\{|\V{k}_1|\ge R/(m+1)\}}\|a(\V{k}_1,\lambda)\Psi_n\|^2\Id\V{k}_1
\\
&\le c^2e^{-2aR/(m+1)}+F(R/(m+1),\infty)\xrightarrow{\;\;R\to\infty\;\;}0.
\end{align*}
Since the expression tending to zero in the last line is $n$-independent, we see that
$\cK_\delta^{(m,\ul{\lambda})}$ is uniformly integrable. 

Next, let $S_{(\V{y},\V{h}_{[m]})}$ be the unitary operator that shifts the variables
$(\V{x},\V{k}_{[m]})$ by the vector $(\V{y},\V{h}_{[m]})\in\RR^{3(m+1)}$ and let
$S^{(1)}_{\V{h}_j}$ be the unitary operator shifting only the variable $\V{k}_1$ by $\V{h}_j$.
By a telescopic sum argument and the permutation symmetry of $\Psi_{\delta,n}^{(m)}$ 
in its last $m$ variables, we then obtain
\begin{align*}
\|\Psi_{\delta,n}^{(m,\ul{\lambda})}-S_{(\V{y},\V{h}_{[m]})}\Psi_{\delta,n}^{(m,\ul{\lambda})}\|
&\le\|\Psi_{n}-S_{(\V{y},\V{0})}\Psi_{n}\|
+\sum_{j=1}^m\|\Psi_{\delta,n}^{(m)}-S_{\V{h}_{j}}^{(1)}\Psi_{\delta,n}^{(m)}\|.
\end{align*}
By the definition of the pointwise annihilation operator,
\begin{align*}
&\|\Psi_{\delta,n}^{(m)}-S_{\V{h}}^{(1)}\Psi_{\delta,n}^{(m)}\|^2
\\
&=\frac{1}{m}\sum_{\lambda=0}^1\int_{\RR^3}
\Big\|\chi_\delta^{\otimes_{m-1}}\Big((\chi_\delta a\Psi_n)^{(m-1)}(\V{k},\lambda)
-(\chi_\delta a\Psi_n)^{(m-1)}(\V{k}+\V{h},\lambda)\Big)\Big\|^2\Id\V{k},
\end{align*}
where the norm under the $\Id\V{k}$-integral is the one on 
$L^2(\RR^{3}\times(\RR^3\times\{0,1\})^{3(m-1)})$
and $\chi_\delta^{\otimes_{m-1}}$ is the multiplication operator
associated with the function 
$(\V{k}_1,\ldots,\V{k}_{m-1})\mapsto\chi_\delta(\V{k}_1)\cdots\chi_\delta(\V{k}_{m-1})$.
Therefore, 
$\sup_{n}\|\Phi_{\delta,n}^{(m)}-S_{\V{h}}^{(1)}\Phi_{\delta,n}^{(m)}\|^2
\le\triangle_\delta(\V{h})$. We finally observe that
$\|\Psi_{n}-S_{(\V{y},\V{0})}\Psi_{n}\|\le|\V{y}|\|\hat{\vxi}\hat{\Psi}_n\|\le c'''|\V{y}|$,
with the $n$-independent constant $c'''$ appearing in \eqref{tizian5}.
(Here $\hat{\Psi}_n$ is the $\sF$-valued Fourier transform of $\Psi_n$.)
Putting these remarks together and applying Cor.~\ref{cor-IR}, we deduce that
$$
\sup_{n\in\NN}\|\Psi_{\delta,n}^{(m,\ul{\lambda})}
-S_{(\V{y},\V{h}_{[m]})}\Phi_{\delta,n}^{(m,\ul{\lambda})}\|
\xrightarrow{\;\;\;|(\V{y},\V{h}_{[m]})|\to0\;\;\;}0.
$$
By the well-known characterization of relatively compact sets in $L^2(\RR^{3(m+1)})$, 
this proves the assertion.
\end{proof}

Finally, we can keep a promise that will also conclude the proof of Thm.~\ref{thm-cont-GS}:

\begin{proof}[Proof of Prop.~\ref{prop-GS}]
Let $1\le n_1<n_2<\ldots\,$ be integers. Then $\{\Psi_{n_j}\}_{j\in\NN}$ contains a weakly
converging subsequence, say $\{\Psi_{\ell_j}\}_{j\in\NN}$, whose weak limit we denote by 
$\Psi_\infty$. We shall show that $\|\Psi_\infty\|=1$, which will imply that $\{\Psi_{\ell_j}\}_{j\in\NN}$
is actually norm convergent.

To this end, let $m_0\in\NN$ and $\delta\in(0,1]$. By virtue of Lem.~\ref{lem-cpt1} we may
successively, in a finite number of steps,
select subsequences to find $\kappa_\ell\in\NN$, $\ell\in\NN$,
with $\kappa_\ell\to\infty$, $\ell\to\infty$, such that every sequence
$\{\Psi^{(m,\ul{\lambda})}_{\delta,\kappa_\ell}\}_{\ell\in\NN}$ with $m\in\{0,\ldots,m_0\}$ and
$\ul{\lambda}\in\{0,1\}^m$ converges strongly to its weak limit 
$(\Gamma(\chi_\delta)\Psi_\infty)^{(m,\ul{\lambda})}$. 
Let $p_{m_0}$ be the orthogonal projection onto the first $m_0+1$ direct summands is
$\HR=\bigoplus_{m=0}^\infty L^2(\RR^3,\sF^{(m)})$. Then
\begin{align}\nonumber
\|\Psi_\infty\|&\ge
\big\|p_{m_0}\Gamma(\chi_\delta)\Psi_\infty\big\|
=\lim_{\ell\to\infty}\|p_{m_0}\Gamma(\chi_\delta)\Psi_{\kappa_\ell}\|
\\\nonumber
&\ge\lim_{\ell\to\infty}\|p_{m_0}\Psi_{\kappa_\ell}\|
-\sup_{n}\SPb{\Psi_{n}}{(1-\Gamma(\chi_\delta^2))\Psi_{n}}^\eh
\\\label{pat1}
&\ge1-\sup_{n\in\NN}\|(\id-p_{m_0})\Psi_{n}\|
-\sup_{n}\big\|\Id\Gamma(\ol{\chi}_\delta^2)^\eh\Psi_{n}\big\|.
\end{align}
Here Cor.~\ref{cor-IRa} implies
$\sup_{n}\|(\id-p_{m_0})\Psi_{n}\|
\le m_0^{\mh}\sup_{n}\|\Id\Gamma(1)^\eh\Psi_n\|\le cm_0^{\mh}$, while Cor.~\ref{cor-IR} entails
\begin{align*}
\big\|\Id\Gamma(\ol{\chi}_\delta^2)^\eh\Psi_{n}\big\|^2
&=\sum_{\lambda=0}^1\int_{\RR^3}\ol{\chi}_\delta^2(\V{k})
\big\|a(\V{k},\lambda)\Psi_n\|^2\Id\V{k}
\le F(0,2\delta)\xrightarrow{\;\;\delta\downarrow0\;\;}0,
\end{align*}
where $c$ and $F(0,2\delta)$ are $n$-independent.
Since $m_0\in\NN$ was arbitrary large and $\delta\in(0,1]$ arbitrary small in \eqref{pat1}, 
we conclude that $\|\Psi_\infty\|=1$.
\end{proof}


\appendix


\section{Multiple commutator estimates}\label{app-comm}

\noindent
In this appendix we derive some norm bounds on commutators
between creation and annihilation operators and functions of
second quantized multiplication operators in the boson Hilbert
space. They are applied in Sect.~\ref{sec-weights}
to estimate the norms of the operators in \eqref{eva1}--\eqref{eva4} and in particular of
\begin{align}\nonumber
2T_1(s)
&=\vt^\mh\Theta_s^{-1}(\Ad_{\vp(\V{G}_{\V{B}^{\V{q}}_s})}^2\Theta_s^2)\Theta_s^{-1}\vt^\mh
\\\nonumber
&=2\{\vt^\mh\Theta_s^{-1}\Ad_{\vp(\V{G}_{\V{B}^{\V{q}}_s})}\Theta_s\}
\{(\Ad_{\vp(\V{G}_{\V{B}^{\V{q}}_s})}\Theta_s)\Theta_s^{-1}\vt^\mh\}
\\\label{for-T1c}
&\quad
+\vt^\mh\Theta_s^{-1}\Ad_{\vp(\V{G}_{\V{B}^{\V{q}}_s})}^2\Theta_s\vt^\mh
+\vt^\mh(\Ad_{\vp(\V{G}_{\V{B}^{\V{q}}_s})}^2\Theta_s)\Theta_s^{-1}\vt^\mh,
\end{align}
where $\Ad_ST:=[S,T]$ and
where we used the product rule $\Ad_S(TT')=T\Ad_ST'+(\Ad_ST)T'$.
The lower the power of $\Theta_s$, the weaker the conditions imposed on
the coefficient vector $\V{c}$ in our bounds below, whence we wrote commutators with
$\Theta_s^2$ as combinations of commutators involving only $\Theta_s$. Likewise,
\begin{align}\nonumber
T_2(s)&=-i(\Ad_{\vp(q_{\V{B}^{\V{q}}_s})}\Theta_s)\Theta_s^{-1}\vt^\mh
\\\label{for-T2}
&\quad
-(\Ad_{\vsigma\cdot\vp(\V{F}_{\V{B}^{\V{q}}_s})}\Theta_s)\Theta_s^{-1}\vt^\mh
+\Theta_s^{-1}(\Ad_{\vsigma\cdot\vp(\V{F}_{\V{B}^{\V{q}}_s})}\Theta_s)\vt^\mh,
\\\label{for-vecT}
\V{T}(s)&=-2(\Ad_{\vp(\V{G}_{\V{B}^{\V{q}}_s})}\Theta_s)\Theta_s^{-1}\vt^\mh.
\end{align}
Let us also note for later reference that, for every $g\in\HP$,
\begin{align}\label{comm-inv1}
(\Ad_{\vp(g)}\Theta_s^{\mp1})\Theta_s^{\pm1}&=-\Theta_s^{\mp1}\Ad_{\vp(g)}\Theta_s^{\pm1},
\\\label{comm-inv2}
(\Ad_{\vp(g)}^2\Theta_s^{\mp1})\Theta_s^{\pm1}&=-\Theta_s^{\mp1}\Ad_{\vp(g)}^2\Theta_s^{\pm1}
+2(\Theta_s^{-1}\Ad_{\vp(g)}\Theta_s)^2.
\end{align}
In view of \eqref{def-vp-vo} all
simple and double commutators appearing above can be written as linear
combinations of terms of the form
\begin{equation*}
(1+\Id\Gamma(\omega))^{-\nf{n}{2}}
\Theta_s^{-\beta+\kappa}
\Big\{\Big(\prod_{j=1}^N\Ad_{\ad(f_j)}\Big)\Big(\prod_{\ell=1}^M\Ad_{a(g_\ell)}\Big)\,\Theta_s\Big\}
\Theta_s^{-\gamma-\kappa}(1+\Id\Gamma(\omega))^{-\nf{m}{2}}.
\end{equation*}
with $\kappa=0$, $\beta+\gamma=1$, $M+N\le2$, and $m+n\le M+N$.
Here we used that, on account of \eqref{CCR} and the Jacobi identity,
the order of the $M+N$ commutations with the creation or annihilation
operators is immaterial. In all cases of our present interest $\Theta_s$ is some function
of a second quantized multiplication operator. Since we need bounds on 
simple and double commutators involving all combinations of
the creation and annihilation operators and different functions of various second quantized
operators, a systematic treatment seems to be in order.
We shall consider the general case of multiple commutators
with arbitrary $M$ and $N$ right away. The resulting bounds might also be useful elsewhere.

Let $\sL$ be some finite index set, $M:=\#\sL$, and let $p_{\sL}:=(p_\ell)_{\ell\in\sL}\in\cM^\sL$.
For $\psi\in\dom(\Id\Gamma(1)^{M/2})$, or rather a representative of it, and for $n\in\NN_0$, 
we write
$$
(a(p_{\sL})\,\psi)^{(n)}(k_{[n]})
:=(n+M)^\eh\dots(n+1)^\eh\,\phi^{(n+M)}(k_{[n]},p_{\sL}),
\quad k_{[n]}\in\cM^{n}.
$$
For almost every $p_{\sL}$, this defines a new element $a(p_{\sL})\,\psi$ of $\sF$.
Notice that the order of the variables $p_\ell$ is immaterial
because of the permutation symmetry of $\phi^{(n+M)}$ and that
$a(p_{\sL})=\prod_{\ell\in\sL}a(p_\ell)$; compare \eqref{def-a(k)}.
The identity
\begin{align}\label{a(P)ad(f)}
a(p_{\sL})\,\ad(f)\,\psi=\ad(f)\,a(p_{\sL})\,\psi+\sum_{\ell\in\sL}f(p_\ell)\,a(p_{\sL\setminus\{\ell\}})\,\psi
\end{align}
holds almost everywhere, if $\psi\in\dom(\Id\Gamma(1)^{(M+1)/2})$
and $f\in\HP$.
If $\V{v}:\cM\to\RR^L$ is a vector of multiplication
operators and if $F:\RR^L\to\RR$ is measurable, then both sides of
the following {\em pull-through formula},
\begin{align}\label{pull-through-mult}
a(p_{\sL})\,F(\Id\Gamma(\V{v}))\,\psi
&=F\Big(\Id\Gamma(\V{v})+\sum_{\ell\in\sL}\V{v}(p_\ell)\Big)\,a(p_{\sL})\,\psi,
\end{align}
define the same elements of $\sF$, for almost every $p_{\sL}$, provided that
$\psi$ belongs to the domain of $\Id\Gamma(1)^{M/2}F(\Id\Gamma(\V{v}))$.
This is a tautological consequence of the definitions.
Together with \eqref{def-a(f)} an $M$-fold repeated application of
the pull-through formula for each $a(p_\ell)$ gives a handy formula
for multiple commutators with annihilation operators of boson
wave functions $g_\ell\in\HP$, $\ell\in\sL$,
\begin{align}\nonumber
\SPB{\phi}{&\Big(\prod_{\ell\in\sL}\Ad_{a(g_\ell)}\Big)F(\Id\Gamma(\V{v}))\,\psi}
\\\label{hansi}
&=\int\Big(\prod_{\ell\in\sL}\ol{g(p_\ell)}\Big)
\SPb{\phi}{\triangle_{p_{\sL}}F(\Id\Gamma(\V{v}))\,a(p_{\sL})\,\psi}\,\Id\mu^M(p_{\sL}),
\end{align}
for $\psi$ as above and all $\phi\in\sF$, where, for any $\wt{F}:\RR^L\to\RR$, 
\begin{align*}
\triangle_{p_\ell}\wt{F}(\Id\Gamma(\V{v})):=
\wt{F}(\Id\Gamma(\V{v})+\V{v}(p_\ell))-\wt{F}(\Id\Gamma(\V{v})),
\qquad\triangle_{p_{\sL}}:=\prod_{\ell\in\sL}\triangle_{p_\ell}.
\end{align*}
Let $\sJ$ be another finite index set disjoint from $\sL$, $N:=\#\sJ$,
and pick additional wave functions, $g_j\in\HP$, $j\in\sJ$.
Applying \eqref{hansi} once with $\phi$ replaced by
$a(g_j)\,\phi$ 
and another time with $\psi$ substituted by $\ad(g_j)\,\psi$
using \eqref{a(P)ad(f)}, and substracting the results we find, 
for all $\phi,\psi\in\dom\big(\Id\Gamma(1)^{\frac{M+1}{2}}F(\Id\Gamma(\V{v}))\big)$,
\begin{align*}
\SPB{&\phi}{\Ad_{\ad(g_j)}\Big(\prod_{\ell\in\sL}\Ad_{a(g_\ell)}\Big)\,F(\Id\Gamma(\V{v}))\,\psi}
\\
&=
-\int\Big(\prod_{\ell\in\sL}\ol{g(p_\ell)}\Big)
\,\SPb{\Ad_{a(g_j)}\triangle_{p_{\sL}}F(\Id\Gamma(\V{v}))\,\phi}{a(p_{\sL})\,\psi}\Id\mu^M(p_{\sL})
\\
&\quad-\sum_{\ell'\in\sL}
\int\Big(\prod_{\ell\in\sL}\ol{g(p_\ell)}\Big)\,g_j(p_{\ell'})\,
\SPb{\phi}{\triangle_{p_{\sL}}F(\Id\Gamma(\V{v}))\,a(p_{\sL\setminus\{\ell'\}})\,\psi}
\Id\mu^M(p_{\sL})
\\
&=-\int\Big(\prod_{\ell\in\sL}\ol{g(p_\ell)}\Big)
\,g_j(p_j)\,\SPb{a(p_j)\,\phi}{\triangle_{p_{\{j\}\cup\sL}}F(\Id\Gamma(\V{v}))\,a(p_{\sL})\,\psi}
\Id\mu^{M+1}(p_{\{j\}\cup\sL})
\\
&\quad-\sum_{\ell'\in\sL}
\int\Big(\prod_{\ell\in\sL}\ol{g(p_\ell)}\Big)\,g_j(p_{\ell'})\,
\SPb{\phi}{\triangle_{p_{\sL}}F(\Id\Gamma(\V{v}))\,a(p_{\sL\setminus\{\ell'\}})\,\psi}
\Id\mu^M(p_{\sL}).
\end{align*}
Repeating this procedure using $a(p_i)\,a(g_j)=a(g_j)\,a(p_i)$ and applying
\eqref{pull-through-mult} to $F_1$ and $F_3$ in \eqref{clelia0} below we obtain the following result:

\begin{lem}\label{lem-clelia}
For $\iota=1,2,3$, let $F_\iota:\RR^{L_\iota}\to\RR$ be measurable and  let
$\V{v}_\iota:\cM\to\RR^{L_\iota}$ be a vector of 
multiplication operators. With the notation explained above we have, for all 
$\phi,\psi\in\dom\big(\Id\Gamma(1)^{\frac{M+N}{2}}\prod_{\iota=1}^3
F_\iota(\Id\Gamma(\V{v}_\iota))\big)$,
\begin{align}\nonumber
\SPB{F_1&(\Id\Gamma(\V{v}_1))\,\phi}{\Big\{\Big(\prod_{j\in\sJ}\Ad_{\ad(g_j)}\Big)
\Big(\prod_{\ell\in\sL}\Ad_{a(g_\ell)}\Big)\,F_2(\Id\Gamma(\V{v}_2))\Big\}
\,F_3(\Id\Gamma(\V{v}_3))\,\psi}
\\\label{clelia0}
&=(-1)^N\!\!\!\sum_{{{\sA\cup\sB=\sJ\atop\sC\cup\sD=\sL}\atop\#\sB=\#\sD}}\int
\SPb{a(p_\sA)\phi}{M_{\sA,\sC}(p_{\sA\cup\sL})\,a(p_\sC)\psi}
\Id\mu^{\#(\sA\cup\sL)}(p_{\sA\cup\sL}),
\end{align}
with the following family of multiplication operators, parametrized by 
$\sA$, $\sC$, $p_{\sA\cup\sL}$,
and the states $g_\ell\in\HP$ (which are dropped in the notation), 
\begin{align}\nonumber
M_{\sA,\sC}(&p_{\sA},p_{\sC},p_{\sD})
:=\sum_{\pi\in\mathrm{Bij}(\sB,\sD)}\Big(\prod_{\ell\in\sL}\ol{g_\ell(p_\ell)}\Big)
\Big(\prod_{b\in\sB}g_{b}(p_{\pi(b)})\Big)
\Big(\prod_{a\in\sA}g_{a}(p_a)\Big)
\times
\\\label{def-MAB}
&\times F_1\Big(\Id\Gamma\big(\V{v}_1+\!\!{\textstyle{\sum\limits_{a\in
\sA}\V{v}_1(p_a)}}\big)\Big)
\{\triangle_{p_{\sA\cup\sL}}F_2(\Id\Gamma(\V{v}_2))\}F_3\Big(\Id\Gamma\big(\V{v}_3
+\!\!{\textstyle{\sum\limits_{c\in
\sC}\V{v}_3(p_c)}}\big)\Big).
\end{align}
Here $\mathrm{Bij}(\sB,\sD)$ denotes the set of bijections of $\sB$
onto $\sD$.
\end{lem}

The next lemma reduces the problem to find a bound on the expression \eqref{clelia0}
to the estimation of real functions of several variables.

\begin{lem}\label{lem-a(K)}
For fixed $p_\sB$, let $\sM(p_{\sA},p_{\sC}):=M_{\sA,\sC}(p_{\sA},p_{\sB},p_{\sC})$ 
denote one of the multiplication operator-valued functions in
\eqref{def-MAB} (or any other reasonable family of operators
parametrized by $p_{\sA}$ and $p_{\sC}$).
Let $\theta:=1+\Id\Gamma(\omega)$, where we suppose that
the measurable function $\omega:\cM\to\RR$
is strictly positive almost everywhere. Then, if $m,n\in\NN_0$ with $n\le\#\sA$ and $m\le\#\sC$,
\begin{align*}
&\int\big|\SPb{a(p_{\sA})\,\theta^{-\nf{n}{2}}\,\phi}{
\sM(p_{\sA},p_\sC)\,a(p_\sC)\,\theta^{-\nf{m}{2}}\,\psi}\big|\,
\frac{\Id\mu^{\#(\sA\cup\sC)}(p_{\sA\cup\sC})}{((\#\sA-n)!(\#\sC-m)!)^\eh}
\\
&\le \|\phi\|\|\psi\|\!
\sum_{{\fa\subset \sA\atop\#\fa\ge n}}\sum_{{\fc\subset\sC\atop\#\fc\ge m}}\!
\bigg[\int\frac{\|\theta^{(\#\fa-n)/2}\sM(p_{\sA},p_{\sC})\theta^{(\#\fc-m)/2}\|^2}{
(\prod_{a\in\fa}\omega(p_a))\prod_{c\in\fc}\omega(p_c)}\Id\mu^{\#(\sA\cup\sC)}(
p_{\sA\cup\sC})\bigg]^{\eh}\!\!.
\end{align*}
If $n>\#\sA$, then the sum over $\fa$ reduces to only one summand
with $\fa=\sA$, $\theta^{(\#\sA-\#\fa)/2}$ is replaced
by $(\theta+\sum_{a\in\sA}\omega(p_a))^{-(n-\#\sA)/2}$, and
$(\#\sA-n)!$ is replaced by $1$. Analogous replacements
occur in the case $m>\#\sC$. (The latter remarks apply in particular when
$n>\#\sA$ {\em and} $m>\#\sC$.)
\end{lem}

\begin{proof}
Let us first assume that $n\le\#\sA$ and $m\le\#\sC$.
For notational simplicity we may also assume that $\sA=[N]:=\{1,\dots,N\}$ and $\sC=[M]$.
Moreover, we shall use the letter $k$ to denote variables labeled by $\sA=[N]$. 
We choose a partition $\sM^N=\bigcup_{\pi\in\mathrm{Bij}([N],[N])}D_{\pi}{(N)}$
into disjoint measurable sets $D_\pi(N)$ such that 
$$
k_{[N]}\in D_{\pi}{(N)}\quad\Rightarrow\quad
\omega(k_{\pi(1)})\ge\omega(k_{\pi(2)})\ge\dots\ge\omega(k_{\pi(N)}),
$$
for all $\pi\in\mathrm{Bij}([N],[N])$.
We choose a partition $\sM^M=\bigcup_{\pi'\in\mathrm{Bij}([M],[M])}D_{\pi'}{(M)}$ 
with the analogous property. We now fix permutations 
$\pi\in\mathrm{Bij}([N],[N])$ and $\pi'\in\mathrm{Bij}([M],[M])$, write
$D_{\pi,\pi'}:=D_\pi(N)\times D_{\pi'}(M)$, and consider
\begin{align*}
I_{\pi,\pi'}&:=\int_{D_{\pi,\pi'}}\!\big|\SPb{a(k_{[N]})\theta^{-\nf{n}{2}}\phi}{
\sM(k_{[N]},p_{[M]})a(p_{[M]})\theta^{-\nf{m}{2}}\,\psi}\big|
\Id\mu^{M+N}(k_{[N]},p_{[M]}).
\end{align*}
We use the less space consuming notation $k_j^\pi:=k_{\pi(j)}$.
By the pull-through formula,
\begin{align*}
&a(k_1)\dots a(k_N)\theta^{-\nf{n}{2}}=a(k_N^\pi)\dots a(k_1^\pi)\theta^{-\nf{n}{2}}
\\
&=(\theta+\omega(k_{n+1}^\pi)+\dots+\omega(k_N^\pi))^\eh\dots(\theta+\omega(k_N^\pi))^\eh
a(k_N^\pi)\theta^\mh\dots
a(k_{n+1}^\pi)\theta^\mh\times
\\
&\quad\times a(k_n^\pi)(\theta+\omega(k_1^\pi)+\dots+\omega(k_{n-1}^\pi))^\mh a(k_{n-1}^\pi)
\dots(\theta+\omega(k_1^\pi))^\mh a(k_1^\pi)\theta^\mh,
\end{align*}
where the second line has to be ignored when $n=N$. If $N>n$, then
the multiplication operator in the second line can be bounded, pointwise on $D_\pi(N)$, as
\begin{align*}
(\theta&+\omega(k_{n+1}^\pi)+\dots+\omega(k_N^\pi))^\eh\dots(\theta+\omega(k_N^\pi))^\eh
\\
&\le((N-n)!)^\eh(\theta^\eh+\omega(k_{n+1}^\pi)^\eh)\dots(\theta^\eh+\omega(k_{N}^\pi)^\eh)
\\
&=
((N-n)!)^\eh\sum_{{\fb\subset[N]\setminus[n]}}\theta^{\#\fb/2}\prod_{j\in([N]\setminus[n])\setminus
\fb}\omega(k_j^\pi)^\eh,
\end{align*}
where we repeatedly used $({a+b})^\eh\le{a}^\eh+{b}^\eh$
and $\omega(k_N^\pi)\le\dots\le\omega(k_{n+1}^\pi)$ on $D_\pi(N)$. 
An analogous estimate certainly holds for the variables $p_{[M]}$ and the
permutation $\pi'$. Applying the Cauchy-Schwarz inequality to each summand thus
contributing to $I_{\pi,\pi'}$ we now see that
\begin{align*}
&\sum_{\pi,\pi'}I_{\pi,\pi'}\le \big((N-n)!(M-m)!\big)^\eh J_{n,N}(\phi)^\eh J_{m,M}(\psi)^\eh
\\
&\;\;\;\cdot\sum_{{\fb\subset[N]\setminus[n]\atop \fd\subset[M]\setminus[m]}}
\bigg[\sum_{\pi,\pi'}\int_{D_{\pi,\pi'}}\frac{\|\theta^{\#\fb/2}\sM(k_{[N]},p_{[M]})
\theta^{\#\fd/2}\|^2}{(\prod_{b\in\fb\cup[n]}\omega(k_b^\pi))
\prod_{d\in\fd\cup[m]}\omega(p_d^{\pi'})}\Id\mu^{M+N}(k_{[N]},p_{[M]})\bigg]^\eh,
\end{align*}
with
\begin{align*}
&J_{n,N}(\phi):=
\int\omega(k_1)\dots\int\omega(k_N)\big\|a(k_N)\,\theta^\mh\dots
a(k_{n+1})\theta^\mh\times
\\
&\times a(k_n)(\theta+\omega(k_1)+\dots+\omega(k_{n-1}))^\mh
\dots a(k_1)\theta^\mh\phi\big\|^2\Id\mu(k_N)\dots\Id\mu(k_1),
\end{align*}
and a similar definition for $J_{m,M}(\psi)$.
Applying the following familiar consequence of Fubini's theorem and the
permutation symmetry of $\eta\in\dom(\Id\Gamma(\omega)^\eh)$ repeatedly, 
\begin{equation*}
\int\omega(k)\|a(k)\eta\|^2\Id\mu(k)=\|\Id\Gamma(\omega)^\eh\eta\|^2
\le\|(\theta+t)^\eh\eta\|^2,\quad t\ge0,
\end{equation*}
we see that $J_{n,N}(\phi)\le\|\phi\|^2$ and, analogously, $J_{m,M}(\psi)\le\|\psi\|^2$.

If $n>N$, then the second line in the formula for
$a(k_N)\dots a(k_1)\,\theta^{-\nf{n}{2}}$ on $D_\pi(N)$ above is replaced
by $(\theta+\omega(k_1^\pi)+\dots+\omega(k_N^\pi))^{-(n-N)/2}$ and
$n$ is substituted by $N$ in the third line.
It is then clear how to conclude in this case.
\end{proof}

The multiplication operators we have to bound according to the
previous lemma involve the higher order difference operations
appearing in the formula \eqref{def-MAB}. To deal with such
expressions we introduce some more notation.

For any real function $f:\RR\to\RR$, we set
$$
\triangle_s f:=(\tau_{s}-\id)f,\qquad
(\tau_sf)(t):=f(t+s),\qquad s,t\in\RR,
$$
and we shall write 
$$
\triangle_{s_{\sL}}:=\prod_{\ell\in\sL}\triangle_{s_{\ell}},
\quad\tau_{s_\sL}:=\prod_{\ell\in\sL}\tau_{s_{\ell}},
\qquad s_{\sL}=(s_\ell)_{\ell\in\sL}\in\RR^\sL.
$$
For $f_1,f_2:\RR\to\RR$, we notice that the product rule,
\begin{equation}\label{Leibniz1}
\triangle_s(f_1\,f_2)=(\triangle_sf_1)\,\tau_s f_2+f_1\,(\triangle_sf_2),
\end{equation}
and the translation invariance of $\triangle_s$ entail
\begin{equation}\label{Leibniz2}
\triangle_{s_{\sL}}(f_1\,f_2)=\sum_{\sA\cup\sB=\sL}(\triangle_{s_\sA}f_1)
\,\tau_{s_\sA}\triangle_{s_\sB}f_2,
\end{equation}
where $\sA$ and $\sB$ are always disjoint and possibly empty.
Furthermore, we readily verify by induction that
\begin{equation}\label{Leibniz-pol}
\triangle_{s_{\sL}}t^n=\!\!\sum_{{\vk_\sL\in\NN^\sL\atop|\vk_\sL|\le n}}{n\choose\vk_\sL}
t^{n-|\vk_\sL|}\prod_{\ell\in\sL}s_{\ell}^{\vk_\ell},
\quad {n\choose\vk_\sL}:=\frac{n!}{(n-|\vk_\sL|)!\prod_{\ell\in\sL}\vk_\ell!}.
\end{equation}
where empty sums are zero and empty products are one.
Notice that $\vk_\ell\ge1$, for each $\ell\in\sL$, in \eqref{Leibniz-pol}.

One is inclined to estimate higher order difference quotients 
by repeatedly applying the mean value theorem and using $L^\infty$-bounds on derivatives. 
A too naive estimation of the derivatives would, however, not be
sufficient for the function $t^\gamma$ to find the canonical
conditions on the photon states $g_j$ in Lem.~\ref{lem-comm-omega}
below, whence we shall argue more carefully in the next lemma.

\begin{lem}\label{lem-Leibniz-t}
Let $a,\alpha>0$,  $\ve\in[0,1]$, $s_0:=0$, $s_\ell>0$ and $\delta_\ell\in[0,1]$, for $\ell\in\sL$, 
and write $|\delta_{\sL}|:=\sum_{\ell\in\sL}\delta_\ell$,
$s_\sL^{\delta_\sL}:=\prod_{\ell\in\sL}s_\ell^{\delta_\ell}$, etc., and
\begin{equation}\label{def-Feps}
F_\ve(t):=t/(1+\ve t),\qquad t>0.
\end{equation}
Then the following bound holds, for all $t>0$,
\begin{align}\label{Leibniz7}
|\triangle_{s_{\sL}}F_\ve^{-\alpha}(t)|
&\le c(\alpha,\#\sL,\delta_\sL)\,F_\ve^{-\alpha}(t)\,t^{-|\delta_{\sL}|}\,s_\sL^{\delta_\sL},
\\\label{Leibniz77}
|\triangle_{s_{\sL}}F_\ve^{\alpha}(t)|
&\le c'(\alpha,\#\sL,\delta_{\sL})
\sum_{\ell\in\{0\}\cup\sL} F_\ve^\alpha(t+s_\ell)
(t+s_\ell)^{-|\delta_{\sL}|}\,s_\sL^{\delta_\sL},
\\\label{Leibniz-exp}
|\triangle_{s_{\sL}}e^{a F_\ve(t)}|
&\le c(\#\sL)\,a^{|\delta_{\sL}|}e^{aF_\ve(t)}e^{a|s_\sL|}\,s_\sL^{\delta_\sL},
\quad\text{if}\;\ve\le a.
\end{align}
\end{lem}

\begin{proof}
First, we consider \eqref{Leibniz7} and \eqref{Leibniz77} in the case $\ve=0$. In view of
\begin{equation}\label{triangle-exp}
\triangle_se^{-r t}=e^{-r t}(e^{-r s}-1)
\end{equation}
it is convenient to represent $t^{-\alpha}$ as a superposition of exponentials,
\begin{align}
t^{-\alpha}&=\Gamma(\alpha)^{-1}\int_0^\infty
e^{-rt}\,\frac{dr}{r^{1-\alpha}}.
\end{align}
Together with \eqref{triangle-exp} the above formula yields
\begin{align}\label{Leibniz4}
\triangle_{s_{\sL}}t^{-\alpha}
&=\Gamma(\alpha)^{-1}
\int_0^\infty e^{-rt}\Big(\prod_{\ell\in\sL}(e^{-s_\ell r}-1)\Big)\,\frac{dr}{r^{1-\alpha}}.
\end{align}
Estimating first $1-e^{-s_\ell r}\le(s_\ell r)^{\delta_\ell }$ and substituting $r\to r/t$
thereafter we arrive at \eqref{Leibniz7} with $\ve=0$.

To prove \eqref{Leibniz77} with $\ve=0$ we pick some some $n\in\NN$ such that $\alpha<n$
and write $t^\alpha=t^nt^{-(n-\alpha)}$. 
Together with \eqref{Leibniz2} and \eqref{Leibniz-pol} this yields
\begin{align*}
\triangle_{s_{\sL}}t^{\alpha}
&=\sum_{\sA\cup \sB=\sL}
\sum_{{\vk_\sA\in\NN^\sA\atop|\vk_\sA|\le n}}{n\choose\vk_\sA}
t^{n-|\vk_\sA|}s_\sA^{\vk_\sA}\,\tau_{s_{\sA}}\triangle_{s_{\sB}} t^{-(n-\alpha)}.
\end{align*}
Applying \eqref{Leibniz7} with $\ve=0$ and $\alpha$ replaced by $n-\alpha$ we see that
\begin{align*}
t^{n-|\vk_\sA|}\,s_\sA^{\vk_\sA}|\tau_{s_{\sA}}\triangle_{s_{\sB}}t^{-(n-\alpha)}|
&\le
\const(n,\alpha,\delta_{\sB})\,(t+|s_\sA|)^{\alpha}\,
\frac{t^{n-|\vk_\sA|}\,s_\sA^{\vk_\sA}\,s_\sB^{\delta_\sB}}{(t+|s_\sA|)^{n+|\delta_\sB|}}
\\
&\le
\const(n,\alpha,\delta_{\sB})\,(t+|s_\sA|)^{\alpha-|\delta_{\sL}|}\,s_\sL^{\delta_\sL}
\\
&\le\const(n,\alpha,\delta_{\sB})\,(\#\sA)^{(\alpha-|\delta_{\sL}|)\vee0}
\sum_{a\in\sA}(t+s_{a})^{\alpha-|\delta_{\sL}|}\,s_\sL^{\delta_\sL}.
\end{align*}
In the second step we also used that $\vk_a\ge1$, for each component of $\vk_\sA=(\vk_a)_{a\in\sA}$,
that $t+|s_\sA|$ is bigger than $t$ and each $s_{a}$, 
and that $\sA\cup\sB=\sL$, which implies
$$
\frac{t^{n-|\vk_\sA|}\,s_\sA^{\vk_\sA}\,s_\sB^{\delta_\sB}}{(t+|s_\sA|)^{n+|\delta_\sB|}}
\le\frac{s_\sA^{\vk_\sA-\delta_{\sA}}}{(t+|s_\sA|)^{|\vk_\sA-\delta_\sA|}}\,
\frac{s_\sL^{\delta_\sL}}{(t+|s_\sA|)^{|\delta_\sL|}}\le\frac{s_\sL^{\delta_\sL}}{(t+|s_\sA|)^{|\delta_\sL|}}.
$$
This proves \eqref{Leibniz77}, for $\ve=0$.

Now, let $\tilde{\ve}>0$. If we replace $t$ by $1+\tilde{\ve}\,t$ in
\eqref{Leibniz7}
and \eqref{Leibniz77} with $\ve=0$, then each $s_\ell$ has to be replaced
by $\tilde{\ve}\,s_\ell$. Combining \eqref{Leibniz7}\&\eqref{Leibniz77}
with $\ve=0$ and the product rule \eqref{Leibniz2} we thus obtain
\begin{align*}
|\triangle_{s_{\sL}}t^\alpha(1+\tilde{\ve}t)^{-\alpha}|
&\le\const(\alpha,\#\sL,\delta_{\sL})\sum_{{\sA\cup\sB=\sL\atop\sA\not=\varnothing}}
\sum_{a\in\sA}
\frac{\tilde{\ve}^{|\delta_\sB|}\,(t+s_{a})^{\alpha-|\delta_\sA|}}{
(1+\tilde{\ve}\,(t+|s_{\sA}|))^{\alpha+|\delta_\sB|}}\,
s_\sL^{\delta_\sL}
\\
&\quad+\const(\alpha,\#\sL,\delta_{\sL})\,
\tilde{\ve}^{|\delta_{\sL}|}\,t^\alpha\,(1+\tilde{\ve}\,t)^{-\alpha-|\delta_{\sL}|}\,s_\sL^{\delta_\sL}.
\end{align*}
Since $\tilde{\ve}/(1+\tilde{\ve}\,(t+|s_{\sA}|))\le1/(t+s_{a})$,
for every $a\in\sA$, this yields \eqref{Leibniz77}
with $\ve=\tilde{\ve}$. We may analogously extend \eqref{Leibniz7} to
positive $\ve$ using that
$(1+\ve t+\ve s_{a})/(t+|s_\sA|)\le(1+\ve t)/t$ in addition.

To prove \eqref{Leibniz-exp} we first recall Fa\`{a} di Bruno's formula,
$$
(g\circ f)^{(k)}(t)=\sum_{\ell=1}^k\frac{g^{(\ell)}(f(t))}{\ell!}\sum_{\alpha\in\NN^\ell: |\alpha|=k}
\frac{k!}{\alpha_1!\dots\alpha_\ell!}\,f^{(\alpha_1)}(t)\dots f^{(\alpha_\ell)}(t),
$$
for $k$-times differentiable functions $f,g:\RR\to\RR$,
which implies
\begin{align*}
\frac{\Id^k}{\Id t^k}e^{aF_\ve(t)}=e^{aF_\ve(t)}\sum_{\ell=1}^kc_{k,\ell}\,
\frac{a^\ell(-\ve)^{k-\ell}}{(1+\ve t)^{k+\ell}},
\qquad c_{k,\ell}:=\frac{k!}{\ell!}\sum_{\alpha\in\NN^\ell: |\alpha|=k}1.
\end{align*}
whence (with $[k]=\{1,\ldots,k\}$, $|u_{[k]}|=u_1+\dots+u_k$)
\begin{align*}
\triangle_{s_{[k]}}e^{aF_\ve(t)}
&=\sum_{\ell=1}^kc_{k,\ell}\int_0^{s_k}\dots\int_{0}^{s_1}e^{aF_\ve(t+|u_{[k]}|)}
\frac{a^\ell(-\ve)^{k-\ell}}{(1+\ve t+\ve|u_{[k]}|)^{k+\ell}}\Id u_1\dots \Id u_k.
\end{align*}
Using $F_\ve(t+r)\le F_\ve(t)+r$, $r,t\ge0$, repeatedly and taking the condition $\ve\le a$
into account we obtain
\begin{align*}
|\triangle_{s_{[k]}}e^{aF_\ve(t)}|
&\le c_k e^{aF_\ve(t)}\prod_{\ell=1}^ke^{as_\ell}\int_0^{s_\ell}ae^{a(u_\ell-s_\ell)}\Id u_\ell
\le c_ke^{aF_\ve(t)}\prod_{\ell=1}^ke^{|a|s_\ell}(as_\ell)^{\delta_\ell}.
\end{align*}
\end{proof}

We are now prepared to derive bounds on the type of commutators appearing in our applications:

\begin{lem}\label{lem-comm-omega}
Let $\sJ$ and $\sL$ be disjoint finite, possibly empty index sets and set $N:=\#\sJ$, $M:=\#\sL$.
Let $m,n\in\NN_0$ with $n\le N$ and $m\le M$, $\alpha\ge1/2$, $\beta,\gamma,\sigma,\tau\ge0$
with $\alpha=\beta+\gamma$ and $\sigma+\tau\le(M+N)/2$, and let $\kappa\in\RR$.
For $\ve\in[0,1]$ and $E\ge1$, define 
\begin{equation}\label{def-FveE}
F_{\ve,E}(t):=(E+t)/(1+\ve(E+t)),\quad t\ge0.
\end{equation}
Pick $g_\ell\in\HP$ satisfying $v^{\mu+\eh}g_\ell\in\HP$, for all $\ell\in\sJ\cup\sL$, with 
$$
\mu:=(|\kappa|+\sigma+\tau)\vee(\alpha+|\kappa|+\sigma+\tau-\tfrac{M+N+m+n}{2}),
$$
where the measurable function $v:\cM\to\RR$ is strictly positive almost
everywhere and monotonically increasing in $|\V{k}|$.
Finally, set $v_\ve:=v/(1+\ve v)$. Then the densely defined operator
\begin{align}\nonumber
T:=(1+\Id\Gamma(v))^{-\nf{n}{2}}&F_{\ve,E}^{\sigma-\beta+\kappa}(\Id\Gamma(v_\ve))
\Big\{\Big(\prod_{j\in\sJ}\Ad_{\ad(g_j)}\Big)\Big(\prod_{\ell\in\sL}\Ad_{a(g_\ell)}\Big)
F_{\ve,E}^\alpha(\Id\Gamma(v_\ve))\Big\}
\\\nonumber
&\qquad\times\,F_{\ve,E}^{\tau-\gamma-\kappa}(\Id\Gamma(v_\ve))(1+\Id\Gamma(v))^{-\nf{m}{2}},
\end{align}
extends uniquely to a bounded operator on $\sF$ with
\begin{align*}
\|T\|&\le c(\alpha,|\kappa|,\sigma,\tau,N,M)\,E^{-\frac{M+N+m+n}{2}+\sigma+\tau}
\\
&\quad\cdot\Big\{
\sum_{{a,b\in\sJ\cup\sL\atop
    a\not=b}}\big\|v^\eh(1+\tfrac{v}{E})^{\mu}g_a\big\|
\big\|v^\eh(1+\tfrac{v}{E})^{|\kappa|+\sigma+\tau}g_b\big\|\!\!\prod_{{c\in\sJ\cup\sL\atop
    c\not=a,b}}\!\|v^\eh g_c\|
 \\
 &\qquad+\sum_{a\in\sJ\cup\sL} 
 \big\|v^\eh(1+\tfrac{v}{E})^{\mu}g_a\big\|\!\!\prod_{{c\in\sJ\cup\sL\atop
    c\not=a}}\!\|v^\eh g_c\|\Big\}.
\end{align*}
Here empty sums should be read as $0$ and empty products are $1$.
\end{lem}

\begin{proof}
We abbreviate $s_\ell:=v_\ve(p_\ell)$, $\ell\in\sJ\cup\sL$, 
and $|s_\sA|:=\sum_{a\in\sA}s_a$, etc. According to \eqref{def-MAB} and 
Lem.~\ref{lem-a(K)} (which we apply with $\omega=v_\ve$ and $m=n=0$)
we have to find bounds on the norms
\begin{align*}
N_{\sA,\sC}^{\fa,\fc}&:=\big\|(1+\Id\Gamma(v_\ve))^{\frac{\#\fa-n+\#\fc-m}{2}}
\,F_{\ve,E}^{\sigma-\beta+\kappa}(\Id\Gamma(v_\ve)+|s_\sA|)
\\
&\qquad\quad\times 
(\triangle_{s_{\sA\cup\sL}}F_{\ve,E}^\alpha)(\Id\Gamma(v_\ve))
\,F_{\ve,E}^{\tau-\gamma-\kappa}(\Id\Gamma(v_\ve)+|s_\sC|)\big\|,
\end{align*}
where $\fa\subset\sA$, $\fc\subset\sC$ satisfy $\#\fa\ge n$, $\#\fc\ge m$, and $\sA\cup\sB=\sJ$,
$\sC\cup\sD=\sL$, are partitions with $\#\sB=\#\sD$.
We again set $s_0:=0$, assuming without loss of generality that
$0\notin\sJ\cup\sL$. Employing \eqref{Leibniz77} and taking
$\alpha=\beta+\gamma$ into account we obtain
\begin{align*}
N_{\sA,\sC}^{\fa,\fc}\le c_0\!\!\!\sum_{\ell\in\{0\}\cup\sA\cup\sL}
&\sup_{t\ge0}\Big(\frac{E+t+s_\ell}{E+t+|s_\sA|}\,\frac{1+\ve E+\ve
  t+\ve|s_\sA|}{1+\ve E+\ve t+\ve s_\ell}\Big)^\beta
\\
\cdot&\sup_{t\ge0}\Big(\frac{E+t+s_\ell}{E+t+|s_\sC|}\,\frac{1+\ve E+\ve
  t+\ve|s_\sC|}{1+\ve E+\ve t+\ve s_\ell}\Big)^\gamma
\\
\cdot&\sup_{t\ge0}\Big(\frac{E+t+|s_\sA|}{E+t+|s_\sC|}\,
\frac{1+\ve E+\ve t+\ve |s_\sC|}{1+\ve E+\ve t+\ve |s_\sA|}\Big)^\kappa
\\
\cdot&\sup_{t\ge0}\frac{(1+t)^{\frac{\#\fa-n+\#\fc-m}{2}}(E+t+|s_\sA|)^\sigma
(E+t+|s_\sC|)^\tau}{(E+t+s_\ell)^{|\delta_{\sA\cup\sL}|}}\,s_{\sA\cup\sL}^{\delta_{\sA\cup\sL}}
\\
\le c_1\!\!\!\sum_{\ell\in\{0\}\cup\sA\cup\sL}\!\!
\Big(\frac{E+s_\ell}{E}&\Big)^\alpha\Big(\frac{E+|s_{\sA\cup\sC}|}{E}\Big)^{|\kappa|}
\frac{(E+|s_{(\sA\cup\sC)\setminus\{\ell\}}|)^{\sigma+\tau}/E^{\sigma+\tau}}{
(E+s_\ell)^{|\delta_{\sA\cup\sL}|-\sigma-\tau-\frac{\#\fa-n+\#\fc-m}{2}}}\,
s_{\sA\cup\sL}^{\delta_{\sA\cup\sL}},
\end{align*}
where the parameters $\delta_a\in[0,1]$ may be chosen depending on
$\sA$, $\sC$, $\fa$, and $\fc$ and where we used that $\ve\,s_a\le1$ and
$$
\forall\,s,v\ge0\::\quad
\sup_{t\ge0}\frac{E+t+v}{E+t+s}=
\left\{
\begin{array}{ll}
(E+v-s)/E,&s< v,\\1,&s\ge v.
\end{array}
\right.
$$
Notice that the fractions containing the $\ve$'s in the previous
inequality are all uniformly bounded because of $\ve\,v_\ve\le1$.
We now choose $\delta_\ell=1$, if $\ell\in \fa\cup \fc\cup \sD$, and $\delta_\ell=1/2$,
if $\ell\in(\sA\setminus\fa)\cup(\sC\setminus\fc)$.
Then $|\delta_{\sA\cup\sL}|-(\#\fa+\#\fc)/2=(M+N)/2\ge\sigma+\tau$, 
since $\#\sB=\#\sD$, and using also
$$
E+|s_{\sA\cup\sC}|\le(E+s_\ell)\,\frac{E+|s_{(\sA\cup\sC)\setminus\{\ell\}}|}{E},
$$
we conclude that
\begin{align*}
\frac{(N_{\sA,\sC}^{\fa,\fc})^2}{(\prod_{a\in\fa\cup\fc}v(p_a))}
&\le c_2 \,\frac{\sS(p_{\sA\cup\sL})}{E^{M+N+m+n-2(\sigma+\tau)}}
\Big(\prod_{d\in\sD}v(p_d)^2\Big)\prod_{a\in \sA\cup \sC}v(p_a),
\end{align*}
where the $(\alpha,\sigma,\tau,\kappa,\sA\cup\sL,N,M,n,m)$-dependent function $\sS$ is given by
\begin{align*}
\sS(p_{\sA\cup\sL})
&:=\sum_{\ell\in\sA\cup\sL}\Big(\frac{E+v(p_\ell)}{E}\Big)^{2\mu}\!
\Big(1+\!\!\sum_{{j\in\sA\cup\sL\atop j\not=\ell}}\!\!
\Big(\frac{E+v(p_j)}{E}\Big)^{2(|\kappa|+\sigma+\tau)}\Big).
\end{align*}
Putting everything together we arrive at
\begin{align*}
&\|T\|
\\
&\le c_3\!\!\sum_{{{\sA\cup \sB=\sJ\atop\sC\cup\sD=\sL}\atop\#\sB=\#\sD}}
\sum_{\pi\in\mathrm{Bij}(\sB,\sD)}\int\Big(\prod_{d\in\sD}
v(p_d)^\eh|g_d(p_d)|\Big)\Big(\prod_{b\in\sB}v(p_{\pi(b)})^\eh|g_{b}(p_{\pi(b)})|\Big)
\\
&\cdot E^{-(M+N)/2+\sigma+\tau}
\bigg[\int\sS(p_{\sA\cup\sL})\!\!\prod_{a\in\sA\cup\sC}\!\big\{v(p_a)\,|g_{a}(p_a)|^2
\Id\mu(p_a)\big\}\bigg]^\eh\Id\mu^{\#\sD}(p_\sD),
\end{align*}
from which we easily infer the asserted estimate.
\end{proof}

\begin{rem}\label{rem-comm-omega}
If we additionally assume that $\sJ=\varnothing$ or $\sL=\varnothing$, 
then the whole statement of Lem.~\ref{lem-comm-omega} still holds true with $\mu$ replaced by 
$$
\tilde{\mu}:=(|\kappa|+\sigma+\tau)\vee(\beta\vee\gamma+|\kappa|+\sigma+\tau-\tfrac{M+N+m+n}{2}).
$$
In fact, in this case we always have $\sB=\sC=\varnothing$, i.e.
$\sA=\sJ$ and $\sC=\sL$ in the proof of Lem.~\ref{lem-comm-omega}. This implies that
$s_\ell\le|s_\sA|$ or $s_\ell\le|s_\sC|$, for all $\ell\in\{0\}\cup\sA\cup\sL$, so that $\alpha$
can be replaced by $\beta\vee\gamma$ in the estimate on $N_{\sA,\sC}^{\mathfrak{a},\mathfrak{c}}$.
Following the proof without any further changes we arrive at the assertion with
$\tilde{\mu}$ in place of $\mu$.
\end{rem}

\begin{ex}\label{ex-isi}
Let $f\in\HP$ with $\omega^\eh f\in\HP$ and set $\theta:=1+\Id\Gamma(\omega)$.
Choosing $\ve=0$, $E=1$, $n=m=0$, $M+N=1$, $\beta=\sigma=\tau=\kappa=0$, and
$\alpha=\gamma=1/2$ in Lem.~\ref{lem-comm-omega}, we see that 
$Y(f):=\big[\theta^\eh,\vp(f)\big]\theta^\mh$,
which is well-defined a priori on $\dom(\Id\Gamma(\omega)^\eh)$, extends to a bounded operator
on $\FHR$ with $\|Y(f)\|\le c\|\omega^\eh f\|$. This implies that the map
$
\RR^\nu\ni\V{x}\mapsto\vp(\V{G}_{\V{x}})^2\in\LO(\dom(\Id\Gamma(\omega)),\FHR)
$
is continuous. In fact, using \eqref{rb-a}\&\eqref{rb-ad} repeatedly, 
we obtain, for all $\V{x},\V{y}\in\RR^\nu$ and $\psi\in\dom(\Id\Gamma(\omega))$,
\begin{align*}
\big\|\vp(\V{G}_{\V{x}})^2\psi-\vp(\V{G}_{\V{y}})^2\psi\big\|
&=\big\|(\vp(\V{G}_{\V{x}})-\vp(\V{G}_{\V{y}}))\vp(\V{G}_{\V{x}}+\V{G}_{\V{y}})\psi\big\|
\\
&\le
c'\|\V{G}_{\V{x}}-\V{G}_{\V{y}}\|_{\mathfrak{k}^\nu}
\big\|\big(Y(\V{G}_{\V{x}}+\V{G}_{\V{y}})+\vp(\V{G}_{\V{x}}+\V{G}_{\V{y}})\big)\theta^\eh\psi\big\|
\\
&\le c''\|\V{G}_{\V{x}}-\V{G}_{\V{y}}\|_{\mathfrak{k}^\nu}
\sup_{\V{z}}\|\V{G}_{\V{z}}\|_{\mathfrak{k}^\nu}
\|\theta\psi\|,
\end{align*}
and the claim follows from the continuity of $\V{x}\mapsto\V{G}_{\V{x}}\in\mathfrak{k}^\nu$
required in Hyp.~\ref{hyp-G}.
\end{ex}

\begin{ex}\label{ex-comm-omega}
Let us explain how to read off the proof of Lem.~\ref{lem-bd-T} from Lem.~\ref{lem-comm-omega}.
First, we show that
\begin{align}\nonumber
\|\V{T}(s)\|&\le c_\alpha\big\|(\tfrac{v_{\alpha,s}}{1+\iota\tau_\alpha(s)})^\eh
(1+\tfrac{v_{\alpha,s}}{1+\iota\tau_\alpha(s)})^{|\alpha|-\eh}\V{G}_{\V{B}_s^{\V{q}}}\big\|
\\\label{bd-TTheta-spin4}
&\le c_\alpha\|(\vo+\vk)^\eh(1+\vo+\vk)^{|\alpha|-\eh}\V{G}_{\V{B}_s^{\V{q}}}\|,
\end{align}
where $\iota\tau_\alpha(s)$, $\vo$, $\vk$, and $v_{\alpha,s}$ are introduced in the paragraph 
preceding Lem.~\ref{lem-spin4}. In fact, recalling \eqref{for-vecT} and \eqref{comm-inv1} and 
using the notation of Lem.~\ref{lem-comm-omega} we may write
\begin{align}\nonumber
\V{T}(s)&=-2F_{\ve,E}(\Id\Gamma(v_{\alpha,\ve,s}))^{-\beta}\big(\Ad_{\ad(\V{G}_{\V{B}_s^{\V{q}}})}
F_{\ve,E}^{|\alpha|}(\Id\Gamma(v_{\alpha,\ve,s}))\big)
F_{\ve,E}(\Id\Gamma(v_{\alpha,\ve,s}))^{-\gamma}
\\\label{for-vecTTheta}
&\quad-2F_{\ve,E}(\Id\Gamma(v_{\alpha,\ve,s}))^{-\beta}\big(\Ad_{a(\V{G}_{\V{B}_s^{\V{q}}})}
F_{\ve,E}^{|\alpha|}(\Id\Gamma(v_{\alpha,\ve,s}))\big)
F_{\ve,E}(\Id\Gamma(v_{\alpha,\ve,s}))^{-\gamma}
\end{align}
with $E:=1+\iota\tau_\alpha(s)$ and either $(\beta,\gamma)=(|\alpha|,0)$ or 
$(\beta,\gamma)=(0,|\alpha|)$. 
Of course, the two terms on the right hand side correspond to the choices $N=1$, $M=0$,
and $N=0$, $M=1$, respectively.
Hence, Lem.~\ref{lem-comm-omega} implies the first inequality of \eqref{bd-TTheta-spin4},
and the second one follows from $v_{\alpha,s}\le(1+\iota\tau_\alpha(s))(\vo+\vk)$.
In the same way, taking \eqref{for-T2} and \eqref{comm-inv1} into account, we obtain
\begin{align*}
\|T_2(s)\|&\le c_{\alpha}'\|(\vo+\vk)^\eh(1+\vo+\vk)^{|\alpha|-\eh}
(q_{\V{B}_s^{\V{q}}},\V{F}_{\V{B}_s^{\V{q}}})\|.
\end{align*}
In view of \eqref{for-T1c}, \eqref{comm-inv1}, and \eqref{comm-inv2}, 
Lem.~\ref{lem-comm-omega} finally implies
\begin{align*}
T_1(s)
&=c_\alpha''\|(\vo+\vk)^\eh(1+\vo+\vk)^{|\alpha|-\eh}\V{G}_{\V{B}_s^{\V{q}}}\|^2
\\
&\quad+c_\alpha''
\|(\vo+\vk)^\eh(1+\vo+\vk)^{(|\alpha|-1)\vee0}\V{G}_{\V{B}_s^{\V{q}}}\|
\,\|(\vo+\vk)^\eh\V{G}_{\V{B}_s^{\V{q}}}\|.
\end{align*}
\end{ex}

\begin{lem}\label{lem-comm-exp}
Let $E\ge1$, $\ve\in[0,1]$, let $v:\cM\to\RR$ be measurable, and let $\omega$ 
be a measurable function on $\cM$ which is strictly positive almost everywhere.
Let $\sJ$ and $\sL$ be disjoint finite, possibly empty index sets, $N:=\#\sJ$, $M:=\#\sL$,
let $m,n\in\NN_0$ such that $m+n=M+N$, and let 
$\delta,\beta,\gamma\in[0,1]$ satisfy $\ve\le\delta=\beta+\gamma$. Pick $g_\ell\in\HP$ satisfying
$$
u_\delta^\eh e^{\delta v}g_\ell\in\HP,
\quad\text{where}\quad u:=(\delta v)\vee((\delta v)^2/\omega),
$$
for all $\ell\in\sJ\cup\sL$. Finally, set $v_\ve:=v/(1+\ve v)$ and define $F_{\ve,E}$ by 
\eqref{def-FveE}. Then the densely defined operator
\begin{align}\nonumber
\cT:=(1+\Id\Gamma(\omega))^{-\nf{n}{2}}&e^{-\beta F_{\ve,E}(\Id\Gamma(v_\ve))}
\Big\{\Big(\prod_{j\in\sJ}\Ad_{\ad(g_j)}\Big)\Big(\prod_{\ell\in\sL}\Ad_{a(g_\ell)}\Big)
e^{\delta F_{\ve,E}(\Id\Gamma(v_\ve))}\Big\}
\\\nonumber
&\qquad\times\,e^{-\gamma F_{\ve,E}(\Id\Gamma(v_\ve))}(1+\Id\Gamma(\omega))^{-\nf{m}{2}},
\end{align} 
extends uniquely to a bounded operator on $\sF$ with
\begin{align}\label{bd-comm-exp}
\|\cT\|&\le c_{m,n,M,N}\prod_{\ell\in\sJ\cup\sL}\|u_\delta^\eh e^{\delta v}g_\ell\|.
\end{align}
\end{lem}

\begin{proof}
Consider two partitions $\sA\cup \sB=\sJ$ and $\sC\cup \sD=\sL$ with $\#\sB=\#\sD$.
Since $m+n=M+N$ at least one of the conditions $n\ge\#\sA$ or $m\ge\#\sC$ is satisfied. 
Without loss of generality we may assume for the moment that $n\ge\#\sA$.
Setting  $\omega_\sA:=\sum_{a\in\sA}\omega(p_a)$,
$s_\ell:=v_\ve(p_\ell)$, $|s_\sA|:=\sum_{a\in\sA}v_\ve(p_a)$, etc., 
we then infer from Lem.~\ref{lem-clelia} 
and Lem.~\ref{lem-a(K)} that it suffices to find a bound on the norms
\begin{align*}
\cN_{\sA,\sC}^{\fc}&:=\big\|(1+\Id\Gamma(\omega)+\omega_\sA)^{-\frac{n-\#\sA}{2}}
\,e^{-\beta F_{\ve,E}(\Id\Gamma(v_\ve)+|s_\sA|)}
\\
&\qquad\quad\times 
(\triangle_{s_{\sA\cup\sL}}e^{\delta F_{\ve,E}})(\Id\Gamma(v_\ve))
\,e^{-\gamma F_{\ve,E}(\Id\Gamma(v_\ve)+|s_\sC|)}(1+\Id\Gamma(\omega))^{\frac{\#\fc-m}{2}}\big\|,
\end{align*}
where $\fc\subset\sC$ satisfies $\#\fc\ge m$. On account of $n-\#\sA\ge\#\fc-m$ we have
$$
\|(1+\Id\Gamma(\omega)+\omega_\sA)^{-\frac{n-\#\sA}{2}}
(1+\Id\Gamma(\omega))^{\frac{\#\fc-m}{2}}\|\le1,
$$
and taking also Lem.~\ref{lem-Leibniz-t} into account we see that
\begin{align*}
\cN_{\sA,\sC}^{\fc}&\le
c_{\sA,\sL}\,\delta^{|\delta_{\sA\cup\sL}|}e^{\delta|s_{\sA\cup\sL}|}
s_{\sA\cup\sL}^{\delta_{\sA\cup\sL}}
\sup_{t\ge0}e^{-\beta F_{\ve,E}(t+|s_\sA|)-\gamma F_{\ve,E}(t+|s_{\sC}|)+\delta F_{\ve,E}(t)},
\end{align*}
where we now choose $\delta_\ell=1$, if $\ell\in\sA\cup\fc\cup \sD$, and $\delta_\ell=1/2$,
if $\ell\in\sC\setminus\fc$.
The supremum in the previous inequality is $\le1$ by the assumption $a=b+c$ and the fact that
$F_{\ve,E}$ is monotonically increasing. 
We thus arrive at
\begin{align*}
\frac{(\cN_{\sA,\sC}^{\fc})^2}{(\prod_{a\in\sA\cup\fc}\omega(p_a))}
&\le c
\Big(\prod_{d\in\sD}u_\delta(p_d)^2e^{2\delta v(p_d)}\Big)\prod_{a\in \sA\cup \sC}
u_\delta(p_a)e^{2\delta v(p_a)}.
\end{align*}
In view of \eqref{clelia0}, \eqref{def-MAB}, and Lem.~\ref{lem-a(K)}, 
and since $\delta\le1$, it finally follows that
\begin{align*}
\|\cT\|&\le c\!\!\sum_{{{\sA\cup\sB=\sJ\atop\sC\cup\sD=\sL}\atop\#\sB=\#\sD}}
\sum_{\pi\in\mathrm{Bij}(\sB,\sD)}
\int\Big(\prod_{d\in\sD}u_\delta(p_d)^\eh e^{\delta v(p_d)/2}\ol{g_d(p_d)}\Big)
\\
&\quad\cdot
\Big(\prod_{b\in\sB}u_\delta(p_{\pi(b)})^\eh e^{\delta v(p_{\pi(b)})/2}g_b(p_{\pi(b)})\Big)
\\
&\quad\cdot\!\!\!\!\!
\sum_{{\fa\subset\sA\atop\#\fa\ge n\wedge\#\sA}}
\sum_{{\fc\subset\sC\atop\#\fc\ge m\wedge\#\sC}}\!\!
\bigg[\int\prod_{a\in\sA\cup\sC}\big\{u_\delta(p_a)e^{2\delta v(p_a)}|g_a(p_a)|^2
\Id\mu(p_{a})\big\}\bigg]^\eh\Id\mu^{\#\sD}(p_{\sD}),
\end{align*}
and we conclude that \eqref{bd-comm-exp} holds true.
\end{proof}

\begin{ex}\label{ex-comm-exp}
For positive $\delta$, the bound \eqref{tina2} is an easy consequence of 
Lem.~\ref{lem-comm-exp},
if we choose $E=1+\iota\tau_1(s)$, $v=v_{1,s}$,  so that 
$\Theta_{\ve,s}^{(1)}=F_{\ve,1+\iota\tau_1(s)}(\Id\Gamma(v_{1,\ve,s}))$,
where $\iota\tau_1(s)$, $v_{1,s}$, and $\Theta_{\ve,s}^{(1)}$ are defined in the
beginning of Sect.~\ref{sec-weights}. Similarly as in Ex.~\ref{ex-comm-omega},
the case $\delta<0$ can be reduced to the case of positive $\delta$ by means of 
\eqref{comm-inv1} and \eqref{comm-inv2}.
\end{ex}


\bigskip

\noindent
{\bf Acknowledgements.} It is a pleasure to thank Batu G\"{u}neysu and Jacob Schach M{\o}ller
for many inspiring and helpful discussions.

This work has been supported by the Villum Foundation.


\end{document}